\documentclass[accepted]{uai2026} 
                        
\usepackage[american]{babel}
\usepackage{natbib} 
\usepackage{mathtools} 
\usepackage[ruled,vlined]{algorithm2e}
\usepackage{siunitx} 
\usepackage{booktabs}
\usepackage{tikz}
\usepackage{circuitikz}
\usepackage{multicol,multirow}
\usepackage[table, dvipsnames]{xcolor}
\usepackage{float}
\usepackage{bigints}
\usepackage{tabularx}
\usepackage{chngcntr} 
\usepackage{xltabular}
\usepackage{array}
\usepackage{amsfonts,amsmath,amssymb,amsthm,mathptmx,mathrsfs, bm}
\usepackage{subcaption}
\usepackage[most]{tcolorbox}
\usepackage{pifont}  
\usepackage{hyperref}

\makeatletter
\let\oldaddcontentsline\addcontentsline
\renewcommand{\addcontentsline}[3]{} 
\makeatother

\newcommand{\vasst}{$\mathsf{VaSST}$}
\newcommand{\vasstmain}{\hyperref[alg:vasst]{$\mathsf{VaSST}$}}
\newcommand{\bms}{\texttt{BMS}}
\newcommand{\bsr}{\texttt{BSR}}
\newcommand{\qlattice}{\texttt{QLattice}}
\newcommand{\pysr}{\texttt{PySR}}
\newcommand{\dsr}{\texttt{DSR}}
\newcommand{\operon}{\texttt{operon}}
\newcommand{\gplearn}{\texttt{gplearn}}
\newcommand{\deap}{\texttt{DEAP}}
\newcommand{\sr}{\texttt{SR}}
\newcommand{\vart}{\texttt{VaRT}}

\definecolor{TreeGray}{RGB}{171,171,171}
\definecolor{TreeOuter}{RGB}{20,71,230}
\definecolor{TreeFill}{RGB}{219,234,254}
\definecolor{vasstgray}{gray}{0.94}
\hypersetup{
  colorlinks=true,
  linkcolor=Blue,
  citecolor=Blue,
  urlcolor=Blue,
  filecolor=Blue
}

\SetKwComment{Comment}{$\triangleright$\ }{}

\newcommand{\eqnref}[1]{%
  \hyperref[#1]{\textcolor{RawSienna}{\textbf{(\ref*{#1})}}}%
}
\allowdisplaybreaks

\newtheorem{theorem}{Theorem}[]
\newtheorem{remark}[theorem]{Remark}
\newtheorem{lemma}{Lemma}[section]

\tcbset{
  llmsrpromptbox/.style={
    enhanced,
    breakable,
    colback=Blue!2,
    colframe=Blue!55!black,
    coltitle=Blue!55!black,
    colbacktitle=Blue!2,
    fonttitle=\bfseries\scriptsize,
    fontupper=\scriptsize,
    boxrule=0.65pt,
    arc=1.6mm,
    left=1.2mm,
    right=1.2mm,
    top=1.2mm,
    bottom=1.2mm,
    titlerule=0.4pt,
    toptitle=0.8mm,
    bottomtitle=0.8mm,
    before skip=6pt,
    after skip=6pt,
    title after break={\texttt{LLM-SR} prompt design continued from previous page},
    before upper={
      \setlength{\parskip}{2pt}
      \setlength{\itemsep}{1pt}
      \setlength{\parsep}{0pt}
      \setlength{\topsep}{2pt}
    }
  }
}

\title{\texorpdfstring{\vasst}{VaSST}: Variational Inference for Symbolic Regression using Soft Symbolic Trees}

\author[1]{\href{mailto:sroy_123@tamu.edu}{Somjit Roy}\thanks{Corresponding author: <\href{mailto:sroy_123@tamu.edu}{sroy\_123@tamu.edu}>.}}
\author[1]{\href{mailto:pritam.dey@tamu.edu}{Pritam Dey}{}}
\author[1]{\href{mailto:bmallick@stat.tamu.edu}{Bani K. Mallick}{}}
\affil[1]{
    Department of Statistics\\
    Texas A\&M University\\
    College Station, Texas, USA
}

\begin{document}
\maketitle


\begin{abstract}
\emph{Symbolic regression} (\sr) has gained recent traction in AI-driven scientific discovery for learning closed-form physical laws. Yet existing methods are dominated by heuristic search or data-intensive approaches that often assume low-noise regimes and lack principled uncertainty quantification, while fully probabilistic \sr\ formulations remain scarce.
We introduce a scalable probabilistic framework for \sr, \textcolor{Blue}{\vasst}, based on \emph{variational inference}. \textcolor{Blue}{\vasst}\ uses \emph{soft symbolic trees}, a continuous relaxation of symbolic expression trees in which discrete operator and feature assignments are replaced by probability distributions over allowable components. This transforms combinatorial symbolic search through an astronomically large expression space into efficient gradient-based optimization while preserving a coherent probabilistic interpretation. The learned soft representations induce posterior distributions over symbolic structures, enabling uncertainty quantification across plausible symbolic forms through posterior-aware symbolic model selection. On simulated experiments and the \emph{Feynman Symbolic Regression Database}, \textcolor{Blue}{\vasst}\ achieves strong structural recovery and predictive accuracy compared to state-of-the-art competing \sr\ methods.
\end{abstract}



\section{Introduction}
\label{sec:intro}
\paragraph{Symbolic regression for scientific discovery.}
\emph{Scientific machine learning} has increasingly emerged as a powerful paradigm for integrating domain knowledge with data-driven modeling, enabling advances across materials science~\citep{SciML-materials-1}, climate and weather prediction~\citep{ML-weather-1}, biology~\citep{SciML-biology-1}, and physics~\citep{SciML-2}. A central objective in these domains is the discovery of explicit governing equations that encode mechanistic structure rather than merely predictive relationships. In this setting, \emph{symbolic regression} (\sr) plays a pivotal role~\citep{MakkeChawla2024SRReview}. Unlike classical prediction-centric regression methods~\citep{LASSO-Tibshirani,RasmussenWilliams2006,BART}, \sr\ operates directly over functional forms to recover concise, interpretable mathematical expressions from experimental data. By identifying explicit equations underlying complex scientific phenomena, \sr\ has enabled sparse discovery of nonlinear dynamical systems~\citep{SR-SciML-2}, accelerated materials design~\citep{SR-SciML-3}, and uncovered fundamental scientific laws~\citep{SR-SciML-1}.
\paragraph{Related work.}
Existing \sr\ methods can be broadly categorized into evolutionary, machine learning-based, and Bayesian approaches. Classical \sr\ is dominated by {genetic programming} and related heuristic search algorithms~\citep{deap,stephens2016gplearn,operon,pysr}, which explore expressions through stochastic evolution but often suffer from high computational complexity, sensitivity to initialization, and the generation of overly complex formulas~\citep{Korns2011}. Machine learning-driven methods cast \sr\ as a sequential decision-making problem, learning to generate grammar rules, tree traversals, or executable symbolic strings via neural architectures~\citep{AI-Feynman,Deep-SR,feyn-qlattice,Kamienny2022,xu2024reinforcement}. Although these methods enhance scalability and predictive performance, they remain inherently search-driven and computationally challenging due to the NP-hard nature of \sr~\citep{Virgolin2022}. Moreover, their effectiveness is often contingent on large training datasets and low-noise regimes, as demonstrated in \hyperref[sec:vasst-in-action]{Section~\ref{sec:vasst-in-action}}.
More recent machine learning-centric \sr\ modules use pretrained large language models (\texttt{LLM}s) to guide symbolic expression generation, refinement, or scientific reasoning~\citep{lasr, llmsr, drsr, phye2e}. While these approaches provide a flexible interface for incorporating domain knowledge, they can be highly sensitive to prompt design and model choice, and may produce syntactically invalid or unstable expressions; see~\hyperref[app:LLM-SR]{Appendix~\ref{app:LLM-SR}} for a representative example. The preceding review highlights the necessity for fully probabilistic \sr\ formulations.

Scientific expressions possess an inherent hierarchical and compositional structure that aligns naturally with tree representations. It is important, however, to distinguish such representations from \emph{decision trees} which recursively partition the covariate space using split rules and assign local predictions to terminal regions~\citep{CART,BCART,BCART-Mallick, BART}. In contrast, tree representations of scientific expressions encode recursive compositions of mathematical operators (e.g., $+$, $\times$, $\exp$, $\sin$; assigned to internal nodes) with primitive data features (assigned to leaves)~\citep{Bartlett,hierbosss}; see \hyperref[fig:symbolic-tree-representation]{Figure~\ref{fig:symbolic-tree-representation}}. Among existing tree-based \sr\ formulations, Bayesian Machine Scientist (\bms)~\citep{guimera2020bayesian} represents an important step toward model-based equation discovery. However, it employs an ad hoc structural prior built from corpus-parsing based on a priori knowledge, and performs inference using Metropolis-Hastings (\texttt{MH}) proposals over discrete symbolic tree structures. 
Given the highly multimodal and combinatorial posterior landscape, such local \texttt{MH} updates can exhibit poor mixing in complex discrete spaces~\citep{bhamidi2008mixing, latuszynski2025mcmc}, leading to slow convergence and inefficient exploration of the symbolic expression space. Similarly, Bayesian Symbolic Regression (\bsr)~\citep{BSR} adopts a tree-based partial Bayesian formulation, where it uses plug-in ordinary least squares estimates for the outer regression coefficients. This incomplete parameter uncertainty propagation, combined with local stochastic \texttt{MH} structural updates, further impedes effective traversal of the symbolic expression space and frequently produces overly complicated output symbolic structures, as evidenced in \hyperref[sec:vasst-in-action]{Section~\ref{sec:vasst-in-action}}. Finally, the reliance of these existing Bayesian \sr\ approaches on Markov chain Monte Carlo-based discrete structural exploration can limit their scalability beyond small- to moderate-sized data settings.
\paragraph{Our contributions.}
Motivated by these limitations, we propose \emph{\underline{Va}riational Inference for Symbolic Regression using \underline{S}oft \underline{S}ymbolic \underline{T}rees} (\vasstmain), a \emph{variational inference} framework for \sr\ that combines principled Bayesian modeling with improved computational scalability over existing probabilistic \sr\ methods. Variational inference recasts Bayesian inference as an optimization problem~\citep{Blei-VI-review-JASA}, offering a scalable alternative to traditional Monte Carlo-based approaches for modern data-intensive settings~\citep{jordan1999introduction,Wainwright-Jordan-Graphical-VI,Graves-NN-VI}. However, a na\"{i}ve application of variational inference over the discrete structural space of symbolic expressions results in a combinatorial optimization problem that negates these scalability benefits~\citep{williams1992simple,koza1994genetic}.

To overcome this challenge, \vasstmain\ uses a novel latent representation of symbolic expressions using \emph{soft symbolic trees}, in which discrete operator and feature assignments are replaced by probability distributions over all allowable operators and features. This relaxation transforms the discrete structural search into a continuous optimization problem, enabling efficient gradient-based exploration of the symbolic expression space through black-box variational inference~\citep{ranganath2014bbvi,advi,dadvi}. Although this relaxation is similar in spirit to the soft tree formulation used in \emph{Variational Regression Trees} (\vart)~\citep{VART}, the relaxed objects and inferential goals here are fundamentally different. While \vart\ performs soft routing over data-partitioning trees, thus relaxing discrete decision splits, \vasstmain\ relaxes discrete assignments of operators and features defining a tree-structured scientific expression. Thus, its soft tree representation serves as a differentiable surrogate for symbolic expression evaluation, rather than as a probabilistic partitioning rule. For a detailed comparison of \vart\ and \vasstmain, refer to~\hyperref[app:vart-vs-vasst]{Appendix~\ref{app:vart-vs-vasst}}.

The learned soft symbolic trees naturally define probability distributions over symbolic expressions. Thus, post-optimization, structurally interpretable expressions are sampled from these learned representations and ranked using a posterior-aware evidence score. The top candidates are then retained yielding an Occam's window-based~\citep{madigan1994model} posterior summary of high-support symbolic explanations.
Finally, a primary goal in scientific equation discovery is to favor structural parsimony in accordance to the {Occam's razor} principle~\citep{Occams-Razor-1}. \vasstmain\ achieves this by employing a {depth-dependent regularizing tree prior} that penalizes overly complex sampled symbolic expressions.

Through extensive experiments, we demonstrate the superior performance of \vasstmain\ relative to a range of state-of-the-art \sr\ methods from \texttt{SRBench}~\citep{la2021contemporary,imai2025call,SRBench}, where it effectively balances structural discovery and predictive accuracy while maintaining computational stability. The \texttt{Python} implementation of \vasstmain\ is available at~\cite{VaSST-Github}.


\section{Scientific Expressions as Symbolic Trees}
\label{sec:symbolic-tree-representation}
Scientific expressions can be recursively constructed by combining primary features (e.g., $x_1, x_2$) and elementary mathematical operators (e.g., $\exp, \sin, +, \times$).
We denote by $\bm x = (x_1, \ldots, x_p)^\top$, the vector of $p$ primary features and by $\mathbf{O} = \mathbf{O}_u \cup \mathbf{O}_b$, the set of allowed mathematical operators, where $\mathbf{O}_u$ and $\mathbf{O}_b$ contain the unary and binary mathematical operators, respectively. Typical choices in scientific modeling include $\mathbf{O}_u = \{\sin, \cos, \exp, \log, ^2, ^3\}$ and $\mathbf{O}_b = \{+, \times, -, /\}$~\citep{AI-Feynman}. 

Let $\mathcal{S}$ denote the set of admissible feature-operator compositions.
Any symbolic expression $f \in \mathcal{S}$ is defined recursively as one of, (a) \emph{primitive expression}: $f(\bm x) = x_m$, where $m \in \{1,\ldots, p\}$; (b) \emph{binary composition}: $f(\bm x) = b(f_1(\bm x),f_2(\bm x))$, where $f_1, f_2\in \mathcal{S}$ and $b\in \mathbf{O}_b$, e.g., $x_1 + x_2$; and (c) \emph{unary composition}: $f(\bm x) = u( \tilde{f}(\bm x))$, where $\tilde{f} \in \mathcal{S}$ and $u\in \mathbf{O}_u$, e.g., $\exp(x_1)$.
This set of characterizations allows representations of $f\in \mathcal{S}$ by a symbolic tree structure $\mathsf{T}(f)$, where {internal} ({nonterminal}) nodes either represent unary (with $1$ child node) or binary (with $2$ children nodes) operators, while {leaves} ({terminal} nodes) correspond to primary features; see \hyperref[fig:symbolic-tree-representation]{Figure~\ref{fig:symbolic-tree-representation}}. Such tree representations of symbolic expressions are not necessarily unique, e.g., $x_1(x_2+x_3) = x_1x_2 + x_1x_3$. With this setup, now in the sequel we formally present the \vasstmain\ modeling framework.

\begin{figure}[!htp]
\centering
\resizebox{\linewidth}{!}{%
\begin{circuitikz}
\tikzset{
  every node/.style={font=\Huge},
  treeNode/.style={
    circle,
    draw=TreeOuter,
    line width=4pt,
    minimum size=2.5cm,
    inner sep=0pt
  },
  opNode/.style={
    treeNode,
    fill=TreeFill
  },
  treeEdge/.style={
    line width=2pt,
    short
  }
}

\node[opNode] at (5,-10.25) {$+$};
\node[treeNode] at (2.5,-14.75) {$x_1$};
\node[treeNode] at (7.5,-14.75) {$x_2$};

\draw[treeEdge] (4.25,-11.25) -- (2.5,-13.5);
\draw[treeEdge] (5.75,-11.25) -- (7.5,-13.5);

\node[treeNode] at (-6.25,-14.75) {$x_1$};

\node[treeNode] at (16.25,-15) {$x_1$};
\node[opNode] at (16.25,-10.25) {$\exp$};

\draw[treeEdge] (16.25,-11.5) -- (16.25,-13.75);

\node[font=\Huge] at (-6.75,-17.25) {(a) $f(\bm x) = x_1$};
\node[font=\Huge] at (5,-17.25) {(b) $f(\bm x) = x_1 + x_2$};
\node[font=\Huge] at (16,-17.25) {(c) $f(\bm x) = \exp(x_1)$};

\end{circuitikz}
}%
\caption{Symbolic tree representation.}
\label{fig:symbolic-tree-representation}
\end{figure}


\section{The \texorpdfstring{\vasst}{VaSST} Model}
\label{sec:vasst-model}
We outline the \vasstmain\ model which comprises two major components: (i) a {symbolic ensemble} consisting of an affine (linear) combination of $K$ symbolic trees and (ii) a hierarchical prior specification over the model regression coefficients, the model noise variance, and the symbolic tree structures.
\subsection{The Symbolic Ensemble Component}
\label{subsec:symbolic-ensemble-compoenent}
Let $\mathbb{D}_n = \{(\bm{x}_i, y_i)\}_{i=1}^{n}$ be the collection of observed data units. \vasstmain\ models and structurally learns the hidden symbolic relationship between the responses $y_i\in \mathbb{R}$ and primary feature vector $\bm{x}_i = (x_{i,1},\ldots,x_{i,p})^\top\in \mathbb{R}^{p}$ using $K$ symbolic expressions $\{f_j\}_{j=1}^{K} \subset \mathcal{S}$ evaluated at $\bm{x}_i$, as:
\begin{align}
\label{eq:vasst-model}
y_i = \beta_0 + \sum_{j=1}^{K}\beta_j\;f_j(\bm{x}_i) + \epsilon_i, \quad i=1, \ldots, n,
\end{align}
where $\bm\beta = (\beta_0, \ldots, \beta_K)^{\top} \in \mathbb{R}^{K+1}$ is the model regression coefficient vector and $\epsilon_i\in \mathbb{R}$ is the model noise with $\epsilon_i \sim \mathrm{N}_1(0, \sigma^2)$ independently. From~\hyperref[eq:vasst-model]{\eqnref{eq:vasst-model}}, we obtain the vector representation of the symbolic ensemble component:
\begin{equation}
\label{eq:vasst-model-vector}
\mathbf{y} = \mathbf{T}\bm\beta + \bm\epsilon, \quad \bm\epsilon\sim \mathrm{N}_n(\bm 0_n, \sigma^{2}\mathbf{I_n}),
\end{equation}
where $\mathbf{T} = (f_j(\bm{x}_i))_{1\leq i \leq n,\; 0\leq j \leq K}\in \mathbb{R}^{n\times \overline{K+1}}$ is the {expression design matrix}, $f_0(\bm{x}_i) = 1$ for all $i=1,\ldots, n$, $\mathbf{y} = (y_1, \ldots, y_n)^{\top}\in \mathbb{R}^{n}$ is the response vector, and $\bm\epsilon = (\epsilon_1, \ldots, \epsilon_n)^{\top}\in \mathbb{R}^{n}$ is the model noise vector.

The model regression parameters $(\bm\beta, \sigma^2)$ in~\eqnref{eq:vasst-model-vector} jointly are endowed upon with the conjugate Normal Inverse-Gamma (NIG) prior viz., $\bm\beta \mid \sigma^2 \sim \mathrm{N}_{K+1}(\bm\mu_0, \sigma^{2}\bm{\Sigma}_0)$ and $\sigma^{2}\sim \mathrm{IG}(a_0, b_0)$, where $\bm\mu_0\in \mathbb{R}^{K+1}$, $\bm\Sigma_0$ (positive definite matrix of order $K+1$), and $a_0,\; b_0 > 0$ are the hyperparameters of the Normal and Inverse-Gamma prior distributions, respectively. We complete the \vasstmain\ model specification by describing the symbolic tree prior in \hyperref[subsec:symbolic-tree-prior]{Section~\ref{subsec:symbolic-tree-prior}}.

\subsection{The \texorpdfstring{\vasst}{VaSST} Symbolic Tree Prior}
\label{subsec:symbolic-tree-prior}
The \vasstmain\ symbolic tree prior for each $\mathsf{T}(f_j)$ is constructed in two stages: (a) a probabilistic specification over a {maximal binary tree skeleton} and (b) a {deterministic pruning operation} that maps the skeleton tree to a valid symbolic tree, as depicted in \hyperref[fig:vasst-symbolic-tree-prior]{Figure~\ref{fig:vasst-symbolic-tree-prior}}.
\paragraph{Full binary tree skeleton.}
To decouple structural decisions from operator and feature assignments, we embed each symbolic tree $\mathsf{T}(f_j)$ into a full binary tree skeleton, denoted by $\mathsf{S}_j$, of fixed maximum depth $D\in \mathbb{N}$. Let $\mathsf{Z}_D = \{0, 1, \ldots, N-1\}$ index the skeleton tree nodes in \emph{heap order}, where $0$ is the root and the left and right children of a skeleton node $\zeta \in \mathsf{Z}_D$ are given by $L(\zeta) = 2\zeta+1$ and $R(\zeta) = 2\zeta +2$, respectively. The total number of nodes is $N = 2^{D+1} - 1$. Also, let $d_{\zeta} \in \{0, \ldots, D\}$ denote the depth of node $\zeta$. For each skeleton node $\zeta\in \mathsf{Z}_D$, we introduce the following, (a) {expansion indicator}: $e_{j\zeta}\in \{0, 1\}$, where $e_{j\zeta} = 1$ or $e_{j\zeta} = 0$ indicates that $\zeta$ is an internal node or a leaf of $\mathsf{T}(f_j)$, respectively; (b) {operator assignment}: $o_{j\zeta}\in \mathbf{O}$, specifies the operator assigned to $\zeta$ if $e_{j\zeta}=1$; and (c) {feature assignment}: $h_{j\zeta}$, specifies the primary feature assigned to $\zeta$ if $e_{j\zeta} = 0$. Thus, the skeleton tree $\mathsf S_j$ can be encoded by the structural variables  $\{e_{j\zeta}, o_{j\zeta}, h_{j\zeta}:\zeta\in \mathsf{Z}_D\}$.
\paragraph{Valid symbolic tree via deterministic pruning.}
The skeleton $\mathsf{S}_j$ is an {ambient} representation. We obtain the corresponding valid symbolic tree $\mathsf{T}(f_j)$ by applying a deterministic pruning operator $\mathfrak{p}$ which proceeds as follows, (a) {terminal pruning}: if $e_{j\zeta} = 0$, then all descendants of $\zeta$ in the skeleton are removed, making it a leaf and (b) {unary operator pruning}: if $e_{j\zeta} = 1$ and $o_{j\zeta} \in \mathbf{O}_u$, then the right subtree rooted at $R(\zeta)$ is removed. Thus, the resulting symbolic tree is $\mathsf{T}(f_j) = \mathfrak{p}(\mathsf{S}_j)$ and the corresponding pruned node set is $\mathsf Z(f_j)$.
\paragraph{Prior over full binary tree skeleton.}
The prior specifications over $\{e_{j\zeta}, o_{j\zeta}, h_{j\zeta}:\zeta\in \mathsf{Z}_D\}$ and hence $\mathsf{S}_j$ is:
\begin{equation}
\label{eq:skeleton-prior}
\begin{split}
    &\Pi(e_{j\zeta}) \equiv \mathrm{Ber}(e_{j\zeta}; p_\zeta = \alpha(1 + d_{\zeta})^{-\delta}),\\
    &\Pi(o_{j\zeta}\mid \mathbf{w}_{\mathrm{op}})
    \equiv \mathrm{Cat}(o_{j\zeta};\mathbf{w}_{\mathrm{op}}),\;\Pi(h_{j\zeta}\mid \mathbf{w}_{\mathrm{ft}})
    \equiv \mathrm{Cat}(h_{j\zeta};\mathbf{w}_{\mathrm{ft}}),\\
    &\Pi(\mathsf S_j \mid \mathbf{w}_{\mathrm{op}}, \mathbf{w}_{\mathrm{ft}})=
    \prod_{\zeta \in \mathsf Z_D}
        \Pi(e_{j\zeta})
        \Pi(o_{j\zeta}\mid \mathbf{w}_{\mathrm{op}})
        \Pi(h_{j\zeta}\mid \mathbf{w}_{\mathrm{ft}}),
\end{split}
\end{equation}
for $j=1,\ldots,K$. In \eqnref{eq:skeleton-prior}, $\alpha, \delta > 0$ and $\mathbf{w}_{\mathrm{op}} \in \Delta^{|\mathbf{O}|}, \mathbf{w}_{\mathrm{ft}}\in \Delta^{p}$ are the operator and feature weight vectors, where $\Delta^{m} = \{\mathbf{z}\in [0,\infty)^{m}\mid \mathbf{1}_m^{\top}\mathbf{z}=1\}$ is the $m$-dimensional simplex. Also, $\mathrm{Ber}(\cdot)$ and $\mathrm{Cat}(\cdot)$ denote Bernoulli and Categorical distributions, respectively. Note that, the prior over $\mathsf{S}_j$ in~\eqnref{eq:skeleton-prior} induces a prior over $\mathsf{T}(f_j)$, as:
\begin{align}
\label{eq:OG-symbolic-tree-prior}
\begin{split}
&\Pi(\mathsf{T}(f_j) \mid \mathbf{w}_{\mathrm{op}}, \mathbf{w}_{\mathrm{ft}}) \\
&\qquad = \prod_{\zeta \in \mathsf Z(f_j)}\Pi(e_{j\zeta}) \Pi(o_{j\zeta}\mid \mathbf{w}_{\mathrm{op}}) \Pi(h_{j\zeta}\mid \mathbf{w}_{\mathrm{ft}}).
\end{split}
\end{align}
Further, the priors over $\mathbf{w}_{\mathrm{op}}$ and $\mathbf{w}_{\mathrm{ft}}$ are:
\begin{align}
\label{eq:weight-priors}
\Pi(\mathbf{w}_{\mathrm{op}}) \equiv \mathrm{Dir}(\mathbf{w}_{\mathrm{op}}; \bm{\eta}_{\mathrm{op}}),
\quad 
\Pi(\mathbf{w}_{\mathrm{ft}}) \equiv \mathrm{Dir}(\mathbf{w}_{\mathrm{ft}};\bm\eta_{\mathrm{ft}}),
\end{align}
where $\bm\eta_{\mathrm{op}} \in {(0,\infty)}^{|\mathbf{O}|}$ and $\bm\eta_{\mathrm{ft}} \in {(0,\infty)}^{p}$ are concentration hyperparameters of Dirichlet distributions denoted by $\mathrm{Dir}(\cdot)$. Therefore,~\eqnref{eq:skeleton-prior} and~\eqnref{eq:weight-priors} define the complete prior specification over $\Theta = (\{\mathsf{S}_j\}_{j=1}^{K}, \mathbf{w}_{\mathrm{op}}, \mathbf{w}_{\mathrm{ft}})$. 

We conclude with \hyperref[remark:depth-dependent-split-probability]{Remark~\ref{remark:depth-dependent-split-probability}} emphasizing the role of the {depth-dependent split probability} $p_\zeta = \alpha(1 + d_{\zeta})^{-\delta}$ as a mechanism for controlling symbolic expression complexity.
\begin{figure}[!ht]
\centering
\resizebox{0.4\textwidth}{!}{%
\begin{circuitikz}
\tikzset{
  every node/.style={font=\Huge},
  liveNode/.style={
    circle,
    draw=TreeOuter,
    line width=4pt,
    minimum size=2.5cm,
    inner sep=0pt
  },
  liveOpNode/.style={
    liveNode,
    fill=TreeFill
  },
  prunedNode/.style={
    circle,
    draw=TreeGray,
    dashed,
    line width=1.5pt,
    minimum size=2.5cm,
    inner sep=0pt
  },
  liveEdge/.style={
    line width=1.5pt,
    short
  },
  prunedEdge/.style={
    draw=TreeGray,
    line width=1.5pt,
    dashed
  },
  levelGuide/.style={
    line width=1.1pt,
    dashed
  },
  thinArrow/.style={
    line width=0.6pt,
    ->,
    >=Stealth
  }
}

\node[liveOpNode] at (1.25,15.75) {$\mathrm{root}$};
\node[liveOpNode] at (8.75,10.75) {};
\node[liveNode] at (-6.25,10.75) {};
\node[liveNode] at (5,5.75) {};
\node[liveNode] at (12.5,5.75) {};

\node[prunedNode] at (-10,5.75) {};
\node[prunedNode] at (-2.5,5.75) {};

\draw[liveEdge] (-6.25,12) -- (0.5,14.75);
\draw[liveEdge] (8.75,12) -- (2,14.75);
\draw[liveEdge] (5,7) -- (8,9.75);
\draw[liveEdge] (12.5,7) -- (9.5,9.75);

\draw[prunedEdge] (-10,7) -- (-7,9.75);
\draw[prunedEdge] (-2.5,7) -- (-5.5,9.75);

\draw[levelGuide] (-11.25,14.5) -- (-1.25,14.5);
\draw[levelGuide] (3.75,14.5) -- (13.75,14.5);
\draw[levelGuide] (-11.25,9.5) -- (-8.75,9.5);
\draw[levelGuide] (11.25,9.5) -- (13.75,9.5);
\draw[levelGuide] (13.75,4.5) -- (16.25,4.5);
\draw[levelGuide] (-13.75,4.5) -- (-11.25,4.5);

\node at (15,5) {$d_{\zeta} = 2$};
\node at (12.5,10) {$d_{\zeta} = 1$};
\node at (8.75,15) {$d_{\zeta} = 0$};

\node at (-3.75,10) {$\zeta = 1$};
\node at (6,10) {$\zeta = 2$};
\node at (3.9,15.75) {$\zeta = 0$};
\node[color=TreeGray] at (-10,3.75) {$\zeta = 3$};
\node[color=TreeGray] at (-2.5,3.75) {$\zeta = 4$};
\node at (5,3.75) {$\zeta = 5$};
\node at (12.5,3.75) {$\zeta = 6$};

\draw[thinArrow] (-6.25,9.5) -- (-6.25,2.5);
\draw[thinArrow] (8.75,9.5) -- (8.75,2.5);

\node at (-6,1.25)
{$e_{j1}=0,\;o_{j1}\in \mathbf{O},\;\text{feature }h_{j1}$};

\node at (9.25,1.25)
{$e_{j2}=1,\;o_{j2}\in \mathbf{O}_b,\;\text{feature }h_{j2}$};

\node at (2.95,0)
{$[\Pi(e_{j\zeta}),\;p_{\zeta};\;\Pi(o_{j\zeta}\mid \mathbf{w}_{\mathrm{op}});\;\Pi(h_{j\zeta}\mid \mathbf{w}_{\mathrm{ft}})]$};

\node at (1.5,17.75)
{$(D,N)=(2,7),\;\text{prior over}\;\mathsf{S}_j\;\text{and induced prior over}\;\mathsf{T}(f_j)$};

\node at (2.0,-1.25)
{$\text{Node set of}\;\mathsf{S}_j:\{0,\ldots,6\};\;\text{Node set of}\;\mathsf{T}(f_j):\{0,1,2,5,6\}$};

\end{circuitikz}
}%
\caption{A live representation of the prior over $\mathsf{S}_j$ and the induced prior over $\mathsf{T}(f_j)$. In $\mathsf{S}_j$, the children nodes $\zeta = 3, 4$ of $\zeta=1$ are pruned using $\mathfrak{p}$ as $e_{j1} = 0$ to obtain the valid symbolic tree $\mathsf{T}(f_j)$.}
\label{fig:vasst-symbolic-tree-prior}
\end{figure}
\begin{remark}[Depth-dependent split probability]
\label{remark:depth-dependent-split-probability}
Guided by the Occam's razor principle~\citep{Occams-Razor-1}, we aim to learn interpretable and parsimonious symbolic expressions which adequately captures the underlying scientific mechanism. This is achieved by the depth-dependent split probability $p_\zeta$ in~\eqnref{eq:skeleton-prior} which imparts a regularizing effect on the individual tree depths~\citep{BCART}.
\end{remark}


\section{Variational Inference for \texorpdfstring{\vasst}{VaSST}}
\label{sec:VI-vasst}
Combining the data likelihood $p(\mathbf{y}\mid \mathbf{T}, \bm\beta, \sigma^2)$, the joint prior over the model regression parameters in \hyperref[subsec:symbolic-ensemble-compoenent]{Section~\ref{subsec:symbolic-ensemble-compoenent}}, and the \vasstmain\ symbolic tree prior in \hyperref[subsec:symbolic-tree-prior]{Section~\ref{subsec:symbolic-tree-prior}}, the joint posterior distribution induced over all unknowns $(\Theta,\bm\beta, \sigma^2)$ is:
\begin{align}
\label{eq:joint-vasst-posterior}
\begin{split}
    &\Pi(\Theta,\bm\beta,\sigma^2\mid \mathbb{D}_n) \propto p(\mathbf{y}\mid \mathbf{T}, \bm\beta, \sigma^2)\Pi(\bm\beta, \sigma^2)\Pi(\mathbf{w}_{\mathrm{op}})\\
    &\quad\Pi(\mathbf{w}_{\mathrm{ft}})\prod_{j=1}^{K}\prod_{\zeta\in \mathsf{Z}_D}\Pi(e_{j\zeta})\Pi(o_{j\zeta}\mid \mathbf{w}_{\mathrm{op}})\Pi(h_{j\zeta}\mid \mathbf{w}_{\mathrm{ft}}).
\end{split}
\end{align}
\paragraph{Marginalization over model regression parameters.}
The conjugate NIG prior over $(\bm \beta, \sigma^2)$ enables marginalization in~\eqnref{eq:joint-vasst-posterior} and results into $\Pi(\Theta\mid \mathbb{D}_n)\propto p(\mathbf{y}\mid \mathbf{T})\Pi(\Theta)$, where $p(\mathbf{y}\mid \mathbf{T})$ up to constants is:
\begin{align}
\label{eq:p(y|T)}
\log p(\mathbf{y}\mid \mathbf{T}) = \tfrac{1}{2}\log|\bm{\Sigma}_n| - a_n\log b_n + \log \Gamma(a_n),
\end{align}
where:
\begin{align}
\label{eq:posterior-model-parameters}
\begin{gathered}
\bm\Sigma_n^{-1} = \bm\Sigma_0^{-1} + \mathbf{T}^{\top}\mathbf{T}, \quad 
\bm\mu_n = \bm\Sigma_n(\bm\Sigma_0^{-1}\bm\mu_0 + \mathbf{T}^{\top}\mathbf{y}), \\
a_n = a_0+\tfrac{n}{2}, \\ b_n = b_0 + \tfrac{1}{2}(\mathbf{y}^{\top}\mathbf{y} + \bm\mu_0^{\top}\bm\Sigma_0^{-1}\bm{\mu}_0 - \bm{\mu}_n^{\top}\bm\Sigma_n^{-1}\bm\mu_n),
\end{gathered}
\end{align}
are the hyperparameters of the posterior NIG distribution. Also, $\Gamma(\cdot)$ and $|\bm\Sigma_n|$ are the Gamma function and the determinant of the matrix $\bm\Sigma_n$, respectively.
See \hyperref[app:parameter-marginalization]{Appendix~\ref{app:parameter-marginalization}} for complete derivations of~\eqnref{eq:p(y|T)} and~\eqnref{eq:posterior-model-parameters}. Therefore, it remains to draw inference from $\Pi(\Theta\mid \mathbb{D}_n)$.
\paragraph{Variational family.}
To conduct scalable and efficient inference, we adopt a variational inference routine for approximating the posterior $\Pi(\Theta\mid \mathbb{D}_n)$ with a tractable {mean-field approximation}~\citep{jordan1999introduction, Blei-VI-review-JASA} $q \in \mathbb{Q}$ detailed in~\eqnref{eq:mean-field-variational-family} below:
\begin{align}
\label{eq:mean-field-variational-family}
\begin{gathered}
q(\Theta) = q(\mathbf{w}_{\mathrm{op}})q(\mathbf{w}_{\mathrm{ft}})\prod_{j=1}^{K}\prod_{\zeta \in \mathsf{Z}_D}q(e_{j\zeta})q(o_{j\zeta})q(h_{j\zeta}),\\
q(e_{j\zeta}) \equiv \mathrm{Ber}(e_{j\zeta}; \widetilde{p}_{j\zeta}),\\
q(o_{j\zeta})\equiv \mathrm{Cat}(o_{j\zeta}; \widetilde{\pi}_{j\zeta}^{\mathrm{op}}),\; q(h_{j\zeta})\equiv \mathrm{Cat}(h_{j\zeta}; \widetilde{\pi}_{j\zeta}^{\mathrm{ft}}),\\
q(\mathbf{w}_{\mathrm{op}})\equiv \mathrm{Dir}(\mathbf{w}_{\mathrm{op}};\widetilde{\bm\eta}_{\mathrm{op}}),\; q(\mathbf{w}_{\mathrm{ft}})\equiv \mathrm{Dir}(\mathbf{w}_{\mathrm{ft}};\widetilde{\bm\eta}_{\mathrm{ft}}),
\end{gathered}
\end{align}
for $j=1,\ldots,K$ and $\zeta\in \mathsf{Z}_D$, where $\widetilde{p}_{j\zeta} = \sigma(\ell_{j\zeta}$), $\widetilde{\pi}^{\mathrm{op}}_{j\zeta}=\texttt{smax}(\mathbf{a}_{j\zeta}^{\mathrm{op}})$, and $\widetilde{\pi}^{\mathrm{ft}}_{j\zeta}=\texttt{smax}(\mathbf{a}_{j\zeta}^{\mathrm{ft}})$; with $\sigma(\cdot)$ and $\texttt{smax}(\cdot)$ denoting the sigmoid and softmax functions, respectively, and 
$\ell_{j\zeta}\in \mathbb{R}$, $\mathbf{a}_{j\zeta}^{\mathrm{op}}\in \mathbb{R}^{|\mathbf{O}|}$, and $\mathbf{a}_{j\zeta}^{\mathrm{ft}}\in \mathbb{R}^{p}$. Hence, $q$ is parameterized by:
\begin{align}
\label{eq:variational-parameters}
\phi = \left(
\widetilde{\bm{\eta}}_{\mathrm{op}}, \widetilde{\bm{\eta}}_{\mathrm{ft}},
\{\ell_{j\zeta}, \mathbf{a}_{j\zeta}^{\mathrm{op}},
\mathbf{a}_{j\zeta}^{\mathrm{ft}}\}_{j = 1,\ldots,K; \zeta \in \mathsf Z_D}
\right).
\end{align}
\paragraph{Evidence lower bound.}
In theory, one can obtain the optimal value of the variational parameters $\phi$ in \eqnref{eq:variational-parameters} by maximizing the {evidence lower bound} (\texttt{ELBO}):
\begin{align}
\label{eq:elbo-1}
\begin{split}
\mathcal{E}(\phi) &= \mathbb{E}_{q}[\log p(\mathbf{y}\mid \mathbf{T})] - \mathrm{KL}[q(\Theta)\parallel \Pi(\Theta)],
\end{split}
\end{align}
where $\mathrm{KL}[q(\Theta)\parallel\Pi(\Theta)] = \mathbb{E}_ q[\log q(\Theta) - \log \Pi(\Theta)]$ is the Kullback-Leibler (KL) divergence and $\log p(\mathbf{y}\mid \mathbf{T})$ is as in~\eqnref{eq:p(y|T)}. By~\eqnref{eq:mean-field-variational-family}, $\mathrm{KL}[q(\Theta)\parallel\Pi(\Theta)]$ splits as:
\begin{align}
\label{eq:elbo-2}
\begin{split}
    &\mathrm{KL}[q(\Theta)\parallel \Pi(\Theta)]
    \\&
    = \mathrm{KL}[q(\mathbf{w}_{\mathrm{op}})\parallel \Pi(\mathbf{w}_{\mathrm{op}})]
    + \mathrm{KL}[q(\mathbf{w}_{\mathrm{ft}})\parallel \Pi(\mathbf{w}_{\mathrm{ft}})]
    \\&
    \qquad + \sum_{j=1}^{K}\sum_{\zeta\in \mathsf{Z}_D}\bigg[\mathrm{KL}[q(e_{j\zeta})\parallel \Pi(e_{j\zeta})]
    \\&\quad\qquad
    +\mathbb{E}_{q(\mathbf{w}_{\mathrm{ft}})}\mathrm{KL}[q(h_{j\zeta})\parallel \Pi(h_{j\zeta}|\mathbf{w}_{\mathrm{ft}})]
    \\&\quad\qquad 
    + \mathbb{E}_{q(\mathbf{w}_{\mathrm{op}})}\mathrm{KL}[q(o_{j\zeta})\parallel\Pi(o_{j\zeta}\mid \mathbf{w}_{\mathrm{op}})]\bigg].
\end{split}
\end{align}
The individual terms in~\eqnref{eq:elbo-2} are given by:
\begin{align*}
\begin{split}
&\mathrm{KL}[q(e_{j\zeta})\parallel \Pi(e_{j\zeta})]
\\&\quad
= \widetilde{p}_{j\zeta}\log\tfrac{\widetilde{p}_{j\zeta}}{p_{\zeta}} + (1 - \widetilde{p}_{j\zeta})\log \tfrac{1-\widetilde{p}_{j\zeta}}{1-p_{\zeta}},\\
&\mathrm{KL}[q(\mathbf{w}_{\mathrm{op}})\parallel \Pi(\mathbf{w}_{\mathrm{op}})]= \log\tfrac{\mathcal{B}(\bm\eta_{\mathrm{op}})}{\mathcal{B}(\widetilde{\bm\eta}_{\mathrm{op}})}\\
&\quad +\sum_{k=1}^{|\mathbf O|}(\widetilde{\eta}_{\mathrm{op}, k} - \eta_{\mathrm{op}, k})(\Psi(\widetilde{\eta}_{\mathrm{op}, k}) - \Psi(\mathbf{1}_{|\mathbf O|}^{\top}\widetilde{\bm\eta}_{\mathrm{op}})),\\
&\mathbb{E}_{q(\mathbf{w}_{\mathrm{op}})}\mathrm{KL}[q(o_{j\zeta})\parallel \Pi(o_{j\zeta}\mid \mathbf{w}_{\mathrm{op}})]\\
&\quad = \sum_{k=1}^{|\mathbf O|}\widetilde{\pi}^{\mathrm{op}}_{j\zeta, k}\left[\log(\widetilde{\pi}^{\mathrm{op}}_{j\zeta, k}) - \Psi(\widetilde{\eta}_{\mathrm{op}, k}) + \Psi(\mathbf{1}^{\top}_{|\mathbf O|}\widetilde{\bm\eta}_{\mathrm{op}})\right],
\end{split}
\end{align*}
where $\mathcal{B}(\cdot)$ and $\Psi(\cdot)$ are the multivariate Beta and Digamma functions, respectively. $\mathrm{KL}[q(\mathbf{w}_{\mathrm{ft}})\parallel \Pi(\mathbf{w}_{\mathrm{ft}})]$ and $\mathbb{E}_{q(\mathbf{w}_{\mathrm{ft}})}\mathrm{KL}[q(h_{j\zeta})\parallel \Pi(h_{j\zeta}\mid \mathbf{w}_{\mathrm{ft}})]$ are given analogously; see \hyperref[app:kl-derivations]{Appendix~\ref{app:kl-derivations}} for details. 

The term $\mathbb{E}_{q}[\log p(\mathbf{y}\mid \mathbf{T})]$ in~\eqnref{eq:elbo-1} does not admit an analytical form. A direct stochastic optimization is computationally prohibitive due to the combinatorial nature of the skeleton tree structural variables $\{e_{j\zeta}, o_{j\zeta}, h_{j\zeta}:\zeta\in \mathsf{Z}_D\}_{j=1}^{K}$. The induced discrete model space grows exponentially with both tree depth ($D$) and ensemble size ($K$) having $\mathsf{O}((2p|\mathbf O|)^{N\cdot K})$ possible configurations, where $\mathsf{O}$ is the big-O notation. Combinatorial search over this astronomically large space is therefore infeasible. Moreover, gradient-based optimization over discrete variables would require high-variance score-function estimators~\citep{williams1992simple,ranganath2014bbvi} or specialized structured search procedures~\citep{koza1994genetic,SR-SciML-1}, both of which scale poorly and become impractical even for moderate values of $D$ and $K$. To overcome this computational bottleneck, \vasstmain\ introduces a differentiable relaxation of symbolic trees, regarded as \emph{soft symbolic trees}, transforming discrete structural and labeling assignments into probability distributions~\citep{maddison2017concretedistributioncontinuousrelaxation,liu2018darts}.
\begin{figure*}[!htp]
    \centering
    \includegraphics[width=\linewidth]{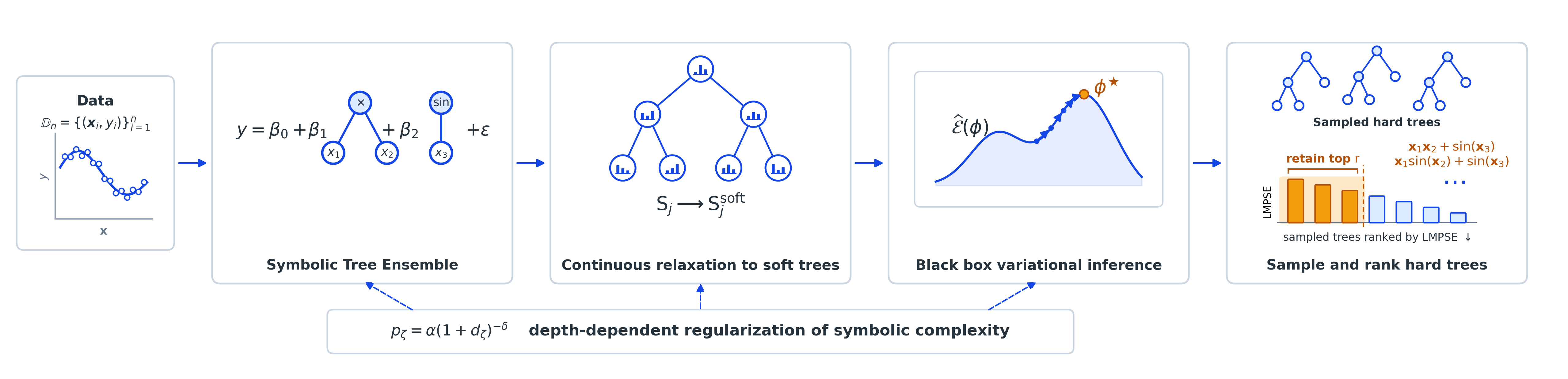}
    \caption{The complete \vasstmain\ pipeline.}
    \label{fig:schematic}
\end{figure*}
\paragraph{Soft symbolic trees via continuous relaxations.}
To construct soft symbolic trees, i.e., $\mathsf{S}_j\to \mathsf{S}_j^{\mathrm{soft}}$ for $j=1,\ldots,K$, we implement the {Binary-Concrete}~\citep{maddison2017concretedistributioncontinuousrelaxation} and {Gumbel-Softmax}~\citep{gumbelsoftmax} {continuous relaxations} of the structural variables $\{e_{j\zeta}, o_{j\zeta}, h_{j\zeta}:\zeta\in \mathsf{Z}_D\}$ of the skeleton tree $\mathsf S_j$ for $j=1,\ldots,K$, as:
\begin{align}
\label{eq:continuous-relaxation}
\begin{split}
\text{\emph{Binary-Concrete}}:\;
&\widetilde{e}_{j\zeta}
= \sigma\!\left(
\tfrac{\ell_{j\zeta} + \log u_{j\zeta}^{\mathrm{ex}} - \log(1 - u_{j\zeta}^{\mathrm{ex}})}
{\tau_\mathrm{ex}}
\right),\\[6pt]
\text{\emph{Gumbel-Softmax}}:\;
&\mathbf{\widetilde{o}}_{j\zeta}
= \texttt{smax}\!\left(
\tfrac{\mathbf{a}_{j\zeta}^{\mathrm{op}} + \mathbf{g}_{j\zeta}^{\mathrm{op}}}
{\tau_{\mathrm{op}}}
\right),\\
&\mathbf{\widetilde{h}}_{j\zeta}
= \texttt{smax}\!\left(
\tfrac{\mathbf{a}_{j\zeta}^{\mathrm{ft}} + \mathbf{g}_{j\zeta}^{\mathrm{ft}}}
{\tau_{\mathrm{ft}}}
\right),
\end{split}
\end{align}
where $u_{j\zeta}^{\mathrm{ex}} \sim \mathrm{Unif}(0,1)$, $\mathbf{g}_{j\zeta}^{\mathrm{op}} = \{-\log(-\log u_{j\zeta, r}^{\mathrm{op}})\}_{r=1}^{|\mathbf O|}$, $\mathbf{g}_{j\zeta}^{\mathrm{ft}} = \{-\log(-\log u_{j\zeta, k}^{\mathrm{ft}})\}_{k=1}^{p}$, $u_{j\zeta, r}^{\mathrm{op}}\sim \mathrm{Unif}(0,1)$, and $u_{j\zeta, k}^{\mathrm{ft}}\sim \mathrm{Unif}(0, 1)$. Also, all uniform random variables are mutually independent. Observe that, $\widetilde{e}_{j\zeta}\in \mathbb{R}$, $\widetilde{\mathbf{o}}_{j\zeta}\in \mathbb{R}^{|\mathbf{O}|}$, and $\widetilde{\mathbf{h}}_{j\zeta}\in \mathbb{R}{^p}$ can be interpreted as the {soft one-hot encoding} of the corresponding original discrete structural variables. In~\eqnref{eq:continuous-relaxation}, $\bm \tau = (\tau_{\mathrm{ex}}, \tau_{\mathrm{op}}, \tau_{\mathrm{ft}})$ are {temperature parameters} controlling the sharpness of the Binary-Concrete and Gumbel-Softmax relaxations. Particularly, higher temperatures yield smooth mixtures of structural configurations, whereas for smaller temperatures, the relaxed variables concentrate toward discrete structures. A careful annealing of these parameters allows for a balance between exploration and structural learning of symbolic expressions. Thus, the resulting soft symbolic tree $\mathsf{S}_j^{\mathrm{soft}}$ is encoded by $\{\widetilde e_{j\zeta}, \widetilde{\mathbf o}_{j\zeta}, \widetilde{\mathbf h}_{j\zeta}:\zeta\in \mathsf{Z}_D\}$.
\paragraph{Evaluation of soft symbolic trees.}
Given the soft symbolic trees $\mathsf{S}_j^{\mathrm{soft}}$ for $j=1,\ldots,K$, as discussed above, we now outline an algorithm for evaluation of $\mathsf{S}_j^{\mathrm{soft}}$ at a given feature vector instance $\bm{x}_i\in \mathbb{R}^{p}$. This evaluation is done recursively over the nodes of $\mathsf{S}_j^{\mathrm{soft}}$. For a given node $\zeta\in \mathsf{Z}_D$, the soft terminal (leaf) contribution is computed as a convex combination of the input features, $ \widetilde{\mathbf{h}}_{j\zeta}^{\top}\bm{x}_i$, which corresponds to a {soft feature evaluation}. If the node acts as nonterminal (internal), its output is obtained by a weighted mixture of unary and binary operations over its children. Specifically, the {unary contribution} aggregates $\widetilde{o}_{j\zeta, u}\cdot u(\mathsf{S}_j^{\mathrm{soft}}(\bm{x}_i; L(\zeta)))$, over $u\in \mathbf{O}_u$, while the {binary contribution} aggregates $\widetilde{o}_{j\zeta, b}\cdot b(\mathsf{S}_j^{\mathrm{soft}}(\bm{x}_i; L(\zeta)), \mathsf{S}_j^{\mathrm{soft}}(\bm{x}_i; R(\zeta)))$, over $b\in \mathbf{O}_b$. The {complete node evaluation} is then obtained via \emph{soft gating}:
\begin{align}
\label{eq:soft-gating}
\begin{split}
    &\mathsf{S}_j^{\mathrm{soft}}(\bm{x}_i; \zeta) = (1-\widetilde{e}_{j\zeta})\widetilde{\mathbf{h}}_{j\zeta}^{\top}\bm{x}_i\\ &+\widetilde{e}_{j\zeta}\sum_{u\in \mathbf{O}_u}\widetilde{o}_{j\zeta, u}\cdot u(\mathsf{S}_j^{\mathrm{soft}}(\bm{x}_i; L(\zeta)))\\
    &+\widetilde{e}_{j\zeta}\sum_{b\in \mathbf{O}_b}\widetilde{o}_{j\zeta, b}\cdot b(\mathsf{S}_j^{\mathrm{soft}}(\bm{x}_i; L(\zeta)), \mathsf{S}_j^{\mathrm{soft}}(\bm{x}_i; R(\zeta))),
\end{split}
\end{align}
thus smoothly interpolating over all possible combinations of operators and features; see \hyperref[alg:soft-eval-at-node]{\textsc{SoftEvalAtNode} Algorithm~\ref{alg:soft-eval-at-node}} in \hyperref[app:soft-evaluation-algorithms]{Appendix~\ref{app:soft-evaluation-algorithms}}. Consequently, evaluating a soft symbolic tree corresponds to computing $\mathsf{S}_j^{\mathrm{soft}}(\bm{x}_i; \zeta)$ from~\eqnref{eq:soft-gating} at the root node $\zeta=0$. Repeating this procedure for the collection of $K$ soft symbolic trees and all observations yields the {soft design matrix} $\mathbf{T}_{\mathrm{soft}}\in \mathbb{R}^{n\times \overline{K+1}}$, which includes the intercept column; see \hyperref[alg:soft-eval]{\textsc{SoftEval} Algorithm~\ref{alg:soft-eval}} in \hyperref[app:soft-evaluation-algorithms]{Appendix~\ref{app:soft-evaluation-algorithms}}.
\paragraph{Stochastic approximation of $\mathcal{E}(\phi)$.}
The \texttt{ELBO} objective $\mathcal{E}(\phi)$ in~\eqnref{eq:elbo-1} involves the analytically intractable term $\mathbb{E}_{q}[\log p(\mathbf{y}\mid \mathbf{T})]$, since $\mathbf{T}$ exhibits nonlinear dependence on the soft symbolic trees. We therefore approximate this term using {Monte Carlo} (\texttt{MC}) sampling from the distribution induced over $\{\widetilde{e}_{j\zeta}, \widetilde{\mathbf{o}}_{j\zeta}, \widetilde{\mathbf{h}}_{j\zeta}:\zeta\in \mathsf{Z}_D\}_{j=1}^{K}$ in~\eqnref{eq:continuous-relaxation}, denoted as $\mathcal{P}_{\phi, \bm \tau}^{\mathrm{soft}}$. In particular, consider $S$ \texttt{MC} samples from $\mathcal{P}_{\phi, \bm \tau}^{\mathrm{soft}}$ and invoking the \hyperref[alg:soft-eval]{\textsc{SoftEval} Algorithm~\ref{alg:soft-eval}} in \hyperref[app:soft-evaluation-algorithms]{Appendix~\ref{app:soft-evaluation-algorithms}}, we compute the soft design matrices $\mathbf{T}_{\mathrm{soft}}^{(s)}$, and hence $\log p(\mathbf{y}\mid \mathbf{T}_{\mathrm{soft}}^{(s)})$ using~\eqnref{eq:p(y|T)}, for $s=1,\ldots,S$. Finally, the \texttt{MC} approximation of $\mathcal{E}(\phi)$ is obtained as:
\begin{align}
\label{eq:MC-approx-elbo}
\widehat{\mathcal{E}}(\phi) = \frac{1}{S}\sum_{s=1}^{S}\log p(\mathbf y\mid \mathbf T_{\mathrm{soft}}^{(s)}) - \mathrm{KL}[q(\Theta)\parallel \Pi(\Theta)],
\end{align}
where $\mathrm{KL}[q_\phi(\Theta)\parallel \Pi(\Theta)]$ is given by~\eqnref{eq:elbo-2}. \hyperref[alg:approx-elbo]{\textsc{ApproxELBO} Algorithm~\ref{alg:approx-elbo}} in \hyperref[app:approx-elbo]{Appendix~\ref{app:approx-elbo}} evaluates the stochastic approximation of $\mathcal{E}(\phi)$ in~\eqnref{eq:elbo-1} by $\widehat{\mathcal{E}}(\phi)$ in~\eqnref{eq:MC-approx-elbo} and \hyperref[app:elbo-gap]{Appendix~\ref{app:elbo-gap}} discusses the corresponding approximation gap. Importantly, $\widehat{\mathcal{E}}(\phi)$  is differentiable with respect to the variational parameters in $\phi$, enabling efficient gradient-based optimization.
\paragraph{Black-box optimization.}
To obtain the optimal variational parameter vector $\phi^{\star}$, we maximize $\widehat{\mathcal{E}}(\phi)$ in~\eqnref{eq:MC-approx-elbo} using {gradient-based black-box} variational inference~\citep{ranganath2014bbvi, advi}, leveraging {automatic differentiation}~\citep{rall1981automatic} to compute $\nabla_{\phi}\widehat{\mathcal{E}}(\phi)$. The resulting objective is optimized using the \texttt{AdamW} optimizer~\citep{AdamW}. To progressively sharpen the continuous relaxations in~\eqnref{eq:continuous-relaxation}, we employ an {annealing schedule} which linearly decreases the temperature across iterations, gradually transitioning from smooth structural mixtures to near-discrete symbolic trees; refer to \hyperref[alg:vasst]{\vasst\ Algorithm~\ref{alg:vasst}} in \hyperref[app:vasst-algorithm]{Appendix~\ref{app:vasst-algorithm}}.
\paragraph{Interpretable hard symbolic ensembles.}
While soft symbolic trees $S_j^{\mathrm{soft}}$ serve as tractable computational intermediaries, to recover interpretable symbolic expressions we sample \emph{hard symbolic trees} using $\phi^{\star}$. Concretely, we draw $H \in \mathbb N$ samples, where for each draw $s=1,\ldots,H$ and each tree $j=1,\ldots,K$, we sample: $\widehat{e}_{j\zeta}^{(s)}\sim \mathrm{Ber}(\sigma(\ell^{\star}_{j\zeta}))$, $\widehat{o}_{j\zeta}^{(s)}\sim \mathrm{Cat}(\texttt{smax}((\mathbf{a}_{j\zeta}^{\mathrm{op}})^{\star}))$, and $\widehat{h}_{j\zeta}^{(s)}\sim \mathrm{Cat}(\texttt{smax}((\mathbf{a}_{j\zeta}^{\mathrm{ft}})^{\star}))$, respectively for all $\zeta\in \mathsf{Z}_D$. This sampling defines a full depth-$D$ hard binary tree skeleton $\widehat{\mathsf{S}}_{j}^{(s)}$, which is then deterministically pruned to obtain $\mathsf{T}(f_{j}^{(s)}) = \mathfrak{p}(\widehat{\mathsf{S}}_{j}^{(s)})$ and hence expressions $f_{j}^{(s)} \in \mathcal{S}$. The expression design matrix $\mathbf{T}^{(s)}$ is computed by evaluating $f_{j}^{(s)}$ at instances $\{\bm{x}_i\}_{i=1}^{n}$. Conditioned on $\mathbf{T}^{(s)}$ and $\mathbf{y}$, the posterior mean estimates $\bm\beta^{(s)}_{\mathrm{PM}} = \mathbb{E}[\bm\beta\mid \sigma^{2}, \mathbf{T}^{(s)}, \mathbf{y}]$ and $(\sigma^{2}_{\mathrm{PM}})^{(s)} = \mathbb{E}[\sigma^{2}\mid \mathbf{T}^{(s)}, \mathbf{y}]$ are computed using~\eqnref{eq:posterior-model-parameters}. See \hyperref[alg:sample-hard-trees]{\textsc{SampleHard} Algorithm~\ref{alg:sample-hard-trees}} in \hyperref[app:sample-hard-symbolic-trees]{Appendix~\ref{app:sample-hard-symbolic-trees}} for details of obtaining hard symbolic expressions.  
\begin{figure*}[!t]
    \centering
    \includegraphics[width=\linewidth]{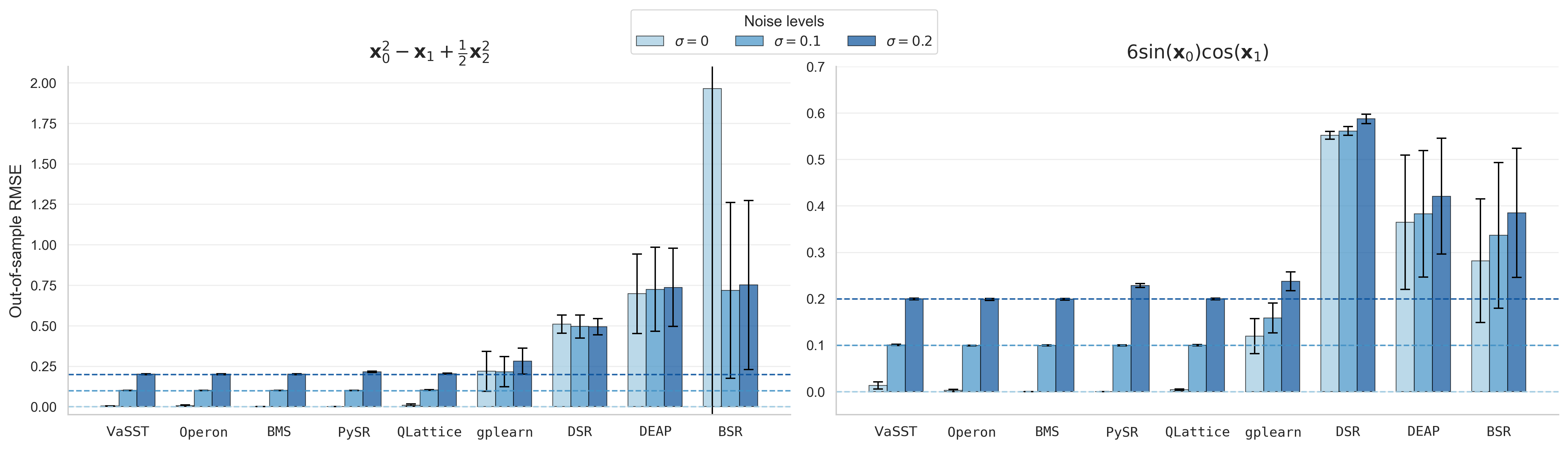}
    \caption{Out-of-sample \texttt{RMSE}s computed on a $10\%$ held-out test set of \vasstmain\ and competitors for learning~\eqnref{eq:simulation-generating-model-1} and~\eqnref{eq:simulation-generating-model-2} across $10$ repetitions under varying noise levels.}
    \label{fig:rmse-barplot-simulations}
\end{figure*}
\paragraph{Posterior-aware ranking of hard symbolic ensembles.}
We complete the \vasstmain\ pipeline by ranking the above $H$ sampled hard symbolic ensembles using a posterior-aware selection criterion. For any ensemble $\{\mathsf{T} (f_j)\}_{j=1}^{K}$, the $\log$-marginal posterior over symbolic ensembles ($\mathrm{LMPSE}$) is:
\begin{align}
\label{eq:LMPSE}
\begin{split}
&\mathrm{LMPSE}(\{\mathsf T(f_j)\}_{j=1}^K) = \log \Pi(\mathsf T(f_j)\}_{j=1}^K \mid \mathbb D_n)
\\& \qquad \quad
= \log p(\mathbf y \mid \mathbf T) + \log \Pi(\{\mathsf T(f_j)\}_{j=1}^K),
\end{split}
\end{align}
where $p(\mathbf y \mid \mathbf T)$ is as in \eqnref{eq:p(y|T)} and $\Pi(\{\mathsf T(f_j)\}_{j=1}^K)$ is: 
\begin{align*}
\begin{split}
&\Pi(\{\mathsf T(f_j)\}_{j=1}^K) = \mathbb{E}_{\mathbf w_{\mathrm{op}},\mathbf w_{\mathrm{ft}}} \left[\prod_{j=1}^K \Pi(\mathsf T(f_j) \mid \mathbf w_{\mathrm{op}}, \mathbf w_{\mathrm{ft}}) \right],
\end{split}
\end{align*}
with $\Pi(\mathsf T(f_j) \mid \mathbf w_{\mathrm{op}}, \mathbf w_{\mathrm{ft}})$ given in \eqnref{eq:OG-symbolic-tree-prior}. \hyperref[subsec:LMPSE-derivation]{Appendix~\ref{subsec:LMPSE-derivation}} gives the analytical formula and derivation of $\mathrm{LMPSE}$.
Note that, $\mathrm{LMPSE}$ in~\eqnref{eq:LMPSE} quantifies an interpretability-accuracy trade-off by balancing data fit, through the marginal likelihood, with structural parsimony, through the marginal prior over the symbolic ensemble.
Ranking hard symbolic ensembles by $\mathrm{LMPSE}$ and consequently retaining the top $\mathsf{r}\in \mathbb{N}$ ensembles with the largest $\mathrm{LMPSE}$ values provides an Occam's window-based posterior summary~\citep{madigan1994model} of competing symbolic structures which account for posterior uncertainty in symbolic model selection. A complete overview of the \vasstmain\ pipeline is in \href{fig:schematic}{Figure~\ref{fig:schematic}}.


\section{\texorpdfstring{\vasst}{VaSST} in Action}
\label{sec:vasst-in-action}

We evaluate \vasstmain\ through simulation studies and a suite of canonical Feynman equations, assessing: (a) structural learning of symbolic expressions, (b) predictive accuracy, (c) stability under experimental noise, and (d) computational scalability relative to existing Bayesian \sr\ methods. \vasstmain\ is compared with state-of-the-art \sr\ modules from \texttt{SRBench}~\citep{SRBench}, including Distributed Evolutionary Algorithms in \texttt{Python} (\deap)~\citep{deap}, \gplearn~\citep{stephens2016gplearn}, \operon~\citep{operon}, \pysr~\citep{pysr}, Deep Symbolic Regression (\dsr)~\citep{Deep-SR}, \qlattice~\citep{feyn-qlattice}, \bms~\citep{guimera2020bayesian}, and \bsr~\citep{BSR}. Across all experiments, we implement \vasstmain\ with $(K, D) = (3, 3)$ and use a $90/10$ train-test split for measuring predictive accuracy by computing out-of-sample root mean square error (\texttt{RMSE}) on the held-out test set. For \vasstmain, after variational optimization, we sample $H=2000$ hard symbolic ensembles from the learned soft representation and rank them using $\mathrm{LMPSE}$ in~\eqnref{eq:LMPSE}. 
Throughout, we use the top ranked symbolic model of \vasstmain\ for empirical comparisons, while the top $\mathsf{r}=10$ ranked models summarize posterior uncertainty over different symbolic structures. For budget equalization, we use common operator sets across methods whenever possible; see~\hyperref[app:configs]{Appendix~\ref{app:configs}} for detailed configurations of all methods. Experiments were conducted in \texttt{Python} on a MacBook Air with an Apple M2 chip and 8GB RAM.

\subsection{Simulation Experiments}
\label{subsec:simulation-study}

We consider the following symbolic expressions of varying levels of structural complexity:
\begin{align}
\label{eq:simulation-generating-model-1}
\mathbf{y}&=\mathbf{x}_0^2 - \mathbf{x}_1 + \tfrac{1}{2}\mathbf{x}_2^2 + \bm\epsilon, \\
\label{eq:simulation-generating-model-2}
\mathbf{y}&=6\sin(\mathbf{x}_0)\cos(\mathbf{x}_1) + \bm\epsilon,
\end{align}
where $\mathbf{x}_j=(x_{1,j},\ldots,x_{n,j})$ denotes the observed values of the $j$th feature, with all operations on these vectors interpreted coordinate-wise. We simulate $x_{i,j}\sim \mathrm{Unif}(2j,2j+1)$, independently for $i=1,\ldots, n$ with $n=2000$ and $j = 0,\ldots,p-1$. Note that, $p=3$ for~\eqnref{eq:simulation-generating-model-1} and $p=2$ for~\eqnref{eq:simulation-generating-model-2}. Two experimental regimes are studied: (a) a noiseless setting, $\bm \epsilon = \bm 0_n$ and (b) noisy settings, where $\bm \epsilon \sim \mathrm{N}_n(\bm 0_n, \sigma^{2}\mathbf{I}_n)$ with $\sigma^2 \in \{0.1^2, 0.2^2\}$.

\hyperref[fig:rmse-barplot-simulations]{Figure~\ref{fig:rmse-barplot-simulations}} reports the out-of-sample \texttt{RMSE}s across $10$ repetitions for both simulation settings, with detailed numerical summaries provided in~\hyperref[app:ooRMSE-sim1]{Appendix~\ref{app:ooRMSE-sim1}}. For~\eqnref{eq:simulation-generating-model-1} and~\eqnref{eq:simulation-generating-model-2}, \vasstmain\ achieves superior predictive accuracy with \texttt{RMSE}s that closely track the corresponding noise floors. Several competing methods, including \operon, \bms, \pysr, \qlattice, and \gplearn, also achieve competitive \texttt{RMSE}s. Conversely, \dsr, \deap, and \bsr\ exhibit degradation in accuracy with substantially larger errors, where \bsr\ is unstable particularly for~\eqnref{eq:simulation-generating-model-1}, while \dsr\ produces high \texttt{RMSE} values across both simulations. 

\begin{figure}[!htp]
    \centering
    \includegraphics[width=\linewidth]{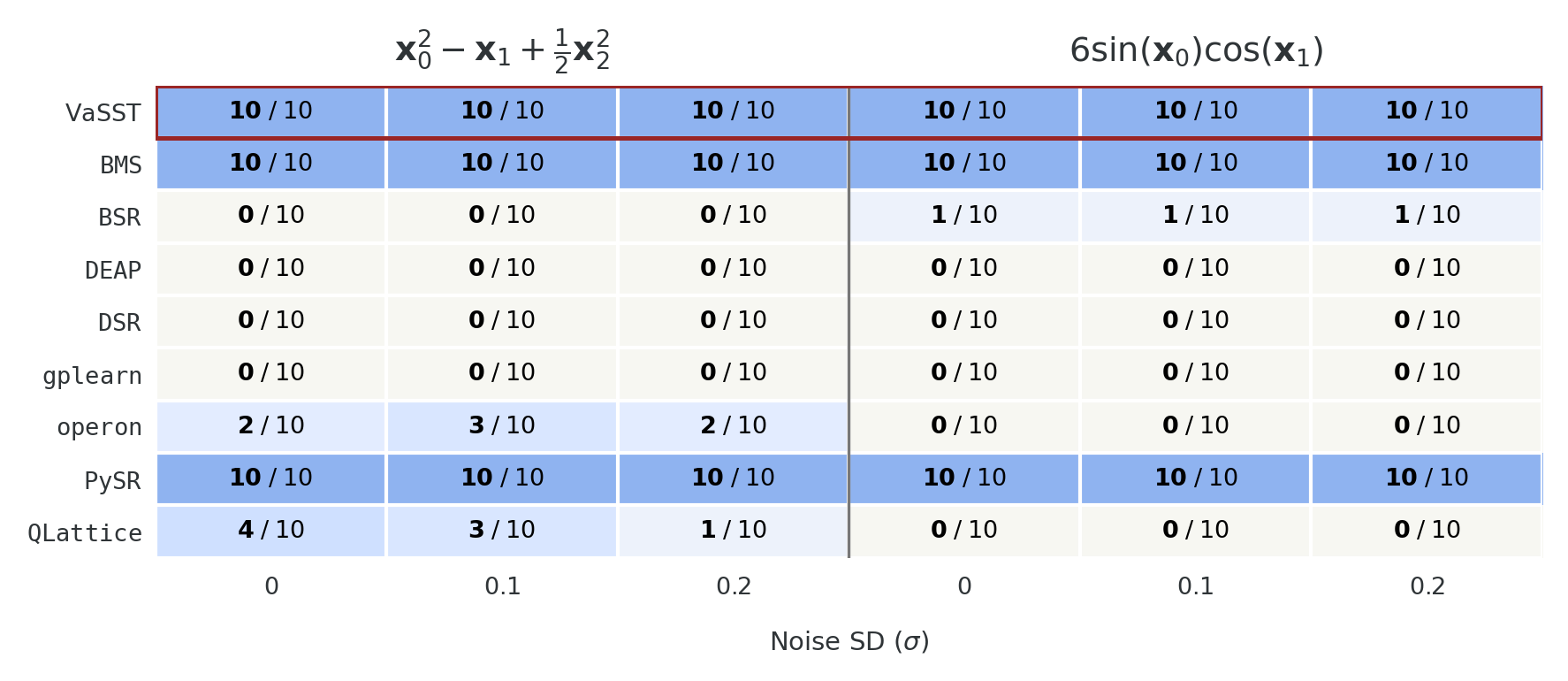}
    \caption{Expression recovery frequency of \vasstmain\ and competitors for learning~\eqnref{eq:simulation-generating-model-1} and~\eqnref{eq:simulation-generating-model-2}.}
    \label{fig:expression-recovery-simulations}
\end{figure}

We next examine whether predictive accuracy translates into structural recovery.
\hyperref[fig:expression-recovery-simulations]{Figure~\ref{fig:expression-recovery-simulations}} summarizes the expression recovery frequencies over same $10$ repetitions. \vasstmain, \bms, and \pysr\ recover the correct structure in all repetitions and across all noise levels while learning~\eqnref{eq:simulation-generating-model-1} and~\eqnref{eq:simulation-generating-model-2}, whereas \bsr, \deap, \dsr, \gplearn, \operon, and \qlattice\ either fail to recover the target expressions or do so inconsistently across repetitions and noise levels.
The symbolic expressions learned by each method reported in~\hyperref[app:additional-symbolic-expressions-sim1]{Appendix~\ref{app:additional-symbolic-expressions-sim1}} further explain the distinction between accuracy and structural recovery. The $\mathrm{LMPSE}$ ranked expressions selected by \vasstmain\ retain the data-generating components of~\eqnref{eq:simulation-generating-model-1} and~\eqnref{eq:simulation-generating-model-2}, while methods such as \operon, \qlattice, and \gplearn\ often achieve low \texttt{RMSE}s through structurally indirect formulas. In particular, \operon\ and \qlattice\ occasionally recover the polynomial structure in~\eqnref{eq:simulation-generating-model-1}, but via unwieldy and complex expressions with multiple high-order powers or nested trigonometric transformations; this recovery also deteriorates under noise and does not transfer to the trigonometric model in~\eqnref{eq:simulation-generating-model-2}. Thus, \hyperref[fig:expression-recovery-simulations]{Figures~\ref{fig:rmse-barplot-simulations}}-\ref{fig:expression-recovery-simulations} and~\hyperref[tab:rep1-expressions]{Tables~\ref{tab:rep1-expressions}}–\ref{tab:sim3-noise2-rep1-expressions} in~\hyperref[app:additional-symbolic-expressions-sim1]{Appendix~\ref{app:additional-symbolic-expressions-sim1}} collectively show that \vasstmain\ balances strong predictive performance with stable symbolic recovery through structurally parsimonious expressions, which is driven by the depth-dependent split probability in~\eqnref{eq:skeleton-prior} penalizing complex symbolic tree forms.

\begin{figure}[!htp]
    \centering
    \includegraphics[width=\linewidth]{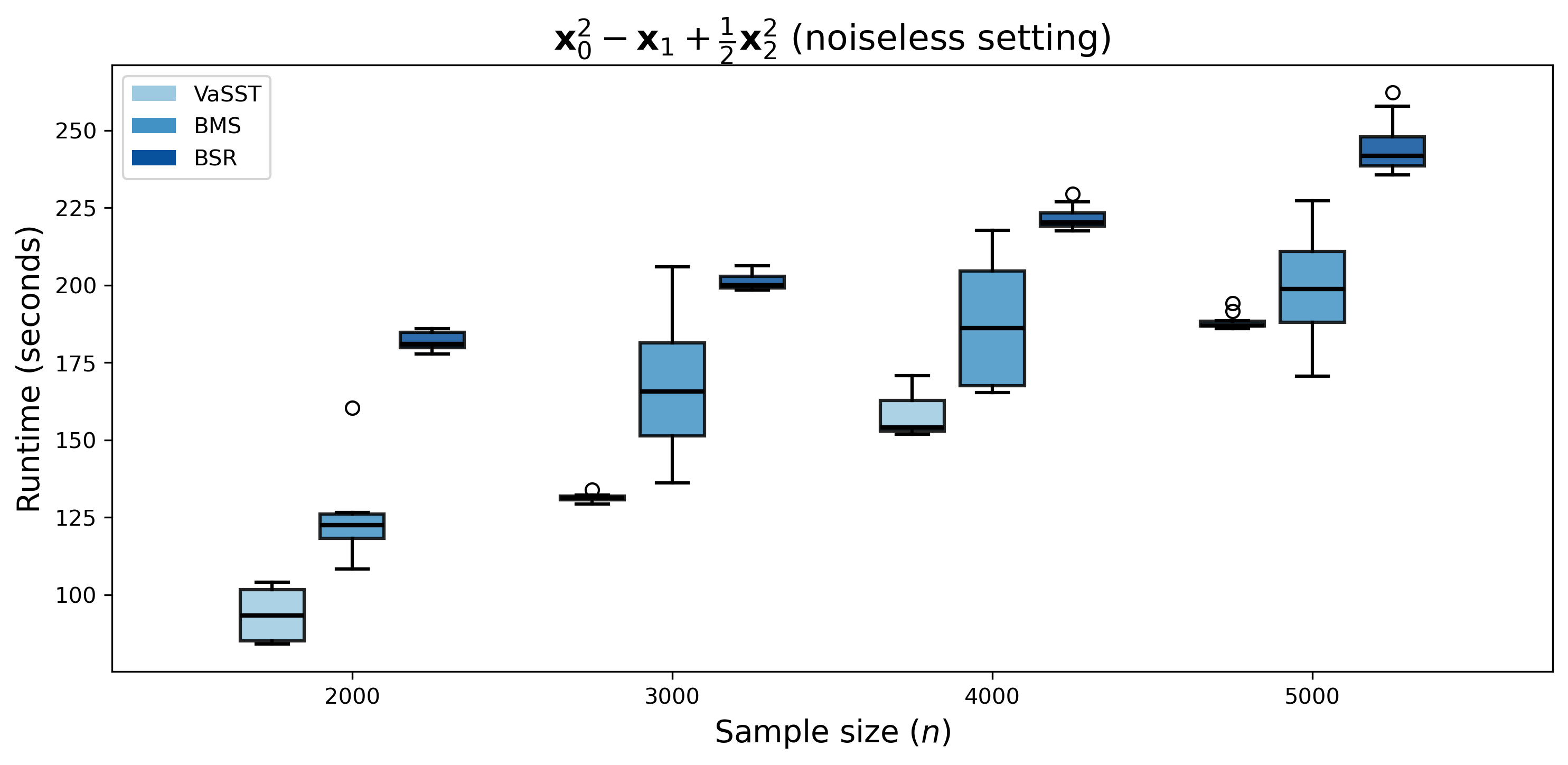}
    \caption{Computational scalability of \vasstmain, \bms, \bsr.}
    \label{fig:runtime_boxplot}
\end{figure}

Furthermore, in learning~\eqnref{eq:simulation-generating-model-1} under the noiseless setting,~\hyperref[fig:runtime_boxplot]{Figure~\ref{fig:runtime_boxplot}} compares the single-run compute time of \vasstmain\ with the Bayesian baselines \bms\ and \bsr\ (maintained at similar configurations to that of \vasstmain). For increasing sample sizes $n\in \{2000, 3000, 4000, 5000\}$ with $10$ repetitions per $n$, \vasstmain\ invariably records the lowest runtime, demonstrating improved scalability.
Finally, the top $\mathsf{r}=10$ hard symbolic ensembles in~\hyperref[tab:vasst-top10-LMPSE]{Tables~\ref{tab:vasst-top10-LMPSE}}-\ref{tab:sim3-noise2-vasst-top10-jmp} of~\hyperref[app:top-symbolic-expressions-vasst]{Appendix~\ref{app:top-symbolic-expressions-vasst}}, ranked by $\mathrm{LMPSE}$, provide a posterior-aware summary of uncertainty over competing symbolic expressions. These ranked candidates show how \vasstmain\ concentrates posterior support on expressions with the correct symbolic structure while retaining structurally similar variants as low rank alternatives.

\subsection{Application to Feynman Equations}
\label{subsec:feynman-data-study}

To further assess symbolic discovery in scientific settings, we apply \vasstmain\ and the competing methods to benchmark \sr\ problems from the \emph{Feynman Symbolic Regression Database} (\texttt{FSReD})~\citep{AI-Feynman}, which contains over $100$ equations from the Feynman Lectures on Physics~\citep{feynman-1}, each with $10^5$ observations. We consider five representative laws with varying symbolic complexity viz., Coulomb’s law~\eqnref{eq:feynman-cl}, change in gravitational potential energy~\eqnref{eq:feynman-cpe}, Lorentz force on a moving charge in an electromagnetic field~\eqnref{eq:feynman-fce}, Fourier’s law of thermal conduction~\eqnref{eq:feynman-ftc}, and angular distribution of scattering with forward-backward asymmetry~\eqnref{eq:feynman-ada}:
\begin{align}
F
&= 0.08\,\tfrac{q_1 q_2}{\epsilon r^{2}},
\tag{{CL}}
\label{eq:feynman-cl}
\\
\Delta U
&= Gm_1m_2\left(\tfrac{1}{r_2}-\tfrac{1}{r_1}\right),
\tag{{CPE}}
\label{eq:feynman-cpe}
\\
F
&= q\left(E_f+vB\sin\theta\right),
\tag{{FCE}}
\label{eq:feynman-fce}
\\
P
&= \tfrac{\kappa A(T_2-T_1)}{d},
\tag{{FTC}}
\label{eq:feynman-ftc}
\\
f
&= \beta^{\dagger}\left(1+\alpha^{\dagger}\cos\theta\right),
\tag{{ADA}}
\label{eq:feynman-ada}
\end{align}
where the input features, responses, and constants are described in \hyperref[app:feynman-input-output]{Appendix~\ref{app:feynman-input-output}}. For each equation, we randomly subsample $n=2000$ observations and study the same regimes as in~\hyperref[subsec:simulation-study]{Section~\ref{subsec:simulation-study}}: the original noiseless response and noisy regimes obtained by adding Gaussian noise with variance $\sigma^2 \in \{0.1^2,0.2^2\}$ to the response variable, mimicking measurement error and/or experimental perturbations.
\begin{table*}[!t]
\centering
\scriptsize
\renewcommand{\arraystretch}{1.15}
\caption{Expressions recovered by \vasstmain\ and competitors for learning \eqnref{eq:feynman-cpe} under $\sigma=0.2$.}
\label{tab:main-feynman-I-13-12-noise0p2-substituted-results}
\begin{tabularx}{\linewidth}{@{}l >{\raggedright\arraybackslash}X r@{}}
\toprule
\toprule
\textbf{Method} & \textbf{Symbolic Expression Learned} & \textbf{\texttt{RMSE}} \\
\midrule

\rowcolor{vasstgray}
\vasstmain &
$\,\underline{1.014765\,\tfrac{Gm_1m_2}{r_2}
-1.000981\,\tfrac{Gm_1m_2}{r_1}}
-0.006813\,\tfrac{r_2}{r_1}
-0.055216$ &
$0.199507$ \\

\addlinespace[1pt]

\operon &
$-0.010000
-0.332\left[
\tfrac{
0.707\,Gm_1^2/r_1
-6.275451\,Gm_1m_2/r_1
}
{2.031\,r_2/r_1}
+3.007296\,\tfrac{Gm_1m_2}{r_1}
+\tfrac{
0.670\,Gm_1^2/r_1
}
{
3.388\,Gm_1^2/r_1
+2.718552\,m_2^2/m_1^2
}
\right]$ &
$0.199207$ \\

\addlinespace[1pt]

\pysr &
$Gm_1m_2\left(\tfrac{1}{r_2}-\tfrac{1}{r_1}\right)$ &
$0.199569$ \\

\addlinespace[1pt]

\gplearn &
$Gm_1m_2\left(\tfrac{1}{r_2} - \tfrac{1}{r_1}\right)$ &
$0.199569$ \\

\addlinespace[1pt]

\deap &
$-\tfrac{Gm_1m_2}{r_1}
-\tfrac{Gm_1^2}{r_1}
+\tfrac{m_2}{m_1}
+\tfrac{2.8770}{-Gm_1^2/r_1+m_2/m_1}
-4.5144$ &
$0.589860$ \\

\addlinespace[1pt]

\bsr &
$8.54074
+0.827664\,\tfrac{m_2r_2}{m_1r_1}
+0.253507\left(
\tfrac{Gm_1m_2r_2}{r_1^2}
-\tfrac{Gm_1^2}{r_1}\right)
+0.379881\,\tfrac{m_2^2}{m_1^2}$ &
$1.107579$ \\

\addlinespace[1pt]

\bms &
$-\tfrac{Gm_1m_2}{r_1}
+\tfrac{a_0Gm_1m_2}{2r_2}$ &
$0.199557$ \\

\addlinespace[1pt]

\qlattice &
$2.11181
\left(
-0.270397\tfrac{Gm_1^2}{r_1}
-0.00788965
\right)
\left(
-0.0392111\tfrac{r_2}{r_1}
+1.45354
+1.86691\exp\!\left\{0.281573\tfrac{r_2}{r_1}\right\}
\right)
\sqrt{
(
\tfrac{m_2}{m_1}
-0.0186931
)^2}
+0.238388$ &
$0.246633$ \\

\addlinespace[1pt]

\dsr &
$Gm_1m_2\left(\tfrac{1}{r_2}-\tfrac{1}{r_1}\right)$ &
$0.199569$ \\

\bottomrule
\bottomrule
\end{tabularx}
{
\scriptsize
{
\emph{Note}: $a_0$ is a constant learned by \bms. The \texttt{RMSE} values correspond to the out-of-sample \texttt{RMSE}s computed on a $10\%$ held-out test set.
}
}
\end{table*}
\begin{figure}[!htp]
    \centering
    \includegraphics[width=\linewidth]{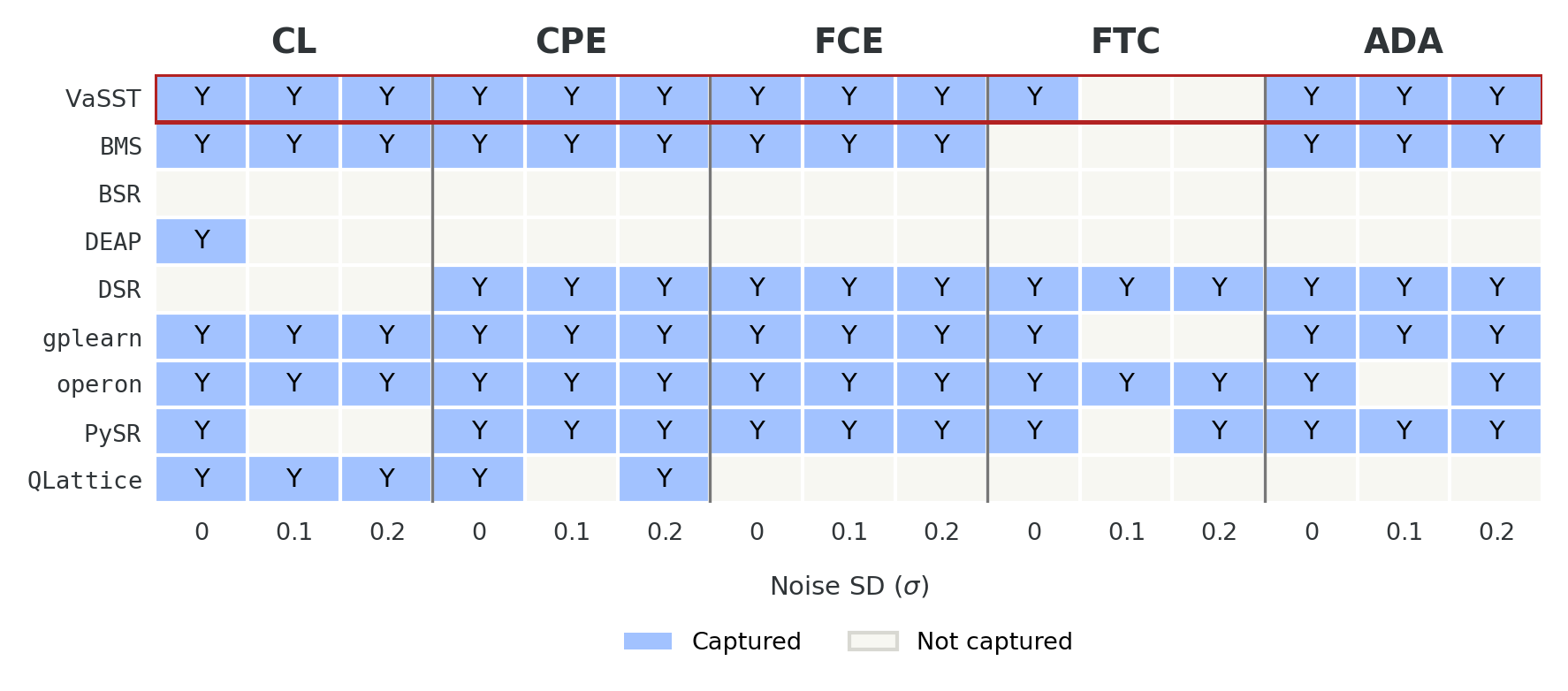}
    \caption{Expression recovery of \vasstmain\ and competitors for learning the Feynman equations.}
    \label{fig:expression-recovery-simulations-feynman}
\end{figure}

\hyperref[fig:expression-recovery-simulations-feynman]{Figure~\ref{fig:expression-recovery-simulations-feynman}} summarizes symbolic recovery across the five Feynman equations and varying noise levels, with full expression-level results reported in~\hyperref[app:additional-feynman-results-expressions-learned]{Appendix~\ref{app:additional-feynman-results-expressions-learned}} and a representative example for \eqnref{eq:feynman-cpe} under $\sigma=0.2$ shown in~\hyperref[tab:feynman-I-13-12-noise0p2-substituted-results]{Table~\ref{tab:main-feynman-I-13-12-noise0p2-substituted-results}}. \vasstmain\ recovers the target structure for~\eqnref{eq:feynman-cl},~\eqnref{eq:feynman-cpe},~\eqnref{eq:feynman-fce}, and~\eqnref{eq:feynman-ada} across noise regimes. For~\eqnref{eq:feynman-ftc}, it learns the exact symbolic form in the noiseless setting and retains dominant thermal conduction components under the noisy regimes. Across these settings, the out-of-sample \texttt{RMSE}s of \vasstmain\ closely track the corresponding noise levels, indicating that its symbolic recovery is accompanied by accurate prediction. The genetic programming modules are competitive on several equations. \pysr\ and \gplearn\ recover exact or near-exact formulas with near-noise-floor \texttt{RMSE} values. \operon\ also achieves accurate prediction in many cases, although its learned formulas are sometimes more algebraically elaborate, involving split terms, fitted constants, or indirect rearrangements of the target law. \bms\ is structurally accurate for~\eqnref{eq:feynman-cl},~\eqnref{eq:feynman-cpe},~\eqnref{eq:feynman-fce}, and~\eqnref{eq:feynman-ada}, but is less successful in recovering the full difference structure of~\eqnref{eq:feynman-ftc}. In contrast, \qlattice, \bsr, and \deap\ return unwieldy expressions involving multiple high-order interactions, consistent with their behavior in~\hyperref[subsec:simulation-study]{Section~\ref{subsec:simulation-study}}. When these methods attain competitive \texttt{RMSE}s, the fit is typically achieved through complex surrogate structures rather than clean recovery of the underlying physical law. Overall, the Feynman results show that \vasstmain\ accurately learns the underlying symbolic structure of the benchmark equations while providing a posterior-aware summary of uncertainty over competing symbolic forms through the top $\mathsf{r}=10$ $\mathrm{LMPSE}$ ranked hard symbolic ensembles reported in~\hyperref[tab:vasst-top10-I-12-2-mapped-noiseless]{Tables~\ref{tab:vasst-top10-I-12-2-mapped-noiseless}}-\ref{tab:vasst-top10-III-17-37-mapped-noise0p2} of~\hyperref[app:top-expressions-vasst-feynman-equations]{Appendix~\ref{app:top-expressions-vasst-feynman-equations}}. As in \hyperref[subsec:simulation-study]{Section~\ref{subsec:simulation-study}}, \vasstmain\ also exhibits substantial runtime gains over \bms\ and \bsr, with detailed runtime comparisons provided in~\hyperref[tab:time_all_vertical]{Table~\ref{tab:time_all_vertical}} of \hyperref[app:runtime_bayes_final]{Appendix~\ref{app:runtime_bayes_final}}.

We conclude with additional diagnostics that further characterize the robustness and computational behavior of \vasstmain. Specifically,~\hyperref[app:ablation-study]{Appendix~\ref{app:ablation-study}} presents ablation studies for \vasstmain\ over $K$ and $D$, providing practical guideline for choosing the symbolic ensemble size and tree depth, and showing that $(K, D)=(3, 3)$ balances symbolic expressivity, recovery accuracy, and computational cost in the applications considered above. Memory footprint and gradient variability analyses in~\hyperref[app:memory-footprint]{Appendix~\ref{app:memory-footprint}} indicate that the variational optimization in \vasstmain\ remains computationally stable as the symbolic search space grows. Finally, \hyperref[app:LLM-SR]{Appendix~\ref{app:LLM-SR}} compares \vasstmain\ with recent machine learning-based \sr\ architectures, including \texttt{LLM-SR}~\citep{llmsr} for learning~\eqnref{eq:feynman-ada} under $\sigma=0.2$. The results highlight a complementary distinction between prompt-driven and probabilistic approaches to \sr. While prompt-driven methods can leverage the generative capacity of language models, the \texttt{LLM-SR} experiments show sensitivity to initialization, language model choice, and prompt specification. In particular, weaker prompts often yield incomplete or non-executable symbolic programs, while more structured prompts improve execution but may still produce overly complex or structurally invalid expressions. In contrast, \vasstmain\ frames symbolic discovery as posterior inference over expression structures, enabling efficient optimization and posterior-aware uncertainty summaries across competing symbolic forms.
%


\section{Conclusion}
We presented \vasstmain, a fully probabilistic framework for \sr\ that leverages variational inference with continuous relaxations of symbolic trees. It enables scalable gradient-based optimization while preserving interpretability equipped with posterior-aware uncertainty quantification. Empirically, \vasstmain\ demonstrates strong structural recovery, competitive predictive accuracy, stability under noise, and notable computational gains over existing Bayesian \sr\ methods.

This work opens promising future avenues for fully probabilistic and scalable variational inference–based \sr. First, extending \vasstmain\ to infer internal numerical constants would broaden its applicability in learning scientific laws with unknown physical coefficients. Although the current formulation can represent simple constants compositionally, such as $\sin(2x)$ through $\sin(x+x)$, it does not directly learn free constants inside nonlinear terms such as $\exp(-\tfrac{1}{2}x^2)$. Second, because algebraically equivalent expressions are naturally pervasive in \sr, future work could incorporate equivalence-aware symbolic model comparison or posterior summaries that group semantically identical expression forms. Finally, more structured optimization strategies for the variational objective may further improve scalability relative to modern machine learning-based \sr\ approaches while preserving the probabilistic interpretation of \vasstmain.

\begin{contributions}
\href{https://roy-sr-007.github.io/}{SR} and \href{https://pritamdey.github.io/}{PD} were responsible for the conceptualization of \vasstmain, the development of the associated methodology, the design, code, and implementation of the experiments, and the writing of the article. \href{https://sites.google.com/tamu.edu/bani-k-mallick/bio}{BKM} provided supervision and guidance throughout the course of the research.
\end{contributions}

\begin{acknowledgements}
Research reported in this article was supported by National Institute of General Medical Sciences of the National  Institutes of Health under award number R01GM163238. Experiments in~\hyperref[app:LLM-SR]{Appendix~\ref{app:LLM-SR}} were conducted with the computing resources provided by Texas A\&M Department of Statistics Arseven Computing Cluster.
The authors are grateful to the anonymous reviewers for their constructive comments, which helped improve the quality and clarity of the article.
\end{acknowledgements}

\bibliography{references}



\newpage
\onecolumn

\begin{center}
\hrule height 4pt
\vspace{0.25in}
{\Large\bfseries Supplementary Materials for\\[0.3em]
\texorpdfstring{\vasst}{VaSST}: Variational Inference for Symbolic Regression using Soft Symbolic Trees}
\vspace{0.25in}
\hrule height 1pt
\end{center}
\vspace{1em}

\appendix

\renewcommand{\theequation}{\Alph{section}.\arabic{equation}}
\counterwithin{equation}{section}
\renewcommand{\thefigure}{\Alph{section}.\arabic{figure}}
\renewcommand{\thetable}{\Alph{section}.\arabic{table}}
\counterwithin{figure}{section}
\counterwithin{table}{section}

\tableofcontents
\makeatletter
\let\addcontentsline\oldaddcontentsline 
\makeatother

\newpage

\section{Notations}
\label{app:notation}

\begin{table}[!htp]
\centering
\small
\renewcommand{\arraystretch}{1.15}
\caption{Special functions and other notations.}
\label{tab:notation}
\begin{tabular}{lll}
\toprule
\toprule
\textbf{Symbol} & \textbf{Name} & \textbf{Definition} \\ 
\midrule

$\mathbb{N}$ 
& Natural numbers 
& Set $\{1,2,3,\dots\}$ \\[10pt]

$\mathbb{R}$ 
& Real numbers 
& Set of real numbers \\[10pt]

$\mathbb{R}^{p}$ 
& $p$-dimensional Euclidean space 
& $\{\bm{x}=(x_1,\dots,x_p)^\top : x_i \in \mathbb{R}\}$ \\[10pt]

$\mathbb{R}^{m\times n}$ 
& Real matrix space 
& Space of real $m\times n$ matrices \\[10pt]

$|A|$ 
& Cardinality 
& Number of elements in a finite set $A$ \\[10pt]

$\Gamma(t)$ 
& Gamma function 
& $\displaystyle \Gamma(t)=\int_{0}^{\infty} x^{t-1} e^{-x}\,dx$ \\[10pt]

$\Psi(t)$ 
& Digamma function 
& $\displaystyle \Psi(t)=\frac{d}{dt}\log\Gamma(t)$ \\[10pt]

$\mathcal{B}(\bm\eta)$ 
& Multivariate Beta function 
& $\displaystyle 
\mathcal{B}(\bm\eta)
=\frac{\prod_{k=1}^{m}\Gamma(\eta_k)}
{\Gamma\!\left(\sum_{k=1}^{m}\eta_k\right)}$ \\[10pt]

$\mathbf{1}_m$ 
& Vector of ones 
& $(1,\dots,1)^\top \in \mathbb{R}^m$ \\[10pt]

$\bm{0}_m$ 
& Vector of zeros 
& $(0,\dots,0)^\top \in \mathbb{R}^m$ \\[10pt]

$\mathbf{I}_m$ 
& Identity matrix of order $m$ 
& $(\mathbf{I}_m)_{ij} = \mathbf{1}\{i=j\}, \quad i,j = 1,\dots,m$ \\[10pt]

$|\mathbf{A}|$ 
& Determinant 
& Determinant of the matrix $\mathbf{A}\in \mathbb{R}^{m\times m}$ \\[10pt]

$\Delta^m$ 
& $m$-dimensional simplex 
& $\displaystyle 
\left\{\mathbf{w}\in[0,\infty)^m :
\mathbf{1}_m^\top \mathbf{w}=1\right\}$ \\[10pt]

$\texttt{smax}(\bm z)$
& Softmax function
& $\displaystyle
\texttt{smax}(\bm z)_k
=
\frac{\exp(z_k)}
{\sum_{\ell=1}^{m}\exp(z_\ell)},
\; k=1,\dots,m$ \\[10pt]

$\sigma(t)$
& Sigmoid function
& $\displaystyle
\sigma(t)=\frac{1}{1+\exp(-t)}$ \\[10pt]

$\mathsf{O}(g(n))$ 
& Order notation 
& $f(n) = \mathsf{O}(g(n))$ if 
$|f(n)| \le C\,|g(n)|$ 
for some constant $C>0$ 
and sufficiently large $n$\\
\bottomrule
\bottomrule
\end{tabular}
\emph{Note}: Notations listed here are used throughout the article.
\end{table}

\newpage

\section{Parameterization Details of Distribution Families used in \texorpdfstring{\vasst}{VaSST}}
\label{app:dist-families}

The \vasstmain\ framework utilizes several standard probability distributions for model specification and prior construction. Although these distributions are widely used in the machine learning literature, multiple parameterizations are common for some of them across various sources. For clarity and reproducibility, we explicitly document the notation and parameterizations adopted in this work.
\paragraph{Normal Inverse-Gamma.}
We use the conjugate Normal Inverse-Gamma ($\mathrm{NIG}$) prior for Gaussian linear regression parameters:
\begin{align}
\label{eq:nig-def}
\bm\beta \mid \sigma^2 &\sim \mathrm{N}_d(\bm\mu_0,\sigma^2\bm\Sigma_0), \quad
\sigma^2 \sim \mathrm{IG}(a_0,b_0),
\end{align}
where $d=K+1$, $\bm\mu_0\in\mathbb{R}^d$, $\bm\Sigma_0\in\mathbb{R}^{d\times d}$ is symmetric positive definite, and $a_0,b_0>0$.
Our Inverse-Gamma parameterization, i.e., $\mathrm{IG}(a,b)$, uses the density:
\begin{align}
\label{eq:ig-pdf}
\Pi(\sigma^2)
=
\frac{b^a}{\Gamma(a)}(\sigma^2)^{-(a+1)}
\exp\!\left(-\frac{b}{\sigma^2}\right),
\quad \sigma^2>0,
\end{align}
where $\Gamma(\cdot)$ is the Gamma function as in \hyperref[tab:notation]{Table~\ref{tab:notation}}. The conditional Gaussian density in~\eqnref{eq:nig-def} is:
\begin{align}
\label{eq:normal-pdf}
\Pi(\bm\beta\mid \sigma^2)
=
(2\pi)^{-\tfrac{d}{2}}(\sigma^2)^{-\tfrac{d}{2}}|\bm\Sigma_0|^{-\tfrac{1}{2}}
\exp\!\left\{-\frac{1}{2\sigma^2}(\bm\beta-\bm\mu_0)^\top \bm\Sigma_0^{-1}(\bm\beta-\bm\mu_0)\right\}.
\end{align}
Equivalently, the joint prior is $\Pi(\bm\beta,\sigma^2)=\Pi(\bm\beta\mid \sigma^2)\Pi(\sigma^2)$.
\paragraph{Bernoulli.} 
A Bernoulli random variable $E\sim \mathrm{Ber}(p)$, with $p\in(0,1)$, has support $\{0,1\}$ and probability mass function:
\begin{align}
\label{eq:bernoulli-pmf}
\Pi(e) = p^{e}(1-p)^{1-e},\quad e\in\{0,1\}.
\end{align}
In our tree prior, the expansion indicator follows $e_{j\zeta}\sim \mathrm{Ber}(p_\zeta)$ with
$p_\zeta=\alpha(1+d_\zeta)^{-\delta}$, where $\alpha\in(0,1)$ and $\delta>0$.
\paragraph{Dirichlet.}
For $m\ge 2$, a vector $\mathbf{w}$ follows a Dirichlet distribution with concentration vector $\bm\eta=(\eta_1,\ldots,\eta_m)^\top\in(0,\infty)^m$, i.e., $\mathbf{w}\sim \mathrm{Dir}(\bm\eta)$, having density:
\begin{align}
\label{eq:dirichlet-pdf}
p(\mathbf{w}\mid \bm\eta)
=
\frac{1}{\mathcal{B}(\bm\eta)}\prod_{k=1}^{m} w_k^{\eta_k-1},
\quad \mathbf{w}\in\Delta^m,
\end{align}
where $\mathcal{B}(\cdot)$ is the multivariate Beta function and $\Delta^{m}$ is the $m$-dimensional simplex, as defined in \hyperref[tab:notation]{Table~\ref{tab:notation}}. In \vasstmain, the operator and feature weight vectors follow
$\mathbf{w}_{\mathrm{op}}\sim \mathrm{Dir}(\bm\eta_{\mathrm{op}})$ and
$\mathbf{w}_{\mathrm{ft}}\sim \mathrm{Dir}(\bm\eta_{\mathrm{ft}})$, where
$\mathbf{w}_{\mathrm{op}}\in\Delta^{|\mathbf{O}|}$ and $\mathbf{w}_{\mathrm{ft}}\in\Delta^{p}$, respectively.
\paragraph{Categorical.}
Let $\mathbf{w}\in\Delta^m$. A Categorical random variable, $c \sim \mathrm{Cat}(\mathbf{w})$,
takes values in $\{1,\ldots,m\}$ with probability mass function:
\begin{align}
\label{eq:categorical-pmf}
\Pi(c = j\mid \mathbf{w}) = w_j,\qquad j\in\{1,\ldots,m\}
\end{align}
In \vasstmain, conditional on weights, we assign operators and features via:
\begin{align}
\label{eq:cat-assignments}
o_{j\zeta}\mid \mathbf{w}_{\mathrm{op}} \sim \mathrm{Cat}(\mathbf{w}_{\mathrm{op}}),\quad
h_{j\zeta}\mid \mathbf{w}_{\mathrm{ft}} \sim \mathrm{Cat}(\mathbf{w}_{\mathrm{ft}}), \quad
o_{j\zeta}\in\mathbf{O},\;
\text{feature }h_{j\zeta}.
\end{align}
When needed, we identify the operator set $\mathbf{O}$ and the set of features with index sets of sizes $|\mathbf{O}|$ and $p$,
respectively, so that~\eqnref{eq:categorical-pmf} applies directly.
\paragraph{Uniform.}
A continuous random variable $u \sim \mathrm{Unif}(a,b)$, with $a<b$, has support on the interval $(a,b)$ and probability density function:
\begin{align}
\label{eq:uniform-pdf}
\Pi(u)
=
\frac{1}{b-a}, \quad a < u < b,
\end{align}
and $\Pi(u)=0$ otherwise. In the Binary-Concrete and Gumbel-Softmax continuous relaxations in \hyperref[sec:VI-vasst]{Section \ref{sec:VI-vasst}}, the $\mathrm{Unif}(0,1)$ distribution was used to reparameterize the discrete structural variables. Further, in simulation experiments in \hyperref[subsec:simulation-study]{Section \ref{subsec:simulation-study}}, uniform distributions were used to generate the input variables.

\newpage

\section{Derivations of Key Quantities}
\label{app:mathematical-details}

\subsection{Marginalization of the Model Regression Parameters}
\label{app:parameter-marginalization}

The full joint posterior distribution as stated in~\eqnref{eq:joint-vasst-posterior} in \hyperref[sec:VI-vasst]{Section~\ref{sec:VI-vasst}} is:
\begin{align}
\label{eq:full-posterior-app}
\begin{split}
\Pi(\Theta,\bm\beta,\sigma^2\mid \mathbb{D}_n)
&\propto 
p(\mathbf{y}\mid \mathbf{T}, \bm\beta, \sigma^2)
\Pi(\bm\beta, \sigma^2)
\Pi(\mathbf{w}_{\mathrm{op}})\Pi(\mathbf{w}_{\mathrm{ft}})\prod_{j=1}^{K}\prod_{\zeta\in \mathsf{Z}_D}\Pi(e_{j\zeta})\Pi(o_{j\zeta}\mid \mathbf{w}_{\mathrm{op}})\Pi(h_{j\zeta}\mid \mathbf{w}_{\mathrm{ft}}).
\end{split}
\end{align}
Under the Gaussian likelihood and the conjugate $\mathrm{NIG}$ prior:
\[
\mathbf{y}\mid \bm\beta,\sigma^2,\mathbf{T}
\sim \mathrm{N}_n(\mathbf{T}\bm\beta,\sigma^2\mathbf{I}_n), \quad \bm\beta\mid\sigma^2 \sim \mathrm{N}_{K+1}(\bm\mu_0,\sigma^2\bm\Sigma_0),
\quad
\sigma^2 \sim \mathrm{IG}(a_0,b_0),
\]
the posterior remains $\mathrm{NIG}$:
\begin{equation}
\label{eq:app-model-posterior}
\bm\beta\mid\sigma^2,\mathbf{T},\mathbf{y}
\sim \mathrm{N}_{K+1}(\bm\mu_n,\sigma^2\bm\Sigma_n),
\quad
\sigma^2\mid \mathbf{T},\mathbf{y}
\sim \mathrm{IG}(a_n,b_n),
\end{equation}
where the posterior hyperparameters are:
\[
\bm\Sigma_n^{-1} = \bm\Sigma_0^{-1} + \mathbf{T}^{\top}\mathbf{T}, \quad \bm\mu_n = \bm\Sigma_n(\bm\Sigma_0^{-1}\bm\mu_0 + \mathbf{T}^{\top}\mathbf{y}), \quad a_n = a_0 + \tfrac{n}{2}, \quad b_n = b_0 + \tfrac{1}{2} \left(\mathbf{y}^{\top}\mathbf{y}
+\bm\mu_0^{\top}\bm\Sigma_0^{-1}\bm\mu_0
-\bm\mu_n^{\top}\bm\Sigma_n^{-1}\bm\mu_n\right).
\]
Integrating out $(\bm\beta,\sigma^2)$ yields the marginal likelihood:
\begin{align}
p(\mathbf{y}\mid \mathbf{T})
&=
\int p(\mathbf{y}\mid \mathbf{T},\bm\beta,\sigma^2)
\Pi(\bm\beta,\sigma^2)
\, d\bm\beta\, d\sigma^2
\nonumber\\
&= (2 \pi)^{-\frac{n}{2}}
\frac{\Gamma(a_n)}{\Gamma(a_0)}
\frac{|\bm\Sigma_n|^{\tfrac{1}{2}}}{|\bm\Sigma_0|^{\tfrac{1}{2}}}
\frac{b_0^{a_0}}{b_n^{a_n}},
\label{eq:marginal-likelihood-app}
\end{align}
where $\Gamma(\cdot)$ is the Gamma function as in \hyperref[tab:notation]{Table~\ref{tab:notation}}. Up to additive constants independent of $\mathbf{T}$, the log-marginal likelihood is therefore:
\begin{align}
\boxed{
\log p(\mathbf{y}\mid \mathbf{T})
=
\tfrac{1}{2}\log|\bm\Sigma_n|
-
a_n \log b_n
+
\log \Gamma(a_n)
+ C,
}
\label{eq:log-marginal-likelihood-app}
\end{align}
where $C$ collects terms independent of $\mathbf{T}$.
Consequently, the marginal posterior over the tree parameters $\Theta$ becomes:
\begin{align}
\boxed{
\Pi(\Theta\mid \mathbb{D}_n)
\propto
p(\mathbf{y}\mid \mathbf{T})
\Pi(\mathbf{w}_{\mathrm{op}})\Pi(\mathbf{w}_{\mathrm{ft}})\prod_{j=1}^{K}\prod_{\zeta\in \mathsf{Z}_D}\Pi(e_{j\zeta})\Pi(o_{j\zeta}\mid \mathbf{w}_{\mathrm{op}})\Pi(h_{j\zeta}\mid \mathbf{w}_{\mathrm{ft}}),
}
\label{eq:marginal-posterior-theta}
\end{align}
which is the quantity of interest we approximate by optimizing the variational objective.

\subsection{Derivation of the Analytical Kullback-Leibler Divergence Terms}
\label{app:kl-derivations}

In this section, we provide a detailed derivation of the Kullback-Leibler ($\mathrm{KL}$) term in~\eqnref{eq:elbo-1} of \hyperref[subsec:symbolic-tree-prior]{Section \ref{subsec:symbolic-tree-prior}}. Specifically, this term corresponds to the $\mathrm{KL}$ divergence between the variational posterior $q(\Theta)$ in \eqnref{eq:mean-field-variational-family} and the prior $\Pi(\Theta)$ over the tree parameters defined by~\eqnref{eq:skeleton-prior} and~\eqnref{eq:weight-priors}.
\begin{align}
\label{eq:KL-app-1}
\begin{split}
\mathrm{KL}&[q(\Theta)\parallel \Pi(\Theta))]
= \mathbb{E}_{q}\left[\log \frac{q(\Theta)}{\Pi(\Theta)} \right]
\\
&=\mathbb{E}_{q}\left[\log \frac{q(\mathbf{w}_{\mathrm{op}})q(\mathbf{w}_{\mathrm{ft}})\prod_{j=1}^{K}\prod_{\zeta \in \mathsf{Z}_D}q(e_{j\zeta})q(o_{j\zeta})q(h_{j\zeta})}{\Pi(\mathbf{w}_{\mathrm{op}})\Pi(\mathbf{w}_{\mathrm{ft}})\prod_{j=1}^{K}\prod_{\zeta\in \mathsf{Z}_D}\Pi(e_{j\zeta})\Pi(o_{j\zeta}\mid \mathbf{w}_{\mathrm{op}})\Pi(h_{j\zeta}\mid \mathbf{w}_{\mathrm{ft}})} \right]
\\
&=\mathbb{E}_{q}\left[\log \frac{q(\mathbf{w}_{\mathrm{op}})}{\Pi(\mathbf{w}_{\mathrm{op}})} + \log \frac{q(\mathbf{w}_{\mathrm{ft}})}{\Pi(\mathbf{w}_{\mathrm{ft}})} +
\sum_{j=1}^{K}\sum_{\zeta \in \mathsf{Z}_{D}}
\left(\log \frac{q(e_{j\zeta})}{\Pi(e_{j\zeta})} + \log \frac{q(o_{j\zeta})}{\Pi(o_{j\zeta}\mid \mathbf{w}_{\mathrm{op}})} + \log \frac{q(h_{j\zeta})}{\Pi(h_{j\zeta}\mid \mathbf{w}_{\mathrm{ft}})} \right)
\right]
\\
&=\mathrm{KL}[q(\mathbf{w}_{\mathrm{op}})\parallel \Pi(\mathbf{w}_{\mathrm{op}})] + \mathrm{KL}[q(\mathbf{w}_{\mathrm{ft}})\parallel \Pi(\mathbf{w}_{\mathrm{ft}})]
+ \sum_{j=1}^{K}\sum_{\zeta\in \mathsf{Z}_D}\bigg[\mathrm{KL}[q(e_{j\zeta})\parallel \Pi(e_{j\zeta})]
\\ & \qquad \qquad \qquad \qquad\qquad\quad
+\mathbb{E}_{q(\mathbf{w}_{\mathrm{op}})}\mathrm{KL}[q(o_{j\zeta})\parallel\Pi(o_{j\zeta}\mid \mathbf{w}_{\mathrm{op}})] 
+ \mathbb{E}_{q(\mathbf{w}_{\mathrm{ft}})}\mathrm{KL}[q(h_{j\zeta})\parallel \Pi(h_{j\zeta}|\mathbf{w}_{\mathrm{ft}})]\bigg].
\end{split}
\end{align}
Now, we derive the analytic forms of each of the terms in the final sum in~\eqnref{eq:KL-app-1}.

\begin{lemma}[KL divergence between Dirichlet distributions]
\label{lem:kl-dirichlet}
Let $q \equiv \mathrm{Dir}(\widetilde{\bm\eta})$ and
$\pi \equiv \mathrm{Dir}(\bm\eta)$ on $\Delta^m$, where
$\widetilde{\bm\eta},\bm\eta\in(0,\infty)^m$. Define $\widetilde\eta_0=\sum_{k=1}^m \widetilde\eta_k$.
Then:
\begin{align}
\mathrm{KL}\!\left[\mathrm{Dir}(\widetilde{\bm\eta})\,\|\,\mathrm{Dir}(\bm\eta)\right]
=
\log\frac{\mathcal{B}(\bm\eta)}{\mathcal{B}(\widetilde{\bm\eta})}
+ \sum_{k=1}^m (\widetilde\eta_k-\eta_k)\bigl[\Psi(\widetilde\eta_k)-\Psi(\widetilde\eta_0)\bigr],
\label{eq:kl-dirichlet}
\end{align}
where $\mathcal{B}(\cdot)$ and $\Psi(\cdot)$ are the multivariate Beta and Digamma functions, respectively, as defined in \hyperref[tab:notation]{Table~\ref{tab:notation}}.
\end{lemma}

\begin{proof}
Let $\mathbf{w} \in \Delta^{m}$ such that $\mathbf{w} \sim q$. From the definition of $\mathrm{KL}$ divergence, we conclude:
\begin{align}
\begin{split}
\mathrm{KL}\!\left[\mathrm{Dir}(\widetilde{\bm\eta})\,\|\,\mathrm{Dir}(\bm\eta)\right]
&= \mathbb{E}_{q}\!\left[\log\frac{q(\mathbf{w})}{\pi(\mathbf{w})}\right] 
= \mathbb{E}_{q}\!\Bigg[
\log\frac{\mathcal{B}(\bm\eta)}{\mathcal{B}(\widetilde{\bm\eta})}
+ \sum_{k=1}^m (\widetilde\eta_k-\eta_k)\log w_k
\Bigg] \\
&= \log\frac{\mathcal{B}(\bm\eta)}{\mathcal{B}(\widetilde{\bm\eta})}
+ \sum_{k=1}^m (\widetilde\eta_k-\eta_k)\,\mathbb{E}_{q}[\log w_k].
\end{split}
\label{eq:kl-dirichlet-step1}
\end{align}
It remains to compute $\mathbb{E}_{q}[\log w_k]$ under $q\equiv\mathrm{Dir}(\widetilde{\bm\eta})$. For any $k\in\{1,\ldots,m\}$, using differentiation under the integral sign:
\begin{align}
0
&=
\frac{\partial}{\partial \widetilde\eta_k}
\int_{\Delta^m} q(\mathbf{w})\,d\mathbf{w}
=
\int_{\Delta^m}
\frac{\partial}{\partial \widetilde\eta_k} q(\mathbf{w})\,d\mathbf{w}.
\label{eq:kl-dirichlet-step2}
\end{align}
Moreover:
\begin{align}
\begin{split}
&\log q(\mathbf{w})
=
-\log \mathcal{B}(\widetilde{\bm\eta})
+ \sum_{r=1}^m (\widetilde\eta_r-1)\log w_r\\
\implies & \frac{\partial}{\partial \widetilde\eta_k}\log q(\mathbf{w})
=
-\frac{\partial}{\partial \widetilde\eta_k}\log \mathcal{B}(\widetilde{\bm\eta})
+ \log w_k \\
\implies & \frac{\partial}{\partial \widetilde\eta_k} q(\mathbf{w})
=
q(\mathbf{w})
\left(
-\frac{\partial}{\partial \widetilde\eta_k}\log \mathcal{B}(\widetilde{\bm\eta})
+ \log w_k
\right).
\end{split}
\label{eq:kl-dirichlet-step3}
\end{align}
Substituting~\eqnref{eq:kl-dirichlet-step3} into \eqnref{eq:kl-dirichlet-step2} and noting that $\int_{\Delta^m} q(\mathbf{w})\,d\mathbf{w}=1$ gives:
\begin{align}
0
&=
-\frac{\partial}{\partial \widetilde\eta_k}\log \mathcal{B}(\widetilde{\bm\eta})
+ \int_{\Delta^m} q(\mathbf{w})\log w_k\,d\mathbf{w}
=
-\frac{\partial}{\partial \widetilde\eta_k}\log \mathcal{B}(\widetilde{\bm\eta})
+ \mathbb{E}_{q}[\log w_k],
\end{align}
and therefore:
\begin{align}
\mathbb{E}_{q}[\log w_k]
=
\frac{\partial}{\partial \widetilde\eta_k}\log \mathcal{B}(\widetilde{\bm\eta}) = \frac{\partial}{\partial \widetilde\eta_k} \left[\sum_{r=1}^m \log \Gamma(\widetilde\eta_r) - \log \Gamma(\widetilde\eta_0)\right] = \Psi(\widetilde\eta_k) - \Psi(\widetilde\eta_0)
\label{eq:dir-logwk-beta-deriv}.
\end{align}
Substituting this expression into \eqnref{eq:kl-dirichlet-step1} gives \eqnref{eq:kl-dirichlet}, completing the proof.
\end{proof}
Now, an application of \hyperref[lem:kl-dirichlet]{Lemma \ref{lem:kl-dirichlet}} yields the following:
\begin{align}
\mathrm{KL}[q(\mathbf{w}_{\mathrm{op}})\parallel \Pi(\mathbf{w}_{\mathrm{op}})]
&=
\log\!\frac{\mathcal{B}(\bm\eta_{\mathrm{op}})}
{\mathcal{B}(\widetilde{\bm\eta}_{\mathrm{op}})}
+ \sum_{k=1}^{|\mathbf{O}|}
(\widetilde{\eta}_{\mathrm{op}, k} - \eta_{\mathrm{op}, k})
\left[
\Psi(\widetilde{\eta}_{\mathrm{op}, k})
-
\Psi(\mathbf{1}_{|\mathbf{O}|}^{\top}\widetilde{\bm\eta}_{\mathrm{op}})
\right],
\label{eq:kl-operator-weights}
\\[6pt]
\mathrm{KL}[q(\mathbf{w}_{\mathrm{ft}})\parallel \Pi(\mathbf{w}_{\mathrm{ft}})]
&=
\log\!\frac{\mathcal{B}(\bm\eta_{\mathrm{ft}})}
{\mathcal{B}(\widetilde{\bm\eta}_{\mathrm{ft}})}
+ \sum_{k=1}^{p}
(\widetilde{\eta}_{\mathrm{ft}, k} - \eta_{\mathrm{ft}, k})
\left[
\Psi(\widetilde{\eta}_{\mathrm{ft}, k})
-
\Psi(\mathbf{1}_{p}^{\top}\widetilde{\bm\eta}_{\mathrm{ft}})
\right],
\label{eq:kl-feature-weights}
\end{align}
where $\mathbf{1}_m$ is a $m$-dimensional vector of ones as in \hyperref[tab:notation]{Table~\ref{tab:notation}}.
By definition, the $\mathrm{KL}$ divergence between two Bernoulli distributions $\mathrm{Ber}(\widetilde p)$ and $\mathrm{Ber}(p)$ is: 
\begin{align}
\begin{split}
\mathrm{KL}\!\left[\mathrm{Ber}(\widetilde p)\,\|\,\mathrm{Ber}(p)\right]
= \widetilde p\log\frac{\widetilde p}{p} + (1-\widetilde p)\log\frac{1-\widetilde p}{1-p}.
\end{split}
\label{eq:kl-bernoulli-step1}
\end{align}
Consequently using~\eqnref{eq:kl-bernoulli-step1}:
\begin{align}
\mathrm{KL}[q(e_{j\zeta})\parallel \Pi(e_{j\zeta})] = \widetilde{p}_{j\zeta}\log \frac{\widetilde{p}_{j\zeta}}{p_{\zeta}} + (1 - \widetilde{p}_{j\zeta})\log \frac{1-\widetilde{p}_{j\zeta}}{1-p_{\zeta}}.
\label{eq:kl-expansion-indicators}
\end{align}
Similarly, the $\mathrm{KL}$ divergence between two categorial distributions $\mathrm{Cat}(\widetilde{\bm\pi})$ and $\mathrm{Cat}(\mathbf{w})$ is:
\begin{align}
\mathrm{KL}\!\left[\mathrm{Cat}(\widetilde{\bm\pi})\,\|\,\mathrm{Cat}(\mathbf{w})\right]
&=\sum_{k=1}^{m} \widetilde{\pi}_{k}\log\frac{\widetilde{\pi}_{k}}{w_k}
=\sum_{k=1}^{m}\widetilde{\pi}_{k}\left[\log \widetilde{\pi}_{k}-\log w_{k}\right],
\label{eq:ekl-cat-step1}
\end{align}
which implies:
\begin{align}
\begin{split}
\mathbb{E}_{q(\mathbf{w}_{\mathrm{op}})}\mathrm{KL}[q(o_{j\zeta})\parallel \Pi(o_{j\zeta}\mid \mathbf{w}_{\mathrm{op}})] &= 
\sum_{k=1}^{|\mathbf{O}|}\widetilde{\pi}^{\mathrm{op}}_{j\zeta, k}\left[\log(\widetilde{\pi}^{\mathrm{op}}_{j\zeta, k}) - \mathbb{E}_{q(\mathbf{w}_{\mathrm{op}})}[\log w_k]\right]
\\
&= 
\sum_{k=1}^{|\mathbf{O}|}\widetilde{\pi}^{\mathrm{op}}_{j\zeta, k}\left[\log(\widetilde{\pi}^{\mathrm{op}}_{j\zeta, k}) - \Psi(\widetilde{\eta}_{\mathrm{op}, k}) + \Psi(\mathbf{1}^{\top}_{|\mathbf{O}|}\widetilde{\bm\eta}_{\mathrm{op}})\right],
\end{split}
\label{eq:kl-operator-indices}
\end{align}
where the last equality follows from~\eqnref{eq:dir-logwk-beta-deriv} in the proof of \hyperref[lem:kl-dirichlet]{Lemma \ref{lem:kl-dirichlet}}. Analogously:
\begin{align}
\begin{split}
\mathbb{E}_{q(\mathbf{w}_{\mathrm{ft}})}\mathrm{KL}[q(h_{j\zeta})\parallel \Pi(h_{j\zeta}\mid \mathbf{w}_{\mathrm{ft}})] =
\sum_{k=1}^{p}\widetilde{\pi}^{\mathrm{ft}}_{j\zeta, k}\left[\log(\widetilde{\pi}^{\mathrm{ft}}_{j\zeta, k}) - \Psi(\widetilde{\eta}_{\mathrm{ft}, k}) + \Psi(\mathbf{1}^{\top}_{p}\widetilde{\bm\eta}_{\mathrm{ft}})\right].
\end{split}
\label{eq:kl-feature-indices}
\end{align}
Putting together~\eqnref{eq:kl-operator-weights}, \eqnref{eq:kl-feature-weights}, \eqnref{eq:kl-expansion-indicators}, \eqnref{eq:kl-operator-indices}, and \eqnref{eq:kl-feature-indices} in \eqnref{eq:KL-app-1}, we get the following composite expression for the $\mathrm{KL}$ component:
\begin{equation}
\boxed{
\begin{aligned}
\mathrm{KL}[q(\Theta)&\parallel \Pi(\Theta)] = \log\! \frac{\mathcal{B}(\bm\eta_{\mathrm{op}})}
{\mathcal{B}(\widetilde{\bm\eta}_{\mathrm{op}})}
+ \sum_{k=1}^{|\mathbf{O}|}
(\widetilde{\eta}_{\mathrm{op}, k} - \eta_{\mathrm{op}, k})
\left[
\Psi(\widetilde{\eta}_{\mathrm{op}, k})
-
\Psi(\mathbf{1}_{|\mathbf{O}|}^{\top}\widetilde{\bm\eta}_{\mathrm{op}})
\right] + \log\!\frac{\mathcal{B}(\bm\eta_{\mathrm{ft}})}
{\mathcal{B}(\widetilde{\bm\eta}_{\mathrm{ft}})}\\
& 
+ \sum_{k=1}^{p}
(\widetilde{\eta}_{\mathrm{ft}, k} - \eta_{\mathrm{ft}, k})
\left[
\Psi(\widetilde{\eta}_{\mathrm{ft}, k})
-
\Psi(\mathbf{1}_{p}^{\top}\widetilde{\bm\eta}_{\mathrm{ft}})
\right] 
+  \sum_{j=1}^{K} \sum_{\zeta \in \mathsf{Z}_D} \Bigg[ \widetilde{p}_{j\zeta}\log \frac{\widetilde{p}_{j\zeta}}{p_{\zeta}} + (1 - \widetilde{p}_{j\zeta})\log \frac{1-\widetilde{p}_{j\zeta}}{1-p_{\zeta}} 
\\ &+ 
\sum_{k=1}^{|\mathbf{O}|}\widetilde{\pi}^{\mathrm{op}}_{j\zeta, k}\left[\log(\widetilde{\pi}^{\mathrm{op}}_{j\zeta, k}) - \Psi(\widetilde{\eta}_{\mathrm{op}, k}) + \Psi(\mathbf{1}^{\top}_{|\mathbf{O}|}\widetilde{\bm\eta}_{\mathrm{op}})\right]
+
\sum_{k=1}^{p}\widetilde{\pi}^{\mathrm{ft}}_{j\zeta, k}\left[\log(\widetilde{\pi}^{\mathrm{ft}}_{j\zeta, k}) - \Psi(\widetilde{\eta}_{\mathrm{ft}, k}) + \Psi(\mathbf{1}^{\top}_{p}\widetilde{\bm\eta}_{\mathrm{ft}})\right]\Bigg].
\end{aligned}
}
\label{eq:KL-analytic-grand}
\end{equation}
%

\subsection{Derivation of LMPSE}
\label{subsec:LMPSE-derivation}

From~\eqnref{eq:vasst-model} and~\eqnref{eq:OG-symbolic-tree-prior}, the joint posterior distribution induced over the tree structures $\{\mathsf T(f_j)\}_{j=1}^K$, coefficient vector $\bm \beta$, noise variance $\sigma^2$, and operator and feature weights $\mathbf{w}_{\mathrm{op}}$ and $\mathbf{w}_{\mathrm{ft}}$ is:
\begin{align}
\label{eq:actual-tree-joint-posterior}
\begin{split}
    &\Pi\left(\{\mathsf T(f_j)\}_{j=1}^K,\bm\beta,\sigma^2, \mathbf{w}_{\mathrm{op}}, \mathbf{w}_{\mathrm{ft}} \mid \mathbb{D}_n\right) \propto p(\mathbf{y}\mid \mathbf{T}, \bm\beta, \sigma^2)\Pi(\bm\beta, \sigma^2)\Pi(\mathbf{w}_{\mathrm{op}})
    \Pi(\mathbf{w}_{\mathrm{ft}})\prod_{j=1}^{K}\prod_{\zeta\in \mathsf{Z}(f_j)}\Pi(e_{j\zeta})\Pi(o_{j\zeta}\mid \mathbf{w}_{\mathrm{op}})\Pi(h_{j\zeta}\mid \mathbf{w}_{\mathrm{ft}}).
\end{split}   
\end{align}
The marginal posterior over $\{\mathsf T(f_j)\}_{j=1}^K$ is:
\begin{align}
\label{eq:actual-tree-marginal-posterior}
\begin{split}
&\Pi\left(\{\mathsf T(f_j)\}_{j=1}^K \mid \mathbb{D}_n\right) 
\propto \mathbb E_{\bm \beta, \sigma^2} \left[p(\mathbf{y}\mid \mathbf{T}, \bm\beta, \sigma^2) \right] \mathbb E_{\mathbf{w}_{\mathrm{op}}, \mathbf{w}_{\mathrm{ft}}} \left[\prod_{j=1}^{K}\prod_{\zeta\in \mathsf{Z}(f_j)}\Pi(e_{j\zeta})\Pi(o_{j\zeta}\mid \mathbf{w}_{\mathrm{op}})\Pi(h_{j\zeta}\mid \mathbf{w}_{\mathrm{ft}})\right]
\\ &
= p(\mathbf{y}\mid \mathbf{T}) \times 
\underbrace{\prod_{d=0}^{D} \left[ \mathrm{ps}(d)^{\mathrm{nonterm}(d)} (1-\mathrm{ps}(d))^{\mathrm{term}(d)} \right]}_{\prod_{j=1}^{K}\prod_{\zeta\in \mathsf{Z}(f_j)}\Pi(e_{j\zeta})} 
\\ & \qquad \qquad \qquad \qquad 
\bigintsss \bigintsss \underbrace{\prod_{\ell=1}^{|\mathbf O|} w_{\mathrm{op}, \ell}^{\gamma_{\ell}}}_{\prod_{j=1}^{K}\prod_{\zeta\in \mathsf{Z}(f_j)}\Pi(o_{j\zeta})} \times
\underbrace{\prod_{m=1}^{p} w_{\mathrm{ft}, m}^{\varrho_{m}}}_{\prod_{j=1}^{K}\prod_{\zeta\in \mathsf{Z}(f_j)}\Pi(o_{j\zeta})} \times
\underbrace{\frac{\prod_{\ell=1}^{|\mathbf O|} w_{\mathrm{op}, \ell}^{\eta_{\mathrm{op},\ell} - 1}}{\mathcal{B}(\bm \eta_{\mathrm{op}})}}_{\mathbf{w}_{\mathrm{op}} \sim \mathrm{Dir}(\bm\eta_{\mathrm{op}})} \times
\underbrace{\frac{\prod_{m=1}^{p} w_{\mathrm{ft}, m}^{\eta_{\mathrm{ft},m} - 1}}{\mathcal{B}(\bm \eta_{\mathrm{ft}})}}_{\mathbf{w}_{\mathrm{op}} \sim \mathrm{Dir}(\bm\eta_{\mathrm{op}})}
d\mathbf{w}_{\mathrm{op}}\; d\mathbf{w}_{\mathrm{ft}},
\end{split}   
\end{align}
where $\mathbf{w}_{\mathrm{op}} = (w_{\mathrm{op}, 1},\ldots,w_{\mathrm{op},|\mathbf O|})^\top$, $\mathbf{w}_{\mathrm{ft}} = (w_{\mathrm{ft}, 1},\ldots,w_{\mathrm{ft}, p})^\top$, $\bm \eta_{\mathrm{op}} = (\eta_{\mathrm{op}, 1},\ldots,\eta_{\mathrm{op},|\mathbf O|})^\top$, $\bm \eta_{\mathrm{ft}} = (\eta_{\mathrm{ft}, 1},\ldots,\eta_{\mathrm{ft}, p})^\top$, $\gamma_\ell$ and $\varrho_m$ denote the total number of occurrences of the $\ell$th operator and the $m$th feature among the $K$ symbolic trees $\left\{\mathsf T(f_j)\right\}_{j=1}^K$, $\mathrm{ps}(d) = \alpha(1 + d)^{-\delta}$ is the split probability at depth $d$, $\mathrm{term}(d)$ and $\mathrm{nonterm}(d)$ denote the total number of terminal and nonterminal nodes at depth $d$ among $\left\{\mathsf T(f_j)\right\}_{j=1}^K$. Define $\bm\gamma = (\gamma_1,\ldots,\gamma_{|\mathbf O|})^\top$ and $\bm\varrho = (\varrho_1,\dots,\varrho_p)^\top$. Then the marginal posterior from~\eqnref{eq:actual-tree-marginal-posterior} is: 
\begin{align}
\label{eq:actual-tree-marginal-posterior-2}
\begin{split}
\Pi\left(\{\mathsf T(f_j)\}_{j=1}^K \mid \mathbb{D}_n\right) 
\propto p(\mathbf{y}\mid \mathbf{T}) \times \prod_{d=0}^{D} \left[ \mathrm{ps}(d)^{\mathrm{nonterm}(d)} (1-\mathrm{ps}(d))^{\mathrm{term}(d)} \right]  \frac{\mathcal{B}(\bm \eta_{\mathrm{op}} + \bm \gamma)}{\mathcal{B}(\bm \eta_{\mathrm{op}})} \frac{\mathcal{B}(\bm \eta_{\mathrm{ft}} + \bm\varrho)}{\mathcal{B}(\bm \eta_{\mathrm{ft}})},
\end{split}   
\end{align}
Therefore the $\mathrm{LMPSE}$ is:
\begin{equation*}
\boxed{
\begin{aligned}
&\mathrm{LMPSE}\left(\{\mathsf T(f_j)\}_{j=1}^K\right) 
= \log \Pi\left(\{\mathsf T(f_j)\}_{j=1}^K \mid \mathbb{D}_n\right) \\
&\qquad = \log p(\mathbf{y}\mid \mathbf{T}) + \sum_{d=0}^{D} \bigg[\mathrm{nonterm}(d) \log  \mathrm{ps}(d) + {\mathrm{term}(d)} \log (1-\mathrm{ps}(d)) \bigg] + \log \frac{\mathcal{B}(\bm \eta_{\mathrm{op}} + \bm \gamma)}{\mathcal{B}(\bm \eta_{\mathrm{op}})} + \log \frac{\mathcal{B}(\bm \eta_{\mathrm{ft}} + \bm\varrho)}{\mathcal{B}(\bm \eta_{\mathrm{ft}})}.
\end{aligned}   
}
\end{equation*}

\newpage

\section{Soft Evaluation Algorithms}
\label{app:soft-evaluation-algorithms}

\subsection{Soft Evaluation at a Node}

\begin{algorithm}[H]
\caption{The \textsc{SoftEvalAtNode}$(\zeta',\bm{x}_i, j, \mathbf{O}, \{\widetilde{e}_{j\zeta}, \widetilde{\mathbf{o}}_{j\zeta}, \widetilde{\mathbf{h}}_{j\zeta}:\zeta\in \mathsf{Z}_D\})$ algorithm for evaluating a soft symbolic subtree rooted at a given node $\zeta'$.}
\label{alg:soft-eval-at-node}
\DontPrintSemicolon
\KwIn{
A given node: $\zeta'$; feature vector: $\bm{x}_i\in \mathbb{R}^{p}$; soft symbolic tree index: $j\in \{1, \ldots, K\}$; the operator set: $\mathbf{O} = \mathbf{O}_u\cup \mathbf{O}_b$; soft relaxations: $\{\widetilde{e}_{j\zeta}, \widetilde{\mathbf{o}}_{j\zeta}, \widetilde{\mathbf{h}}_{j\zeta}:\zeta\in \mathsf{Z}_D\}$.\;
}
\KwOut{$\mathsf{S}_j^{\mathrm{soft}}(\bm{x}_i; \zeta')\in \mathbb{R}$.\;}

\BlankLine
\Comment{\textcolor{Blue}{Soft feature evaluation}}
$\mathrm{Term}(\bm{x}_i, j, \zeta') \leftarrow \widetilde{\mathbf{h}}^{\top}_{j\zeta'}\bm{x}_i$\;

\BlankLine
\Comment{\textcolor{Blue}{Soft operator evaluation}}
\tcp{$L(\zeta')$ and $R(\zeta')$ are the left and right children of $\zeta'$, respectively}
$\mathrm{NontermUnary}(\bm{x}_i, j, \zeta') \leftarrow \sum_{u\in \mathbf{O}_u}\widetilde{o}_{j\zeta', u}\cdot u(\mathsf{S}_j^{\mathrm{soft}}(\bm{x}_i, L(\zeta')))$\;
$\mathrm{NontermBinary}(\bm{x}_i, j, \zeta') \leftarrow \sum_{b\in \mathbf{O}_b}\widetilde{o}_{j\zeta', b}\cdot b(\mathsf{S}_j^{\mathrm{soft}}(\bm{x}_i; L(\zeta')),\; \mathsf{S}_j^{\mathrm{soft}}(\bm{x}_i; R(\zeta')))$\;
$\mathrm{Nonterm}(\bm{x}_i, j, \zeta') \leftarrow \mathrm{NontermUnary}(\bm{x}_i, j, \zeta') + \mathrm{NontermBinary}(\bm{x}_i, j, \zeta')$\;

\BlankLine
\Comment{\textcolor{Blue}{Complete soft evaluation at $\zeta'$}}
$\mathsf{S}_j^{\mathrm{soft}}(\bm{x}_i; \zeta') \leftarrow \widetilde{e}_{j\zeta'}\mathrm{Nonterm}(\bm{x}_i, j, \zeta') + (1 - \widetilde{e}_{j\zeta'})\mathrm{Term}(\bm{x}_i, j, \zeta')$\;

\BlankLine
\Return{$\mathsf{S}_j^{\mathrm{soft}}(\bm{x}_i; \zeta')$}\;
\end{algorithm}

\subsection{Complete Soft Evaluation}
\begin{algorithm}[H]
\caption{The \textsc{SoftEval}$(\{\bm{x}_i\in \mathbb{R}^{p}\}_{i=1}^{n},\mathbf{O},\{\widetilde{e}_{j\zeta}, \widetilde{\mathbf{o}}_{j\zeta}, \widetilde{\mathbf{h}}_{j\zeta}:\zeta\in \mathsf{Z}_D\}_{j=1}^{K})$ algorithm for evaluating an ensemble of $K$ soft symbolic trees.}
\label{alg:soft-eval}
\DontPrintSemicolon
\KwIn{
Feature vector instances: $\{\bm{x}_i\in \mathbb{R}^{p}\}_{i=1}^{n}$; the operator set: $\mathbf{O} = \mathbf{O}_u\cup \mathbf{O}_b$; soft relaxations: $\{\widetilde{e}_{j\zeta}, \widetilde{\mathbf{o}}_{j\zeta}, \widetilde{\mathbf{h}}_{j\zeta}:\zeta\in \mathsf{Z}_D\}_{j=1}^{K}$.\;
}
\KwOut{$\mathbf{T}_{\mathrm{soft}}\in \mathbb{R}^{n\times \overline{K+1}}$.\;}

\BlankLine
\For{$i \leftarrow 1$ \KwTo $n$}{
    \Comment{\textcolor{Blue}{Intercept term set to $1$}}
    $t_{i0} \leftarrow 1$\;
    \For{$j\leftarrow 1$ \KwTo $K$}{
        \Comment{\textcolor{Blue}{$j$th soft symbolic tree evaluated at root node for $\bm{x}_i$}}
        $t_{ij} = \mathsf{S}_{j}^{\mathrm{soft}}(\bm{x}_i; \zeta'=0)$ from \hyperref[alg:soft-eval-at-node]{Algorithm~\ref{alg:soft-eval-at-node}: \textsc{SoftEvalAtNode}}\;
    }
}

\BlankLine
$\mathbf{T}_{\mathrm{soft}}\leftarrow (t_{ij})_{1\leq i \leq n,\; 0\leq j \leq K}$\;

\BlankLine
\Return{$\mathbf{T}_{\mathrm{soft}}$\;}
\end{algorithm}

\newpage

\section{Algorithm for Approximation of \texorpdfstring{$\mathcal{E}(\phi)$}{Ephi}}
\label{app:approx-elbo}
\begin{algorithm}[H]
\caption{The \textsc{ApproxELBO}$(\mathbb{D}_n, (\bm\mu_0, \bm\Sigma_0, a_0, b_0, \bm\eta_\mathrm{op}, \bm\eta_{\mathrm{ft}}, \alpha, \delta), \mathbf{O}, \phi, \bm \tau = (\tau_{\mathrm{ex}}, \tau_{\mathrm{op}}, \tau_{\mathrm{ft}}), S)$ algorithm for computing Monte Carlo approximation of $\mathcal{E}(\phi)$.}
\label{alg:approx-elbo}
\DontPrintSemicolon
\KwIn{
Data: $\mathbb{D}_n = \{(\bm{x}_i, y_i)\}_{i=1}^{n}$; full set of prior hyperparameters: $(\bm\mu_0, \bm\Sigma_0, a_0, b_0, \bm\eta_\mathrm{op}, \bm\eta_{\mathrm{ft}}, \alpha, \delta)$; the operator set: $\mathbf{O} = \mathbf{O}_u\cup \mathbf{O}_b$; the variational parameter vector: $\phi$; the temperature parameters: $\bm \tau = (\tau_\mathrm{ex}, \tau_\mathrm{op}, \tau_{\mathrm{ft}})$; the Monte Carlo sample size: $S$.\;
}
\KwOut{$\widehat{\mathcal{E}}(\phi)$.\;}

\BlankLine
\Comment{\textcolor{Blue}{{Compute the Kullback-Leibler term}}}
$\mathsf{K} \leftarrow \mathrm{KL}(q_{\phi}(\Theta)\parallel \Pi(\Theta))$, using~\eqnref{eq:elbo-2}\;

\BlankLine
\For{$s \leftarrow 1$ \KwTo $S$}{
    \Comment{\textcolor{Blue}{Sample soft tree structural parameters from $\mathcal{P}_{\phi, \bm \tau}^{\mathrm{soft}}$}}
    $\{\widetilde{e}_{j\zeta}^{(s)}, \widetilde{\mathbf{o}}_{j\zeta}^{(s)}, \widetilde{\mathbf{h}}_{j\zeta}^{(s)}:\zeta\in \mathsf{Z}_D\}_{j=1}^{K}\sim \mathcal{P}^{\mathrm{soft}}_{\phi, \bm \tau}$, using the temperature parameters\;
    \Comment{\textcolor{Blue}{Computing soft symbolic design matrix}}
    $\mathbf{T}^{(s)}_{\mathrm{soft}} \leftarrow$ \textsc{SoftEval}$(\{\bm{x}_i\in \mathbb{R}^{p}\}_{i=1}^{n},\mathbf{O},\{\widetilde{e}_{j\zeta}^{(s)}, \widetilde{\mathbf{o}}_{j\zeta}^{(s)}, \widetilde{\mathbf{h}}_{j\zeta}^{(s)}:\zeta\in \mathsf{Z}_D\}_{j=1}^{K})$\;
    \Comment{\textcolor{Blue}{Computing the marginal likelihood at $\mathbf{T}_{\mathrm{soft}}^{(s)}$}}
    $\mathrm{ML}^{(s)} \leftarrow \log p(\mathbf{y}\mid \mathbf{T}_{\mathrm{soft}}^{(s)})$, using~\eqnref{eq:p(y|T)}\;
}

\BlankLine
$\widehat{\mathcal{E}}(\phi) \leftarrow \tfrac{1}{S}\sum_{s=1}^{S}\mathrm{ML}^{(s)} - \mathsf{K}$\;

\BlankLine
\Return{$\widehat{\mathcal{E}}(\phi)$\;}
\end{algorithm}
\newpage

\section{The \texorpdfstring{\vasst}{VaSST} Algorithm}
\label{app:vasst-algorithm}
\begin{algorithm}[H]
\caption{The \vasstmain$(\mathbb{D}_n, (\bm\mu_0, \bm\Sigma_0, a_0, b_0, \bm\eta_{\mathrm{op}}, \bm\eta_{\mathrm{ft}}, \alpha, \delta), \mathbf{O}, \phi^{\mathrm{init}}, T, (\tau_{\mathrm{start}}, \tau_{\mathrm{end}}, T_{\tau}), S, \gamma, c)$ algorithm for \sr\ using automatic differentiation-based black-box variational inference.}
\label{alg:vasst}
\DontPrintSemicolon
\KwIn{
Data: $\mathbb{D}_n = \{(\bm{x}_i,y_i)\}_{i=1}^{n}$; full set of prior hyperparameters: $(\bm\mu_0, \bm\Sigma_0, a_0, b_0, \bm\eta_{\mathrm{op}}, \bm\eta_{\mathrm{ft}}, \alpha, \delta)$; the operator set: $\mathbf{O}=\mathbf{O}_u\cup\mathbf{O}_b$; initial variational parameter vector: $\phi^{\mathrm{init}}$; number of steps: $T$; temperature schedule: $(\tau_{\mathrm{start}}, \tau_{\mathrm{end}}, T_{\tau})$; Monte Carlo sample size: $S$; learning rate: $\gamma$; gradient clipping threshold: $c$.
}
\KwOut{Optimized variational parameter vector $\phi^{\star}$.}

\BlankLine
Set $\phi_0 \leftarrow \phi^{\mathrm{init}}$\;
Initialize \texttt{AdamW} optimizer with learning rate $\gamma$\;

\For{$t \leftarrow 1$ \KwTo $T$}{
    \Comment{\textcolor{Blue}{Anneal temperature parameters}}
    $\tau_t \leftarrow \tau_{\mathrm{start}} + (\tau_{\mathrm{end}} - \tau_{\mathrm{start}})\cdot\min\{\tfrac{t}{T_{\tau}}, 1\}$\;
    Set $(\tau_{\mathrm{ex}},\tau_{\mathrm{op}},\tau_{\mathrm{ft}})\leftarrow(\tau_t,\tau_t,\tau_t)$\;

    Set all gradients to zero\;

    \Comment{\textcolor{Blue}{Compute stochastic ELBO estimate}}
    $\widehat{\mathcal{E}}(\phi) \leftarrow \textsc{ApproxELBO}$$(\mathbb{D}_n, (\bm\mu_0, \bm\Sigma_0, a_0, b_0, \bm\eta_\mathrm{op}, \bm\eta_{\mathrm{ft}}, \alpha, \delta), \mathbf{O}, \phi, (\tau_{\mathrm{ex}}, \tau_{\mathrm{op}}, \tau_{\mathrm{ft}}), S)$\;

    \Comment{\textcolor{Blue}{minimize negative objective}}
    $\mathsf J \leftarrow -\widehat{\mathcal{E}}(\phi)$

    \If{$\mathsf J$ is not finite}{
        \textbf{continue}\;
    }

    Backpropagate: compute $\nabla_{\phi}\mathsf J$\;

    \If{any component of $\nabla_{\phi}\mathsf J$ is not finite}{
        Set all gradients to zero; \textbf{continue}\;
    }

    Clip gradients: $\|\nabla_{\phi}\mathsf J\| \le c$\;

    $\phi_t \leftarrow$ \texttt{AdamW} update step on $\phi$\;
}
Set $\phi^{\star} \leftarrow \phi_T$\;
\Return{$\phi^{\star}$}\;
\end{algorithm}
\newpage

\section{Algorithm for Sampling Hard Symbolic Trees}
\label{app:sample-hard-symbolic-trees}
\begin{algorithm}[H]
\caption{The \textsc{SampleHard}$(\mathbb{D}_n, (\bm\mu_0, \bm\Sigma_0, a_0, b_0, \bm\eta_{\mathrm{op}}, \bm\eta_{\mathrm{ft}}, \alpha, \delta),\phi^{\star}, \mathbf{O}, H)$ algorithm to draw samples of hard symbolic trees from learned soft structural representations.}
\label{alg:sample-hard-trees}
\DontPrintSemicolon
\KwIn{
Data: $\mathbb{D}_n=\{(\bm{x}_i, y_i)\}_{i=1}^{n}$; full set of prior hyperparameters: $(\bm\mu_0, \bm\Sigma_0, a_0, b_0, \bm\eta_{\mathrm{op}}, \bm\eta_{\mathrm{ft}}, \alpha, \delta)$; optimizied variational parameter vector: $\phi^{\star}$; the operator set $\mathbf{O} = \mathbf{O}_u\cup \mathbf{O}_b$; the number of hard symbolic tree samples: H.
}
\KwOut{
A collection of $H$ sampled hard symbolic ensembles and their corresponding posterior mean estimates of the model regression parameters: $\{\{\mathsf{T}(f_j^{(s)})\}_{j=1}^{K}, \bm\beta^{(s)}_{\mathrm{PM}}, (\sigma^2)^{(s)}_{\mathrm{PM}}:s=1,\ldots,H\}$.
}

\BlankLine
\For{$s \leftarrow 1$ \KwTo $H$}{
    \Comment{\textcolor{Blue}{Sample one ensemble of $K$ hard symbolic trees}}
    \For{$j \leftarrow 1$ \KwTo $K$}{
        \Comment{\textcolor{Blue}{(i) sample expansion indicators}}
        $\widehat e^{(s)}_{j\zeta} \sim \mathrm{Ber}(\sigma(\ell^{\star}_{j\zeta}))$ for all $\zeta\in\mathsf{Z}_D$\;
        Enforce leaf constraint: $\widehat e_{j\zeta}^{(s)}\leftarrow 0$ for all $\zeta$ at depth $D$\;

        \Comment{\textcolor{Blue}{(ii) sample operator and feature labels}}
        $\widehat o^{(s)}_{j\zeta} \sim \mathrm{Cat}(\texttt{smax}((\mathbf a^{\mathrm{op}}_{j\zeta})^{\star}))$\;
        $\widehat h^{(s)}_{j\zeta} \sim \mathrm{Cat}(\texttt{smax}((\mathbf a^{\mathrm{ft}}_{j\zeta})^{\star}))$\;

        \Comment{\textcolor{Blue}{(iii) Construct hard symbolic tree skeleton}}
        $(\widehat{\mathsf{S}}_{j}^{(D)})^{(s)} \leftarrow \{\widehat{e}_{j\zeta}^{(s)}, \widehat{o}_{j\zeta}^{(s)}, \widehat{h}_{j\zeta}^{(s)}:\zeta\in \mathsf{Z}_D\}$
        
        \Comment{\textcolor{Blue}{(iv) symbolic trees via deterministic pruning of skeleton}}
        $\mathsf{T}(f_j^{(s)}) \leftarrow \mathfrak{p}\left((\widehat{\mathsf{S}}_{j}^{(D)})^{(s)}\right)$\;

        \For{$i \leftarrow 1$ \KwTo $n$}{
        \Comment{\textcolor{Blue}{Evaluate at sample $\bm{x}_i$}}
        $t_{ij}^{(s)} \leftarrow g(\bm{x}_i; \mathsf{T}(f_j^{(s)}))$
        }
    }
    \Comment{\textcolor{Blue}{Construct expression design matrix}}
    $t_{i0}^{(s)} \leftarrow 1$ for all $i =1,\ldots,n$\;
    $\mathbf{T}^{(s)} \leftarrow ((t_{ij}^{(s)}))_{1\leq i \leq n, 0\leq j \leq K}$\;

    \Comment{\textcolor{Blue}{Compute posterior means of model regression parameters}}
    $\bm\beta^{(s)}_{\mathrm{PM}} \leftarrow \mathbb{E}[\bm\beta\mid \sigma^{2}, \mathbf{T}^{(s)}, \mathbf{y}]$ and $(\sigma^2)^{(s)}_{\mathrm{PM}} \leftarrow \mathbb{E}[\sigma^2\mid \mathbf{T}^{(s)}, \mathbf{y}]$, using~\eqnref{eq:app-model-posterior}\;
}
\Return{$\{\{\mathsf{T}(f_j^{(s)})\}_{j=1}^{K},{\bm{\beta}}^{(s)}_{\mathrm{PM}},({\sigma}^{2})^{(s)}_{\mathrm{PM}}:s=1,\ldots,H\}$}\;
\end{algorithm}
\newpage

\section{Discussion on the Gap induced by \texttt{ELBO} Approximation}
\label{app:elbo-gap}

As discussed in \hyperref[sec:VI-vasst]{Section~\ref{sec:VI-vasst}}, the exact discrete \texttt{ELBO} in~\eqnref{eq:elbo-1} for \vasstmain\ is:
\begin{align}
\label{eq:elbo-supp-exact}
\mathcal{E}(\phi)
=
\mathbb{E}_{q}[\log p(\mathbf y\mid \mathbf T)]
-
\mathrm{KL}[q(\Theta)\parallel \Pi(\Theta)],
\end{align}
where $\mathbf T$ is the hard symbolic design matrix induced by discrete expansion, operator, and feature assignments. Direct optimization of~\eqnref{eq:elbo-supp-exact}, however, is computationally prohibitive because the expectation is over a combinatorial space of symbolic tree structures. Consequently, \vasstmain\ replaces the hard symbolic trees by temperature-controlled soft symbolic trees. Specifically, the Bernoulli expansion indicators $e_{j\zeta}$ are relaxed using Binary-Concrete random variables $\widetilde e_{j\zeta}$, while the categorical operator and feature assignments are relaxed using Gumbel-Softmax random vectors $\widetilde{\mathbf o}_{j\zeta}$ and $\widetilde{\mathbf h}_{j\zeta}$, respectively. These relaxations induce a distribution over soft symbolic tree structures, denoted by $\mathcal P^{\mathrm{soft}}_{\phi,\bm\tau}$, where $\bm\tau=(\tau_{\mathrm{ex}},\tau_{\mathrm{op}},\tau_{\mathrm{ft}})$ is the vector of relaxation temperatures. The corresponding temparture-controlled relaxed \texttt{ELBO} objective is therefore:
\begin{align}
\label{eq:elbo-supp-1}
\mathcal{E}_{\bm \tau}(\phi) = \mathbb E_{\mathcal P_{\phi, \bm \tau}^{\mathrm{soft}}} [\log p(\mathbf y \mid \mathbf T_{\mathrm{soft}})] - \mathrm{KL}[q(\Theta) \parallel \Pi(\Theta)],
\end{align}
where $\mathbf T_{\mathrm{soft}}$ is an instance of the soft symbolic design matrix induced by the recursive soft symbolic tree evaluation in~\eqnref{eq:soft-gating}. However, the expectation in~\eqnref{eq:elbo-supp-1} is still analytically intractable because $\log p(\mathbf y\mid \mathbf T_{\mathrm{soft}})$ depends nonlinearly on the recursively evaluated soft symbolic design matrix. Therefore, \vasstmain\ uses a Monte Carlo (\texttt{MC}) approximation. Given $S$ \texttt{MC} samples from $\mathcal P^{\mathrm{soft}}_{\phi,\bm\tau}$, let $\mathbf T_{\mathrm{soft}}^{(s)}$ denote the corresponding soft symbolic design matrix for $s=1,\ldots,S$. The stochastic objective optimized by \vasstmain\ is:
\begin{align}
\label{eq:elbo-supp-2}
\widehat{\mathcal{E}}(\phi) = \widehat{\mathcal{E}}_{\bm \tau,S}(\phi)
=
\frac{1}{S}\sum_{s=1}^S
\log p(\mathbf y\mid \mathbf T_{\mathrm{soft}}^{(s)})
-
\mathrm{KL}[q(\Theta)\parallel\Pi(\Theta)],
\end{align}
It is important to note that, $\widehat{\mathcal{E}}(\phi)$ in~\eqnref{eq:elbo-supp-2}
is an unbiased \texttt{MC} estimator of $\mathcal{E}_{\bm \tau}(\phi)$ in~\eqnref{eq:elbo-supp-1}, not of the exact discrete \texttt{ELBO} $\mathcal{E}(\phi)$ in \eqnref{eq:elbo-supp-exact}. This distinction leads to following decomposition of the induced approximation gap:
\begin{align}
\label{eq:elbo-supp-3}
\mathcal{E}(\phi)-\widehat{\mathcal{E}}_{\bm \tau,S}(\phi)
=
\underbrace{\mathcal{E}(\phi)-\mathcal{E}_{\bm\tau}(\phi)}_{\text{relaxation bias}}
+
\underbrace{\mathcal{E}_{\bm \tau}(\phi)-\widehat{\mathcal{E}}_{\bm \tau,S}(\phi)}_{\text{Monte Carlo (\texttt{MC}) error}}.
\end{align}
Since the same analytic Kullback-Leibler ($\mathrm{KL}$) penalty is used in both objectives, the relaxation bias arises entirely from replacing the hard tree expected marginal likelihood by the soft tree expected marginal likelihood:
\begin{align}
\label{eq:elbo-supp-4}
\mathcal{E}(\phi)-\mathcal{E}_{\bm\tau}(\phi)
=
\mathbb{E}_{q}[\log p(\mathbf y\mid \mathbf T)]
-
\mathbb{E}_{\mathcal P^{\mathrm{soft}}_{\phi,\bm\tau}}
[\log p(\mathbf y\mid \mathbf T_{\mathrm{soft}})],
\end{align}
which is non-zero in general since soft symbolic evaluation and marginal likelihood evaluation do not commute. This non-commutativity is especially important in \vasstmain, since $\mathbf T_{\mathrm{soft}}$ is produced by recursive mixtures of nonlinear symbolic operators, while $\log p(\mathbf y\mid \mathbf T_{\mathrm{soft}})$ depends on $\mathbf T_{\mathrm{soft}}$ through the collapsed Normal Inverse-Gamma marginal likelihood in~\eqnref{eq:p(y|T)}. Hence, at finite temperature, the relaxed objective targets a smoothed symbolic model that is close in spirit, but not identical to, the original discrete symbolic tree model.

The second term in~\eqnref{eq:elbo-supp-3} is standard stochastic approximation error. For fixed $(\phi,\bm\tau)$, it has mean zero. If:
\begin{align*}
V_{\phi,\bm\tau}
=
\mathrm{Var}_{\mathcal P^{\mathrm{soft}}_{\phi,\bm\tau}}
[\log p(\mathbf y\mid \mathbf T_{\mathrm{soft}})]
<
\infty,
\end{align*}
then $\mathrm{Var}_{\mathcal P^{\mathrm{soft}}_{\phi,\bm\tau}}[\widehat{\mathcal{E}}_{\bm\tau,S}(\phi)]= S^{-1}V_{\phi,\bm\tau}$, because the $\mathrm{KL}$ term is deterministic for fixed $\phi$. Thus, increasing $S$ reduces the \texttt{MC} component of the gap at the usual $S^{-1/2}$ rate, whereas it does not remove the relaxation bias. The relaxation bias is controlled instead through the annealing schedule: as the temperatures decrease, the Binary-Concrete and Gumbel-Softmax variables become increasingly concentrated near Bernoulli and Categorical draws. Under numerically stable symbolic evaluation, this progressively sharpens the soft symbolic trees toward hard symbolic structures.

In light of these observations, this approximation is best interpreted as an algorithmic device for scalable posterior learning over symbolic structures. Direct optimization of the discrete objective in~\eqnref{eq:elbo-supp-exact} would require either score-function estimators over symbolic structures or explicit combinatorial search over an exponentially large expression space. Such approaches may be high-variance or computationally prohibitive in the \sr\ setting. The relaxed objective in~\eqnref{eq:elbo-supp-1}, together with the \texttt{MC} approximation in~\eqnref{eq:elbo-supp-2}, provides a differentiable reparameterized surrogate with low-variance gradients, enabling black-box gradient-based optimization.

Finally, the distinction between the relaxed optimization objective and the original discrete symbolic model is also reflected in the post-optimization inference step of \vasstmain, where the optimized soft representations are not reported as the final outcomes. Instead, post-optimization, hard symbolic ensembles are sampled from the learned variational distributions and then evaluated using the discrete symbolic design matrix $\mathbf T$. These hard symbolic ensembles are ranked by the posterior-aware $\mathrm{LMPSE}$ criterion in~\eqnref{eq:LMPSE}, which combines the collapsed marginal likelihood under the hard design matrix with the symbolic tree prior. Thus, the relaxation is used to make posterior exploration computationally tractable, while the final reported symbolic summaries are obtained by returning to the hard symbolic tree space.
\newpage

\section{Configurations of \texorpdfstring{\vasst}{VaSST} and Competing Methods}
\label{app:configs}

\subsection{Operator Set}
\label{subsec:operator-set}

Following are the operator sets tabulated in~\hyperref[tab:operator-sets-by-method-equation]{Table~\ref{tab:operator-sets-by-method-equation}}, used across all experiments in~\hyperref[subsec:simulation-study]{Sections~\ref{subsec:simulation-study}} and \hyperref[subsec:feynman-data-study]{\ref{subsec:feynman-data-study}}. For budget equalization, we try to maintain a common operator specification across \vasstmain\ and competing \sr\ methods.

\begin{table*}[!htp]
\centering
\scriptsize
\renewcommand{\arraystretch}{1.15}
\caption{Operator set $\mathbf{O}$ used across all experiments in~\hyperref[subsec:simulation-study]{Sections~\ref{subsec:simulation-study}} and \hyperref[subsec:feynman-data-study]{\ref{subsec:feynman-data-study}}.}
\label{tab:operator-sets-by-method-equation}
\begin{tabularx}{\textwidth}{@{}>{\raggedright\arraybackslash}p{0.13\textwidth}
>{\raggedright\arraybackslash}X
>{\raggedright\arraybackslash}X
>{\raggedright\arraybackslash}X
>{\raggedright\arraybackslash}X
>{\raggedright\arraybackslash}X
>{\raggedright\arraybackslash}X
>{\raggedright\arraybackslash}X@{}}
\toprule
\toprule
\textbf{Method} &
\eqnref{eq:simulation-generating-model-1} &
\eqnref{eq:simulation-generating-model-2} &
\eqnref{eq:feynman-cl} &
\eqnref{eq:feynman-cpe} &
\eqnref{eq:feynman-fce} &
\eqnref{eq:feynman-ftc} &
\eqnref{eq:feynman-ada} \\
\midrule

\vasstmain
& $\{+, -,\times, /,^2\}$
& $\{\times, \sin, \cos\}$
& $\{\times, /, ^2\}$
& $\{+,-,\times,/\}$
& $\{+, \times, /,\sin\}$
& $\{+,-,\times,/,\sqrt{}\}$ 
& $\{+, \times, \cos\}$ \\

\addlinespace[2pt]

\operon
& $\{+, -,\times, /,^2\}$
& $\{\times, \sin, \cos\}$
& $\{\times, /, ^2\}$
& $\{+,-,\times,/\}$
& $\{+, \times, /,\sin\}$
& $\{+,-,\times,/,\sqrt{}\}$ 
& $\{+, \times, \cos\}$ \\

\addlinespace[2pt]

\bms
& $\{+,\times, /, \texttt{neg},^2\}$
& $\{\times, \sin, \cos\}$
& $\{\times, /, ^2\}$
& $\{+,\texttt{neg},\times,/\}$
& $\{+, \times, /,\sin\}$
& $\{+,\texttt{neg},\times,/,\sqrt{}\}$ 
& $\{+, \times, \cos\}$ \\

\addlinespace[2pt]

\bsr
& $\{+, -,\times, \texttt{inv},^2\}$
& $\{\times, \sin, \cos\}$
& $\{\times, \texttt{inv}, ^2\}$
& $\{+,-,\times,\texttt{inv}\}$
& $\{+, \times, \texttt{inv},\sin\}$
& $\{+,-,\times,\texttt{inv},\sqrt{}\}$ 
& $\{+, \times, \cos\}$ \\

\addlinespace[2pt]

\qlattice
& $\{+, \times, ^2, \texttt{lin}, \texttt{inv}\}$ 
& \texttt{default} 
& \texttt{default} & \texttt{default} & \texttt{default} & \texttt{default} & \texttt{default} \\

\addlinespace[2pt]

\pysr
& $\{+, -,\times, /,^2\}$
& $\{\times, \sin, \cos\}$
& $\{\times, /, ^2\}$
& $\{+,-,\times,/\}$
& $\{+, \times, /,\sin\}$
& $\{+,-,\times,/,\sqrt{}\}$ 
& $\{+, \times, \cos\}$ \\

\addlinespace[2pt]

\gplearn
& $\{+, -,\times, /,^2\}$
& $\{\times, \sin, \cos\}$
& $\{\times, /, ^2\}$
& $\{+,-,\times,/\}$
& $\{+, \times, /,\sin\}$
& $\{+,-,\times,/,\sqrt{}\}$ 
& $\{+, \times, \cos\}$ \\

\addlinespace[2pt]

\deap
& $\{+, -,\times, /,^2\}$
& $\{\times, \sin, \cos\}$
& $\{\times, /, ^2\}$
& $\{+,-,\times,/\}$
& $\{+, \times, /,\sin\}$
& $\{+,-,\times,/,\sqrt{}\}$ 
& $\{+, \times, \cos\}$ \\

\addlinespace[2pt]

\dsr
& $\{+, -,\times, /,\texttt{n2}\}$
& $\{\times, \sin, \cos\}$
& $\{\times, /, \texttt{n2}\}$
& $\{+,-,\times,/\}$
& $\{+, \times, /,\sin\}$
& $\{+,-,\times,/,\sqrt{}\}$ 
& $\{+, \times, \cos\}$ \\

\bottomrule
\bottomrule
\end{tabularx}

\vspace{2pt}
\parbox{\textwidth}{
\scriptsize
\emph{Note}: \operon\ (\url{https://github.com/heal-research/pyoperon}) and \deap\ (\url{https://github.com/DEAP/deap}) additionally uses variables and ephemeral constants. The operators $\texttt{neg}$ denotes the unary negation operator, $\texttt{neg}(z) = -z$; $\texttt{inv}$ denotes the unary reciprocal operator, $\texttt{inv}(z) = \tfrac{1}{z}$; $\texttt{lin}$ denotes the linear operator, $\texttt{lin}(z) = az + b$; and $\texttt{n2}$ denotes the unary square operator, $\texttt{n2}(z) = z^{2}$. For \qlattice\ (\url{https://docs.abzu.ai/}), \texttt{default} indicates the use of default operator set.
}
\end{table*}

\subsection{Experimental Settings of \texorpdfstring{\vasst}{VaSST}}
\label{subsec:vasst-configs}

Across all experiments in~\hyperref[subsec:simulation-study]{Sections~\ref{subsec:simulation-study}} and \hyperref[subsec:feynman-data-study]{\ref{subsec:feynman-data-study}}, \vasstmain\ was implemented using the common experimental settings tabulated in~\hyperref[tab:vasst-common-hyperparameters]{Table~\ref{tab:vasst-common-hyperparameters}}.

\begingroup
\scriptsize
\renewcommand{\arraystretch}{1.15}

\begin{xltabular}{\linewidth}{@{}>{\raggedright\arraybackslash}X
>{\raggedright\arraybackslash}X
>{\raggedright\arraybackslash}X@{}}
\caption{Common experimental settings of \vasstmain\ used in~\hyperref[subsec:simulation-study]{Sections~\ref{subsec:simulation-study}} and \hyperref[subsec:feynman-data-study]{\ref{subsec:feynman-data-study}}.}
\label{tab:vasst-common-hyperparameters}\\

\toprule
\toprule
\textbf{Component} & \textbf{Hyperparameter} & \textbf{Value} \\
\midrule
\endfirsthead

\multicolumn{3}{@{}l}{\scriptsize\emph{Table~\thetable\ continued from previous page.}}\\
\toprule
\toprule
\textbf{Component} & \textbf{Hyperparameter} & \textbf{Value} \\
\midrule
\endhead

\midrule
\multicolumn{3}{r@{}}{\scriptsize\emph{Continued on next page}}\\
\endfoot

\bottomrule
\bottomrule
\addlinespace[2pt]
\multicolumn{3}{@{}p{\linewidth}@{}}{
\scriptsize
\emph{Note}:
The Binary-Concrete and Gumbel-Softmax temperature parameters $\tau_{\mathrm{ex}}$, $\tau_{\mathrm{op}}$, and $\tau_{\mathrm{ft}}$, respectively, are linearly annealed from $\tau_{\mathrm{start}}=1.0$ to $\tau_{\mathrm{end}}=0.5$ over $T_{\tau}=1500$ steps. Out-of-sample \texttt{RMSE} is computed on a $10\%$ held-out test set with the LMPSE selected symbolic model.
}\\
\endlastfoot

Symbolic ensemble
& Number of trees, $K$
& $3$ \\

& Maximum tree depth, $D$
& $3$ \\

\midrule

Symbolic tree prior
& Split probability parameters, $(\alpha, \delta)$
& $(0.95, 2.00)$ \\

& Operator Dirichlet parameter, $\bm\eta_{\mathrm{op}}$
& $\mathbf 1_{|\mathbf O|}$ \\

& Feature Dirichlet parameter, $\bm\eta_{\mathrm{ft}}$
& $\mathbf 1_p$ \\

\midrule

Outer regression coefficients
& Prior mean, $\bm\mu_0$
& $\bm 0_{K+1}$ \\

& Prior covariance, $\bm\Sigma_0$
& $10\mathbf I_{K+1}$ \\

\midrule

Model noise variance
& Inverse-Gamma shape, $a_0$
& $2$ \\

& Inverse-Gamma scale, $b_0$
& $2$ \\

\midrule

Variational optimization
& Optimizer
& \texttt{AdamW} \\

& Learning rate, $\gamma$
& $5\times 10^{-5}$ \\

& \texttt{AdamW} momentum parameters
& $(0.90,0.99)$ \\

& \texttt{AdamW} weight decay
& $0$ \\

& Number of optimization steps, $T$
& $2000$ \\

& Monte Carlo samples per \texttt{ELBO} step, $S$
& $8$ \\

& Gradient clipping threshold, $c$
& $1.0$ \\

\midrule

Temperature annealing
& Initial temperature, $\tau_{\mathrm{start}}$
& $1.0$ \\

& Final temperature, $\tau_{\mathrm{end}}$
& $0.5$ \\

& Annealing horizon, $T_\tau$
& $1500$ steps \\

\midrule

Numerical stabilization
& Soft symbolic tree output clipping threshold
& $\texttt{value\_clip}=200$ \\

& Hyperbolic tangent clipping
& $\texttt{use\_tanh\_clip}=\texttt{True}$ \\

& Variational logit clipping threshold
& $\texttt{logits\_clip}=10.0$ \\

\midrule

Hard symbolic ensembles
& Number of sampled hard symbolic ensembles, $H$
& $2000$ \\

& Ranking criterion
& Using $\mathrm{LMPSE}$ in~\eqnref{eq:LMPSE} \\

& Number of top ranked symbolic expressions reported, $\mathsf{r}$
& $10$ \\

\midrule

Evaluation
& Training split
& $90\%$ \\

& Held-out test split for out-of-sample \texttt{RMSE}
& $10\%$ \\

\end{xltabular}
\endgroup

\subsection{Experimental Settings of \operon, \gplearn, \deap, and \pysr}
\label{subsec:operon-gplearn-deap-pysr-configs}
\paragraph{\operon.}
We run \operon\ \citep{operon} using the \texttt{SymbolicRegressor} implementation from \texttt{PyOperon} (\url{https://github.com/heal-research/pyoperon}). For each simulation experiments in~\eqnref{eq:simulation-generating-model-1}-\eqnref{eq:simulation-generating-model-2} from~\hyperref[subsec:simulation-study]{Section~\ref{subsec:simulation-study}} and the Feynman equations in~\eqnref{eq:feynman-cl}-\eqnref{eq:feynman-ada} from~\hyperref[subsec:feynman-data-study]{Section~\ref{subsec:feynman-data-study}}, the operator sets used are reported in~\hyperref[tab:operator-sets-by-method-equation]{Table~\ref{tab:operator-sets-by-method-equation}} and \operon\ was implemented with $\texttt{generations}=2000$, $\texttt{population\_size}=1000$, $\texttt{max\_length}=20$, and $\texttt{n\_threads}=1$. The number of generations ($\texttt{generations}=2000$) specifies the total number of evolutionary (e.g., mutation, recombination, and selection) iterations used to update the population of candidate symbolic expressions, while the population size ($\texttt{population\_size}=1000$) determines the number of expressions maintained at each iteration. The maximum length ($\texttt{max\_length}=20$) parameter bounds the size of the symbolic expression tree. Finally, $\texttt{n\_threads}=1$ enforces a single-threaded execution. After fitting, the best symbolic expression returned by \operon\ is extracted from the fitted symbolic model with respect to the in-sample \texttt{RMSE}.
\paragraph{\gplearn.}
We employ the \gplearn~\citep{stephens2016gplearn} \sr\ framework (\url{https://github.com/trevorstephens/gplearn}) implemented in \texttt{Python} and outline its experimental configuration for learning the symbolic expressions in~\eqnref{eq:simulation-generating-model-1}-\eqnref{eq:simulation-generating-model-2} from~\hyperref[subsec:simulation-study]{Section~\ref{subsec:simulation-study}} and the Feynman equations in~\eqnref{eq:feynman-cl}-\eqnref{eq:feynman-ada} from~\hyperref[subsec:feynman-data-study]{Section~\ref{subsec:feynman-data-study}}. The \texttt{SymbolicRegressor} is configured with 
a population size of $2000$ and evolved for $20$ generations using tournament selection with tournament size $20$. The operator sets used are reported in~\hyperref[tab:operator-sets-by-method-equation]{Table~\ref{tab:operator-sets-by-method-equation}}. The fitness metric is set to \texttt{RMSE}, and a parsimony coefficient of $10^{-4}$ is used to penalize overly complex expressions.
\paragraph{Distributed Evolutionary algorithms in \texttt{Python} (\deap).}
We implement genetic programming for \sr\ using the \deap~\citep{deap} framework (\url{https://github.com/DEAP/deap}), a general-purpose evolutionary computation library in \texttt{Python}.
The experimental configurations of \deap\ used for learning the symbolic expressions in~\eqnref{eq:simulation-generating-model-1}-\eqnref{eq:simulation-generating-model-2} from~\hyperref[subsec:simulation-study]{Section~\ref{subsec:simulation-study}} and the Feynman equations in~\eqnref{eq:feynman-cl}-\eqnref{eq:feynman-ada} from~\hyperref[subsec:feynman-data-study]{Section~\ref{subsec:feynman-data-study}} are as follows. Symbolic expressions are evolved using a $(\mu+\lambda)$ evolutionary strategy with $\mu=200$ parent individuals and $\lambda=200$ offspring per generation, for $20$ generations. The crossover and mutation probabilities are set to $0.5$ and $0.2$, respectively. The operator sets used are reported in~\hyperref[tab:operator-sets-by-method-equation]{Table~\ref{tab:operator-sets-by-method-equation}} along with ephemeral random constants sampled from $\mathrm{Unif}(-5, 5)$. Fitness is measured using \texttt{RMSE}, and the best individual is selected from the {Hall-of-Fame archive} after evolution.
\paragraph{\pysr.}
We use the \texttt{PySRRegressor} implementation (\url{https://github.com/MilesCranmer/pysr}) for \pysr~\citep{pysr}, which provides a \texttt{Python} interface to the \texttt{SymbolicRegression.jl} search engine. \pysr\ searches over symbolic expressions that optimize a user-specified objective. For learning the expressions in~\eqnref{eq:simulation-generating-model-1}-\eqnref{eq:simulation-generating-model-2} from~\hyperref[subsec:simulation-study]{Section~\ref{subsec:simulation-study}} and the Feynman equations in~\eqnref{eq:feynman-cl}-\eqnref{eq:feynman-ada} from~\hyperref[subsec:feynman-data-study]{Section~\ref{subsec:feynman-data-study}}, we configure the operator sets as in~\hyperref[tab:operator-sets-by-method-equation]{Table~\ref{tab:operator-sets-by-method-equation}} and set $\texttt{niterations}=200$, $\texttt{populations}=10$, $\texttt{population\_size}=50$, and $\texttt{max\_size}=20$. We use squared-error loss and $\texttt{model\_selection}=\texttt{``best''}$, where the final symbolic expression being extracted using \texttt{get\_best()} corresponds to the best accuracy-complexity trade-off among the discovered Pareto-optimal expressions stored in \texttt{equations\_}.
%

\subsection{Experimental Settings of \bms\ and \bsr}
\label{subsec:bms-bsr-configs}

\paragraph{Bayesian Machine Scientist (\bms).}
For each simulation experiments in~\eqnref{eq:simulation-generating-model-1}-\eqnref{eq:simulation-generating-model-2} from~\hyperref[subsec:simulation-study]{Section~\ref{subsec:simulation-study}} and the Feynman equations in~\eqnref{eq:feynman-cl}-\eqnref{eq:feynman-ada} from~\hyperref[subsec:feynman-data-study]{Section~\ref{subsec:feynman-data-study}}, we use the \texttt{AutoRA} implementation (\url{https://github.com/AutoResearch/autora-theorist-bms}) of \bms~\citep{guimera2020bayesian} with the operator sets reported in~\hyperref[tab:operator-sets-by-method-equation]{Table~\ref{tab:operator-sets-by-method-equation}} and uniform prior operator weights. Model discovery is performed using parallel-tempered Markov chain Monte Carlo (\texttt{MCMC}) over the symbolic grammar space for $2000$ epochs. The final fitted symbolic expression returned by \bms\ is the one having the minimum description length during the \texttt{MCMC} run.
\paragraph{Bayesian Symbolic Regression (\bsr).}
For each simulation experiments in~\eqnref{eq:simulation-generating-model-1}-\eqnref{eq:simulation-generating-model-2} from~\hyperref[subsec:simulation-study]{Section~\ref{subsec:simulation-study}} and the Feynman equations in~\eqnref{eq:feynman-cl}-\eqnref{eq:feynman-ada} from~\hyperref[subsec:feynman-data-study]{Section~\ref{subsec:feynman-data-study}}, we use the \texttt{AutoRA} implementation (\url{https://github.com/AutoResearch/autora-theorist-bsr}) of \bsr~\citep{BSR} with the number of trees aligned to the ensemble size used in \vasstmain; see \hyperref[tab:vasst-common-hyperparameters]{Table~\ref{tab:vasst-common-hyperparameters}}. Uniform prior operator weights were maintained as in \bms. The number of \texttt{MCMC} iterations is set to $2000$ and the operator sets used are reported in~\hyperref[tab:operator-sets-by-method-equation]{Table~\ref{tab:operator-sets-by-method-equation}}. To ensure fairness in comparison, the split probability configuration for node expansion is matched to that used in \vasstmain, as specified in~\hyperref[tab:vasst-common-hyperparameters]{Table~\ref{tab:vasst-common-hyperparameters}}. The outer regression coefficients are estimated using ordinary least squares and the final symbolic expression corresponds to the last fitted ensemble by the \texttt{MCMC}.

\subsection{Experimental Settings of \qlattice\ and \dsr}
\label{subsec:qlattice-dsr-configs}
\paragraph{\qlattice.}
We implement \qlattice~\citep{feyn-qlattice} via the \texttt{Feyn} \texttt{Python} (\url{https://docs.abzu.ai/}), a supervised machine learning framework designed to perform \sr. The default configuration of \qlattice\ with $25$ training epochs were used for learning the symbolic expressions in~\eqnref{eq:simulation-generating-model-1}-\eqnref{eq:simulation-generating-model-2} from~\hyperref[subsec:simulation-study]{Section~\ref{subsec:simulation-study}} and the Feynman equations in~\eqnref{eq:feynman-cl}-\eqnref{eq:feynman-ada} from~\hyperref[subsec:feynman-data-study]{Section~\ref{subsec:feynman-data-study}}.
The operator sets, referred to as \emph{interactions} in the \qlattice–\texttt{Feyn} interface, are configured as in~\hyperref[tab:operator-sets-by-method-equation]{Table~\ref{tab:operator-sets-by-method-equation}}.
\paragraph{Deep Symbolic Regression (\dsr).}
We use the \texttt{DeepSymbolicRegressor} implementation (\url{https://github.com/dso-org/deep-symbolic-optimization}) of \dsr~\citep{Deep-SR} for learning the symbolic expressions in~\eqnref{eq:simulation-generating-model-1}-\eqnref{eq:simulation-generating-model-2} from~\hyperref[subsec:simulation-study]{Section~\ref{subsec:simulation-study}} and the Feynman equations in~\eqnref{eq:feynman-cl}-\eqnref{eq:feynman-ada} from~\hyperref[subsec:feynman-data-study]{Section~\ref{subsec:feynman-data-study}}. \dsr\ parameterizes a policy over symbolic expressions using an autoregressive recurrent neural network (\texttt{RNN}), which sequentially samples symbolic programs from a prescribed function set. The operator sets are configured as in~\hyperref[tab:operator-sets-by-method-equation]{Table~\ref{tab:operator-sets-by-method-equation}}. The regression task is run with protected operators enabled. We maintain $\texttt{n\_samples}=200000$, $\texttt{batch\_size}=500$, and $\texttt{max\_length}=20$ across all experiments in \hyperref[subsec:simulation-study]{Sections~\ref{subsec:simulation-study}} and~\ref{subsec:feynman-data-study}.
Here, $\texttt{n\_samples}=200000$ denotes the total symbolic program sampling budget, $\texttt{batch\_size}=500$ is the number of candidate programs sampled per \texttt{RNN} policy gradient update, and $\texttt{max\_length}=20$ bounds the maximum length of each sampled symbolic expression. These settings correspond to approximately $400$ \texttt{RNN} policy gradient updates, unless early stopping occurs. During training, DSR retains the best-reward symbolic program encountered over the search. The final recovered expression is extracted from the fitted estimator through the attribute \texttt{program\_}.
\newpage

\section{Additional Results for Simulation Experiments}\label{app:additional-simulation-results}

\subsection{Symbolic Expressions Learned by \texorpdfstring{\vasst}{VaSST} and Competing Methods}\label{app:additional-symbolic-expressions-sim1}

\begin{center}
\underline{For $\mathbf{x}_0^2 - \mathbf{x}_1 + \tfrac{1}{2}\mathbf{x}_2^2$ in~\hyperref[subsec:simulation-study]{Section~\ref{subsec:simulation-study}}.}
\end{center}

\begin{table}[H]
\centering
\scriptsize
\renewcommand{\arraystretch}{1.15}
\caption{Expressions learned in a single run of \vasstmain\ and competing methods for learning~\eqnref{eq:simulation-generating-model-1} under the noiseless setting.}
\label{tab:rep1-expressions}
\begin{tabularx}{\linewidth}{@{}l >{\raggedright\arraybackslash}X r@{}}
\toprule
\toprule
\textbf{Method} & \textbf{Symbolic Expression Learned} & \textbf{\texttt{RMSE}} \\
\midrule

\rowcolor{vasstgray}
\vasstmain &
$\,\underline{0.9987\mathbf{x}_0^2
-1.0018\mathbf{x}_1
+0.5125\mathbf{x}_2^2}
-0.0031\mathbf{x}_2
-0.2283$ &
$0.003526$ \\

\addlinespace[1pt]

\operon &
$0.9451\mathbf{x}_0^2
+0.0881\mathbf{x}_0
-0.9999\mathbf{x}_1
+0.4997\mathbf{x}_2^2
-0.0680
-\dfrac{0.2716\mathbf{x}_2^2}{\Delta}$ &
$0.000636$ \\

\addlinespace[1pt]

\bms &
$\mathbf{x}_0^2-\mathbf{x}_1+0.5\mathbf{x}_2^2$ &
$0.000000$ \\

\addlinespace[1pt]

\bsr &
$5.4{\times}10^{-6}
+6.6{\times}10^{-6}\mathbf{x}_0^{-1}\mathbf{x}_1^{-2}
+1.7{\times}10^{-4}\mathbf{x}_2^2
+2.0{\times}10^{-7}
\left(
\mathbf{x}_2^4+\mathbf{x}_0^{-2}\mathbf{x}_2^{-1}
\right)$ &
$13.084523$ \\

\addlinespace[1pt]

\qlattice &
$0.9987\mathbf{x}_0^2
-1.0016\mathbf{x}_1
+0.5028\mathbf{x}_2^2
-0.0285\mathbf{x}_2
+0.0733$ &
$0.002887$ \\

\addlinespace[1pt]

\pysr &
$\mathbf{x}_0^2-\mathbf{x}_1+0.5\mathbf{x}_2^2$ &
$0.000000$ \\

\addlinespace[1pt]

\gplearn &
$0.7890\mathbf{x}_0^4+0.4260\mathbf{x}_2^2$ &
$0.394263$ \\

\addlinespace[1pt]

\deap &
$1.1724(\mathbf{x}_0+1.9821)(\mathbf{x}_2-4.7760)
+2\mathbf{x}_2$ &
$0.375525$ \\

\addlinespace[1pt]

\dsr &
$-\mathbf{x}_1+0.5\mathbf{x}_2^2$ &
$0.453799$ \\

\bottomrule
\bottomrule
\end{tabularx}
{\scriptsize{
\emph{Note}: For \operon,
$\Delta =
3.2569\mathbf{x}_0(20.4510\mathbf{x}_0^4-2.0478\mathbf{x}_2^2)
+1.0080\mathbf{x}_1^2
+0.3566\mathbf{x}_1\mathbf{x}_2
-4.0955\mathbf{x}_2^2
-1.0550\mathbf{x}_2$.
\texttt{RMSE} corresponds to the out-of-sample \texttt{RMSE} computed on a $10\%$ held-out test set.
}}
\end{table}

\begin{table}[H]
\centering
\scriptsize
\renewcommand{\arraystretch}{1.15}
\caption{Expressions learned in a single run of \vasstmain\ and competing methods for learning~\eqnref{eq:simulation-generating-model-1} under $\sigma = 0.1$.}
\label{tab:noise1-rep1-expressions}
\begin{tabularx}{\linewidth}{@{}l >{\raggedright\arraybackslash}X r@{}}
\toprule
\toprule
\textbf{Method} & \textbf{Symbolic Expression Learned} & \textbf{\texttt{RMSE}} \\
\midrule

\rowcolor{vasstgray}
\vasstmain &
$\,\underline{1.0439\mathbf{x}_0^2
-1.0105\mathbf{x}_1
+0.5219\mathbf{x}_2^2}
-0.0289\mathbf{x}_0
+0.0289\mathbf{x}_2
-0.5386$ &
$0.101862$ \\

\addlinespace[1pt]

\operon &
$0.0001\mathbf{x}_0^{16}
+0.0197\mathbf{x}_0^9
+1.0230\mathbf{x}_0^2
-1.1390\mathbf{x}_1
+0.5120\mathbf{x}_2^2
+0.3510
+\dfrac{0.0385\mathbf{x}_2^2}{\Delta}$ &
$0.101284$ \\

\addlinespace[1pt]

\bms &
$\mathbf{x}_0^2-\mathbf{x}_1+a_0^{-1}\mathbf{x}_2^2$ &
$0.101690$ \\

\addlinespace[1pt]

\bsr &
$-26.6898
-0.4150\mathbf{x}_0^2
-0.4150\mathbf{x}_1^2
+4.4860\mathbf{x}_1
+5.3160\mathbf{x}_2
+68.8933\mathbf{x}_1^{-4}$ &
$0.490357$ \\

\addlinespace[1pt]

\qlattice &
$0.0756\mathbf{x}_0^2\mathbf{x}_2
+0.6128\mathbf{x}_0^2
-0.0029\mathbf{x}_0\mathbf{x}_2
-0.0232\mathbf{x}_0
-0.0747\mathbf{x}_1\mathbf{x}_2
-0.6059\mathbf{x}_1
+0.3953\mathbf{x}_2^2
+1.3280\mathbf{x}_2
-4.0760$ &
$0.102598$ \\

\addlinespace[1pt]

\pysr &
$\mathbf{x}_0^2-\mathbf{x}_1+0.4999\mathbf{x}_2^2$ &
$0.101690$ \\

\addlinespace[1pt]

\gplearn &
$0.7890\mathbf{x}_0^4+0.4264\mathbf{x}_2^2$ &
$0.411029$ \\

\addlinespace[1pt]

\deap &
$1.1724\mathbf{x}_0\mathbf{x}_2
-5.5995\mathbf{x}_0
+4.3239\mathbf{x}_2
-11.0990$ &
$0.391945$ \\

\addlinespace[1pt]

\dsr &
$\dfrac{\mathbf{x}_2^3}{\mathbf{x}_1+2\mathbf{x}_2-1}$ &
$0.489848$ \\

\bottomrule
\bottomrule
\end{tabularx}
{\scriptsize{
\emph{Note}: For \operon,
$\Delta =
0.3842\mathbf{x}_0
-1.1130\mathbf{x}_1
-\tfrac{2.5270\mathbf{x}_2}
{-0.3346\mathbf{x}_2^2
+5.9040\mathbf{x}_2
+0.0313\mathbf{x}_1/\mathbf{x}_0}$. $a_0$ is a constant learned by \bms. \texttt{RMSE} corresponds to the out-of-sample \texttt{RMSE} computed on a $10\%$ held-out test set.
}}
\end{table}

\begin{table}[H]
\centering
\scriptsize
\renewcommand{\arraystretch}{1.15}
\caption{Expressions learned in a single run of \vasstmain\ and competing methods for learning~\eqnref{eq:simulation-generating-model-1} under $\sigma = 0.2$.}
\label{tab:noise2-rep1-expressions}
\begin{tabularx}{\linewidth}{@{}l >{\raggedright\arraybackslash}X r@{}}
\toprule
\toprule
\textbf{Method} & \textbf{Symbolic Expression Learned} & \textbf{\texttt{RMSE}} \\
\midrule

\rowcolor{vasstgray}
\vasstmain &
$\,\underline{1.0425\mathbf{x}_0^2
-1.0203\mathbf{x}_1
+0.5212\mathbf{x}_2^2}
-0.0378\mathbf{x}_0
+0.0378\mathbf{x}_2
-0.5396$ &
$0.203331$ \\

\addlinespace[2pt]

\operon &
$1.0125\mathbf{x}_0^2
-0.9759\mathbf{x}_1
+0.4987\mathbf{x}_2^2
-0.0454
+\dfrac{-0.8539\mathbf{x}_1+0.4504\mathbf{x}_2}{\Delta}$ &
$0.202378$ \\

\addlinespace[2pt]

\bms &
$\mathbf{x}_0^2-\mathbf{x}_1+a_0\mathbf{x}_2^2$ &
$0.203380$ \\

\addlinespace[2pt]

\bsr &
$-0.4211\mathbf{x}_0^2
-0.4211\mathbf{x}_1^2
+4.4754\mathbf{x}_1
+5.3177\mathbf{x}_2
-26.6139
+68.2127\mathbf{x}_1^{-4}$ &
$0.520460$ \\

\addlinespace[2pt]

\qlattice &
$1.0382\mathbf{x}_0^2
-0.0673\mathbf{x}_0
-1.0220\mathbf{x}_1
+5.5120\mathbf{x}_2
-15.0740$ &
$0.206777$ \\

\addlinespace[2pt]

\pysr &
$\mathbf{x}_0-\mathbf{x}_1+0.4944\mathbf{x}_2^2$ &
$0.218974$ \\

\addlinespace[2pt]

\gplearn &
$0.7890\mathbf{x}_0^4+0.4264\mathbf{x}_2^2$ &
$0.448464$ \\

\addlinespace[2pt]

\deap &
$1.1724\mathbf{x}_0\mathbf{x}_2
-5.5995\mathbf{x}_0
+4.3239\mathbf{x}_2
-11.0988$ &
$0.432334$ \\

\addlinespace[2pt]

\dsr &
$\dfrac{0.5\mathbf{x}_2^3}{\mathbf{x}_2+1}$ &
$0.504085$ \\

\bottomrule
\bottomrule
\end{tabularx}
{\scriptsize{
\emph{Note}: For \operon,
$\Delta =
54.196\mathbf{x}_0^2
+1.2928\mathbf{x}_1^2
-10.808\mathbf{x}_1
+1.4352\mathbf{x}_2^2
-2.8313\mathbf{x}_2
-0.4920$. $a_0$ is a constant learned by \bms. \texttt{RMSE} corresponds to the out-of-sample \texttt{RMSE} computed on a $10\%$ held-out test set.
}}
\end{table}
\newpage
\begin{center}
\underline{For $6\sin(\mathbf{x}_0)\cos(\mathbf{x}_1)$ in~\hyperref[subsec:simulation-study]{Section~\ref{subsec:simulation-study}}.}
\end{center}

\begin{table}[H]
\centering
\scriptsize
\renewcommand{\arraystretch}{1.15}
\caption{Expressions recovered in a single run of \vasstmain\ and competing methods for learning~\eqnref{eq:simulation-generating-model-2} under the noiseless setting.}
\label{tab:sim3-noiseless-rep1-expressions}
\begin{tabularx}{\linewidth}{@{}l >{\raggedright\arraybackslash}X r@{}}
\toprule
\toprule
\textbf{Method} & \textbf{Symbolic Expression Learned} & \textbf{\texttt{RMSE}} \\
\midrule

\rowcolor{vasstgray}
\vasstmain &
$\,\underline{5.6818\sin(\mathbf{x}_0)\cos(\mathbf{x}_1)}
+0.0787\mathbf{x}_0\mathbf{x}_1\cos(\mathbf{x}_1)
-0.0281\cos(\mathbf{x}_1)\cos(\mathbf{x}_0\mathbf{x}_1)
-0.0397$ &
$0.009577$ \\

\addlinespace[2pt]

\operon &
$-0.002 + 6.527\cos(-0.982 \mathbf x_1)\sin(1.026 \mathbf x_0)\cos(A(\mathbf x)B(\mathbf x))$ &
$0.003086$ \\

\addlinespace[2pt]

\bms &
$6\sin(\mathbf{x}_0)\cos(\mathbf{x}_1)$ &
$0.000000$ \\

\addlinespace[2pt]

\bsr &
$-103.575
+1.9976\sin(\mathbf{x}_1)
+101.883\cos(\mathbf{x}_0)
+43.7129\mathbf{x}_0^2$ &
$0.340724$ \\

\addlinespace[2pt]

\qlattice &
$-2.011(3.114\mathbf x_0 + 0.002)\exp\left\{-0.185(0.0482-\mathbf{x}_0)^2-2.891(0.289\mathbf x_1 + 0.832L(\mathbf{x})-1)^2-0.0004\right\}$ &
$0.004568$ \\

\addlinespace[2pt]

\pysr &
$6\sin(\mathbf{x}_0)\cos(\mathbf{x}_1)$ &
$0.000000$ \\

\addlinespace[2pt]

\gplearn &
$-0.6513\mathbf{x}_0\mathbf{x}_1^2$ &
$0.172140$ \\

\addlinespace[2pt]

\deap &
$-1.9881\mathbf{x}_1\sin\{\sin(\mathbf{x}_0)\}$ &
$0.279846$ \\

\addlinespace[2pt]

\dsr &
$\mathbf{x}_0\mathbf{x}_1^2\cos(\mathbf{x}_0)\cos(\mathbf{x}_1)$ &
$0.548577$ \\

\bottomrule
\bottomrule
\end{tabularx}
{\scriptsize{
\emph{Note}: For \operon, $A(\mathbf{x}) = \cos\!\left(\cos\!\left[\cos(-1.142\mathbf{x}_1)\cos(-1.142\mathbf{x}_1)\right]\right)
\cos\!\left(\cos\!\left[1.686\mathbf{x}_0\cos(0.258\mathbf{x}_1)\right]\right)
\cos\!\left(\cos(-1.134\mathbf{x}_1)\right)$ and $B(\mathbf{x})=\cos\!\left[-0.295\cos(-3.325\mathbf{x}_1)\sin(1.351\mathbf{x}_0)\right]
\sin(-0.159\mathbf{x}_1)$. For \qlattice, $L(\mathbf{x}) = -0.002(-0.201\mathbf{x}_0-1)^2 - 1.832(0.349\mathbf x_1 - 1)^2$. \texttt{RMSE} corresponds to the out-of-sample \texttt{RMSE} computed on a $10\%$ held-out test set.
}}
\end{table}

\begin{table}[H]
\centering
\scriptsize
\renewcommand{\arraystretch}{1.15}
\caption{Expressions recovered in a single run of \vasstmain\ and competing methods for learning~\eqnref{eq:simulation-generating-model-2} under $\sigma=0.1$.}
\label{tab:sim3-noise1-rep1-expressions}
\begin{tabularx}{\linewidth}{@{}l >{\raggedright\arraybackslash}X r@{}}
\toprule
\toprule
\textbf{Method} & \textbf{Symbolic Expression Learned} & \textbf{\texttt{RMSE}} \\
\midrule

\rowcolor{vasstgray}
\vasstmain &
$\,\underline{5.9931\sin(\mathbf{x}_0)\cos(\mathbf{x}_1)}
-0.0023\mathbf{x}_1
-0.0009\mathbf{x}_1\cos\{\cos(\mathbf{x}_1)\}
+0.0023$ &
$0.099728$ \\

\addlinespace[2pt]

\operon &
$-0.001
+1.949\cos\{0.318\mathbf{x}_0\sin(-0.894\mathbf{x}_0)\}
\{1.353\mathbf{x}_1\sin(0.917\mathbf{x}_0)\cos(-1.127\mathbf{x}_1)\}
A(\mathbf{x})B(\mathbf{x})$ &
$0.099540$ \\

\addlinespace[2pt]

\bms &
$a_0\sin(\mathbf{x}_0)\cos(\mathbf{x}_1)$ &
$0.099723$ \\

\addlinespace[2pt]

\bsr &
$-9.5947
+2.6779\cos(\mathbf{x}_1)
+10.3204\cos\{\mathbf{x}_0\sin(\mathbf{x}_0)\}$ &
$0.652872$ \\

\addlinespace[2pt]

\qlattice &
$-3.647(1.665\mathbf{x}_0+0.003)
\exp\left\{
-0.235(0.182-\mathbf{x}_0)^2
-2\exp\{-1.544(1-0.866\mathbf{x}_1)^2\}
\right\}
+0.001$ &
$0.100564$ \\

\addlinespace[2pt]

\pysr &
$6.0049\sin(\mathbf{x}_0)\cos(\mathbf{x}_1)$ &
$0.099723$ \\

\addlinespace[2pt]

\gplearn &
$-0.6513\mathbf{x}_0\mathbf{x}_1^2$ &
$0.202208$ \\

\addlinespace[2pt]

\deap &
$-1.9881\mathbf{x}_1\sin\{\sin(\mathbf{x}_0)\}$ &
$0.296343$ \\

\addlinespace[2pt]

\dsr &
$\mathbf{x}_0\mathbf{x}_1^2\cos(\mathbf{x}_0)\cos(\mathbf{x}_1)$ &
$0.556807$ \\

\bottomrule
\bottomrule
\end{tabularx}
{\scriptsize{
\emph{Note}: For \operon,
$A(\mathbf{x})
=
\cos\!\left[
\sin\{\cos(\sin(1.353\mathbf{x}_1))1.353\mathbf{x}_1\}
\sin(1.352\mathbf{x}_1)\sin(-2.031\mathbf{x}_1)
\right]$
and
$B(\mathbf{x})
=
\sin\!\left[
\cos\{-0.474\sin[\sin\{\sin(1.353\mathbf{x}_1)\}]\sin(-2.001\mathbf{x}_1)\}
\right]$. $a_0$ is a constant learned by \bms. \texttt{RMSE} corresponds to the out-of-sample \texttt{RMSE} computed on a $10\%$ held-out test set.
}}
\end{table}
\begin{table}[H]
\centering
\scriptsize
\renewcommand{\arraystretch}{1.15}
\caption{Expressions recovered in a single run of \vasstmain\ and competing methods for learning~\eqnref{eq:simulation-generating-model-2} under $\sigma=0.2$.}
\label{tab:sim3-noise2-rep1-expressions}
\begin{tabularx}{\linewidth}{@{}l >{\raggedright\arraybackslash}X r@{}}
\toprule
\toprule
\textbf{Method} & \textbf{Symbolic Expression Learned} & \textbf{\texttt{RMSE}} \\
\midrule

\rowcolor{vasstgray}
\vasstmain &
$\,\underline{5.9933\sin(\mathbf{x}_0)\cos(\mathbf{x}_1)}
-0.0031\mathbf{x}_1
-0.0039\mathbf{x}_1\cos\{\cos(\mathbf{x}_1)\}
+0.0074$ &
$0.199438$ \\

\addlinespace[2pt]

\operon &
$88.541A(\mathbf{x})B(\mathbf{x})C(\mathbf{x})$ &
$0.199028$ \\

\addlinespace[2pt]

\bms &
$a_0\sin(\mathbf{x}_0)\cos(\mathbf{x}_1)$ &
$0.199446$ \\

\addlinespace[2pt]

\bsr &
$-9.5709
+2.6796\cos(\mathbf{x}_1)
+10.2934\cos\{\mathbf{x}_0\sin(\mathbf{x}_0)\}$ &
$0.679095$ \\

\addlinespace[2pt]

\qlattice &
$-8.4212\tanh(0.8270\mathbf{x}_0-0.0052)
\tanh(1.0045\mathbf{x}_1-1.6330)
+0.0021$ &
$0.201085$ \\

\addlinespace[2pt]

\pysr &
$5.4158\mathbf{x}_0\cos(\mathbf{x}_1)$ &
$0.235483$ \\

\addlinespace[2pt]

\gplearn &
$-0.7930\mathbf{x}_0\mathbf{x}_1^2\cos(0.7930\mathbf{x}_0)$ &
$0.262413$ \\

\addlinespace[2pt]

\deap &
$-1.9881\mathbf{x}_1\sin\{\sin(\mathbf{x}_0)\}$ &
$0.342378$ \\

\addlinespace[2pt]

\dsr &
$\mathbf{x}_0\mathbf{x}_1^2\cos(\mathbf{x}_0)\cos(\mathbf{x}_1)$ &
$0.582262$ \\

\bottomrule
\bottomrule
\end{tabularx}
{\scriptsize{
\emph{Note}: For \operon,
$A(\mathbf{x})
=
\cos\!\left[
\cos\{\cos(1.134\mathbf{x}_1)\cos[\cos(1.134\mathbf{x}_1)]\}
\right]
\cos\{0.2027\mathbf{x}_0^3\sin^2(0.863\mathbf{x}_1)\}$,
$B(\mathbf{x})
=
\cos\{\cos(0.079\mathbf{x}_1)\}
\cos\!\left[
\cos\{\cos(0.079\mathbf{x}_1)\}\cos(0.079\mathbf{x}_1)
\right]$,
and
$C(\mathbf{x})
=
0.079\mathbf{x}_1
\sin(0.917\mathbf{x}_0)
\cos\{0.3141\mathbf{x}_0^2\sin(0.917\mathbf{x}_0)\}
\cos(1.134\mathbf{x}_1)$. $a_0$ is a constant learned by \bms. \texttt{RMSE} corresponds to the out-of-sample \texttt{RMSE} computed on a $10\%$ held-out test set.
}}
\end{table}
\newpage

\subsection{Out-of-sample \texttt{RMSE} for \texorpdfstring{\vasst}{VaSST} and Competing Methods}\label{app:ooRMSE-sim1}

\begin{table}[H]
\centering
\scriptsize
\renewcommand{\arraystretch}{1.15}
\caption{Out-of-sample \texttt{RMSE}s (computed on a $10\%$ held-out test set) of \vasstmain and competing methods over $10$ repetitions (mean $\pm$ standard deviation) for learning $\mathbf{x}_0^{2}-\mathbf{x}_1+\frac{1}{2}\mathbf{x}_2^{2}$ in~\eqnref{eq:simulation-generating-model-1} of~\hyperref[subsec:simulation-study]{Section~\ref{subsec:simulation-study}} across all noise settings.}
\label{tab:rmse-summary-all}
\begin{minipage}[t]{0.32\linewidth}
\centering
\caption*{\textbf{Noiseless}}
\begin{tabular}{@{}l c@{}}
\toprule
\toprule
\textbf{Method} & \textbf{\texttt{RMSE}, mean $\pm$ sd} \\
\midrule
\rowcolor{vasstgray}
\vasstmain & $0.00412 \pm 0.00133$ \\
\addlinespace[2pt]
\operon & $0.00644 \pm 0.00394$ \\
\bms & $7.74{\times}10^{-17} \pm 2.45{\times}10^{-16}$ \\
\bsr & $1.96367 \pm 3.94562$ \\
\qlattice & $0.00975 \pm 0.00684$ \\
\pysr & $0.00000 \pm 0.00000$ \\
\gplearn & $0.21830 \pm 0.12335$ \\
\deap & $0.69742 \pm 0.24570$ \\
\dsr & $0.50942 \pm 0.05589$ \\
\bottomrule
\bottomrule
\end{tabular}
\end{minipage}
\hfill
\begin{minipage}[t]{0.32\linewidth}
\centering
\caption*{\textbf{$\sigma=0.1$}}
\begin{tabular}{@{}l c@{}}
\toprule
\toprule
\textbf{Method} & \textbf{\texttt{RMSE}, mean $\pm$ sd} \\
\midrule
\rowcolor{vasstgray}
\vasstmain & $0.10036 \pm 0.00158$ \\
\addlinespace[2pt]
\operon & $0.10040 \pm 0.00152$ \\
\bms & $0.10032 \pm 0.00152$ \\
\bsr & $0.71729 \pm 0.54286$ \\
\qlattice & $0.10219 \pm 0.00199$ \\
\pysr & $0.10032 \pm 0.00152$ \\
\gplearn & $0.21616 \pm 0.09241$ \\
\deap & $0.72482 \pm 0.25997$ \\
\dsr & $0.49501 \pm 0.07044$ \\
\bottomrule
\bottomrule
\end{tabular}
\end{minipage}
\hfill
\begin{minipage}[t]{0.32\linewidth}
\centering
\caption*{\textbf{$\sigma=0.2$}}
\begin{tabular}{@{}l c@{}}
\toprule
\toprule
\textbf{Method} & \textbf{\texttt{RMSE}, mean $\pm$ sd} \\
\midrule
\rowcolor{vasstgray}
\vasstmain & $0.20061 \pm 0.00309$ \\
\addlinespace[2pt]
\operon & $0.20025 \pm 0.00309$ \\
\bms & $0.20063 \pm 0.00305$ \\
\bsr & $0.75153 \pm 0.52180$ \\
\qlattice & $0.20347 \pm 0.00319$ \\
\pysr & $0.21481 \pm 0.00351$ \\
\gplearn & $0.28095 \pm 0.08016$ \\
\deap & $0.73682 \pm 0.24150$ \\
\dsr & $0.49416 \pm 0.04973$ \\
\bottomrule
\bottomrule
\end{tabular}
\end{minipage}
\end{table}

\begin{table}[H]
\centering
\scriptsize
\renewcommand{\arraystretch}{1.15}
\caption{Out-of-sample \texttt{RMSE}s (computed on a $10\%$ held-out test set) of \vasstmain and competing methods over $10$ repetitions (mean $\pm$ standard deviation) for learning $6\sin(\mathbf{x}_0)\cos(\mathbf{x}_1)$ in~\eqnref{eq:simulation-generating-model-2} of~\hyperref[subsec:simulation-study]{Section~\ref{subsec:simulation-study}} across all noise settings.}
\label{tab:sim3-rmse-summary-all}
\begin{minipage}[t]{0.32\linewidth}
\centering
\caption*{\textbf{Noiseless}}
\begin{tabular}{@{}l c@{}}
\toprule
\toprule
\textbf{Method} & \textbf{\texttt{RMSE}, mean $\pm$ sd} \\
\midrule
\rowcolor{vasstgray}
\vasstmain & $0.01306 \pm 0.00767$ \\
\addlinespace[2pt]
\operon & $0.00244 \pm 0.00241$ \\
\bms & $0.00000 \pm 0.00000$ \\
\bsr & $0.28186 \pm 0.13300$ \\
\qlattice & $0.00400 \pm 0.00127$ \\
\pysr & $0.00000 \pm 0.00000$ \\
\gplearn & $0.11964 \pm 0.03772$ \\
\deap & $0.36454 \pm 0.14455$ \\
\dsr & $0.55193 \pm 0.00870$ \\
\bottomrule
\bottomrule
\end{tabular}
\end{minipage}
\hfill
\begin{minipage}[t]{0.32\linewidth}
\centering
\caption*{\textbf{$\sigma=0.1$}}
\begin{tabular}{@{}l c@{}}
\toprule
\toprule
\textbf{Method} & \textbf{\texttt{RMSE}, mean $\pm$ sd} \\
\midrule
\rowcolor{vasstgray}
\vasstmain & $0.10054 \pm 0.00149$ \\
\addlinespace[2pt]
\operon & $0.09904 \pm 0.00095$ \\
\bms & $0.09941 \pm 0.00091$ \\
\bsr & $0.33652 \pm 0.15671$ \\
\qlattice & $0.09994 \pm 0.00123$ \\
\pysr & $0.09941 \pm 0.00091$ \\
\gplearn & $0.15850 \pm 0.03238$ \\
\deap & $0.38290 \pm 0.13653$ \\
\dsr & $0.56119 \pm 0.00939$ \\
\bottomrule
\bottomrule
\end{tabular}
\end{minipage}
\hfill
\begin{minipage}[t]{0.32\linewidth}
\centering
\caption*{\textbf{$\sigma=0.2$}}
\begin{tabular}{@{}l c@{}}
\toprule
\toprule
\textbf{Method} & \textbf{\texttt{RMSE}, mean $\pm$ sd} \\
\midrule
\rowcolor{vasstgray}
\vasstmain & $0.19951 \pm 0.00170$ \\
\addlinespace[2pt]
\operon & $0.19828 \pm 0.00189$ \\
\bms & $0.19881 \pm 0.00181$ \\
\bsr & $0.38499 \pm 0.13881$ \\
\qlattice & $0.19971 \pm 0.00173$ \\
\pysr & $0.22836 \pm 0.00429$ \\
\gplearn & $0.23735 \pm 0.02038$ \\
\deap & $0.42092 \pm 0.12471$ \\
\dsr & $0.58739 \pm 0.01000$ \\
\bottomrule
\bottomrule
\end{tabular}
\end{minipage}
\end{table}
\newpage

\subsection{Top Ranked Symbolic Expressions Learned by \texorpdfstring{\vasst}{VaSST} with respect to LMPSE}
\label{app:top-symbolic-expressions-vasst}

\begin{center}
\underline{For $\mathbf{x}_0^2 - \mathbf{x}_1 + \tfrac{1}{2}\mathbf{x}_2^2$ in~\hyperref[subsec:simulation-study]{Section~\ref{subsec:simulation-study}}.}
\end{center}

\begin{table}[H]
\centering
\scriptsize
\renewcommand{\arraystretch}{1.15}
\caption{Top $\mathsf{r}=10$ ranked symbolic expressions recovered by a single run of \vasstmain\ under the noiseless setting for learning~\eqnref{eq:simulation-generating-model-1}. Expressions are ranked by $\mathrm{LMPSE}$ among $H=2000$ sampled hard symbolic trees.}
\label{tab:vasst-top10-LMPSE}
\begin{tabularx}{\linewidth}{@{}c >{\raggedright\arraybackslash}X r r@{}}
\toprule
\toprule
\textbf{Rank} & \textbf{\vasstmain\ Expression Learned} & \textbf{$\mathrm{LMPSE}$} & \textbf{\texttt{RMSE}} \\
\midrule
1 &
$\,{0.9987\mathbf{x}_0^2
-1.0018\mathbf{x}_1
+0.5125\mathbf{x}_2^2}
-0.0031\mathbf{x}_2
-0.2283$ &
$3296.45$ & $0.003526$ \\

\addlinespace[2pt]

2 &
$1.0009\mathbf{x}_0^2
-1.0009\mathbf{x}_1
+0.5331\mathbf{x}_2^2
-0.2247\mathbf{x}_2
+0.3650$ &
$3285.09$ & $0.002155$ \\

\addlinespace[2pt]

3 &
$0.9996\mathbf{x}_0^2
-0.9996\mathbf{x}_1
+0.5122\mathbf{x}_2^2
-0.0001\mathbf{x}_1^8
-0.2236$ &
$3276.92$ & $0.003551$ \\

\addlinespace[2pt]

4 &
$0.9993\mathbf{x}_0^2
-1.0001\mathbf{x}_1
+0.5122\mathbf{x}_2^2
-0.2409$ &
$3272.47$ & $0.003524$ \\

\addlinespace[2pt]

5 &
$0.9993\mathbf{x}_0^2
-1.0014\mathbf{x}_1
+0.5122\mathbf{x}_2^2
-0.2384$ &
$3270.59$ & $0.003529$ \\

\addlinespace[2pt]

6 &
$0.9997\mathbf{x}_0^2
-0.9987\mathbf{x}_1
+0.5737\mathbf{x}_2^2
-0.5737\mathbf{x}_2
+1.1346$ &
$3268.11$ & $0.002609$ \\

\addlinespace[2pt]

7 &
$0.9991\mathbf{x}_0^2
-1.0388\mathbf{x}_1
+0.5194\mathbf{x}_2^2
+0.1120\mathbf{x}_2
-0.8127$ &
$3239.00$ & $0.010701$ \\

\addlinespace[2pt]

8 &
$1.0013\mathbf{x}_0^2
+0.0114\mathbf{x}_0
-1.0013\mathbf{x}_1
+0.5426\mathbf{x}_2^2
-0.8243$ &
$3205.50$ & $0.011400$ \\

\addlinespace[2pt]

9 &
$0.8460\mathbf{x}_0^2
+0.1669\mathbf{x}_0
-0.1669\mathbf{x}_1^2
-0.1669\mathbf{x}_1
+0.5473\mathbf{x}_2^2
-0.3785\mathbf{x}_2
-0.2803$ &
$3184.82$ & $0.017184$ \\

\addlinespace[2pt]

10 &
$-0.0719\mathbf{x}_0^4
+0.1438\mathbf{x}_0^2\mathbf{x}_1
+0.1438\mathbf{x}_0^2\mathbf{x}_2
-0.0018\mathbf{x}_0
-0.0719\mathbf{x}_1^2
-0.1438\mathbf{x}_1\mathbf{x}_2
+0.0006\mathbf{x}_1
+0.5476\mathbf{x}_2^2
-0.0006\mathbf{x}_2
-1.3782$ &
$3153.72$ & $0.014262$ \\

\bottomrule
\bottomrule
\end{tabularx}
{
\scriptsize{
\emph{Note}: \texttt{RMSE} corresponds to the out-of-sample \texttt{RMSE} computed on a $10\%$ held-out test set with the LMPSE selected symbolic model.
}
}
\end{table}

\begin{table}[H]
\centering
\scriptsize
\renewcommand{\arraystretch}{1.15}
\caption{Top $\mathsf{r}=10$ ranked symbolic expressions recovered by a single run of \vasstmain\ under $\sigma=0.1$ for learning~\eqnref{eq:simulation-generating-model-1}. Expressions are ranked by $\mathrm{LMPSE}$ among $H=2000$ sampled hard symbolic trees.}
\label{tab:noise1-vasst-top10-LMPSE}
\begin{tabularx}{\linewidth}{@{}c >{\raggedright\arraybackslash}X r r@{}}
\toprule
\toprule
\textbf{Rank} & \textbf{\vasstmain\ Expression Learned} & \textbf{$\mathrm{LMPSE}$} & \textbf{\texttt{RMSE}} \\
\midrule

1 &
$\,{1.0439\mathbf{x}_0^2
-1.0105\mathbf{x}_1
+0.5219\mathbf{x}_2^2}
-0.0289\mathbf{x}_0
+0.0289\mathbf{x}_2
-0.5386$ &
$1506.54$ & $0.101862$ \\

\addlinespace[2pt]

2 &
$0.9914\mathbf{x}_0^2
-1.0107\mathbf{x}_1
+0.5129\mathbf{x}_2^2
-0.2253$ &
$1503.94$ & $0.101688$ \\

\addlinespace[2pt]

3 &
$0.9912\mathbf{x}_0^2
-1.0099\mathbf{x}_1
+0.5124\mathbf{x}_2^2
-0.2196$ &
$1502.74$ & $0.101685$ \\

\addlinespace[2pt]

4 &
$1.0124\mathbf{x}_0^2
-0.0345\mathbf{x}_0
-1.0124\mathbf{x}_1
+0.5238\mathbf{x}_2^2
-0.4311$ &
$1498.05$ & $0.101826$ \\

\addlinespace[2pt]

5 &
$1.0804\mathbf{x}_0^2
-0.0941\mathbf{x}_0
-0.0941\mathbf{x}_1^2
-0.5402\mathbf{x}_1
+0.5123\mathbf{x}_2^2
-0.7809$ &
$1495.14$ & $0.101806$ \\

\addlinespace[2pt]

6 &
$0.9813\mathbf{x}_0^2
-1.0114\mathbf{x}_1
+0.5361\mathbf{x}_2^2
-0.0909\mathbf{x}_2
-0.2825$ &
$1490.35$ & $0.101942$ \\

\addlinespace[2pt]

7 &
$0.9907\mathbf{x}_0^2
-1.0088\mathbf{x}_1
+0.4684\mathbf{x}_2^2
+0.5404\mathbf{x}_2
-1.8103$ &
$1489.77$ & $0.102035$ \\

\addlinespace[2pt]

8 &
$0.9912\mathbf{x}_0^2
-1.0099\mathbf{x}_1
+0.5129\mathbf{x}_2^2
-0.2270$ &
$1489.37$ & $0.101688$ \\

\addlinespace[2pt]

9 &
$1.0113\mathbf{x}_0^2
+0.0224\mathbf{x}_0
-1.0113\mathbf{x}_1
+0.5257\mathbf{x}_2^2
-0.4911$ &
$1486.86$ & $0.101895$ \\

\addlinespace[2pt]

10 &
$0.9915\mathbf{x}_0^2
+0.0147\mathbf{x}_1^2
-1.0913\mathbf{x}_1
+0.5809\mathbf{x}_2^2
-0.5809\mathbf{x}_2
+1.1513$ &
$1484.55$ & $0.101748$ \\

\bottomrule
\bottomrule
\end{tabularx}
{
\scriptsize{
\emph{Note}: \texttt{RMSE} corresponds to the out-of-sample \texttt{RMSE} computed on a $10\%$ held-out test set with the LMPSE selected symbolic model.
}
}
\end{table}

\begin{table}[H]
\centering
\scriptsize
\renewcommand{\arraystretch}{1.15}
\caption{Top $\mathsf{r}=10$ ranked symbolic expressions recovered by a single run of \vasstmain\ under $\sigma=0.2$ for learning~\eqnref{eq:simulation-generating-model-1}. Expressions are ranked by $\mathrm{LMPSE}$ among $H=2000$ sampled hard symbolic trees.}
\label{tab:noise2-vasst-top10-LMPSE}
\begin{tabularx}{\linewidth}{@{}c >{\raggedright\arraybackslash}X r r@{}}
\toprule
\toprule
\textbf{Rank} & \textbf{\vasstmain\ Expression Learned} & \textbf{$\mathrm{LMPSE}$} & \textbf{\texttt{RMSE}} \\
\midrule
1 &
$\,{1.0425\mathbf{x}_0^2
-1.0203\mathbf{x}_1
+0.5212\mathbf{x}_2^2}
-0.0378\mathbf{x}_0
+0.0378\mathbf{x}_2
-0.5396$ &
$255.26$ & $0.203331$ \\

\addlinespace[2pt]

2 &
$0.9826\mathbf{x}_0^2
-1.0187\mathbf{x}_1
+0.4675\mathbf{x}_2^2
+0.5512\mathbf{x}_2
-1.8176$ &
$251.41$ & $0.203457$ \\

\addlinespace[2pt]

3 &
$0.9832\mathbf{x}_0^2
-1.0198\mathbf{x}_1
+0.5125\mathbf{x}_2^2
-0.1982$ &
$251.27$ & $0.203272$ \\

\addlinespace[2pt]

4 &
$0.9833\mathbf{x}_0^2
-1.0205\mathbf{x}_1
+0.5130\mathbf{x}_2^2
-0.2039$ &
$249.55$ & $0.203274$ \\

\addlinespace[2pt]

5 &
$0.9841\mathbf{x}_0^2
-0.2033\mathbf{x}_1^2
+0.6355\mathbf{x}_2^2
-1.2709\mathbf{x}_2
+1.8591$ &
$248.64$ & $0.203483$ \\

\addlinespace[2pt]

6 &
$1.0234\mathbf{x}_0^2
-0.0561\mathbf{x}_0
-1.0234\mathbf{x}_1
+0.5239\mathbf{x}_2^2
-0.4025$ &
$244.71$ & $0.203299$ \\

\addlinespace[2pt]

7 &
$1.0745\mathbf{x}_0^2
-0.0968\mathbf{x}_0
-0.0968\mathbf{x}_1^2
-0.5373\mathbf{x}_1
+0.5125\mathbf{x}_2^2
-0.7737$ &
$243.29$ & $0.203169$ \\

\addlinespace[2pt]

8 &
$1.0224\mathbf{x}_0^2
+0.0008\mathbf{x}_0
-1.0224\mathbf{x}_1
+0.5258\mathbf{x}_2^2
-0.4624$ &
$241.63$ & $0.203338$ \\

\addlinespace[2pt]

9 &
$0.9827\mathbf{x}_0^2
-1.0124\mathbf{x}_1
+0.6067\mathbf{x}_2^2
-1.0124\mathbf{x}_2
+2.5075$ &
$239.59$ & $0.203264$ \\

\addlinespace[2pt]

10 &
$0.9733\mathbf{x}_0^2
-1.0212\mathbf{x}_1
+0.5372\mathbf{x}_2^2
-0.1011\mathbf{x}_2
-0.2343$ &
$238.96$ & $0.203410$ \\

\bottomrule
\bottomrule
\end{tabularx}
{
\scriptsize{
\emph{Note}: \texttt{RMSE} corresponds to the out-of-sample \texttt{RMSE} computed on a $10\%$ held-out test set with the LMPSE selected symbolic model.
}
}
\end{table}
\newpage

\begin{center}
\underline{For $6\sin(\mathbf{x}_0)\cos(\mathbf{x}_1)$ in~\hyperref[subsec:simulation-study]{Section~\ref{subsec:simulation-study}}.}
\end{center}

\begin{table}[H]
\centering
\scriptsize
\renewcommand{\arraystretch}{1.15}
\caption{Top $\mathsf{r}=10$ ranked symbolic expressions recovered by a single run of \vasstmain\ under the noiseless setting for learning~\eqnref{eq:simulation-generating-model-2}. Expressions are ranked by $\mathrm{LMPSE}$ among $H=2000$ sampled hard symbolic trees.}
\label{tab:sim3-noiseless-vasst-top10-jmp}
\begin{tabularx}{\linewidth}{@{}c >{\raggedright\arraybackslash}X r r@{}}
\toprule
\toprule
\textbf{Rank} & \textbf{\vasstmain\ Expression Learned} & $\mathrm{LMPSE}$ & \textbf{\texttt{RMSE}} \\
\midrule
1 &
$\,0.0787\mathbf{x}_0\mathbf{x}_1\cos(\mathbf{x}_1)
+{5.6818\sin(\mathbf{x}_0)\cos(\mathbf{x}_1)}
-0.0281\cos(\mathbf{x}_1)\cos(\mathbf{x}_0\mathbf{x}_1)
-0.0397$ &
$2699.67$ & $0.009557$ \\

\addlinespace[2pt]

2 &
$-0.0951\mathbf{x}_0
+5.8568\sin(\mathbf{x}_0)\cos(\mathbf{x}_1)
+0.0977\cos\{\cos(\mathbf{x}_1)\}
-0.0725$ &
$2693.37$ & $0.006417$ \\

\addlinespace[2pt]

3 &
$-0.0178\mathbf{x}_1
+5.9318\sin(\mathbf{x}_0)\cos(\mathbf{x}_1)
-0.1651\sin[\cos\{\cos(\mathbf{x}_0)\}]
+0.1208$ &
$2687.91$ & $0.004291$ \\

\addlinespace[2pt]

4 &
$0.0021\mathbf{x}_1\cos\{\cos(\mathbf{x}_1)\}
-0.0014\mathbf{x}_1
+5.9929\sin(\mathbf{x}_0)\cos(\mathbf{x}_1)
-0.0028$ &
$2687.49$ & $0.001527$ \\

\addlinespace[2pt]

5 &
$5.9249\sin(\mathbf{x}_0)\cos(\mathbf{x}_1)
+0.0343\cos(\mathbf{x}_1)
+0.1327\cos\{\sin(\mathbf{x}_0)\}
-0.1156$ &
$2681.23$ & $0.004575$ \\

\addlinespace[2pt]

6 &
$5.9704\sin(\mathbf{x}_0)\cos(\mathbf{x}_1)
-0.0567\cos\{\cos(\mathbf{x}_0)\}
+0.0267$ &
$2680.73$ & $0.002973$ \\

\addlinespace[2pt]

7 &
$0.0019\mathbf{x}_1\cos(\mathbf{x}_1)
+5.9929\sin(\mathbf{x}_0)\cos(\mathbf{x}_1)
-0.0028\sin(\mathbf{x}_1)
+0.0026$ &
$2677.57$ & $0.001530$ \\

\addlinespace[2pt]

8 &
$5.9709\sin(\mathbf{x}_0)\cos(\mathbf{x}_1)
+0.0443\sin\{\cos(\mathbf{x}_0)\}\cos(\mathbf{x}_0)
-0.0359\cos[\cos\{\cos(\mathbf{x}_0)\}]
-0.0104$ &
$2676.48$ & $0.002941$ \\

\addlinespace[2pt]

9 &
$0.0270\sin\{\cos(\mathbf{x}_0\mathbf{x}_1)\}
+5.9225\sin(\mathbf{x}_0)\cos(\mathbf{x}_1)
-0.0011\mathbf{x}_1\cos(\mathbf{x}_1)\cos(\mathbf{x}_0\mathbf{x}_1)
-0.0378$ &
$2664.00$ & $0.004804$ \\

\addlinespace[2pt]

10 &
$6.3825\cos(\mathbf{x}_1)\sin\{\sin(\mathbf{x}_0)\}
+0.9289\sin\{\cos(\mathbf{x}_0)\}
-0.2152\sin\{\sin(\mathbf{x}_1\mathbf{x}_0)\}
-0.9499$ &
$2568.17$ & $0.009784$ \\

\bottomrule
\bottomrule
\end{tabularx}
{
\scriptsize{
\emph{Note}: \texttt{RMSE} corresponds to the out-of-sample \texttt{RMSE} computed on a $10\%$ held-out test set with the LMPSE selected symbolic model.
}
}
\end{table}

\begin{table}[H]
\centering
\scriptsize
\renewcommand{\arraystretch}{1.15}
\caption{Top $\mathsf{r}=10$ ranked symbolic expressions recovered by a single run of \vasstmain\ under $\sigma=0.1$ for learning~\eqnref{eq:simulation-generating-model-2}. Expressions are ranked by $\mathrm{LMPSE}$ among $H=2000$ sampled hard symbolic trees.}
\label{tab:sim3-noise1-vasst-top10-jmp}
\begin{tabularx}{\linewidth}{@{}c >{\raggedright\arraybackslash}X r r@{}}
\toprule
\toprule
\textbf{Rank} & \textbf{\vasstmain\ Expression Learned} & $\mathrm{LMPSE}$ & \textbf{\texttt{RMSE}} \\
\midrule
1 &
${5.9931\sin(\mathbf{x}_0)\cos(\mathbf{x}_1)}
-0.0023\mathbf{x}_1
-0.0009\mathbf{x}_1\cos\{\cos(\mathbf{x}_1)\}
+0.0023$ &
$1405.78$ & $0.099728$ \\

\addlinespace[2pt]

2 &
$5.8878\sin(\mathbf{x}_0)\cos(\mathbf{x}_1)
-0.0735\mathbf{x}_0
+0.0778\cos\{\cos(\mathbf{x}_1)\}
-0.0606$ &
$1402.42$ & $0.099916$ \\

\addlinespace[2pt]

3 &
$5.9670\sin(\mathbf{x}_0)\cos(\mathbf{x}_1)
-0.0706\sin[\cos\{\cos(\mathbf{x}_0)\}]
-0.0088\mathbf{x}_1
+0.0508$ &
$1401.50$ & $0.099788$ \\

\addlinespace[2pt]

4 &
$5.9895\sin(\mathbf{x}_0)\cos(\mathbf{x}_1)
-0.0112\cos\{\cos(\mathbf{x}_0)\}
+0.0011$ &
$1397.16$ & $0.099739$ \\

\addlinespace[2pt]

5 &
$5.7657\sin(\mathbf{x}_0)\cos(\mathbf{x}_1)
-0.0181\cos(\mathbf{x}_0\mathbf{x}_1)\cos(\mathbf{x}_1)
+0.0593\mathbf{x}_0\mathbf{x}_1\cos(\mathbf{x}_1)
-0.0299$ &
$1396.37$ & $0.100150$ \\

\addlinespace[2pt]

6 &
$5.9927\sin(\mathbf{x}_0)\cos(\mathbf{x}_1)
+0.0109\mathbf{x}_1\cos(\mathbf{x}_1)
-0.0279\sin(\mathbf{x}_1)
+0.0324$ &
$1395.93$ & $0.099724$ \\

\addlinespace[2pt]

7 &
$5.9616\sin(\mathbf{x}_0)\cos(\mathbf{x}_1)
+0.0188\cos(\mathbf{x}_1)
+0.0608\cos\{\sin(\mathbf{x}_0)\}
-0.0545$ &
$1394.34$ & $0.099799$ \\

\addlinespace[2pt]

8 &
$5.9902\sin(\mathbf{x}_0)\cos(\mathbf{x}_1)
+0.0578\sin\{\cos(\mathbf{x}_0)\}\cos(\mathbf{x}_0)
-0.1480\cos[\cos\{\cos(\mathbf{x}_0)\}]
+0.0736$ &
$1393.06$ & $0.099722$ \\

\addlinespace[2pt]

9 &
$0.0166\sin\{\cos(\mathbf{x}_0\mathbf{x}_1)\}
+5.9499\sin(\mathbf{x}_0)\cos(\mathbf{x}_1)
-0.0007\mathbf{x}_1\cos(\mathbf{x}_1)\cos(\mathbf{x}_0\mathbf{x}_1)
-0.0240$ &
$1377.14$ & $0.099818$ \\

\addlinespace[2pt]

10 &
$6.3906\sin\{\sin(\mathbf{x}_0)\}\cos(\mathbf{x}_1)
+0.8903\cos\{\sin(\mathbf{x}_0)\}
+0.1985\sin[\sin(\mathbf{x}_0\mathbf{x}_1)]
-0.9055$ &
$1365.28$ & $0.100016$ \\

\bottomrule
\bottomrule
\end{tabularx}
{
\scriptsize{
\emph{Note}: \texttt{RMSE} corresponds to the out-of-sample \texttt{RMSE} computed on a $10\%$ held-out test set with the LMPSE selected symbolic model.
}
}
\end{table}

\begin{table}[H]
\centering
\scriptsize
\renewcommand{\arraystretch}{1.15}
\caption{Top $\mathsf{r}=10$ ranked symbolic expressions recovered by a single run of \vasstmain\ under $\sigma=0.2$ for learning~\eqnref{eq:simulation-generating-model-2}. Expressions are ranked by $\mathrm{LMPSE}$ among $H=2000$ sampled hard symbolic trees.}
\label{tab:sim3-noise2-vasst-top10-jmp}
\begin{tabularx}{\linewidth}{@{}c >{\raggedright\arraybackslash}X r r@{}}
\toprule
\toprule
\textbf{Rank} & \textbf{\vasstmain\ Expression Learned} & $\mathrm{LMPSE}$ & \textbf{\texttt{RMSE}} \\
\midrule
1 &
${5.9933\sin(\mathbf{x}_0)\cos(\mathbf{x}_1)}
-0.0031\mathbf{x}_1
-0.0039\mathbf{x}_1\cos\{\cos(\mathbf{x}_1)\}
+0.0074$ &
$251.27$ & $0.199438$ \\

\addlinespace[2pt]

2 &
$6.0022\sin(\mathbf{x}_0)\cos(\mathbf{x}_1)
+0.0240\sin[\cos\{\cos(\mathbf{x}_0)\}]
+0.0002\mathbf{x}_1
-0.0193$ &
$246.74$ & $0.199424$ \\

\addlinespace[2pt]

3 &
$5.9187\sin(\mathbf{x}_0)\cos(\mathbf{x}_1)
-0.0518\mathbf{x}_0
+0.0578\cos\{\cos(\mathbf{x}_1)\}
-0.0486$ &
$246.36$ & $0.199518$ \\

\addlinespace[2pt]

4 &
$6.0085\sin(\mathbf{x}_0)\cos(\mathbf{x}_1)
+0.0342\cos\{\cos(\mathbf{x}_0)\}
-0.0245$ &
$242.70$ & $0.199409$ \\

\addlinespace[2pt]

5 &
$5.9926\sin(\mathbf{x}_0)\cos(\mathbf{x}_1)
+0.0199\mathbf{x}_1\cos(\mathbf{x}_1)
-0.0530\sin(\mathbf{x}_1)
+0.0621$ &
$241.43$ & $0.199430$ \\

\addlinespace[2pt]

6 &
$5.9982\sin(\mathbf{x}_0)\cos(\mathbf{x}_1)
+0.0033\cos(\mathbf{x}_1)
-0.0110\cos\{\sin(\mathbf{x}_0)\}
+0.0066$ &
$239.50$ & $0.199429$ \\

\addlinespace[2pt]

7 &
$6.0095\sin(\mathbf{x}_0)\cos(\mathbf{x}_1)
+0.0712\sin\{\cos(\mathbf{x}_0)\}\cos(\mathbf{x}_0)
-0.2601\cos[\cos\{\cos(\mathbf{x}_0)\}]
+0.1577$ &
$238.64$ & $0.199379$ \\

\addlinespace[2pt]

8 &
$5.8496\sin(\mathbf{x}_0)\cos(\mathbf{x}_1)
-0.0080\cos(\mathbf{x}_0\mathbf{x}_1)\cos(\mathbf{x}_1)
+0.0399\mathbf{x}_0\mathbf{x}_1\cos(\mathbf{x}_1)
-0.0200$ &
$238.59$ & $0.199601$ \\

\addlinespace[2pt]

9 &
$5.3764\mathbf{x}_0\cos(\mathbf{x}_1)
-2.6272\sin\{\cos(\mathbf{x}_0)\}
-0.7567\sin(\mathbf{x}_0)
+2.2304$ &
$236.51$ & $0.201138$ \\

\addlinespace[2pt]

10 &
$5.3674\mathbf{x}_0\cos(\mathbf{x}_1)
-0.7371\sin(\mathbf{x}_0)\cos(\mathbf{x}_0)
-1.6969\sin\{\cos(\mathbf{x}_0)\}
+1.4529$ &
$236.17$ & $0.201004$ \\

\bottomrule
\bottomrule
\end{tabularx}
{
\scriptsize{
\emph{Note}: \texttt{RMSE} corresponds to the out-of-sample \texttt{RMSE} computed on a $10\%$ held-out test set with the LMPSE selected symbolic model.
}
}
\end{table}
\newpage

\section{Description of Feynman Equations}
\label{app:feynman-input-output}

\begin{table}[!htp]
\centering
\caption{Description of the response variables and input features for the Feynman equations in \hyperref[subsec:feynman-data-study]{Section~\ref{subsec:feynman-data-study}}.}
\label{tab:feynman-features}
\scriptsize
\begin{tabular}{lllll}
\toprule
\toprule
\textbf{Label} & \textbf{Feynman Equation} & \textbf{Variable} & \textbf{Type} & \textbf{Description} \\
\midrule

\eqnref{eq:feynman-cl} & Coulomb's law
& $F=0.08\tfrac{q_1q_2}{\epsilon r^{2}}$ & Response & Electrostatic force between two charges \\
& & $q_1, q_2$ & Input & Charges of two particles \\
& & $r$ & Input & Distance between the two charges \\
& & $\epsilon$ & Constant & Permittivity of the medium \\

\midrule

\eqnref{eq:feynman-cpe} & Change in gravitational potential energy
& $\Delta U = Gm_1m_2(\tfrac{1}{r_2} - \tfrac{1}{r_1})$ & Response & Change in gravitational potential energy \\
& & $m_1, m_2$ & Input & Masses of two objects \\
& & $r_1, r_2$ & Input & Initial and final separation distances \\
& & $G$ & Constant & Gravitational constant \\

\midrule

\eqnref{eq:feynman-fce} & Force on a charge in electromagnetic field
& $F=q(E_f + vB\sin\theta)$ & Response & Force on a moving charge in electromagnetic field \\
& & $q$ & Input & Electric charge of the particle \\
& & $E_f$ & Input & Electric field strength \\
& & $v$ & Input & Particle velocity magnitude \\
& & $B$ & Input & Magnetic field strength \\
& & $\theta$ & Input & Angle between velocity and magnetic field \\

\midrule

\eqnref{eq:feynman-ftc} & Fourier's law of thermal conduction
& $P=\tfrac{\kappa A(T_2-T_1)}{d}$ & Response & Rate of heat transfer or thermal power \\
& & $A$ & Input & Cross-sectional area for heat flow \\
& & $T_1, T_2$ & Input & Temperatures at the two boundaries \\
& & $d$ & Input & Distance between the two temperature points \\
& & $\kappa$ & Constant & Thermal conductivity of the material \\

\midrule

\eqnref{eq:feynman-ada} & Angular distribution asymmetry
& $f=\beta^{\dagger}(1+\alpha^{\dagger}\cos\theta)$ & Response & Angular distribution as a function of $\theta$ \\
& & $\beta^{\dagger}$ & Input & Baseline scale parameter of the angular distribution \\
& & $\alpha^{\dagger}$ & Input & Angular asymmetry parameter \\
& & $\theta$ & Input & Observation angle \\

\bottomrule
\bottomrule
\end{tabular}
\end{table}

\newpage

\section{Additional Results for Feynman Equations}
\label{app:additional-feynman-results}

\subsection{Symbolic Expressions Learned by \texorpdfstring{\vasst}{VaSST} and Competing Methods}
\label{app:additional-feynman-results-expressions-learned}

\begin{center}
\underline{For learning \eqnref{eq:feynman-cl}: $F = 0.08\tfrac{q_1 q_2}{\epsilon r^{2}}$ in~\hyperref[subsec:feynman-data-study]{Section~\ref{subsec:feynman-data-study}}}.
\end{center}

\begin{table}[H]
\centering
\scriptsize
\renewcommand{\arraystretch}{1.15}
\caption{Expressions recovered by \vasstmain\ and competing methods for learning \eqnref{eq:feynman-cl} under the noiseless setting.}
\label{tab:feynman-I-12-2-substituted-results}
\begin{tabularx}{\linewidth}{@{}l >{\raggedright\arraybackslash}X r@{}}
\toprule
\toprule
\textbf{Method} & \textbf{Symbolic Expression Learned} & \textbf{\texttt{RMSE}} \\
\midrule

\rowcolor{vasstgray}
\vasstmain &
$\,\underline{0.0797\,\tfrac{q_1q_2}{\epsilon r^2}}
-0.0001\,\tfrac{q_1^2}{\epsilon r^2}
-0.0003$ &
$0.000138$ \\

\addlinespace[1pt]

\operon &
$0.0799\,\tfrac{q_1q_2}{\epsilon r^2}$ &
$1.08{\times}10^{-7}$ \\

\addlinespace[1pt]

\pysr &
$0.0796\,\tfrac{q_1q_2}{\epsilon r^2}$ &
$1.77{\times}10^{-8}$ \\

\addlinespace[1pt]

\gplearn &
$0.0793\,\tfrac{q_1q_2}{\epsilon r^2}$ &
$0.000554$ \\

\addlinespace[1pt]

\deap &
$0.2055\,\tfrac{q_1^2}{\epsilon r^2}$ &
$0.046139$ \\

\addlinespace[1pt]

\bsr &
$0.1469
-0.0058\,\tfrac{q_1^4q_2^2}{\epsilon^3r^6}
+0.0043\,\tfrac{q_1q_2^3}{\epsilon^2r^4}
+0.0122\,\tfrac{q_1^5q_2}{\epsilon^3r^6}$ &
$0.0264$ \\

\addlinespace[1pt]

\bms &
$a_0\,\tfrac{q_1q_2}{\epsilon r^2}$ &
$3.70{\times}10^{-17}$ \\

\addlinespace[1pt]

\qlattice &
$-0.7418
\left(
0.5132\tfrac{q_1^2}{\epsilon r^2}
+4.1738{\times}10^{-6}
\right)
\left(
-0.2091\tfrac{q_2}{q_1}
-1.9028{\times}10^{-6}
\right)
-1.5291{\times}10^{-6}$ &
$4.64{\times}10^{-7}$ \\

\addlinespace[1pt]

\dsr &
$\dfrac{q_2^2}{\epsilon r^2}$ &
$13.248890$ \\

\bottomrule
\bottomrule
\end{tabularx}
{
\scriptsize{
\emph{Note}: $a_0$ is a constant learned by \bms. \texttt{RMSE} corresponds to the out-of-sample \texttt{RMSE} computed on a $10\%$ held-out test set.
}
}
\end{table}

\begin{table}[H]
\centering
\scriptsize
\renewcommand{\arraystretch}{1.15}
\caption{Expressions recovered by \vasstmain\ and competing methods for learning \eqnref{eq:feynman-cl} under $\sigma=0.1$.}
\label{tab:feynman-I-12-2-noise0p1-substituted-results}
\begin{tabularx}{\linewidth}{@{}l >{\raggedright\arraybackslash}X r@{}}
\toprule
\toprule
\textbf{Method} & \textbf{Symbolic Expression Learned} & \textbf{\texttt{RMSE}} \\
\midrule

\rowcolor{vasstgray}
\vasstmain &
$\,\underline{0.080738\,\dfrac{q_1q_2}{\epsilon r^2}}
+0.001896\,\dfrac{q_1^2}{\epsilon r^2}
-0.008712$ &
$0.099706$ \\

\addlinespace[1pt]

\operon &
$0.081301\,\tfrac{q_1q_2}{\epsilon r^2}
-0.008000$ &
$0.099706$ \\

\addlinespace[1pt]

\pysr &
$0.211932\,\tfrac{q_1^2}{\epsilon r^2}$ &
$0.109963$ \\

\addlinespace[1pt]

\gplearn &
$0.079296\,\tfrac{q_1q_2}{\epsilon r^2}$ &
$0.099793$ \\

\addlinespace[1pt]

\deap &
$0.190037\,\tfrac{q_1^2}{\epsilon r^2}$ &
$0.117293$ \\

\addlinespace[1pt]

\bsr &
$0.143050
-0.005868\,\tfrac{q_1^4q_2^2}{\epsilon^3r^6}
+0.004369\,\tfrac{q_1q_2^3}{\epsilon^2r^4}
+0.012211\,\tfrac{q_1^5q_2}{\epsilon^3r^6}$ &
$0.103577$ \\

\addlinespace[1pt]

\bms &
$a_0\,\tfrac{q_1q_2}{\epsilon r^2}$ &
$0.099776$ \\

\addlinespace[1pt]

\qlattice &
$-0.612279
\left(
0.121789
-0.711697\tfrac{q_1^2}{\epsilon r^2}
\right)
\left(
0.200641\tfrac{q_2}{q_1}
-0.0385421
\right)
+0.0278299$ &
$0.099699$ \\

\addlinespace[1pt]

\dsr &
$\tfrac{q_2^2}{\epsilon r^2}$ &
$13.247801$ \\

\bottomrule
\bottomrule
\end{tabularx}
{
\scriptsize{
\emph{Note}: $a_0$ is a constant learned by \bms. \texttt{RMSE} corresponds to the out-of-sample \texttt{RMSE} computed on a $10\%$ held-out test set.
}
}
\end{table}
\begin{table}[H]
\centering
\scriptsize
\renewcommand{\arraystretch}{1.15}
\caption{Expressions recovered by \vasstmain\ and competing methods for learning \eqnref{eq:feynman-cl} under $\sigma=0.2$.}
\label{tab:feynman-I-12-2-noise0p2-substituted-results}
\begin{tabularx}{\linewidth}{@{}l >{\raggedright\arraybackslash}X r@{}}
\toprule
\toprule
\textbf{Method} & \textbf{Symbolic Expression Learned} & \textbf{\texttt{RMSE}} \\
\midrule

\rowcolor{vasstgray}
\vasstmain &
$\,\underline{0.081740\,\tfrac{q_1q_2}{\epsilon r^2}}
+0.003941\,\tfrac{q_1^2}{\epsilon r^2}
-0.017138$ &
$0.199411$ \\

\addlinespace[1pt]

\operon &
$0.082857\,\tfrac{q_1q_2}{\epsilon r^2}
-0.016000$ &
$0.199412$ \\

\addlinespace[1pt]

\pysr &
$0.212564\,\tfrac{q_1^2}{\epsilon r^2}$ &
$0.205145$ \\

\addlinespace[1pt]

\gplearn &
$0.080850\,\tfrac{q_1q_2}{\epsilon r^2}$ &
$0.199569$ \\

\addlinespace[1pt]

\deap &
$0.190037\,\tfrac{q_1^2}{\epsilon r^2}$ &
$0.209396$ \\

\addlinespace[1pt]

\bsr &
$0.139186
-0.005879\,\tfrac{q_1^4q_2^2}{\epsilon^3r^6}
+0.004428\,\tfrac{q_1q_2^3}{\epsilon^2r^4}
+0.012213\,\tfrac{q_1^5q_2}{\epsilon^3r^6}$ &
$0.201605$ \\

\addlinespace[1pt]

\bms &
$a_0\,\tfrac{q_1q_2}{\epsilon r^2}$ &
$0.199551$ \\

\addlinespace[1pt]

\qlattice &
$-0.714214
\left(
0.138143
-0.685898\tfrac{q_1^2}{\epsilon r^2}
\right)
\left(
0.182712\tfrac{q_2}{q_1}
-0.0355843
\right)
+0.0228361$ &
$0.199442$ \\

\addlinespace[1pt]

\dsr &
$\tfrac{q_2^2}{\epsilon r^2}$ &
$13.247464$ \\

\bottomrule
\bottomrule
\end{tabularx}
{
\scriptsize{
\emph{Note}: $a_0$ is a constant learned by \bms. \texttt{RMSE} corresponds to the out-of-sample \texttt{RMSE} computed on a $10\%$ held-out test set.
}
}
\end{table}

\newpage

\begin{center}
\underline{For learning \eqnref{eq:feynman-cpe}: $\Delta U = G m_1 m_2(\tfrac{1}{r_2} - \tfrac{1}{r_1})$ in~\hyperref[subsec:feynman-data-study]{Section~\ref{subsec:feynman-data-study}}}.
\end{center}

\begin{table}[H]
\centering
\scriptsize
\renewcommand{\arraystretch}{1.15}
\caption{Expressions recovered by \vasstmain\ and competing methods for learning \eqnref{eq:feynman-cpe} under the noiseless setting.}
\label{tab:feynman-I-13-12-substituted-results}
\begin{tabularx}{\linewidth}{@{}l >{\raggedright\arraybackslash}X r@{}}
\toprule
\toprule
\textbf{Method} & \textbf{Symbolic Expression Learned} & \textbf{\texttt{RMSE}} \\
\midrule

\rowcolor{vasstgray}
\vasstmain &
$\,\underline{1.014839\,\tfrac{Gm_1m_2}{r_2}
-0.998513\,\tfrac{Gm_1m_2}{r_1}}
-0.003078\,\tfrac{r_2}{r_1}
-0.021220$ &
$0.003817$ \\

\addlinespace[1pt]

\operon &
$0.998287\,\tfrac{Gm_1m_2}{r_2}
-0.999754\,\tfrac{Gm_1m_2}{r_1}$ &
$0.000155$ \\

\addlinespace[1pt]

\pysr &
$Gm_1m_2\left(\tfrac{1}{r_2}-\tfrac{1}{r_1}\right)$ &
$8.57{\times}10^{-16}$ \\

\addlinespace[1pt]

\gplearn &
$Gm_1m_2\left(\tfrac{1}{r_2}-\tfrac{1}{r_1}\right)$ &
$0.000000$ \\

\addlinespace[1pt]

\deap &
$-\tfrac{Gm_1m_2}{r_1}
-\tfrac{Gm_1^2}{r_1}
+\tfrac{m_2}{m_1}
+\tfrac{2.8770}{-Gm_1^2/r_1+m_2/m_1}
-4.5144$ &
$0.556532$ \\

\addlinespace[1pt]

\bsr &
$8.54998
+0.826934\,\tfrac{m_2r_2}{m_1r_1}
+0.253004\left(
\tfrac{Gm_1m_2r_2}{r_1^2}
-\tfrac{Gm_1^2}{r_1}
\right)
+0.379514\,\tfrac{m_2^2}{m_1^2}$ &
$1.090235$ \\

\addlinespace[1pt]

\bms &
$-\tfrac{Gm_1m_2}{r_1}
+\tfrac{Gm_1m_2}{a_0r_2}$ &
$0.000000$ \\

\addlinespace[1pt]

\qlattice &
$15.5019
\left(
-0.511191\tfrac{Gm_1^2}{r_1}
-0.0180004
\right)
\left(
0.200553\tfrac{m_2}{m_1}
-0.0060554
\right)
\left(
-0.0263949\tfrac{r_2}{r_1}
-1.70128
\right)
\{
-0.104696\tfrac{r_2}{r_1}
-\tanh\!(-0.114106\tfrac{r_2}{r_1}-1.02825)
-1.43374
\}
+0.161139$ &
$0.072571$ \\

\addlinespace[1pt]

\dsr &
$Gm_1m_2\left(\tfrac{1}{r_2}-\tfrac{1}{r_1}\right)$ &
$8.58{\times}10^{-16}$ \\

\bottomrule
\bottomrule
\end{tabularx}
{
\scriptsize{
\emph{Note}: $a_0$ is a constant learned by \bms. \texttt{RMSE} corresponds to the out-of-sample \texttt{RMSE} computed on a $10\%$ held-out test set.
}
}
\end{table}

\begin{table}[H]
\centering
\scriptsize
\renewcommand{\arraystretch}{1.15}
\caption{Expressions recovered by \vasstmain\ and competing methods for learning \eqnref{eq:feynman-cpe} under $\sigma=0.1$.}
\label{tab:feynman-I-13-12-noise0p1-substituted-results}
\begin{tabularx}{\linewidth}{@{}l >{\raggedright\arraybackslash}X r@{}}
\toprule
\toprule
\textbf{Method} & \textbf{Symbolic Expression Learned} & \textbf{\texttt{RMSE}} \\
\midrule

\rowcolor{vasstgray}
\vasstmain &
$\,\underline{1.014802\,\tfrac{Gm_1m_2}{r_2}
-0.999747\,\tfrac{Gm_1m_2}{r_1}}
-0.004945\,\tfrac{r_2}{r_1}
-0.038218$ &
$0.099843$ \\

\addlinespace[1pt]

\operon &
$0.997823\,\tfrac{Gm_1m_2}{r_2}
-1.003906\,\tfrac{Gm_1m_2}{r_1}
-0.009000$ &
$0.099712$ \\

\addlinespace[1pt]

\pysr &
$Gm_1m_2\left(\tfrac{1}{r_2}-\tfrac{1}{r_1}\right)$ &
$0.099784$ \\

\addlinespace[1pt]

\gplearn &
$Gm_1m_2\left(\tfrac{1}{r_2}-\tfrac{1}{r_1}\right)$ &
$0.099784$ \\

\addlinespace[1pt]

\deap &
$-\tfrac{Gm_1m_2}{r_1}
-\tfrac{Gm_1^2}{r_1}
+\tfrac{m_2}{m_1}
+\tfrac{2.8770}{-Gm_1^2/r_1+m_2/m_1}
-4.5144$ &
$0.564690$ \\

\addlinespace[1pt]

\bsr &
$8.54536
+0.827299\,\tfrac{m_2r_2}{m_1r_1}
+0.253255\left(
\tfrac{Gm_1m_2r_2}{r_1^2}
-\tfrac{Gm_1^2}{r_1}
\right)
+0.379697\,\tfrac{m_2^2}{m_1^2}$ &
$1.094409$ \\

\addlinespace[1pt]

\bms &
$-\tfrac{Gm_1m_2}{r_1}
+\tfrac{a_0Gm_1m_2}{2r_2}$ &
$0.099778$ \\

\addlinespace[1pt]

\qlattice &
$15.5533
\left(
-0.230750\tfrac{Gm_1^2}{r_1}
+0.00218109
\right)
\left(
0.355623\tfrac{m_2}{m_1}
-0.144215
+1.54939\exp\!\left\{0.638983\tfrac{r_2}{r_1}\right\}
\right)
+15.5533\exp\![
-0.734346(0.428116\tfrac{m_2}{m_1}+1)^2
-0.747608(-0.00406677\tfrac{r_2}{r_1}+1)^2
]
-0.100993$ &
$0.126651$ \\

\addlinespace[1pt]

\dsr &
$Gm_1m_2\left(\tfrac{1}{r_2}-\tfrac{1}{r_1}\right)$ &
$0.099784$ \\

\bottomrule
\bottomrule
\end{tabularx}
{
\scriptsize
{
\emph{Note}: $a_0$ is a constant learned by \bms. \texttt{RMSE} corresponds to the out-of-sample \texttt{RMSE} computed on a $10\%$ held-out test set.
}
}
\end{table}

\begin{table}[H]
\centering
\scriptsize
\renewcommand{\arraystretch}{1.15}
\caption{Expressions recovered by \vasstmain\ and competing methods for learning \eqnref{eq:feynman-cpe} under $\sigma=0.2$.}
\label{tab:feynman-I-13-12-noise0p2-substituted-results}
\begin{tabularx}{\linewidth}{@{}l >{\raggedright\arraybackslash}X r@{}}
\toprule
\toprule
\textbf{Method} & \textbf{Symbolic Expression Learned} & \textbf{\texttt{RMSE}} \\
\midrule

\rowcolor{vasstgray}
\vasstmain &
$\,\underline{1.014765\,\tfrac{Gm_1m_2}{r_2}
-1.000981\,\tfrac{Gm_1m_2}{r_1}}
-0.006813\,\tfrac{r_2}{r_1}
-0.055216$ &
$0.199507$ \\

\addlinespace[1pt]

\operon &
$-0.010000
-0.332\left[
\tfrac{
0.707\,Gm_1^2/r_1
-6.275451\,Gm_1m_2/r_1
}
{2.031\,r_2/r_1}
+3.007296\,\tfrac{Gm_1m_2}{r_1}
+\tfrac{
0.670\,Gm_1^2/r_1
}
{
3.388\,Gm_1^2/r_1
+2.718552\,m_2^2/m_1^2
}
\right]$ &
$0.199207$ \\

\addlinespace[1pt]

\pysr &
$Gm_1m_2\left(\tfrac{1}{r_2}-\tfrac{1}{r_1}\right)$ &
$0.199569$ \\

\addlinespace[1pt]

\gplearn &
$Gm_1m_2\left(\tfrac{1}{r_2} - \tfrac{1}{r_1}\right)$ &
$0.199569$ \\

\addlinespace[1pt]

\deap &
$-\tfrac{Gm_1m_2}{r_1}
-\tfrac{Gm_1^2}{r_1}
+\tfrac{m_2}{m_1}
+\tfrac{2.8770}{-Gm_1^2/r_1+m_2/m_1}
-4.5144$ &
$0.589860$ \\

\addlinespace[1pt]

\bsr &
$8.54074
+0.827664\,\tfrac{m_2r_2}{m_1r_1}
+0.253507\left(
\tfrac{Gm_1m_2r_2}{r_1^2}
-\tfrac{Gm_1^2}{r_1}\right)
+0.379881\,\tfrac{m_2^2}{m_1^2}$ &
$1.107579$ \\

\addlinespace[1pt]

\bms &
$-\tfrac{Gm_1m_2}{r_1}
+\tfrac{a_0Gm_1m_2}{2r_2}$ &
$0.199557$ \\

\addlinespace[1pt]

\qlattice &
$2.11181
\left(
-0.270397\tfrac{Gm_1^2}{r_1}
-0.00788965
\right)
\left(
-0.0392111\tfrac{r_2}{r_1}
+1.45354
+1.86691\exp\!\left\{0.281573\tfrac{r_2}{r_1}\right\}
\right)
\sqrt{
(
\tfrac{m_2}{m_1}
-0.0186931
)^2}
+0.238388$ &
$0.246633$ \\

\addlinespace[1pt]

\dsr &
$Gm_1m_2\left(\tfrac{1}{r_2}-\tfrac{1}{r_1}\right)$ &
$0.199569$ \\

\bottomrule
\bottomrule
\end{tabularx}
{
\scriptsize
{
\emph{Note}: $a_0$ is a constant learned by \bms. \texttt{RMSE} corresponds to the out-of-sample \texttt{RMSE} computed on a $10\%$ held-out test set.
}
}
\end{table}
\newpage

\begin{center}
\underline{For learning \eqnref{eq:feynman-fce}: $F = q(E_f + Bv \sin \theta)$ in~\hyperref[subsec:feynman-data-study]{Section~\ref{subsec:feynman-data-study}}}.
\end{center}

\begin{table}[H]
\centering
\scriptsize
\renewcommand{\arraystretch}{1.15}
\caption{Expressions recovered by \vasstmain\ and competing methods for learning \eqnref{eq:feynman-fce} under the noiseless setting.}
\label{tab:feynman-I-12-11-substituted-results}
\begin{tabularx}{\linewidth}{@{}l >{\raggedright\arraybackslash}X r@{}}
\toprule
\toprule
\textbf{Method} & \textbf{Symbolic Expression Learned} & \textbf{\texttt{RMSE}} \\
\midrule

\rowcolor{vasstgray}
\vasstmain &
$\,\underline{0.979592\,qE_f
+1.066203\,qBv\sin\theta}
-0.002652\,\theta
+0.001420$ &
$0.001574$ \\

\addlinespace[1pt]

\operon &
$0.988416\,qE_f
+0.439639\,qBv\sin(0.982\theta)
+0.565508\,qBv\sin(1.014\theta)
+0.074$ &
$0.277625$ \\

\addlinespace[1pt]

\pysr &
$qE_f+qBv\sin\theta$ &
$0.000000$ \\

\addlinespace[1pt]

\gplearn &
$qE_f+qBv\sin\theta$ &
$0.000000$ \\

\addlinespace[1pt]

\deap &
$\left\{
qBv
+4\sin\theta
+\sin(\sin\theta)
+\sin(\sin(\sin\theta))
\right\}\sin\theta$ &
$5.511325$ \\

\addlinespace[1pt]

\bsr &
$14.9902
+0.245458\left[
\left(\tfrac{Bv}{E_f}+\theta^2\right)\sin\theta
+\sin\left(\tfrac{Bv}{E_f}\right)
\right]
+0.182367\,\tfrac{Bv}{E_f}\sin(qE_f\theta)
-2.18703\left\{
\tfrac{Bv}{E_f}
+\sin\left(\tfrac{Bv\theta}{E_f}\right)
\right\}$ &
$19.116390$ \\

\addlinespace[1pt]

\bms &
$qBv\sin\theta+a_0^2qE_f$ &
$0.000000$ \\

\addlinespace[1pt]

\qlattice &
$95.8349(0.0317161qE_f-1.20835)
\exp\![
-2\exp\!\{
-1.83685(0.0751253\tfrac{Bv}{E_f}-1)^2
-23.7706(0.125944\theta-1)^2
\}
-2\exp\!\{
-1.96799(0.634099\theta-1)^2
\}
]
+74.3577$ &
$13.875759$ \\

\addlinespace[1pt]

\dsr &
$qE_f+qBv\sin\theta$ &
$0.000000$ \\

\bottomrule
\bottomrule
\end{tabularx}
{
\scriptsize
{
\emph{Note}: $a_0$ is a constant learned by \bms. \texttt{RMSE} corresponds to the out-of-sample \texttt{RMSE} computed on a $10\%$ held-out test set.
}
}
\end{table}

\begin{table}[H]
\centering
\scriptsize
\renewcommand{\arraystretch}{1.15}
\caption{Expressions recovered by \vasstmain\ and competing methods for learning \eqnref{eq:feynman-fce} under $\sigma=0.1$.}
\label{tab:feynman-I-12-11-noise0p1-substituted-results}
\begin{tabularx}{\linewidth}{@{}l >{\raggedright\arraybackslash}X r@{}}
\toprule
\toprule
\textbf{Method} & \textbf{Symbolic Expression Learned} & \textbf{\texttt{RMSE}} \\
\midrule

\rowcolor{vasstgray}
\vasstmain &
$\,\underline{0.979179\,qE_f
+1.066295\,qBv\sin\theta}
-0.002220\,\theta
+0.001432$ &
$0.101577$ \\

\addlinespace[1pt]

\operon &
$1.003148\,qE_f
+1.002258\,qBv\sin\theta
-0.043937\sin(0.883\theta)
-0.012000$ &
$0.101138$ \\

\addlinespace[1pt]

\pysr &
$qE_f+qBv\sin\theta$ &
$0.098647$ \\

\addlinespace[1pt]

\gplearn &
$qE_f+qBv\sin\theta$ &
$0.098647$ \\

\addlinespace[1pt]

\deap &
$\left\{
qBv
+4\sin\theta
+\sin(\sin\theta)
+\sin(\sin(\sin\theta))
\right\}\sin\theta$ &
$5.513179$ \\

\addlinespace[1pt]

\bsr &
$14.9833
+0.245459\left[
\left(\tfrac{Bv}{E_f}+\theta^2\right)\sin\theta
+\sin\left(\tfrac{Bv}{E_f}\right)
\right]
+0.182854\,\tfrac{Bv}{E_f}\sin(qE_f\theta)
-2.18523\left\{
\tfrac{Bv}{E_f}
+\sin\left(\tfrac{Bv\theta}{E_f}\right)
\right\}$ &
$19.119963$ \\

\addlinespace[1pt]

\bms &
$qE_f+qBv\sin\theta$ &
$0.098647$ \\

\addlinespace[1pt]

\qlattice &
$255.963
\exp\!\left[
-5.76788\left\{
0.446894
-\exp\!\left[-2.16672(1-0.213868\theta)^2\right]
\right\}^2
-1.78239(1-0.0236962qE_f)^2
\right]
-35.3727$ &
$14.287933$ \\

\addlinespace[1pt]

\dsr &
$qE_f+qBv\sin\theta$ &
$0.098647$ \\

\bottomrule
\bottomrule
\end{tabularx}
{
\scriptsize
{
\emph{Note}: \texttt{RMSE} corresponds to the out-of-sample \texttt{RMSE} computed on a $10\%$ held-out test set.
}
}
\end{table}

\begin{table}[H]
\centering
\scriptsize
\renewcommand{\arraystretch}{1.15}
\caption{Expressions recovered by \vasstmain\ and competing methods for learning \eqnref{eq:feynman-fce} under $\sigma=0.2$.}
\label{tab:feynman-I-12-11-noise0p2-substituted-results}
\begin{tabularx}{\linewidth}{@{}l >{\raggedright\arraybackslash}X r@{}}
\toprule
\toprule
\textbf{Method} & \textbf{Symbolic Expression Learned} & \textbf{\texttt{RMSE}} \\
\midrule

\rowcolor{vasstgray}
\vasstmain &
$\,\underline{0.978765\,qE_f
+1.066388\,qBv\sin\theta}
-0.001789\,\theta
+0.001444$ &
$0.201587$ \\

\addlinespace[1pt]

\operon &
$0.995742\,qE_f
+0.427023\,qBv\sin(0.989\theta)
+0.576095\,qBv\sin(1.008\theta)
+0.033000$ &
$0.216361$ \\

\addlinespace[1pt]

\pysr &
$1.000169\,qE_f+qBv\sin\theta$ &
$0.197290$ \\

\addlinespace[1pt]

\gplearn &
$qE_f+qBv\sin\theta$ &
$0.197294$ \\

\addlinespace[1pt]

\deap &
$\left\{
qBv
+4\sin\theta
+\sin(\sin\theta)
+\sin(\sin(\sin\theta))
\right\}\sin\theta$ &
$5.516797$ \\

\addlinespace[1pt]

\bsr &
$14.9764
+0.245460\left[
\left(\tfrac{Bv}{E_f}+\theta^2\right)\sin\theta
+\sin\left(\tfrac{Bv}{E_f}\right)
\right]
+0.183341\,\tfrac{Bv}{E_f}\sin(qE_f\theta)
-2.18344\left\{
\tfrac{Bv}{E_f}
+\sin\left(\tfrac{Bv\theta}{E_f}\right)
\right\}$ &
$19.124043$ \\

\addlinespace[1pt]

\bms &
$qE_f+qBv\sin\theta$ &
$0.197294$ \\

\addlinespace[1pt]

\qlattice &
$71.7853
-111.307\exp\![
-2\{
0.0512725\theta
+\exp\![
-0.887081(1-0.0224396qE_f)^2
-5.71614(
-0.022938\tfrac{Bv}{E_f}
+0.150368\theta
-1
)^2
]
-0.225009
\}^2
-2\exp\!\left\{
-1.30897(0.61342\theta-1)^2
\right\}
]$ &
$13.574922$ \\

\addlinespace[1pt]

\dsr &
$qE_f+qBv\sin\theta$ &
$0.197294$ \\

\bottomrule
\bottomrule
\end{tabularx}
{
\scriptsize
{
\emph{Note}: \texttt{RMSE} corresponds to the out-of-sample \texttt{RMSE} computed on a $10\%$ held-out test set.
}
}
\end{table}

\newpage

\begin{center}
\underline{For learning \eqnref{eq:feynman-ftc}: $P = \frac{\kappa A(T_2-T_1)}{d}$ in~\hyperref[subsec:feynman-data-study]{Section~\ref{subsec:feynman-data-study}}}.
\end{center}

\begin{table}[H]
\centering
\scriptsize
\renewcommand{\arraystretch}{1.15}
\caption{Expressions recovered by \vasstmain\ and competing methods for learning \eqnref{eq:feynman-ftc} under the noiseless setting.}
\label{tab:feynman-II-2-42-substituted-results}
\begin{tabularx}{\linewidth}{@{}l >{\raggedright\arraybackslash}X r@{}}
\toprule
\toprule
\textbf{Method} & \textbf{Symbolic Expression Learned} & \textbf{\texttt{RMSE}} \\
\midrule

\rowcolor{vasstgray}
\vasstmain &
$\,\underline{
-0.805371\,\tfrac{AT_1\kappa}{d}
+0.998520\,\tfrac{A\kappa T_2}{d}}
-0.000254\left\{
\tfrac{1}{\sqrt A\,T_1\kappa}
-\tfrac{d/\sqrt A-T_2/T_1}{\sqrt A\,T_1\kappa-T_2/T_1}
\right\}
+0.014782$ &
$0.012602$ \\

\addlinespace[1pt]

\operon &
$6.66730\,\tfrac{AT_1\kappa}{d}
\left(
0.150000\tfrac{T_2}{T_1}
-0.149937
\right)$ &
$6.30{\times}10^{-5}$ \\

\addlinespace[1pt]

\pysr &
$\tfrac{A\kappa(T_2-T_1)}{d}$ &
$1.59{\times}10^{-16}$ \\

\addlinespace[1pt]

\gplearn &
$\tfrac{A\kappa(T_2-T_1)}{d}$ &
$1.71{\times}10^{-16}$ \\

\addlinespace[1pt]

\deap &
$0.675078\,\tfrac{A\kappa T_2}{d}$ &
$0.199012$ \\

\addlinespace[1pt]

\bsr &
$1.06626
-0.113972\,\tfrac{dT_2}{\sqrt A\,T_1}
+0.031160\left(
\tfrac{d\kappa T_2}{\sqrt A}
-\sqrt A\,T_1\kappa
\right)
+0.069105\,\tfrac{\sqrt A\,T_2^2}{dT_1^2}$ &
$0.178899$ \\

\addlinespace[1pt]

\bms &
$a_0\,\tfrac{A\kappa T_2}{d}$ &
$0.074917$ \\

\addlinespace[1pt]

\qlattice &
$-3.01307
\left(
0.262817
-0.263207\tfrac{T_2}{T_1}
\right)
\left(
-0.597221\sqrt A\,T_1\kappa
-0.00029899
\right)
\tfrac{
0.000501386\,d/\sqrt A
-0.584919
}{
0.274851\,d/\sqrt A
+0.00563991
}
-0.000457$ &
$0.000277$ \\

\addlinespace[1pt]

\dsr &
$\tfrac{A\kappa(T_2-T_1)}{d}$ &
$2.18{\times}10^{-16}$ \\

\bottomrule
\bottomrule
\end{tabularx}
{
\scriptsize
{
\emph{Note}: $a_0$ is a constant learned by \bms. \texttt{RMSE} corresponds to the out-of-sample \texttt{RMSE} computed on a $10\%$ held-out test set.
}
}
\end{table}

\begin{table}[H]
\centering
\scriptsize
\renewcommand{\arraystretch}{1.15}
\caption{Expressions recovered by \vasstmain\ and competing methods for learning \eqnref{eq:feynman-ftc} under $\sigma=0.1$.}
\label{tab:feynman-II-2-42-noise0p1-substituted-results}
\begin{tabularx}{\linewidth}{@{}l >{\raggedright\arraybackslash}X r@{}}
\toprule
\toprule
\textbf{Method} & \textbf{Symbolic Expression Learned} & \textbf{\texttt{RMSE}} \\
\midrule

\rowcolor{vasstgray}
\vasstmain &
$\,
\underline{0.871347\tfrac{A\kappa T_2}{d}}
-0.019212\,\sqrt A\,T_1\kappa
-0.046585\left(
\sqrt A\,T_1\kappa
-1
-\tfrac{T_2}{T_1}
\right)
-0.293006$ &
$0.101254$ \\

\addlinespace[2pt]

\operon &
$1.007879\,\tfrac{A\kappa T_2}{d}
-1.015839\,\tfrac{AT_1\kappa}{d}
+0.249325\,\tfrac{\sqrt A}{d}
-0.063014$ &
$0.099538$ \\

\addlinespace[2pt]

\pysr &
$\tfrac{A\kappa(T_2-1.001694T_1)}{d}$ &
$0.099782$ \\

\addlinespace[2pt]

\gplearn &
$0.900000\,\tfrac{A\kappa(T_2-0.720000T_1)}{d}$ &
$0.111707$ \\

\addlinespace[2pt]

\deap &
$0.675078\,\tfrac{A\kappa T_2}{d}$ &
$0.223923$ \\

\addlinespace[2pt]

\bsr &
$1.09040
-0.117204\,\tfrac{dT_2}{\sqrt A\,T_1}
+0.031340\left(
\tfrac{d\kappa T_2}{\sqrt A}
-\sqrt A\,T_1\kappa
\right)
+0.070823\,\tfrac{\sqrt A\,T_2^2}{dT_1^2}$ &
$0.206912$ \\

\addlinespace[2pt]

\bms &
$a_0\,\tfrac{A\kappa T_2}{d}$ &
$0.126948$ \\

\addlinespace[2pt]

\qlattice &
$-1.19133
\left(
0.409023
-0.399587\tfrac{T_2}{T_1}
\right)
\left(
1.18272\sqrt A\,T_1\kappa
-0.0111254
\right)
\exp\!\left(
-0.203901\tfrac{d}{\sqrt A}
\right)
+0.0143714$ &
$0.101341$ \\

\addlinespace[2pt]

\dsr &
$\tfrac{A\kappa(T_2-T_1)}{d}$ &
$0.099784$ \\

\bottomrule
\bottomrule
\end{tabularx}
{
\scriptsize
{
\emph{Note}: $a_0$ is a constant learned by \bms. \texttt{RMSE} corresponds to the out-of-sample \texttt{RMSE} computed on a $10\%$ held-out test set.
}
}
\end{table}

\begin{table}[H]
\centering
\scriptsize
\renewcommand{\arraystretch}{1.15}
\caption{Expressions recovered by \vasstmain\ and competing methods for learning \eqnref{eq:feynman-ftc} under $\sigma=0.2$.}
\label{tab:feynman-II-2-42-noise0p2-substituted-results}
\begin{tabularx}{\linewidth}{@{}l >{\raggedright\arraybackslash}X r@{}}
\toprule
\toprule
\textbf{Method} & \textbf{Symbolic Expression Learned} & \textbf{\texttt{RMSE}} \\
\midrule

\rowcolor{vasstgray}
\vasstmain &
$\,\underline{
0.903072\,\tfrac{A\kappa T_2}{d}}
-0.094430\,\sqrt A\,T_1\kappa
-0.661969\,\tfrac{T_1}{T_2}
+0.112071$ &
$0.199821$ \\

\addlinespace[2pt]

\operon &
$-0.021000
+0.100\left[
2.053\tfrac{\sqrt{A}}{d}\left\{
1.071212\,\tfrac{T_1}{T_2}
+6.962264\,\sqrt A\,T_1\kappa
\right\}
\left(
0.707\,\tfrac{T_2}{T_1}
-0.693
\right)
-0.002139\left(\tfrac{d}{\sqrt A}\right)^3
\right]$ &
$0.198948$ \\

\addlinespace[2pt]

\pysr &
$\tfrac{A\kappa(T_2-1.003386T_1)}{d}$ &
$0.199565$ \\

\addlinespace[2pt]

\gplearn &
$\tfrac{A\kappa(T_2-0.900000T_1)}{d}$ &
$0.203313$ \\

\addlinespace[2pt]

\deap &
$0.675078\,\tfrac{A\kappa T_2}{d}$ &
$0.283886$ \\

\addlinespace[2pt]

\bsr &
$1.11454
-0.120437\,\tfrac{dT_2}{\sqrt A\,T_1}
+0.031521\left(
\tfrac{d\kappa T_2}{\sqrt A}
-\sqrt A\,T_1\kappa
\right)
+0.072541\,\tfrac{\sqrt A\,T_2^2}{dT_1^2}$ &
$0.271019$ \\

\addlinespace[2pt]

\bms &
$\tfrac{A\kappa T_2}{a_0d}$ &
$0.215714$ \\

\addlinespace[2pt]

\qlattice &
$2.2079
(
0.36293\tfrac{T_2}{T_1}
-0.419621
)
[
0.167743\sqrt A\,T_1\kappa
+(
0.0940613\tfrac{d}{\sqrt A}
-0.639193
)
(
-0.581647\sqrt A\,T_1\kappa
+0.248409\tfrac{d}{\sqrt A}
-0.775808
)
+0.0568153
]
+0.0549773$ &
$0.200905$ \\

\addlinespace[2pt]

\dsr &
$\tfrac{A\kappa(T_2-T_1)}{d}$ &
$0.199569$ \\

\bottomrule
\bottomrule
\end{tabularx}
{
\scriptsize
{
\emph{Note}: $a_0$ is a constant learned by \bms. \texttt{RMSE} corresponds to the out-of-sample \texttt{RMSE} computed on a $10\%$ held-out test set.
}
}
\end{table}

\newpage

\begin{center}
\underline{For learning \eqnref{eq:feynman-ada}: $f = \beta^{\dagger}(1 + \alpha^{\dagger}\cos \theta)$ in~\hyperref[subsec:feynman-data-study]{Section~\ref{subsec:feynman-data-study}}}.
\end{center}

\begin{table}[H]
\centering
\scriptsize
\renewcommand{\arraystretch}{1.15}
\caption{Expressions recovered by \vasstmain\ and competing methods for learning \eqnref{eq:feynman-ada} under the noiseless setting.}
\label{tab:feynman-III-17-37-substituted-results}
\begin{tabularx}{\linewidth}{@{}l >{\raggedright\arraybackslash}X r@{}}
\toprule
\toprule
\textbf{Method} & \textbf{Symbolic Expression Learned} & \textbf{\texttt{RMSE}} \\
\midrule

\rowcolor{vasstgray}
\vasstmain &
$\,\underline{
1.023842\,\beta^\dagger\{1+\alpha^\dagger\cos\theta\}}
+0.001029\,\theta
+0.000125\cos(2\alpha^\dagger)
-0.000899$ &
$0.001349$ \\

\addlinespace[1pt]

\operon &
$1.91884\,\beta^\dagger
\left\{
0.520\,\alpha^\dagger\cos(0.994\theta)
+0.002\,\theta
+0.366
\right\}
+0.417217\,\beta^\dagger\cos(0.594\theta)\cos(0.583\theta)
+0.004$ &
$0.133614$ \\

\addlinespace[1pt]

\pysr &
$\beta^\dagger\{1+\alpha^\dagger\cos\theta\}$ &
$0.000000$ \\

\addlinespace[1pt]

\gplearn &
$\beta^\dagger+\alpha^\dagger\beta^\dagger\cos\theta$ &
$1.39{\times}10^{-15}$ \\

\addlinespace[1pt]

\deap &
$\cos\theta
\left\{
\alpha^\dagger\beta^\dagger
+3\cos\theta
+\cos(1.020634\theta)
\right\}$ &
$2.220867$ \\

\addlinespace[1pt]

\bsr &
$88.5628
-11.4517\,\theta
+0.643483\cos(\alpha^\dagger)
+0.065444\left\{
\cos\big((\beta^\dagger)^2\big)
+2\beta^\dagger\theta
\right\}$ &
$7.863577$ \\

\addlinespace[1pt]

\bms &
$\beta^\dagger+\alpha^\dagger\beta^\dagger\cos\theta$ &
$1.80{\times}10^{-15}$ \\

\addlinespace[1pt]

\qlattice &
$-225.967
[
-0.0466202\theta
-0.307711
+\exp\!\{
-2\exp\!(
-0.302799(0.0450206\alpha^\dagger+1)^2
-11.5544(0.164397\theta-1)^2
)
\}
]
\exp\!\{
-2.16269(0.14257\beta^\dagger-1)^2
-2.24866(0.0880658\alpha^\dagger-1)^2
\}
+0.330437$ &
$0.478315$ \\

\addlinespace[1pt]

\dsr &
$\beta^\dagger+\alpha^\dagger\beta^\dagger\cos\theta$ &
$1.41{\times}10^{-15}$ \\

\bottomrule
\bottomrule
\end{tabularx}
{
\scriptsize
{
\emph{Note}: \texttt{RMSE} corresponds to the out-of-sample \texttt{RMSE} computed on a $10\%$ held-out test set.
}
}
\end{table}

\begin{table}[H]
\centering
\scriptsize
\renewcommand{\arraystretch}{1.15}
\caption{Expressions recovered by \vasstmain\ and competing methods for learning \eqnref{eq:feynman-ada} under $\sigma=0.1$.}
\label{tab:feynman-III-17-37-noise0p1-substituted-results}
\begin{tabularx}{\linewidth}{@{}l >{\raggedright\arraybackslash}X r@{}}
\toprule
\toprule
\textbf{Method} & \textbf{Symbolic Expression Learned} & \textbf{\texttt{RMSE}} \\
\midrule

\rowcolor{vasstgray}
\vasstmain &
$\,\underline{
1.024106\,\beta^\dagger\{1+\alpha^\dagger\cos\theta\}}
+0.000106\,\theta
+0.000302\cos(2\alpha^\dagger)
-0.000092$ &
$0.108279$ \\

\addlinespace[1pt]

\operon &
$1.015367\,\alpha^\dagger\beta^\dagger\cos(0.996\theta)
+1.46096\,\beta^\dagger
\cos\!\left[
\cos\left\{
-0.872
+0.607\theta
+\cos(0.291\theta)
\right\}
\right]
+0.015$ &
$0.158603$ \\

\addlinespace[1pt]

\pysr &
$\beta^\dagger\{\alpha^\dagger\cos\theta+0.999743\}$ &
$0.099780$ \\

\addlinespace[1pt]

\gplearn &
$\beta^\dagger+\alpha^\dagger\beta^\dagger\cos\theta$ &
$0.099784$ \\

\addlinespace[1pt]

\deap &
$\cos\theta
\left\{
\alpha^\dagger\beta^\dagger
+3\cos\theta
+\cos(1.020634\theta)
\right\}$ &
$2.225102$ \\

\addlinespace[1pt]

\bsr &
$88.5585
-11.4506\,\theta
+0.646516\cos(\alpha^\dagger)
+0.065331\left\{
\cos\big((\beta^\dagger)^2\big)
+2\beta^\dagger\theta
\right\}$ &
$7.866517$ \\

\addlinespace[1pt]

\bms &
$\beta^\dagger+\alpha^\dagger\beta^\dagger\cos\theta$ &
$0.099784$ \\

\addlinespace[1pt]

\qlattice &
$-39.2975
(
0.00653208
-0.317349\beta^\dagger
)
(
0.108825
-0.269264\alpha^\dagger
)
[
-0.00439472\beta^\dagger
-0.678794
+\exp\!\{
-2\exp\![
-15.4029(1-0.159096\theta)^2
-0.419709(0.0349345\alpha^\dagger+1)^2
]
\}
]
+1.02986$ &
$0.467908$ \\

\addlinespace[1pt]

\dsr &
$\beta^\dagger+\alpha^\dagger\beta^\dagger\cos\theta$ &
$0.099784$ \\

\bottomrule
\bottomrule
\end{tabularx}
{
\scriptsize
{
\emph{Note}: \texttt{RMSE} corresponds to the out-of-sample \texttt{RMSE} computed on a $10\%$ held-out test set.
}
}
\end{table}

\begin{table}[H]
\centering
\scriptsize
\renewcommand{\arraystretch}{1.15}
\caption{Expressions recovered by \vasstmain\ and competing methods for learning \eqnref{eq:feynman-ada} under $\sigma=0.2$.}
\label{tab:feynman-III-17-37-noise0p2-substituted-results}
\begin{tabularx}{\linewidth}{@{}l >{\raggedright\arraybackslash}X r@{}}
\toprule
\toprule
\textbf{Method} & \textbf{Symbolic Expression Learned} & \textbf{\texttt{RMSE}} \\
\midrule

\rowcolor{vasstgray}
\vasstmain &
$\,\underline{
1.024369\,\beta^\dagger\{1+\alpha^\dagger\cos\theta\}}
+0.001094\,\theta
+0.00048\cos(2\alpha^\dagger)
-0.000952$ &
$0.200414$ \\

\addlinespace[1pt]

\operon &
$0.995938\,\beta^\dagger
+1.000753\,\alpha^\dagger\beta^\dagger\cos\theta
+0.007000$ &
$0.199415$ \\

\addlinespace[1pt]

\pysr &
$\beta^\dagger\{\alpha^\dagger\cos\theta+0.999471\}$ &
$0.199560$ \\

\addlinespace[1pt]

\gplearn &
$\beta^\dagger+\alpha^\dagger\beta^\dagger\cos\theta$ &
$0.199569$ \\

\addlinespace[1pt]

\deap &
$\cos\theta
\left\{
\alpha^\dagger\beta^\dagger
+3\cos\theta
+\cos(1.020634\theta)
\right\}$ &
$2.233791$ \\

\addlinespace[1pt]

\bsr &
$88.5542
-11.4494\,\theta
+0.649548\cos(\alpha^\dagger)
+0.065219\left\{
\cos\big((\beta^\dagger)^2\big)
+2\beta^\dagger\theta
\right\}$ &
$7.870720$ \\

\addlinespace[1pt]

\bms &
$\beta^\dagger+\alpha^\dagger\beta^\dagger\cos\theta$ &
$0.199569$ \\

\addlinespace[1pt]

\qlattice &
$-41.1561
(0.691287\beta^\dagger-0.000452184)
[
-0.000217418\beta^\dagger
-0.285239\theta
+\tanh\!\{
(6.8735-0.874459\theta)
(-0.0417498\alpha^\dagger-0.324422)
\}
+2.20758
]
+0.0668375$ &
$0.416285$ \\

\addlinespace[1pt]

\dsr &
$\beta^\dagger+\alpha^\dagger\beta^\dagger\cos\theta$ &
$0.199569$ \\

\bottomrule
\bottomrule
\end{tabularx}
{
\scriptsize
{
\emph{Note}: \texttt{RMSE} corresponds to the out-of-sample \texttt{RMSE} computed on a $10\%$ held-out test set.
}
}
\end{table}

\newpage

\subsection{Compute Runtimes for \texorpdfstring{\vasst}{VaSST} and Competing Bayesian \sr\ Methods}
\label{app:runtime_bayes_final}
\begin{table*}[!htp]
\centering
\small
\renewcommand{\arraystretch}{1.1}
\caption{Single-run compute runtimes (in seconds) of \vasstmain\ and competing Bayesian methods for learning the Feynman equations in \hyperref[subsec:feynman-data-study]{Section~\ref{subsec:feynman-data-study}} under different noise settings.}
\label{tab:time_all_vertical}
\begin{tabularx}{\textwidth}{@{}>{\raggedright\arraybackslash}X l c c c@{}}
\toprule
\toprule
\textbf{Feynman Equation} & \textbf{Method} & $\boldsymbol{\textbf{Noiseless}}$ & ${\sigma^{2}=0.1^{2}}$ & ${\sigma^{2}=0.2^{2}}$ \\
\midrule

\multirow{3}{=}{\eqnref{eq:feynman-cl}:\;$F=0.08\tfrac{q_1q_2}{\epsilon r^{2}}$}
& \cellcolor{vasstgray}\vasstmain & \cellcolor{vasstgray}$96.41$ & \cellcolor{vasstgray}$94.45$  & \cellcolor{vasstgray}$96.04$ \\
& \bms   & $143.55$  & $144.23$  & $144.48$ \\
& \bsr   & $319.96$ & $315.43$ & $328.92$ \\

\midrule

\multirow{3}{=}{\eqnref{eq:feynman-cpe}:\;$\Delta U=Gm_1m_2(\tfrac{1}{r_2} - \tfrac{1}{r_1})$}
& \cellcolor{vasstgray}\vasstmain & \cellcolor{vasstgray}$129.45$ & \cellcolor{vasstgray}$108.75$ & \cellcolor{vasstgray}$109.66$ \\
& \bms   & $155.50$ & $224.50$ & $232.54$ \\
& \bsr   & $242.06$ & $240.68$ & $244.30$ \\

\midrule

\multirow{3}{=}{\eqnref{eq:feynman-fce}:\;$F=q(E_f+vB\sin\theta)$}
& \cellcolor{vasstgray}\vasstmain & \cellcolor{vasstgray}$121.54$ & \cellcolor{vasstgray}$114.00$ & \cellcolor{vasstgray}$109.91$ \\
& \bms   & $184.79$ & $222.07$ & $340.54$ \\
& \bsr   & $246.79$ & $252.23$ & $256.57$ \\

\midrule

\multirow{3}{=}{\eqnref{eq:feynman-ftc}:\;$P = \tfrac{\kappa A(T_2-T_1)}{d}$}
& \cellcolor{vasstgray}\vasstmain & \cellcolor{vasstgray}$108.87$ & \cellcolor{vasstgray}$117.82$ & \cellcolor{vasstgray}$115.23$ \\
& \bms   & $143.06$  & $141.42$  & $144.24$ \\
& \bsr   & $236.09$ & $249.89$ & $246.77$ \\

\midrule

\multirow{3}{=}{\eqnref{eq:feynman-ada}:\;$f = \beta^{\dagger}(1 + \alpha^{\dagger}\cos \theta)$}
& \cellcolor{vasstgray}\vasstmain & \cellcolor{vasstgray}$104.72$ & \cellcolor{vasstgray}$116.34$ & \cellcolor{vasstgray}$107.88$ \\
& \bms   & $261.41$ & $268.45$ & $264.55$ \\
& \bsr   & $252.54$ & $257.96$ & $259.78$ \\

\bottomrule
\bottomrule
\end{tabularx}

\vspace{0.35em}

\parbox{\textwidth}{
\footnotesize
\emph{Note}: Runtimes may vary depending on the machine architecture and computational environment. The reported runtimes were obtained on a MacBook Air with an Apple M2 chip and 8GB RAM.
}
\end{table*}

\newpage

\subsection{Top Ranked Symbolic Expressions Learned by \texorpdfstring{\vasst}{VaSST} with respect to LMPSE}
\label{app:top-expressions-vasst-feynman-equations}

\begin{center}
\underline{For learning \eqnref{eq:feynman-cl}: $F = 0.08\tfrac{q_1 q_2}{\epsilon r^{2}}$ in~\hyperref[subsec:feynman-data-study]{Section~\ref{subsec:feynman-data-study}}}.
\end{center}

\begingroup
\scriptsize
\renewcommand{\arraystretch}{1.10}

\begin{xltabular}{\linewidth}{@{}c >{\raggedright\arraybackslash}X r r@{}}
\caption{{Top $\mathsf{r}=10$ ranked symbolic expressions recovered by a single run of \vasstmain\ under the noiseless setting for learning \eqnref{eq:feynman-cl}. Expressions are ranked by $\mathrm{LMPSE}$ among $H=2000$ sampled hard symbolic trees.}}
\label{tab:vasst-top10-I-12-2-mapped-noiseless}\\

\toprule
\toprule
\textbf{Rank} & \textbf{\vasstmain\ Expression Learned} & \textbf{\texttt{RMSE}} & $\mathrm{LMPSE}$ \\
\midrule
\endfirsthead

\multicolumn{4}{@{}l}{\scriptsize\emph{Table~\thetable\ continued from previous page.}}\\
\toprule
\toprule
\textbf{Rank} & \textbf{\vasstmain\ Expression Learned} & \textbf{\texttt{RMSE}} & $\mathrm{LMPSE}$ \\
\midrule
\endhead

\midrule
\multicolumn{4}{r@{}}{\scriptsize\emph{Continued on next page}}\\
\endfoot

\bottomrule
\bottomrule
\addlinespace[2pt]
\multicolumn{4}{@{}p{\linewidth}@{}}{
\scriptsize
\emph{Note}: \texttt{RMSE} corresponds to the out-of-sample \texttt{RMSE} computed on a $10\%$ held-out test set with the LMPSE selected symbolic model.
}\\
\endlastfoot

1 &
$\displaystyle
-0.0003
- \tfrac{0.0001q_1^{2}}{\epsilon r^{2}}
+ {\tfrac{0.0797q_1q_2}{\epsilon r^{2}}}$ &
$0.0001$ & $3365.2575$ \\

\addlinespace[1pt]

2 &
$\displaystyle
-0.0003
- \tfrac{2.4261{\times}10^{-5}q_2}{q_1}
+ \tfrac{0.0797q_1q_2}{\epsilon r^{2}}$ &
$0.0001$ & $3364.8864$ \\

\addlinespace[1pt]

3 &
$\displaystyle
-0.0003
- \tfrac{0.0001q_1^{2}}{\epsilon r^{2}}
+ \tfrac{0.0797q_1q_2}{\epsilon r^{2}}$ &
$0.0001$ & $3363.8731$ \\

\addlinespace[1pt]

4 &
$\displaystyle
-0.0003
- \tfrac{2.2908{\times}10^{-5}q_2}{q_1}
+ \tfrac{0.0797q_1q_2}{\epsilon r^{2}}$ &
$0.0001$ & $3363.5020$ \\

\addlinespace[1pt]

5 &
$\displaystyle
-0.0003
- \tfrac{2.2908{\times}10^{-5}q_2}{q_1}
+ \tfrac{0.0797q_1q_2}{\epsilon r^{2}}$ &
$0.0001$ & $3363.5020$ \\

\addlinespace[1pt]

6 &
$\displaystyle
-0.0003
+ \tfrac{0.0796q_1q_2}{\epsilon r^{2}}
+ \tfrac{2.8973{\times}10^{-5}q_2^{2}}{\epsilon r^{2}}$ &
$0.0001$ & $3362.8026$ \\

\addlinespace[1pt]

7 &
$\displaystyle
0.0003
- \tfrac{8.6888{\times}10^{-5}q_2^{2}}{q_1^{2}}
- \tfrac{0.0007q_1^{2}}{\epsilon r^{2}}
+ \tfrac{0.0799q_1q_2}{\epsilon r^{2}}$ &
$0.0001$ & $3360.9947$ \\

\addlinespace[1pt]

8 &
$\displaystyle
-0.0003
+ \tfrac{0.0061q_1^{2}}{\epsilon r^{2}}
+ \tfrac{0.0749q_1q_2}{\epsilon r^{2}}
+ \tfrac{0.0009q_2^{2}}{\epsilon r^{2}}$ &
$0.0002$ & $3360.1902$ \\

\addlinespace[1pt]

9 &
$\displaystyle
0.0012
- \tfrac{0.0006q_2}{q_1}
+ \tfrac{0.0793q_1q_2}{\epsilon r^{2}}
+ \tfrac{0.0001q_2^{2}}{\epsilon r^{2}}$ &
$0.0001$ & $3358.6486$ \\

\addlinespace[1pt]

10 &
$\displaystyle
-9.3659{\times}10^{-5}
+ \tfrac{0.0796q_1q_2}{\epsilon r^{2}}
+ \tfrac{5.2089{\times}10^{-6}q_1^{8}}{\epsilon^{4}r^{8}}$ &
$0.0001$ & $3358.2041$ \\

\end{xltabular}
\endgroup

\begingroup
\scriptsize
\renewcommand{\arraystretch}{1.10}

\begin{xltabular}{\linewidth}{@{}c >{\raggedright\arraybackslash}X r r@{}}
\caption{{Top $\mathsf{r}=10$ ranked symbolic expressions recovered by a single run of \vasstmain\ under $\sigma=0.1$ for learning \eqnref{eq:feynman-cl}. Expressions are ranked by $\mathrm{LMPSE}$ among $H=2000$ sampled hard symbolic trees.}}
\label{tab:vasst-top10-I-12-2-mapped-noise0p1}\\

\toprule
\toprule
\textbf{Rank} & \textbf{\vasstmain\ Expression Learned} & \textbf{\texttt{RMSE}} & $\mathrm{LMPSE}$ \\
\midrule
\endfirsthead

\multicolumn{4}{@{}l}{\scriptsize\emph{Table~\thetable\ continued from previous page.}}\\
\toprule
\toprule
\textbf{Rank} & \textbf{\vasstmain\ Expression Learned} & \textbf{\texttt{RMSE}} & $\mathrm{LMPSE}$ \\
\midrule
\endhead

\midrule
\multicolumn{4}{r@{}}{\scriptsize\emph{Continued on next page}}\\
\endfoot

\bottomrule
\bottomrule
\addlinespace[2pt]
\multicolumn{4}{@{}p{\linewidth}@{}}{
\scriptsize
\emph{Note}: \texttt{RMSE} corresponds to the out-of-sample \texttt{RMSE} computed on a $10\%$ held-out test set with the LMPSE selected symbolic model.
}\\
\endlastfoot

1 &
$\displaystyle
\tfrac{
-0.0087\epsilon r^{2}
+0.0019q_{1}^{2}
+{0.0807q_{1}q_{2}}
}{\epsilon r^{2}}$ &
$0.0997$ & $1574.9014$ \\

\addlinespace[1pt]

2 &
$\displaystyle
0.0018
- \tfrac{0.0043q_{2}}{q_{1}}
+ \tfrac{0.0816q_{1}q_{2}}{\epsilon r^{2}}$ &
$0.0997$ & $1574.6367$ \\

\addlinespace[1pt]

3 &
$\displaystyle
\tfrac{
-0.0087\epsilon r^{2}
+0.0019q_{1}^{2}
+0.0807q_{1}q_{2}
}{\epsilon r^{2}}$ &
$0.0997$ & $1573.5842$ \\

\addlinespace[1pt]

4 &
$\displaystyle
0.0018
- \tfrac{0.0043q_{2}}{q_{1}}
+ \tfrac{0.0816q_{1}q_{2}}{\epsilon r^{2}}$ &
$0.0997$ & $1573.3191$ \\

\addlinespace[1pt]

5 &
$\displaystyle
0.0018
- \tfrac{0.0043q_{2}}{q_{1}}
+ \tfrac{0.0816q_{1}q_{2}}{\epsilon r^{2}}$ &
$0.0997$ & $1573.3191$ \\

\addlinespace[1pt]

6 &
$\displaystyle
\tfrac{
-0.0086\epsilon r^{2}
+0.0819q_{1}q_{2}
-0.0002q_{2}^{2}
}{\epsilon r^{2}}$ &
$0.0997$ & $1572.4406$ \\

\addlinespace[1pt]

7 &
$\displaystyle
0.0127
- \tfrac{0.0031q_{2}^{2}}{q_{1}^{2}}
- \tfrac{0.0186q_{1}^{2}}{\epsilon r^{2}}
+ \tfrac{0.0885q_{1}q_{2}}{\epsilon r^{2}}$ &
$0.0997$ & $1571.0526$ \\

\addlinespace[1pt]

8 &
$\displaystyle
\tfrac{
-0.0091\epsilon r^{2}
+0.0421q_{1}^{2}
+0.0499q_{1}q_{2}
+0.0059q_{2}^{2}
}{\epsilon r^{2}}$ &
$0.0997$ & $1569.9638$ \\

\addlinespace[1pt]

9 &
$\displaystyle
\tfrac{
-0.0097\epsilon r^{2}
+0.1073q_{1}^{2}
+0.0154q_{2}^{2}
}{\epsilon r^{2}}$ &
$0.0997$ & $1569.3446$ \\

\addlinespace[1pt]

10 &
$\displaystyle
0.0253
- \tfrac{0.0131q_{2}}{q_{1}}
+ \tfrac{0.0757q_{1}q_{2}}{\epsilon r^{2}}
+ \tfrac{0.0022q_{2}^{2}}{\epsilon r^{2}}$ &
$0.0997$ & $1568.6428$ \\

\end{xltabular}
\endgroup

\begingroup
\scriptsize
\renewcommand{\arraystretch}{1.10}

\begin{xltabular}{\linewidth}{@{}c >{\raggedright\arraybackslash}X r r@{}}
\caption{{Top $\mathsf{r}=10$ ranked symbolic expressions recovered by a single run of \vasstmain\ under $\sigma=0.2$ for learning \eqnref{eq:feynman-cl}. Expressions are ranked by $\mathrm{LMPSE}$ among $H=2000$ sampled hard symbolic trees.}}
\label{tab:vasst-top10-I-12-2-mapped-noise0p2}\\

\toprule
\toprule
\textbf{Rank} & \textbf{\vasstmain\ Expression Learned} & \textbf{\texttt{RMSE}} & $\mathrm{LMPSE}$ \\
\midrule
\endfirsthead

\multicolumn{4}{@{}l}{\scriptsize\emph{Table~\thetable\ continued from previous page.}}\\
\toprule
\toprule
\textbf{Rank} & \textbf{\vasstmain\ Expression Learned} & \textbf{\texttt{RMSE}} & $\mathrm{LMPSE}$ \\
\midrule
\endhead

\midrule
\multicolumn{4}{r@{}}{\scriptsize\emph{Continued on next page}}\\
\endfoot

\bottomrule
\bottomrule
\addlinespace[2pt]
\multicolumn{4}{@{}p{\linewidth}@{}}{
\scriptsize
\emph{Note}: \texttt{RMSE} corresponds to the out-of-sample \texttt{RMSE} computed on a $10\%$ held-out test set with the LMPSE selected symbolic model.
}\\
\endlastfoot

1 &
$\displaystyle
\tfrac{
-0.0171\epsilon r^2
+0.0039q_1^2
+{0.0817q_1q_2}
}{\epsilon r^2}$ &
$0.1994$ & $320.3584$ \\

\addlinespace[1pt]

2 &
$\displaystyle
0.0039
-\tfrac{0.0085q_2}{q_1}
+\tfrac{0.0835q_1q_2}{\epsilon r^2}$ &
$0.1994$ & $320.1072$ \\

\addlinespace[1pt]

3 &
$\displaystyle
\tfrac{
-0.0171\epsilon r^2
+0.0040q_1^2
+0.0817q_1q_2
}{\epsilon r^2}$ &
$0.1994$ & $319.0511$ \\

\addlinespace[1pt]

4 &
$\displaystyle
0.0039
-\tfrac{0.0085q_2}{q_1}
+\tfrac{0.0835q_1q_2}{\epsilon r^2}$ &
$0.1994$ & $318.7996$ \\

\addlinespace[1pt]

5 &
$\displaystyle
0.0039
-\tfrac{0.0085q_2}{q_1}
+\tfrac{0.0835q_1q_2}{\epsilon r^2}$ &
$0.1994$ & $318.7996$ \\

\addlinespace[1pt]

6 &
$\displaystyle
\tfrac{
-0.0169\epsilon r^2
+0.0843q_1q_2
-0.0004q_2^2
}{\epsilon r^2}$ &
$0.1994$ & $317.8968$ \\

\addlinespace[1pt]

7 &
$\displaystyle
0.0251
-\tfrac{0.0061q_2^2}{q_1^2}
-\tfrac{0.0366q_1^2}{\epsilon r^2}
+\tfrac{0.0971q_1q_2}{\epsilon r^2}$ &
$0.1994$ & $316.5585$ \\

\addlinespace[1pt]

8 &
$\displaystyle
\tfrac{
-0.0178\epsilon r^2
+0.0780q_1^2
+0.0249q_1q_2
+0.0108q_2^2
}{\epsilon r^2}$ &
$0.1994$ & $315.4286$ \\

\addlinespace[1pt]

9 &
$\displaystyle
\tfrac{
-0.0181\epsilon r^2
+0.1105q_1^2
+0.0156q_2^2
}{\epsilon r^2}$ &
$0.1994$ & $314.9703$ \\

\addlinespace[1pt]

10 &
$\displaystyle
0.0494
-\tfrac{0.0255q_2}{q_1}
+\tfrac{0.0721q_1q_2}{\epsilon r^2}
+\tfrac{0.0042q_2^2}{\epsilon r^2}$ &
$0.1994$ & $314.1365$ \\

\end{xltabular}
\endgroup

\begin{center}
\underline{For learning \eqnref{eq:feynman-cpe}: $\Delta U = G m_1 m_2(\tfrac{1}{r_2} - \tfrac{1}{r_1})$ in~\hyperref[subsec:feynman-data-study]{Section~\ref{subsec:feynman-data-study}}}.
\end{center}

\begingroup
\scriptsize
\renewcommand{\arraystretch}{1.10}

\begin{xltabular}{\linewidth}{@{}c >{\raggedright\arraybackslash}X r r@{}}
\caption{Top $\mathsf{r}=10$ ranked symbolic expressions recovered by a single run of \vasstmain\ under the noiseless setting for learning \eqnref{eq:feynman-cpe}. Expressions are ranked by $\mathrm{LMPSE}$ among $H=2000$ sampled hard symbolic trees.}
\label{tab:vasst-top10-I-13-12-mapped-noiseless}\\

\toprule
\toprule
\textbf{Rank} & \textbf{\vasstmain\ Expression Learned} & \textbf{\texttt{RMSE}} & $\mathrm{LMPSE}$ \\
\midrule
\endfirsthead

\multicolumn{4}{@{}l}{\scriptsize\emph{Table~\thetable\ continued from previous page.}}\\
\toprule
\toprule
\textbf{Rank} & \textbf{\vasstmain\ Expression Learned} & \textbf{\texttt{RMSE}} & $\mathrm{LMPSE}$ \\
\midrule
\endhead

\midrule
\multicolumn{4}{r@{}}{\scriptsize\emph{Continued on next page}}\\
\endfoot

\bottomrule
\bottomrule
\addlinespace[2pt]
\multicolumn{4}{@{}p{\linewidth}@{}}{
\scriptsize
\emph{Note}: \texttt{RMSE} corresponds to the out-of-sample \texttt{RMSE} computed on a $10\%$ held-out test set with the LMPSE selected symbolic model.
}\\
\endlastfoot

1 &
$\displaystyle
{\tfrac{1.0148Gm_1m_2}{r_2}
-\tfrac{0.9985Gm_1m_2}{r_1}}
-0.0212
-\tfrac{0.0031r_2}{r_1}$ &
$0.0038$ & $3278.4007$ \\

\addlinespace[1pt]

2 &
$\displaystyle
\tfrac{0.0009G^{2}m_1^{3}m_2}{r_1^{2}}
+\tfrac{1.0107Gm_1m_2}{r_2}
-\tfrac{0.9965Gm_1m_2}{r_1}
-0.0007$ &
$0.0037$ & $3199.2274$ \\

\addlinespace[1pt]

3 &
$\displaystyle
\tfrac{0.0064G^{2}m_1^{5}}{m_2r_1r_2}
+\tfrac{0.0064Gm_1^{3}}{m_2r_1}
-\tfrac{0.0010Gm_1^{2}\{2Gm_1^{2}/r_1-1\}}{r_1}
-\tfrac{1.0073Gm_1^{2}\{-m_2r_1/(m_1r_2)+m_2/m_1\}}{r_1}
-0.0086
-\tfrac{0.0064m_2}{m_1}$ &
$0.0152$ & $3174.0205$ \\

\addlinespace[1pt]

4 &
$\displaystyle
\tfrac{2.2506G^{2}m_1^{3}m_2}{r_1^{2}\{-2Gm_1^{2}/r_1+1\}}
+\tfrac{0.0873Gm_1m_2\{-Gm_1^{2}/r_1+m_2/m_1\}}{r_2}
+0.9087
-\tfrac{0.1176m_2r_1\{-Gm_1^{2}/r_1+m_2/m_1\}}{m_1r_2}$ &
$0.1501$ & $770.4722$ \\

\addlinespace[1pt]

5 &
$\displaystyle
\tfrac{1.0486Gm_1^{2}}{r_1}
-\tfrac{1.8288Gm_1m_2}{r_1}
+1.0821
+\tfrac{0.6100r_2\{Gm_1m_2/r_1+m_2r_1/(m_1r_2)\}}
{r_1\{Gm_1^{2}/r_1+r_2/r_1\}}$ &
$0.1915$ & $357.5249$ \\

\addlinespace[1pt]

6 &
$\displaystyle
\tfrac{0.6894Gm_1^{2}}{r_1}
-\tfrac{1.4571Gm_1m_2}{r_1}
+1.5451
+\tfrac{0.0898r_2}{r_1}
-\tfrac{1.4571m_2r_1}{m_1r_2}
-\tfrac{0.6894m_2}{m_1}$ &
$0.2020$ & $245.0738$ \\

\addlinespace[1pt]

7 &
$\displaystyle
\tfrac{0.0361Gm_1^{2}}{r_1}
-\tfrac{1.2526Gm_1m_2}{r_1}
+\tfrac{17.2393r_1^{2}}{r_2^{2}}
-0.6275
+\tfrac{0.0542r_2}{r_1}
-\tfrac{0.0181m_2}{m_1}$ &
$0.1932$ & $13.0193$ \\

\addlinespace[1pt]

8 &
$\displaystyle
-\tfrac{0.8731Gm_1^{2}\{-r_2/r_1+m_2/m_1\}}{r_1}
-\tfrac{0.7120Gm_1^{2}r_2}{r_1^{2}}
-0.8223
+\tfrac{0.8731m_2r_1\{Gm_1^{2}/r_1+r_2/r_1\}}{m_1r_2}
+\tfrac{0.1539m_2r_2}{m_1r_1}$ &
$0.2330$ & $-31.8476$ \\

\addlinespace[1pt]

9 &
$\displaystyle
\tfrac{1.0196Gm_1^{2}}{r_1}
-\tfrac{1.4805Gm_1m_2}{r_1}
+2.4189
+\tfrac{0.1999r_2}{r_1}
-\tfrac{0.3998m_2}{m_1}$ &
$0.2354$ & $-48.6053$ \\

\addlinespace[1pt]

10 &
$\displaystyle
\tfrac{0.0025Gm_1^{2}}{r_2}
+\tfrac{0.8171Gm_1m_2\{-r_2/r_1+m_2/m_1\}}
{r_1\{r_2/r_1+m_2r_1/(Gm_1^{3})\}}
+0.0206
-\tfrac{0.0051r_2}{r_1}
+\tfrac{0.0245m_2\{-2r_2/r_1+2m_2/m_1\}}{m_1}
+\tfrac{0.0025m_2}{m_1}$ &
$0.2363$ & $-66.4971$ \\

\end{xltabular}
\endgroup

\begingroup
\scriptsize
\renewcommand{\arraystretch}{1.10}

\begin{xltabular}{\linewidth}{@{}c >{\raggedright\arraybackslash}X r r@{}}
\caption{Top $\mathsf{r}=10$ ranked symbolic expressions recovered by a single run of \vasstmain\ under $\sigma=0.1$ for learning \eqnref{eq:feynman-cpe}. Expressions are ranked by $\mathrm{LMPSE}$ among $H=2000$ sampled hard symbolic trees.}
\label{tab:vasst-top10-I-13-12-mapped-noise0p1}\\

\toprule
\toprule
\textbf{Rank} & \textbf{\vasstmain\ Expression Learned} & \textbf{\texttt{RMSE}} & $\mathrm{LMPSE}$ \\
\midrule
\endfirsthead

\multicolumn{4}{@{}l}{\scriptsize\emph{Table~\thetable\ continued from previous page.}}\\
\toprule
\toprule
\textbf{Rank} & \textbf{\vasstmain\ Expression Learned} & \textbf{\texttt{RMSE}} & $\mathrm{LMPSE}$ \\
\midrule
\endhead

\midrule
\multicolumn{4}{r@{}}{\scriptsize\emph{Continued on next page}}\\
\endfoot

\bottomrule
\bottomrule
\addlinespace[2pt]
\multicolumn{4}{@{}p{\linewidth}@{}}{
\scriptsize
\emph{Note}: \texttt{RMSE} corresponds to the out-of-sample \texttt{RMSE} computed on a $10\%$ held-out test set with the LMPSE selected symbolic model.
}\\
\endlastfoot

1 &
$\displaystyle
{\tfrac{1.0148Gm_1m_2}{r_2}
-\tfrac{0.9997Gm_1m_2}{r_1}}
-0.0382
-\tfrac{0.0049r_2}{r_1}$ &
$0.0998$ & $1534.4676$ \\

\addlinespace[1pt]

2 &
$\displaystyle
\tfrac{0.0011G^{2}m_1^{3}m_2}{r_1^{2}}
+\tfrac{1.0077Gm_1m_2}{r_2}
-\tfrac{0.9980Gm_1m_2}{r_1}
-0.0077$ &
$0.0998$ & $1513.0846$ \\

\addlinespace[1pt]

3 &
$\displaystyle
\tfrac{0.0046G^{2}m_1^{5}}{m_2r_1r_2}
-\tfrac{0.0009G^{2}m_1^{4}}{r_1^{2}}
+\tfrac{0.0046Gm_1^{3}}{m_2r_1}
+\tfrac{0.0004Gm_1^{2}}{r_1}
+\tfrac{1.0071Gm_1m_2}{r_2}
-\tfrac{1.0071Gm_1m_2}{r_1}
-0.0194
-\tfrac{0.0046m_2}{m_1}$ &
$0.1010$ & $1489.7953$ \\

\addlinespace[1pt]

4 &
$\displaystyle
\tfrac{
2.2538G^{2}m_1^{5}m_2r_2
-0.9018m_1^{2}r_1r_2\{2Gm_1^{2}-r_1\}
+m_2\{0.0872Gm_1^{2}-0.1176r_1\}\{2Gm_1^{2}-r_1\}\{Gm_1^{3}-m_2r_1\}
}
{-m_1^{2}r_1r_2\{2Gm_1^{2}-r_1\}}$ &
$0.1781$ & $455.0898$ \\

\addlinespace[1pt]

5 &
$\displaystyle
\tfrac{
-0.6088Gm_1^{2}m_2r_2
+m_1\{Gm_1^{2}+r_2\}\{-1.0444Gm_1^{2}+1.8285Gm_1m_2-1.0711r_1\}
-0.6088m_2r_1^{2}
}
{-m_1r_1\{Gm_1^{2}+r_2\}}$ &
$0.2157$ & $130.8494$ \\

\addlinespace[1pt]

6 &
$\displaystyle
\tfrac{0.6888Gm_1^{2}}{r_1}
-\tfrac{1.4581Gm_1m_2}{r_1}
+1.5241
+\tfrac{0.0876r_2}{r_1}
-\tfrac{1.4581m_2r_1}{m_1r_2}
-\tfrac{0.6888m_2}{m_1}$ &
$0.2255$ & $35.3900$ \\

\addlinespace[1pt]

7 &
$\displaystyle
\tfrac{0.0354Gm_1^{2}}{r_1}
-\tfrac{1.2537Gm_1m_2}{r_1}
+\tfrac{17.1818r_1^{2}}{r_2^{2}}
-0.6402
+\tfrac{0.0531r_2}{r_1}
-\tfrac{0.0177m_2}{m_1}$ &
$0.2187$ & $-161.2716$ \\

\addlinespace[1pt]

8 &
$\displaystyle
\tfrac{0.1616Gm_1^{2}r_2}{r_1^{2}}
+\tfrac{0.8735Gm_1m_2}{r_2}
-\tfrac{0.8735Gm_1m_2}{r_1}
-0.8325
+\tfrac{0.8735m_2}{m_1}
+\tfrac{0.1539m_2r_2}{m_1r_1}$ &
$0.2529$ & $-190.5952$ \\

\addlinespace[1pt]

9 &
$\displaystyle
\tfrac{1.0151Gm_1^{2}}{r_1}
-\tfrac{1.4807Gm_1m_2}{r_1}
+2.4027
+\tfrac{0.1993r_2}{r_1}
-\tfrac{0.3986m_2}{m_1}$ &
$0.2555$ & $-205.9995$ \\

\addlinespace[1pt]

10 &
$\displaystyle
\tfrac{
0.8178G^{2}m_1^{5}m_2r_2\{m_1r_2-m_2r_1\}
-0.0038Gm_1^{4}r_1\{Gm_1^{3}r_2+m_2r_1^{2}\}
+r_2\{Gm_1^{3}r_2+m_2r_1^{2}\}
\{0.0018m_1^{2}r_1+0.0038m_1r_2^{2}-0.0245m_2r_1\{2m_1r_2-m_2r_1\}\}
}
{-r_1r_2\{Gm_1^{3}r_2+m_2r_1^{2}\}}$ &
$0.2569$ & $-228.4333$ \\

\end{xltabular}
\endgroup

\begingroup
\scriptsize
\renewcommand{\arraystretch}{1.10}

\begin{xltabular}{\linewidth}{@{}c >{\raggedright\arraybackslash}X r r@{}}
\caption{Top $\mathsf{r}=10$ ranked symbolic expressions recovered by a single run of \vasstmain\ under $\sigma=0.2$ for learning \eqnref{eq:feynman-cpe}. Expressions are ranked by $\mathrm{LMPSE}$ among $H=2000$ sampled hard symbolic trees.}
\label{tab:vasst-top10-I-13-12-mapped-noise0p2}\\

\toprule
\toprule
\textbf{Rank} & \textbf{\vasstmain\ Expression Learned} & \textbf{\texttt{RMSE}} & $\mathrm{LMPSE}$ \\
\midrule
\endfirsthead

\multicolumn{4}{@{}l}{\scriptsize\emph{Table~\thetable\ continued from previous page.}}\\
\toprule
\toprule
\textbf{Rank} & \textbf{\vasstmain\ Expression Learned} & \textbf{\texttt{RMSE}} & $\mathrm{LMPSE}$ \\
\midrule
\endhead

\midrule
\multicolumn{4}{r@{}}{\scriptsize\emph{Continued on next page}}\\
\endfoot

\bottomrule
\bottomrule
\addlinespace[2pt]
\multicolumn{4}{@{}p{\linewidth}@{}}{
\scriptsize
\emph{Note}: \texttt{RMSE} corresponds to the out-of-sample \texttt{RMSE} computed on a $10\%$ held-out test set with the LMPSE selected symbolic model.
}\\
\endlastfoot

1 &
$\displaystyle
{\tfrac{1.0148Gm_1m_2}{r_2}
-\tfrac{1.0010Gm_1m_2}{r_1}}
-0.0552
-\tfrac{0.0068r_2}{r_1}$ &
$0.1995$ & $287.4354$ \\

\addlinespace[1pt]

2 &
$\displaystyle
\tfrac{0.0012G^2m_1^3m_2}{r_1^2}
+\tfrac{1.0048Gm_1m_2}{r_2}
-\tfrac{0.9994Gm_1m_2}{r_1}
-0.0148$ &
$0.1995$ & $275.0103$ \\

\addlinespace[1pt]

3 &
$\displaystyle
\tfrac{0.0028G^2m_1^5}{m_2r_1r_2}
+\tfrac{0.0002G^2m_1^4}{r_1^2}
+\tfrac{0.0028Gm_1^3}{m_2r_1}
-\tfrac{0.0001Gm_1^2}{r_1}
+\tfrac{1.0068Gm_1m_2}{r_2}
-\tfrac{1.0068Gm_1m_2}{r_1}
-0.0302
-\tfrac{0.0028m_2}{m_1}$ &
$0.2002$ & $252.9171$ \\

\addlinespace[1pt]

4 &
$\displaystyle
\tfrac{
2.2570G^2m_1^5m_2r_2
-0.8948m_1^2r_1r_2\{2Gm_1^2-r_1\}
+m_2\{0.0871Gm_1^2-0.1175r_1\}\{2Gm_1^2-r_1\}\{Gm_1^3-m_2r_1\}
}
{-m_1^2r_1r_2\{2Gm_1^2-r_1\}}$ &
$0.2466$ & $-163.8244$ \\

\addlinespace[1pt]

5 &
$\displaystyle
\tfrac{
-m_1\{0.6077Gm_1m_2r_2-\{Gm_1^2+r_2\}(-1.0402Gm_1^2+1.8282Gm_1m_2-1.0601r_1)\}
-0.6077m_2r_1^2
}
{-m_1r_1\{Gm_1^2+r_2\}}$ &
$0.2762$ & $-346.8534$ \\

\addlinespace[1pt]

6 &
$\displaystyle
\tfrac{0.6883Gm_1^2}{r_1}
-\tfrac{1.4591Gm_1m_2}{r_1}
+1.5032
+\tfrac{0.0854r_2}{r_1}
-\tfrac{1.4591m_2r_1}{m_1r_2}
-\tfrac{0.6883m_2}{m_1}$ &
$0.2842$ & $-412.0230$ \\

\addlinespace[1pt]

7 &
$\displaystyle
\tfrac{0.0346Gm_1^2}{r_1}
-\tfrac{1.2549Gm_1m_2}{r_1}
+\tfrac{17.1240r_1^2}{r_2^2}
-0.6530
+\tfrac{0.0519r_2}{r_1}
-\tfrac{0.0173m_2}{m_1}$ &
$0.2796$ & $-544.3299$ \\

\addlinespace[1pt]

8 &
$\displaystyle
\tfrac{0.1622Gm_1^2r_2}{r_1^2}
+\tfrac{0.8740Gm_1m_2}{r_2}
-\tfrac{0.8740Gm_1m_2}{r_1}
-0.8426
+\tfrac{0.8740m_2}{m_1}
+\tfrac{0.1540m_2r_2}{m_1r_1}$ &
$0.3058$ & $-561.1996$ \\

\addlinespace[1pt]

9 &
$\displaystyle
\tfrac{1.0107Gm_1^2}{r_1}
-\tfrac{1.4810Gm_1m_2}{r_1}
+2.3864
+\tfrac{0.1987r_2}{r_1}
-\tfrac{0.3974m_2}{m_1}$ &
$0.3083$ & $-570.1379$ \\

\addlinespace[1pt]

10 &
$\displaystyle
\tfrac{
0.8185G^2m_1^5m_2r_2\{m_1r_2-m_2r_1\}
-0.0051Gm_1^4r_1\{Gm_1^3r_2+m_2r_1^2\}
}
{-r_1r_2\{Gm_1^3r_2+m_2r_1^2\}}
+
\tfrac{
r_2\{Gm_1^3r_2+m_2r_1^2\}
\{+0.0241m_1^2r_1+0.0102m_1^2r_2-0.0051m_1m_2r_1-0.0490m_2^2r_1\}
}
{-r_1r_2\{Gm_1^3r_2+m_2r_1^2\}}$ &
$0.3099$ & $-594.8235$ \\

\end{xltabular}
\endgroup

\begin{center}
\underline{For learning \eqnref{eq:feynman-fce}: $F = q(E_f + Bv \sin \theta)$ in~\hyperref[subsec:feynman-data-study]{Section~\ref{subsec:feynman-data-study}}}.
\end{center}

\begingroup
\scriptsize
\renewcommand{\arraystretch}{1.10}

\begin{xltabular}{\linewidth}{@{}c >{\raggedright\arraybackslash}X r r@{}}
\caption{Top $\mathsf{r}=10$ ranked symbolic expressions recovered by a single run of \vasstmain\ under the noiseless setting for learning \eqnref{eq:feynman-fce}. Expressions are ranked by $\mathrm{LMPSE}$ among $H=2000$ sampled hard symbolic trees.}
\label{tab:vasst-top10-I-12-11-mapped-noiseless}\\

\toprule
\toprule
\textbf{Rank} & \textbf{\vasstmain\ Expression Learned} & \textbf{\texttt{RMSE}} & $\mathrm{LMPSE}$ \\
\midrule
\endfirsthead

\multicolumn{4}{@{}l}{\scriptsize\emph{Table~\thetable\ continued from previous page.}}\\
\toprule
\toprule
\textbf{Rank} & \textbf{\vasstmain\ Expression Learned} & \textbf{\texttt{RMSE}} & $\mathrm{LMPSE}$ \\
\midrule
\endhead

\midrule
\multicolumn{4}{r@{}}{\scriptsize\emph{Continued on next page}}\\
\endfoot

\bottomrule
\bottomrule
\addlinespace[2pt]
\multicolumn{4}{@{}p{\linewidth}@{}}{
\scriptsize
\emph{Note}: \texttt{RMSE} corresponds to the out-of-sample \texttt{RMSE} computed on a $10\%$ held-out test set with the LMPSE selected symbolic model.
}\\
\endlastfoot

1 &
$\displaystyle
1.0662Bvq\sin\theta
+0.9796E_fq
-0.0027\theta
+0.001420$ &
$0.0016$ & $-3800.7159$ \\

\addlinespace[1pt]

2 &
$\displaystyle
0.1920Bvq\theta
-\tfrac{0.1550Bv\theta(E_fq+\theta)}{E_f}
+0.8769(Bvq+E_fq)\sin\theta
+0.1920\sin\{\sin(E_fq)\}
+0.0121$ &
$0.0055$ & $-4807.8071$ \\

\addlinespace[1pt]

3 &
$\displaystyle
1.0511Bvq\sin\theta
+0.1952Bvq
+\tfrac{0.1952Bv\theta}{E_f}
+\tfrac{1.0511Bv\sin(Bv/E_f)}{E_f}
+\tfrac{0.1952Bv}{E_f}
+0.3903\theta
-0.0395(Bv/E_f+2\theta)(Bv/E_f+\theta^2)
+0.0001$ &
$0.0065$ & $-4896.5806$ \\

\addlinespace[1pt]

4 &
$\displaystyle
1.0389Bvq\sin\theta
+0.0979\theta
+1.6425\sin(E_fq)
+0.0063$ &
$0.0080$ & $-4954.2483$ \\

\addlinespace[1pt]

5 &
$\displaystyle
1.1034Bvq\sin\theta
-\tfrac{0.3719Bv}{E_f}
-3.7914\sin\{\sin\theta\}
+0.097$ &
$0.0081$ & $-4966.7859$ \\

\addlinespace[1pt]

6 &
$\displaystyle
0.9973Bvq\sin\theta
+0.0990\theta
+0.0084(E_fq+\theta)^2\{\sin\theta+\sin(E_fq)\}
+0.0990\sin(Bv/E_f)
+0.0211$ &
$0.0086$ & $-5032.7066$ \\

\addlinespace[1pt]

7 &
$\displaystyle
1.0184Bvq\sin\theta
-0.0215E_fq\theta\{\theta+\sin(Bv/E_f)\}
+0.1353\sin(Bv/E_f)
+0.0051$ &
$0.0106$ & $-5045.5640$ \\

\addlinespace[1pt]

8 &
$\displaystyle
1.0134Bvq\sin\theta
+0.1554\theta
+0.1554\sin\theta
-0.0688\sin(2\theta)\sin\{\sin(E_fq)\}
+0.1554\sin(E_fq)
+0.0071$ &
$0.0101$ & $-5048.4944$ \\

\addlinespace[1pt]

9 &
$\displaystyle
0.0950BvE_fq^2
+\tfrac{0.0950Bv(E_fq+\theta)}{E_f}
+1.0087(Bvq+E_fq)\sin\{\sin\theta\}
-0.8010\sin\{E_f^2q^2(Bv/E_f+\theta)\}
-0.0073$ &
$0.0104$ & $-5072.3435$ \\

\addlinespace[1pt]

10 &
$\displaystyle
1.1751Bvq\sin\{\sin\theta\}
-0.1238\theta\{\sin(Bv/E_f)+\sin(E_fq)\}\sin(E_fq)
+0.1329\sin(Bv/E_f)
+0.0995$ &
$0.0128$ & $-5201.4007$ \\

\end{xltabular}
\endgroup

\begingroup
\scriptsize
\renewcommand{\arraystretch}{1.10}

\begin{xltabular}{\linewidth}{@{}c >{\raggedright\arraybackslash}X r r@{}}
\caption{Top $\mathsf{r}=10$ ranked symbolic expressions recovered by a single run of \vasstmain\ under $\sigma=0.1$ for learning \eqnref{eq:feynman-fce}. Expressions are ranked by $\mathrm{LMPSE}$ among $H=2000$ sampled hard symbolic trees.}
\label{tab:vasst-top10-I-12-11-mapped-noise0p1}\\

\toprule
\toprule
\textbf{Rank} & \textbf{\vasstmain\ Expression Learned} & \textbf{\texttt{RMSE}} & $\mathrm{LMPSE}$ \\
\midrule
\endfirsthead

\multicolumn{4}{@{}l}{\scriptsize\emph{Table~\thetable\ continued from previous page.}}\\
\toprule
\toprule
\textbf{Rank} & \textbf{\vasstmain\ Expression Learned} & \textbf{\texttt{RMSE}} & $\mathrm{LMPSE}$ \\
\midrule
\endhead

\midrule
\multicolumn{4}{r@{}}{\scriptsize\emph{Continued on next page}}\\
\endfoot

\bottomrule
\bottomrule
\addlinespace[2pt]
\multicolumn{4}{@{}p{\linewidth}@{}}{
\scriptsize
\emph{Note}: \texttt{RMSE} corresponds to the out-of-sample \texttt{RMSE} computed on a $10\%$ held-out test set with the LMPSE selected symbolic model.
}\\
\endlastfoot

1 &
$\displaystyle
1.0663Bvq\sin{\left(\theta \right)}
+0.9792E_{f}q
-0.0022\theta
+0.0014$ &
$0.1016$ & $-3805.6060$ \\

\addlinespace[1pt]

2 &
$\displaystyle
0.0370Bvq\theta
+0.8770Bvq\sin{\left(\theta \right)}
-\tfrac{0.1549Bv\theta^{2}}{E_{f}}
+0.8770E_{f}q\sin{\left(\theta \right)}
+0.1919\sin{\left(\sin{\left(E_{f}q \right)} \right)}
+0.0051$ &
$0.1021$ & $-4810.8994$ \\

\addlinespace[1pt]

3 &
$\displaystyle
-\tfrac{0.0395Bv^{2}}{E_{f}^{2}}
+1.0512Bvq\sin{\left(\theta \right)}
+0.1954Bvq
-\tfrac{0.0395Bv\theta^{2}}{E_{f}}
+\tfrac{0.1163Bv\theta}{E_{f}}
+\tfrac{1.0512Bv\sin{\left(\tfrac{Bv}{E_{f}} \right)}}{E_{f}}
+0.0908\theta
+0.030$ &
$0.1037$ & $-4899.0665$ \\

\addlinespace[1pt]

4 &
$\displaystyle
1.0390Bvq\sin{\left(\theta \right)}
+0.0983\theta
+1.6438\sin{\left(E_{f}q \right)}
+0.0091$ &
$0.1028$ & $-4954.6377$ \\

\addlinespace[1pt]

5 &
$\displaystyle
1.1036Bvq\sin{\left(\theta \right)}
-\tfrac{0.3703Bv}{E_{f}}
-3.7934\sin{\left(\sin{\left(\theta \right)} \right)}
+0.0001$ &
$0.1045$ & $-4967.3706$ \\

\addlinespace[1pt]

6 &
$\displaystyle
0.9974Bvq\sin{\left(\theta \right)}
+0.0084E_{f}^{2}q^{2}\sin{\left(\theta \right)}
+0.0084E_{f}^{2}q^{2}\sin{\left(E_{f}q \right)}
+0.0168E_{f}q\theta\sin{\left(\theta \right)}
+0.0168E_{f}q\theta\sin{\left(E_{f}q \right)}
+0.0084\theta^{2}\sin{\left(\theta \right)}
+0.0084\theta^{2}\sin{\left(E_{f}q \right)}
+0.0992\theta
+0.0992\sin{\left(\tfrac{Bv}{E_{f}} \right)}
+0.5121$ &
$0.1074$ & $-5033.1722$ \\

\addlinespace[1pt]

7 &
$\displaystyle
1.0185Bvq\sin{\left(\theta \right)}
-0.0214E_{f}q\theta^{2}
-0.0214E_{f}q\theta\sin{\left(\tfrac{Bv}{E_{f}} \right)}
+0.1274\sin{\left(\tfrac{Bv}{E_{f}} \right)}
+0.1023$ &
$0.1334$ & $-5046.4801$ \\

\addlinespace[1pt]

8 &
$\displaystyle
1.0135Bvq\sin{\left(\theta \right)}
+0.1558\theta
+0.1558\sin{\left(\theta \right)}
-0.0744\sin{\left(2\theta \right)}
\sin{\left(\sin{\left(E_{f}q \right)} \right)}
+0.1558\sin{\left(E_{f}q \right)}
+6.0358$ &
$0.1903$ & $-5048.9328$ \\

\addlinespace[1pt]

9 &
$\displaystyle
0.0950BvE_{f}q^{2}
+1.0088Bvq\sin{\left(\sin{\left(\theta \right)} \right)}
+0.0950Bvq
+\tfrac{0.0950Bv\theta}{E_{f}}
+1.0088E_{f}q\sin{\left(\sin{\left(\theta \right)} \right)}
-0.7980\sin{\left(E_{f}q\left(\tfrac{Bv}{E_{f}}+\theta\right) \right)}
-1.1405$ &
$0.2031$ & $-5073.3466$ \\

\addlinespace[1pt]

10 &
$\displaystyle
1.1752Bvq\sin{\left(\sin{\left(\theta \right)} \right)}
-0.1240\theta\sin{\left(\tfrac{Bv}{E_{f}} \right)}\sin{\left(E_{f}q \right)}
-0.1240\theta\sin^{2}{\left(E_{f}q \right)}
+0.1251\sin{\left(\tfrac{Bv}{E_{f}} \right)}
+1.6583$ &
$0.1949$ & $-5201.4565$ \\

\end{xltabular}
\endgroup

\begingroup
\scriptsize
\renewcommand{\arraystretch}{1.10}

\begin{xltabular}{\linewidth}{@{}c >{\raggedright\arraybackslash}X r r@{}}
\caption{Top $\mathsf{r}=10$ ranked symbolic expressions recovered by a single run of \vasstmain\ under $\sigma=0.2$ for learning \eqnref{eq:feynman-fce}. Expressions are ranked by $\mathrm{LMPSE}$ among $H=2000$ sampled hard symbolic trees.}
\label{tab:vasst-top10-I-12-11-mapped-noise0p2}\\

\toprule
\toprule
\textbf{Rank} & \textbf{\vasstmain\ Expression Learned} & \textbf{\texttt{RMSE}} & $\mathrm{LMPSE}$ \\
\midrule
\endfirsthead

\multicolumn{4}{@{}l}{\scriptsize\emph{Table~\thetable\ continued from previous page.}}\\
\toprule
\toprule
\textbf{Rank} & \textbf{\vasstmain\ Expression Learned} & \textbf{\texttt{RMSE}} & $\mathrm{LMPSE}$ \\
\midrule
\endhead

\midrule
\multicolumn{4}{r@{}}{\scriptsize\emph{Continued on next page}}\\
\endfoot

\bottomrule
\bottomrule
\addlinespace[2pt]
\multicolumn{4}{@{}p{\linewidth}@{}}{
\scriptsize
\emph{Note}: \texttt{RMSE} corresponds to the out-of-sample \texttt{RMSE} computed on a $10\%$ held-out test set with the LMPSE selected symbolic model.
}\\
\endlastfoot

1 &
$\displaystyle
1.0664Bqv\sin\theta
+0.9788E_fq
-0.0018\theta
+0.0014$ &
$0.2016$ & $-3818.2221$ \\

\addlinespace[1pt]

2 &
$\displaystyle
0.0370Bqv\theta
+0.8771Bqv\sin\theta
-\tfrac{0.1549Bv\theta^2}{E_f}
+0.8771E_fq\sin\theta
+0.1919\sin\{\sin(E_fq)\}
+0.0021$ &
$0.2034$ & $-4816.9309$ \\

\addlinespace[1pt]

3 &
$\displaystyle
-\tfrac{0.0395B^2v^2}{E_f^2}
+1.0513Bqv\sin\theta
+0.1956Bqv
-\tfrac{0.0395Bv\theta^2}{E_f}
+\tfrac{0.1166Bv\theta}{E_f}
+\tfrac{1.0513Bv\sin(Bv/E_f)}{E_f}
+\tfrac{0.1956Bv}{E_f}
-0.0790\theta^3
+0.3912\theta
+0.0048$ &
$0.2032$ & $-4904.2802$ \\

\addlinespace[1pt]

4 &
$\displaystyle
1.0391Bqv\sin\theta
+0.0987\theta
+1.6451\sin(E_fq)
+0.0714$ &
$0.2050$ & $-4957.4885$ \\

\addlinespace[1pt]

5 &
$\displaystyle
1.1037Bqv\sin\theta
-\tfrac{0.3686Bv}{E_f}
-3.7954\sin\{\sin\theta\}
+0.1075$ &
$0.2049$ & $-4970.4218$ \\

\addlinespace[1pt]

6 &
$\displaystyle
0.9974Bqv\sin\theta
+0.0084E_f^2q^2\sin\theta
+0.0084E_f^2q^2\sin(E_fq)
+0.0168E_fq\theta\sin\theta
+0.0168E_fq\theta\sin(E_fq)
+0.0084\theta^2\sin\theta
+0.0084\theta^2\sin(E_fq)
+0.0994\theta
+0.0994\sin(Bv/E_f)
+1.0122$ &
$0.2293$ & $-5036.0061$ \\

\addlinespace[1pt]

7 &
$\displaystyle
1.0186Bqv\sin\theta
-0.0213E_fq\theta^2
-0.0213E_fq\theta\sin(Bv/E_f)
+0.1195\sin(Bv/E_f)
+0.9636$ &
$0.2531$ & $-5049.6806$ \\

\addlinespace[1pt]

8 &
$\displaystyle
1.0135Bqv\sin\theta
+0.1562\theta
+0.1562\sin\theta
-0.0800\sin(2\theta)\sin\{\sin(E_fq)\}
+0.1562\sin(E_fq)
+1.0345$ &
$0.2714$ & $-5051.6859$ \\

\addlinespace[1pt]

9 &
$\displaystyle
0.0951BE_fq^2v
+1.0089Bqv\sin\{\sin\theta\}
+0.0951Bqv
+\tfrac{0.0951Bv\theta}{E_f}
+1.0089E_fq\sin\{\sin\theta\}
-0.7951\sin(BE_fq^2v+E_f^2q^2\theta)
-1.0005$ &
$0.2612$ & $-5076.6401$ \\

\addlinespace[1pt]

10 &
$\displaystyle
1.1753Bqv\sin\{\sin\theta\}
-0.1241\theta\sin^2(E_fq)
-0.1241\theta\sin(E_fq)\sin(Bv/E_f)
+0.1174\sin(Bv/E_f)
+1.9158$ &
$0.3121$ & $-5203.5019$ \\

\end{xltabular}
\endgroup

\begin{center}
\underline{For learning \eqnref{eq:feynman-ftc}: $P = \frac{\kappa A(T_2-T_1)}{d}$ in~\hyperref[subsec:feynman-data-study]{Section~\ref{subsec:feynman-data-study}}}.
\end{center}

\begingroup
\scriptsize
\renewcommand{\arraystretch}{1.10}

\begin{xltabular}{\linewidth}{@{}c >{\raggedright\arraybackslash}X r r@{}}
\caption{Top $\mathsf{r}=10$ ranked symbolic expressions recovered by a single run of \vasstmain\ under the noiseless setting for learning \eqnref{eq:feynman-ftc}. Expressions are ranked by $\mathrm{LMPSE}$ among $H=2000$ sampled hard symbolic trees.}
\label{tab:vasst-top10-II-2-42-mapped-noiseless}\\

\toprule
\toprule
\textbf{Rank} & \textbf{\vasstmain\ Expression Learned} & \textbf{\texttt{RMSE}} & $\mathrm{LMPSE}$ \\
\midrule
\endfirsthead

\multicolumn{4}{@{}l}{\scriptsize\emph{Table~\thetable\ continued from previous page.}}\\
\toprule
\toprule
\textbf{Rank} & \textbf{\vasstmain\ Expression Learned} & \textbf{\texttt{RMSE}} & $\mathrm{LMPSE}$ \\
\midrule
\endhead

\midrule
\multicolumn{4}{r@{}}{\scriptsize\emph{Continued on next page}}\\
\endfoot

\bottomrule
\bottomrule
\addlinespace[2pt]
\multicolumn{4}{@{}p{\linewidth}@{}}{
\scriptsize
\emph{Note}: \texttt{RMSE} corresponds to the out-of-sample \texttt{RMSE} computed on a $10\%$ held-out test set with the LMPSE selected symbolic model.
}\\
\endlastfoot

1 &
$\displaystyle
-0.805371\,\tfrac{AT_1\kappa}{d}
+0.998520\,\tfrac{A\kappa T_2}{d}
-0.000254\{
\tfrac{1}{\sqrt A\,T_1\kappa}
-\tfrac{d/\sqrt A-T_2/T_1}{\sqrt A\,T_1\kappa-T_2/T_1}
\}
+0.014782$
&
$0.0126$ & $3205.1951$ \\

\addlinespace[1pt]

2 &
$\displaystyle
-0.0484\sqrt A\,T_1\kappa
-\tfrac{0.8506\sqrt A(-\sqrt A\,T_2\kappa+d/\sqrt A)}{d}
+0.5893
+\tfrac{0.0508T_2}{T_1}$ &
$0.0207$ & $3125.1013$ \\

\addlinespace[1pt]

3 &
$\displaystyle
\tfrac{2.0449A^{3/2}T_2\kappa}{d^2}
-0.2527\sqrt A\,T_1\kappa
+0.1263\sqrt A\,T_2\kappa
-0.1748
+\tfrac{0.0409d}{\sqrt A}$ &
$0.0173$ & $3111.3952$ \\

\addlinespace[1pt]

4 &
$\displaystyle
0.0032AT_1^2\kappa^2(-\sqrt A\,T_1\kappa+T_2/T_1)
+\tfrac{0.7742AT_2\kappa}{d}
-0.3208
+\tfrac{0.0679T_2}{T_1}$ &
$0.0224$ & $3097.3072$ \\

\addlinespace[1pt]

5 &
$\displaystyle
0.8039\sqrt A\,T_1\kappa
-\tfrac{0.8675AT_1\kappa(-T_2/T_1+d/\sqrt A)}{d}
-\tfrac{0.7636T_1}{T_2}
+0.1411$ &
$0.0226$ & $3078.9767$ \\

\addlinespace[1pt]

6 &
$\displaystyle
0.0047\sqrt A\,T_1\kappa
+\tfrac{0.0047\sqrt A\,T_2}{T_1d}
-0.0047AT_1^2\kappa^2
+\tfrac{0.8076AT_2\kappa}{d}
-0.3331
+\tfrac{0.0645T_2}{T_1}$ &
$0.0237$ & $3075.3877$ \\

\addlinespace[1pt]

7 &
$\displaystyle
-0.0045AT_1^2\kappa^2
-\tfrac{0.8146AT_2\kappa(-T_2/T_1+d/\sqrt A)}{d(T_2/T_1-d/\sqrt A)}
-0.3277
+\tfrac{0.0637T_2}{T_1}$ &
$0.0241$ & $3055.1463$ \\

\addlinespace[1pt]

8 &
$\displaystyle
\tfrac{0.8197AT_2\kappa}{d}
-0.0010T_2d\kappa
-0.3395
+\tfrac{0.0679T_2}{T_1}
+\tfrac{0.0010d}{\sqrt A}$ &
$0.0241$ & $3050.3108$ \\

\addlinespace[1pt]

9 &
$\displaystyle
-0.0005AT_1T_2\kappa^2
+\tfrac{0.8000AT_2\kappa}{d}
-0.3437
+\tfrac{0.0699T_2}{T_1}
+\tfrac{0.0005d}{\sqrt A}$ &
$0.0248$ & $3049.2081$ \\

\addlinespace[1pt]

10 &
$\displaystyle
0.0034\sqrt A\,T_1\kappa
-\tfrac{0.7927\sqrt A\,T_2(-\sqrt A\,T_1\kappa+d/\sqrt A)}{T_1d}
-0.3544
+\tfrac{0.8633T_2}{T_1}
-\tfrac{0.0034(\sqrt A\,T_1\kappa-d/\sqrt A)}{\sqrt A\,T_1\kappa}$ &
$0.0248$ & $3029.5726$ \\

\end{xltabular}
\endgroup

\begingroup
\scriptsize
\renewcommand{\arraystretch}{1.10}

\begin{xltabular}{\linewidth}{@{}c >{\raggedright\arraybackslash}X r r@{}}
\caption{Top $\mathsf{r}=10$ ranked symbolic expressions recovered by a single run of \vasstmain\ under $\sigma=0.1$ for learning \eqnref{eq:feynman-ftc}. Expressions are ranked by $\mathrm{LMPSE}$ among $H=2000$ sampled hard symbolic trees.}
\label{tab:vasst-top10-II-2-42-mapped-noise0p1}\\

\toprule
\toprule
\textbf{Rank} & \textbf{\vasstmain\ Expression Learned} & \textbf{\texttt{RMSE}} & $\mathrm{LMPSE}$ \\
\midrule
\endfirsthead

\multicolumn{4}{@{}l}{\scriptsize\emph{Table~\thetable\ continued from previous page.}}\\
\toprule
\toprule
\textbf{Rank} & \textbf{\vasstmain\ Expression Learned} & \textbf{\texttt{RMSE}} & $\mathrm{LMPSE}$ \\
\midrule
\endhead

\midrule
\multicolumn{4}{r@{}}{\scriptsize\emph{Continued on next page}}\\
\endfoot

\bottomrule
\bottomrule
\addlinespace[2pt]
\multicolumn{4}{@{}p{\linewidth}@{}}{
\scriptsize
\emph{Note}: \texttt{RMSE} corresponds to the out-of-sample \texttt{RMSE} computed on a $10\%$ held-out test set with the LMPSE selected symbolic model.
}\\
\endlastfoot

1 &
$\displaystyle
-0.0658\sqrt{A}T_{1}\kappa
+\tfrac{0.8713AT_{2}\kappa}{d}
-0.2464
+\tfrac{0.0466T_{2}}{T_{1}}$ &
$0.1013$ & $1519.3407$ \\

\addlinespace[1pt]

2 &
$\displaystyle
-0.0790\sqrt{A}T_{1}\kappa
+\tfrac{0.8853AT_{2}\kappa}{d}
-\tfrac{0.7128T_{1}}{T_{2}}
+0.1266$ &
$0.1015$ & $1519.0129$ \\

\addlinespace[1pt]

3 &
$\displaystyle
\tfrac{2.0411A^{3/2}T_{2}\kappa}{d^{2}}
-0.2565\sqrt{A}T_{1}\kappa
+0.1283\sqrt{A}T_{2}\kappa
-0.1242
+\tfrac{0.0285d}{\sqrt{A}}$ &
$0.1012$ & $1508.1725$ \\

\addlinespace[1pt]

4 &
$\displaystyle
-0.0041A^{3/2}T_{1}^{3}\kappa^{3}
+0.0041AT_{1}T_{2}\kappa^{2}
+\tfrac{0.7700AT_{2}\kappa}{d}
-0.3306
+\tfrac{0.0702T_{2}}{T_{1}}$ &
$0.1023$ & $1499.4539$ \\

\addlinespace[1pt]

5 &
$\displaystyle
0.0076\sqrt{A}T_{1}\kappa
+\tfrac{0.0076\sqrt{A}T_{2}}{T_{1}d}
-0.0076AT_{1}^{2}\kappa^{2}
+\tfrac{0.8186AT_{2}\kappa}{d}
-0.3429
+\tfrac{0.0633T_{2}}{T_{1}}$ &
$0.1024$ & $1499.0290$ \\

\addlinespace[1pt]

6 &
$\displaystyle
-\tfrac{0.8347\sqrt{A}T_{2}\kappa}{-T_{2}/T_{1}+T_{2}d/(\sqrt{A}T_{1})}
+\tfrac{1.0066AT_{2}\kappa}{d}
-\tfrac{0.0001T_{2}}{\sqrt{A}T_{1}^{2}\kappa-T_{2}}
+\tfrac{0.0001d}{AT_{1}\kappa-\sqrt{A}T_{2}/T_{1}}
+0.0148$ &
$0.1013$ & $1496.6708$ \\

\addlinespace[1pt]

7 &
$\displaystyle
\tfrac{0.8381AT_{2}\kappa}{d}
-0.0017T_{2}d\kappa
-0.3522
+\tfrac{0.0688T_{2}}{T_{1}}
+\tfrac{0.0017d}{\sqrt{A}}$ &
$0.1024$ & $1495.2980$ \\

\addlinespace[1pt]

8 &
$\displaystyle
-0.0071AT_{1}^{2}\kappa^{2}
+\tfrac{0.8274AT_{2}\kappa}{d}
-0.3337
+\tfrac{0.0625T_{2}}{T_{1}}$ &
$0.1025$ & $1486.1853$ \\

\addlinespace[1pt]

9 &
$\displaystyle
-0.0012AT_{1}T_{2}\kappa^{2}
+\tfrac{0.8147AT_{2}\kappa}{d}
-0.3606
+\tfrac{0.0700T_{2}}{T_{1}}
+\tfrac{0.0012d}{\sqrt A}$ &
$0.1030$ & $1485.0374$ \\

\addlinespace[1pt]

10 &
$\displaystyle
-0.0837\sqrt A\,T_{1}\kappa
+\tfrac{0.0837\sqrt A\,T_{2}}{T_{1}d}
-\tfrac{0.0837AT_{1}\kappa}{d}
+\tfrac{0.9207AT_{2}\kappa}{d}
-0.1414$ &
$0.1024$ & $1478.0320$ \\

\end{xltabular}
\endgroup

\begingroup
\scriptsize
\renewcommand{\arraystretch}{1.10}

\begin{xltabular}{\linewidth}{@{}c >{\raggedright\arraybackslash}X r r@{}}
\caption{Top $\mathsf{r}=10$ ranked symbolic expressions recovered by a single run of \vasstmain\ under $\sigma=0.2$ for learning \eqnref{eq:feynman-ftc}. Expressions are ranked by $\mathrm{LMPSE}$ among $H=2000$ sampled hard symbolic trees.}
\label{tab:vasst-top10-II-2-42-mapped-noise0p2}\\

\toprule
\toprule
\textbf{Rank} & \textbf{\vasstmain\ Expression Learned} & \textbf{\texttt{RMSE}} & $\mathrm{LMPSE}$ \\
\midrule
\endfirsthead

\multicolumn{4}{@{}l}{\scriptsize\emph{Table~\thetable\ continued from previous page.}}\\
\toprule
\toprule
\textbf{Rank} & \textbf{\vasstmain\ Expression Learned} & \textbf{\texttt{RMSE}} & $\mathrm{LMPSE}$ \\
\midrule
\endhead

\midrule
\multicolumn{4}{r@{}}{\scriptsize\emph{Continued on next page}}\\
\endfoot

\bottomrule
\bottomrule
\addlinespace[2pt]
\multicolumn{4}{@{}p{\linewidth}@{}}{
\scriptsize
\emph{Note}: \texttt{RMSE} corresponds to the out-of-sample \texttt{RMSE} computed on a $10\%$ held-out test set with the LMPSE selected symbolic model.
}\\
\endlastfoot

1 &
$\displaystyle
-0.0944\sqrt A\,T_1\kappa
+\tfrac{0.9031AT_2\kappa}{d}
-\tfrac{0.6620T_1}{T_2}
+0.1121$ &
$0.1998$ & $296.1263$ \\

\addlinespace[1pt]

2 &
$\displaystyle
-0.0832\sqrt A\,T_1\kappa
+\tfrac{0.8921AT_2\kappa}{d}
-0.2315
+\tfrac{0.0424T_2}{T_1}$ &
$0.1998$ & $290.2211$ \\

\addlinespace[1pt]

3 &
$\displaystyle
\tfrac{2.0374A^{3/2}T_2\kappa}{d^2}
-0.2604\sqrt A\,T_1\kappa
+0.1302\sqrt A\,T_2\kappa
-0.0737
+\tfrac{0.0160d}{\sqrt A}$ &
$0.2000$ & $285.3669$ \\

\addlinespace[1pt]

4 &
$\displaystyle
0.0105\sqrt A\,T_1\kappa
+\tfrac{0.0105\sqrt A\,T_2}{T_1d}
-0.0105AT_1^2\kappa^2
+\tfrac{0.8296AT_2\kappa}{d}
-0.3527
+\tfrac{0.0620T_2}{T_1}$ &
$0.2007$ & $279.9495$ \\

\addlinespace[1pt]

5 &
$\displaystyle
\tfrac{0.8566AT_2\kappa}{d}
-0.0024T_2d\kappa
-0.3648
+\tfrac{0.0697T_2}{T_1}
+\tfrac{0.0024d}{\sqrt A}$ &
$0.2005$ & $279.1597$ \\

\addlinespace[1pt]

6 &
$\displaystyle
0.9642\left(\tfrac{AT_2\kappa}{d}-\sqrt A\,T_1\kappa\right)
+0.8266\sqrt A\,T_1\kappa
-0.0003
\left\{
\tfrac{\sqrt A\,T_1\kappa}
{\sqrt A\,T_1\kappa-\sqrt A\,T_1\kappa}
\right\}$ &
$0.2009$ & $277.5255$ \\

\addlinespace[1pt]

7 &
$\displaystyle
-0.0050A^{3/2}T_1^3\kappa^3
+0.0050AT_1T_2\kappa^2
+\tfrac{0.7657AT_2\kappa}{d}
-0.3404
+\tfrac{0.0726T_2}{T_1}$ &
$0.2007$ & $276.9519$ \\

\addlinespace[1pt]

8 &
$\displaystyle
-0.0019AT_1T_2\kappa^2
+\tfrac{0.8293AT_2\kappa}{d}
-0.3775
+\tfrac{0.0701T_2}{T_1}
+\tfrac{0.0019d}{\sqrt A}$ &
$0.2011$ & $270.7425$ \\

\addlinespace[1pt]

9 &
$\displaystyle
-0.0096AT_1^2\kappa^2
+\tfrac{0.8401AT_2\kappa}{d}
-0.3396
+\tfrac{0.0612T_2}{T_1}$ &
$0.2007$ & $267.8984$ \\

\addlinespace[1pt]

10 &
$\displaystyle
-0.0911\sqrt A\,T_1\kappa
+\tfrac{0.0911\sqrt A\,T_2}{T_1d}
-\tfrac{0.0911AT_1\kappa}{d}
+\tfrac{0.9303AT_2\kappa}{d}
-0.1489$ &
$0.2001$ & $263.8642$ \\

\end{xltabular}
\endgroup

\begin{center}
\underline{For learning \eqnref{eq:feynman-ada}: $f = \beta^{\dagger}(1 + \alpha^{\dagger}\cos \theta)$ in~\hyperref[subsec:feynman-data-study]{Section~\ref{subsec:feynman-data-study}}}.
\end{center}

\begingroup
\scriptsize
\renewcommand{\arraystretch}{1.10}

\begin{xltabular}{\linewidth}{@{}c >{\raggedright\arraybackslash}X r r@{}}
\caption{Top $\mathsf{r}=10$ ranked symbolic expressions recovered by a single run of \vasstmain\ under the noiseless setting for learning \eqnref{eq:feynman-ada}. Expressions are ranked by $\mathrm{LMPSE}$ among $H=2000$ sampled hard symbolic trees.}
\label{tab:vasst-top10-III-17-37-mapped-noiseless}\\

\toprule
\toprule
\textbf{Rank} & \textbf{\vasstmain\ Expression Learned} & \textbf{\texttt{RMSE}} & $\mathrm{LMPSE}$ \\
\midrule
\endfirsthead

\multicolumn{4}{@{}l}{\scriptsize\emph{Table~\thetable\ continued from previous page.}}\\
\toprule
\toprule
\textbf{Rank} & \textbf{\vasstmain\ Expression Learned} & \textbf{\texttt{RMSE}} & $\mathrm{LMPSE}$ \\
\midrule
\endhead

\midrule
\multicolumn{4}{r@{}}{\scriptsize\emph{Continued on next page}}\\
\endfoot

\bottomrule
\bottomrule
\addlinespace[2pt]
\multicolumn{4}{@{}p{\linewidth}@{}}{
\scriptsize
\emph{Note}: \texttt{RMSE} corresponds to the out-of-sample \texttt{RMSE} computed on a $10\%$ held-out test set with the LMPSE selected symbolic model.
}\\
\endlastfoot

1 &
$\displaystyle
1.0238\alpha^\dagger\beta^\dagger\cos\theta
+1.0238\beta^\dagger
+0.00103\theta
+0.0012\cos(2\alpha^\dagger)
-0.0009$ &
$0.0013$ & $1000.8680$ \\

\addlinespace[1pt]

2 &
$\displaystyle
0.0975\alpha^\dagger
+0.0505\beta^\dagger(\beta^\dagger+2\theta)
+0.0325\theta
+0.8895(\alpha^\dagger\beta^\dagger+\beta^\dagger)\cos\theta
+0.0325\cos(\beta^\dagger)
-0.5612$ &
$0.5480$ & $-1720.7743$ \\

\addlinespace[1pt]

3 &
$\displaystyle
0.0440\alpha^\dagger\beta^\dagger(\beta^\dagger+\theta)
-0.1722\beta^\dagger\cos(\alpha^\dagger)
+1.0081\{\alpha^\dagger\beta^\dagger+\cos(\alpha^\dagger)\}\cos\theta
+0.0440\cos(\beta^\dagger+\theta)
-1.6680$ &
$0.7250$ & $-2283.3134$ \\

\addlinespace[1pt]

4 &
$\displaystyle
0.8978\alpha^\dagger\beta^\dagger\cos\theta
+0.0826\alpha^\dagger
+0.0476\beta^\dagger\theta^2
+0.8978\beta^\dagger\cos\theta
+0.0476\theta
-2.9743$ &
$0.7862$ & $-2428.7997$ \\

\addlinespace[1pt]

5 &
$\displaystyle
1.0678\alpha^\dagger\beta^\dagger\cos\theta
+0.1879\alpha^\dagger\theta
-1.1075\alpha^\dagger
+0.5636\beta^\dagger
-0.5537\cos(\beta^\dagger)
-0.0448$ &
$0.8996$ & $-2700.6837$ \\

\addlinespace[1pt]

6 &
$\displaystyle
1.0092\alpha^\dagger\beta^\dagger\cos\theta
+0.2084\cos(\alpha^\dagger)
+0.2084\cos(2\alpha^\dagger)
+0.2084\cos(\beta^\dagger)
-0.0643\cos(2\alpha^\dagger\beta^\dagger\theta)
+3.4579$ &
$1.1653$ & $-3228.7870$ \\

\addlinespace[1pt]

7 &
$\displaystyle
1.0434\alpha^\dagger\beta^\dagger\cos\theta
-0.1611\beta^\dagger\cos\theta
-0.0155\cos\{(\alpha^\dagger)^2\}
-0.1611\cos(2\beta^\dagger)
+3.4235$ &
$1.1933$ & $-3266.0467$ \\

\addlinespace[1pt]

8 &
$\displaystyle
1.1812\beta^\dagger
+4.3937(\alpha^\dagger+\beta^\dagger)\cos\theta
-20.0205\cos\theta
+4.3937\cos(\cos\theta)
-3.9709$ &
$1.7949$ & $-4090.8640$ \\

\addlinespace[1pt]

9 &
$\displaystyle
-0.1459\alpha^\dagger
+0.8872\beta^\dagger\{\alpha^\dagger+\cos(\alpha^\dagger)\}\cos\theta
-0.0613\theta^2
-0.0613\theta
-0.0613\cos\{\cos(\alpha^\dagger)\}
+8.1100$ &
$1.8462$ & $-4140.0795$ \\

\addlinespace[1pt]

10 &
$\displaystyle
0.5482\beta^\dagger(\alpha^\dagger+\beta^\dagger)\cos\theta
+0.8815\beta^\dagger\cos(\cos\theta)
-1.0354\theta
-1.0354\cos(\cos\theta)
+9.9864$ &
$1.9516$ & $-4253.3533$ \\

\end{xltabular}
\endgroup

\begingroup
\scriptsize
\renewcommand{\arraystretch}{1.10}

\begin{xltabular}{\linewidth}{@{}c >{\raggedright\arraybackslash}X r r@{}}
\caption{Top $\mathsf{r}=10$ ranked symbolic expressions recovered by a single run of \vasstmain\ under $\sigma=0.1$ for learning \eqnref{eq:feynman-ada}. Expressions are ranked by $\mathrm{LMPSE}$ among $H=2000$ sampled hard symbolic trees.}
\label{tab:vasst-top10-III-17-37-mapped-noise0p1}\\

\toprule
\toprule
\textbf{Rank} & \textbf{\vasstmain\ Expression Learned} & \textbf{\texttt{RMSE}} & $\mathrm{LMPSE}$ \\
\midrule
\endfirsthead

\multicolumn{4}{@{}l}{\scriptsize\emph{Table~\thetable\ continued from previous page.}}\\
\toprule
\toprule
\textbf{Rank} & \textbf{\vasstmain\ Expression Learned} & \textbf{\texttt{RMSE}} & $\mathrm{LMPSE}$ \\
\midrule
\endhead

\midrule
\multicolumn{4}{r@{}}{\scriptsize\emph{Continued on next page}}\\
\endfoot

\bottomrule
\bottomrule
\addlinespace[2pt]
\multicolumn{4}{@{}p{\linewidth}@{}}{
\scriptsize
\emph{Note}: \texttt{RMSE} corresponds to the out-of-sample \texttt{RMSE} computed on a $10\%$ held-out test set with the LMPSE selected symbolic model.
}\\
\endlastfoot

1 &
$\displaystyle
1.0241\alpha^\dagger\beta^\dagger\cos\theta
+1.0241\beta^\dagger
+0.0001\theta
+0.0003\cos(2\alpha^\dagger)
-0.0001$ &
$0.1083$ & $596.2956$ \\

\addlinespace[1pt]

2 &
$\displaystyle
0.8896\alpha^\dagger\beta^\dagger\cos\theta
+0.0999\alpha^\dagger
+0.0504(\beta^\dagger)^2
+0.1009\beta^\dagger\theta
+0.8896\beta^\dagger\cos\theta
+0.0333\theta
+0.0333\cos(\beta^\dagger)
-0.5769$ &
$0.5599$ & $-1763.2871$ \\

\addlinespace[1pt]

3 &
$\displaystyle
0.0440\alpha^\dagger(\beta^\dagger)^2
+0.0440\alpha^\dagger\beta^\dagger\theta
+1.0081\alpha^\dagger\beta^\dagger\cos\theta
-0.1715\beta^\dagger\cos(\alpha^\dagger)
+1.0081\cos(\alpha^\dagger)\cos\theta
+0.0440\cos(\beta^\dagger+\theta)
-1.6715$ &
$0.7327$ & $-2297.4057$ \\

\addlinespace[1pt]

4 &
$\displaystyle
0.8978\alpha^\dagger\beta^\dagger\cos\theta
+0.0847\alpha^\dagger
+0.0476\beta^\dagger\theta^2
+0.8978\beta^\dagger\cos\theta
+0.0476\theta
-2.9871$ &
$0.7872$ & $-2433.3839$ \\

\addlinespace[1pt]

5 &
$\displaystyle
1.0679\alpha^\dagger\beta^\dagger\cos\theta
+0.1878\alpha^\dagger\theta
-1.1050\alpha^\dagger
+0.5634\beta^\dagger
-0.5525\cos(\beta^\dagger)
-0.0551$ &
$0.9001$ & $-2704.0075$ \\

\addlinespace[1pt]

6 &
$\displaystyle
1.0092\alpha^\dagger\beta^\dagger\cos\theta
+0.2101\cos(\alpha^\dagger)
+0.2101\cos(2\alpha^\dagger)
+0.2101\cos(\beta^\dagger)
-0.0647\cos(2\alpha^\dagger\beta^\dagger\theta)
+3.4477$ &
$1.1641$ & $-3226.8548$ \\

\addlinespace[1pt]

7 &
$\displaystyle
1.0438\alpha^\dagger\beta^\dagger\cos\theta
-0.1627\beta^\dagger\cos\theta
-0.0180\cos\{(\alpha^\dagger)^2\}
-0.1627\cos(2\beta^\dagger)
+3.4140$ &
$1.1915$ & $-3263.6318$ \\

\addlinespace[1pt]

8 &
$\displaystyle
4.3980\alpha^\dagger\cos\theta
+4.3980\beta^\dagger\cos\theta
+1.1795\beta^\dagger
-20.0600\cos\theta
+4.3980\cos(\cos\theta)
-3.9615$ &
$1.7949$ & $-4091.1167$ \\

\addlinespace[1pt]

9 &
$\displaystyle
0.8875\alpha^\dagger\beta^\dagger\cos\theta
-0.1441\alpha^\dagger
+0.8875\beta^\dagger\cos(\alpha^\dagger)\cos\theta
-0.0610\theta^2
-0.0610\theta
-0.0610\cos\{\cos(\alpha^\dagger)\}
+8.0903$ &
$1.8461$ & $-4140.0206$ \\

\addlinespace[1pt]

10 &
$\displaystyle
0.5484\alpha^\dagger\beta^\dagger\cos\theta
+0.5484(\beta^\dagger)^2\cos\theta
+0.8827\beta^\dagger\cos(\cos\theta)
-1.0310\theta
-1.0310\cos(\cos\theta)
+9.9696$ &
$1.9507$ & $-4252.5163$ \\

\end{xltabular}
\endgroup

\begingroup
\scriptsize
\renewcommand{\arraystretch}{1.10}

\begin{xltabular}{\linewidth}{@{}c >{\raggedright\arraybackslash}X r r@{}}
\caption{Top $\mathsf{r}=10$ ranked symbolic expressions recovered by a single run of \vasstmain\ under $\sigma=0.2$ for learning \eqnref{eq:feynman-ada}. Expressions are ranked by $\mathrm{LMPSE}$ among $H=2000$ sampled hard symbolic trees.}
\label{tab:vasst-top10-III-17-37-mapped-noise0p2}\\

\toprule
\toprule
\textbf{Rank} & \textbf{\vasstmain\ Expression Learned} & \textbf{\texttt{RMSE}} & $\mathrm{LMPSE}$ \\
\midrule
\endfirsthead

\multicolumn{4}{@{}l}{\scriptsize\emph{Table~\thetable\ continued from previous page.}}\\
\toprule
\toprule
\textbf{Rank} & \textbf{\vasstmain\ Expression Learned} & \textbf{\texttt{RMSE}} & $\mathrm{LMPSE}$ \\
\midrule
\endhead

\midrule
\multicolumn{4}{r@{}}{\scriptsize\emph{Continued on next page}}\\
\endfoot

\bottomrule
\bottomrule
\addlinespace[2pt]
\multicolumn{4}{@{}p{\linewidth}@{}}{
\scriptsize
\emph{Note}: \texttt{RMSE} corresponds to the out-of-sample \texttt{RMSE} computed on a $10\%$ held-out test set with the LMPSE selected symbolic model.
}\\
\endlastfoot

1 &
$\displaystyle
1.0244\alpha^\dagger\beta^\dagger\cos\theta
+1.0244\beta^\dagger
+0.0011\theta
+0.0005\cos(2\alpha^\dagger)
-0.0009$ &
$0.2004$ & $-73.3202$ \\

\addlinespace[1pt]

2 &
$\displaystyle
0.8896\alpha^\dagger\beta^\dagger\cos\theta
+0.1024\alpha^\dagger
+0.0504(\beta^\dagger)^2
+0.1007\beta^\dagger\theta
+0.8896\beta^\dagger\cos\theta
+0.0341\theta
+0.0341\cos(\beta^\dagger)
-0.5925$ &
$0.5797$ & $-1871.4716$ \\

\addlinespace[1pt]

3 &
$\displaystyle
0.0440\alpha^\dagger(\beta^\dagger)^2
+0.0440\alpha^\dagger\beta^\dagger\theta
+1.0082\alpha^\dagger\beta^\dagger\cos\theta
-0.1708\beta^\dagger\cos(\alpha^\dagger)
+1.0082\cos(\alpha^\dagger)\cos\theta
+0.0440\cos(\beta^\dagger+\theta)
-1.6750$ &
$0.7474$ & $-2343.8105$ \\

\addlinespace[1pt]

4 &
$\displaystyle
0.8979\alpha^\dagger\beta^\dagger\cos\theta
+0.0867\alpha^\dagger
+0.0477\beta^\dagger\theta^2
+0.8979\beta^\dagger\cos\theta
+0.0477\theta
-2.9999$ &
$0.7891$ & $-2443.9352$ \\

\addlinespace[1pt]

5 &
$\displaystyle
1.0679\alpha^\dagger\beta^\dagger\cos\theta
+0.1877\alpha^\dagger\theta
-1.1024\alpha^\dagger
+0.5631\beta^\dagger
-0.5512\cos(\beta^\dagger)
-0.0654$ &
$0.9011$ & $-2711.0094$ \\

\addlinespace[1pt]

6 &
$\displaystyle
1.0093\alpha^\dagger\beta^\dagger\cos\theta
+0.2119\cos(\alpha^\dagger)
+0.2119\cos(2\alpha^\dagger)
+0.2119\cos(\beta^\dagger)
-0.0523\cos(2\alpha^\dagger\beta^\dagger\theta)
+3.4568$ &
$1.1614$ & $-3217.6691$ \\

\addlinespace[1pt]

7 &
$\displaystyle
1.0441\alpha^\dagger\beta^\dagger\cos\theta
-0.1643\beta^\dagger\cos\theta
-0.0206\cos\{(\alpha^\dagger)^2\}
-0.1643\cos(2\beta^\dagger)
+3.4217$ &
$1.1877$ & $-3257.2784$ \\

\addlinespace[1pt]

8 &
$\displaystyle
4.4023\alpha^\dagger\cos\theta
+4.4023\beta^\dagger\cos\theta
+1.1778\beta^\dagger
-20.1000\cos\theta
+4.4023\cos(\cos\theta)
-3.9658$ &
$1.7950$ & $-4092.0921$ \\

\addlinespace[1pt]

9 &
$\displaystyle
0.8878\alpha^\dagger\beta^\dagger\cos\theta
-0.1423\alpha^\dagger
+0.8878\beta^\dagger\cos(\alpha^\dagger)\cos\theta
-0.0607\theta^2
-0.0607\theta
-0.0607\cos\{\cos(\alpha^\dagger)\}
+8.0502$ &
$1.8460$ & $-4139.9095$ \\

\addlinespace[1pt]

10 &
$\displaystyle
0.5487\alpha^\dagger\beta^\dagger\cos\theta
+0.5487(\beta^\dagger)^2\cos\theta
+0.8838\beta^\dagger\cos(\cos\theta)
-1.0266\theta
-1.0266\cos(\cos\theta)
+9.9047$ &
$1.9489$ & $-4249.8368$ \\

\end{xltabular}
\endgroup
\newpage

\section{Distinction between \texorpdfstring{\vasst}{VaSST} and \vart}
\label{app:vart-vs-vasst}

Although \vasstmain\ and \vart~\citep{VART} both use variational inference over relaxed tree-structured objects, the underlying statistical objects, computational relaxations, and inferential targets are fundamentally different. \vart\ is designed for variational inference in decision trees, which consist of recursive partitions of the covariate space followed by prediction through routing of observations to terminal regions. In contrast, \vasstmain\ addresses \sr, where trees encode compositional mathematical expressions; see \hyperref[sec:symbolic-tree-representation]{Section~\ref{sec:symbolic-tree-representation}}. This distinction changes the meaning of every node of the underlying tree structures, the form of the forward evaluation, the nature of the posterior distribution, and the role of uncertainty quantification, as summarized below.

First, as noted in \hyperref[sec:intro]{Sections~\ref{sec:intro}} and \ref{sec:symbolic-tree-representation}, \vasstmain\ is based on symbolic tree representations. Internal nodes in a symbolic tree represent mathematical operators such as $+$, $-$, $\times$, $/$, $\sin$, $\cos$, $\log$, $\exp$, and powers, while terminal nodes represent primitive features. Thus, a symbolic tree evaluates to an analytical expression. This is fundamentally different from a decision tree, where internal nodes represent split rules and the tree defines a partition of the covariate space rather than a closed-form scientific law.

Second, \vasstmain\ requires a new \emph{soft symbolic evaluation} mechanism; see \hyperref[alg:soft-eval]{\textsc{SoftEval} Algorithm ~\ref{alg:soft-eval}} of \hyperref[app:soft-evaluation-algorithms]{Appendix~\ref{app:soft-evaluation-algorithms}}. In \vart, relaxation is achieved through soft routing, where an observation is probabilistically routed through a decision tree. In \vasstmain, there is no routing interpretation. Instead, the relaxation must propagate mixtures over operator and feature choices through a compositional expression tree. The resulting soft tree is therefore a differentiable surrogate for symbolic expression evaluation, not a probabilistic partitioning rule.

Third, \vasstmain\ uses reparameterized relaxations tailored to symbolic structures. In~\eqnref{eq:continuous-relaxation}, Binary-Concrete variables are used for expansion decisions, while Gumbel-Softmax relaxations are used for operator and feature assignments. These relaxations must respect the compositional nature of \sr, where invalid or numerically unstable operations, such as division by near-zero quantities or undefined logarithmic evaluations, require careful treatment. Such issues do not arise in standard soft decision tree routing.

Fourth, \vasstmain\ incorporates a collapsed Bayesian linear layer that is absent in the \vart\ modeling framework. Specifically, \vasstmain\ places a conjugate Normal Inverse-Gamma prior on the outer regression coefficients and noise variance, and analytically marginalizes these quantities. This yields a collapsed structural objective over symbolic ensembles, allowing the variational optimization to focus on expression structure while preserving uncertainty over the regression layer. This collapsed formulation is central to the underlying Bayesian model of \vasstmain.

Finally, the posterior object in \vasstmain\ is a distribution over symbolic expressions, not a distribution over predictive partitions. This enables posterior-aware sampling, ranking, and selection of candidate scientific laws after optimization. Consequently, uncertainty in \vasstmain\ is uncertainty over competing analytical explanations, whereas uncertainty in \vart\ concerns alternative tree partitions and their associated predictions. Thus, \vasstmain\ is not a direct application of \vart\ to \sr; it introduces a distinct probabilistic framework for differentiable, uncertainty-aware discovery of compositional scientific expressions.
\newpage

\section{Ablation Studies for \texorpdfstring{\vasst}{VaSST}}
\label{app:ablation-study}

Using ablation studies, we investigate the sensitivity of \vasstmain\ with respect to the number of symbolic trees $K$ and the maximum tree depth $D$. These two architectural hyperparameters control complementary aspects of the symbolic representation. The number of trees $K$ determines the number of additive symbolic components in the final ensemble, whereas the depth $D$ controls the maximum compositional complexity allowed within each symbolic tree.

We consider two representative Feynman equations~\citep{AI-Feynman} from~\hyperref[subsec:feynman-data-study]{Section~\ref{subsec:feynman-data-study}}:
\begin{align*}
\eqnref{eq:feynman-cpe}:\; \Delta U = G m_1 m_2\left(\frac{1}{r_2} - \frac{1}{r_1}\right), \quad \eqnref{eq:feynman-ftc}:\;P=\frac{\kappa A(T_2-T_1)}{d}.
\end{align*}
For each equation in the preceding display, $n=2000$ observations are subsampled from $10^{5}$ observations, and Gaussian noise with variance $\sigma^{2} = 0.1^{2}$ is added to the response. We then run \vasstmain\ over the grid $K\in \{2, 3, 4, 5\}$ and $D\in \{2, 3, 4, 5\}$, while keeping all other prior, relaxation, and optimization settings configured to be the same as in~\hyperref[subsec:vasst-configs]{Appendix~\ref{subsec:vasst-configs}}. The operator sets used for \eqnref{eq:feynman-cpe} and \eqnref{eq:feynman-ftc} are $\mathbf{O} = \{+,\times,/\}$ and $\mathbf{O} = \{\times, -, /\}$, respectively.

\begin{table}[htbp]
\centering
\caption{Runtime and out-of-sample \texttt{RMSE} (computed on a $10\%$ held-out test set) of \vasstmain\ for $\eqnref{eq:feynman-cpe}:\;\Delta U = Gm_1m_2\left(\frac{1}{r_2} - \frac{1}{r_1}\right)$ across different choices of $K$ and $D$.}
\label{tab:ablation-I_13_12}

\begin{minipage}[t]{0.47\textwidth}
\centering
\textbf{Runtime}

\vspace{4pt}

\begin{tabular}{@{}ccccc@{}}
\toprule
\toprule
$K \backslash D$ & $2$ & $3$ & $4$ & $5$ \\
\midrule
$2$ & $54.7$ & $98.8$ & $192.4$ & $399.6$ \\
$3$ & $61.5$ & $112.0$ & $221.4$ & $491.3$ \\
$4$ & $67.9$ & $123.1$ & $251.8$ & $600.9$ \\
$5$ & $72.3$ & $134.8$ & $286.4$ & $800.5$ \\
\bottomrule
\bottomrule
\end{tabular}
\end{minipage}
\hfill
\begin{minipage}[t]{0.47\textwidth}
\centering
\textbf{\texttt{RMSE}}

\vspace{4pt}

\begin{tabular}{@{}ccccc@{}}
\toprule
\toprule
$K \backslash D$ & $2$ & $3$ & $4$ & $5$ \\
\midrule
$2$ & $0.2330$ & $0.1012$ & $0.1959$ & $0.2349$ \\
$3$ & $0.1002$ & $0.1019$ & $0.1051$ & $0.1918$ \\
$4$ & $0.1168$ & $0.1002$ & $0.1005$ & $0.1012$ \\
$5$ & $0.1002$ & $0.1005$ & $0.1004$ & $0.1010$ \\
\bottomrule
\bottomrule
\end{tabular}
\end{minipage}
\end{table}

\begin{table}[htbp]
\centering
\caption{Runtime and out-of-sample \texttt{RMSE} (computed on a $10\%$ held-out test set) of \vasstmain\ for $\eqnref{eq:feynman-ftc}:\;P = \frac{\kappa A(T_2 - T_1)}{d}$ across different choices of $K$ and $D$.}
\label{tab:ablation-II_2_42}

\begin{minipage}[t]{0.47\textwidth}
\centering
\textbf{Runtime}

\vspace{4pt}

\begin{tabular}{@{}ccccc@{}}
\toprule
\toprule
$K \backslash D$ & $2$ & $3$ & $4$ & $5$ \\
\midrule
$2$ & $50.6$ & $97.2$ & $191.3$ & $401.6$ \\
$3$ & $62.1$ & $112.0$ & $221.8$ & $498.5$ \\
$4$ & $72.2$ & $125.0$ & $249.8$ & $628.2$ \\
$5$ & $83.1$ & $146.0$ & $299.4$ & $803.4$ \\
\bottomrule
\bottomrule
\end{tabular}
\end{minipage}
\hfill
\begin{minipage}[t]{0.47\textwidth}
\centering
\textbf{\texttt{RMSE}}

\vspace{4pt}

\begin{tabular}{@{}ccccc@{}}
\toprule
\toprule
$K \backslash D$ & $2$ & $3$ & $4$ & $5$ \\
\midrule
$2$ & $0.1019$ & $0.1002$ & $0.1002$ & $0.1070$ \\
$3$ & $0.1000$ & $0.1011$ & $0.1012$ & $0.1019$ \\
$4$ & $0.1000$ & $0.0999$ & $0.1000$ & $0.1000$ \\
$5$ & $0.0999$ & $0.0999$ & $0.0998$ & $0.1007$ \\
\bottomrule
\bottomrule
\end{tabular}
\end{minipage}
\end{table}

\hyperref[tab:ablation-I_13_12]{Tables~\ref{tab:ablation-I_13_12}} and~\ref{tab:ablation-II_2_42} report the runtimes and out-of-sample \texttt{RMSE}s of \vasstmain\ across the chosen $(K, D)$ grid while learning \eqnref{eq:feynman-cpe} and \eqnref{eq:feynman-ftc}. The key findings are summarized as follows.

\paragraph{Performance stabilizes with model complexity.}
For both the Feynman equations, out-of-sample \texttt{RMSE} improves significantly and tracks the noise level $\sigma^{2} = 0.1^{2}$ closely when moving from $K=2$ to $K\geq 3$, and stabilizes for $K\geq 4$ and $D\geq 3$. This indicates that \vasstmain\ does not require large model sizes to achieve optimal predictive accuracy.

\paragraph{Low capacity regimes can limit symbolic recovery.}
As expected, insufficient model capacity leads to occasional failures in recovering the true symbolic law. For \eqnref{eq:feynman-ftc}, failure occurs at $(K=2, D=5)$, and for \eqnref{eq:feynman-cpe}, failure occurs at $(K=2, D\in \{2, 4, 5\})$ and $(K=3, D=5)$. This highlights that both sufficient number of trees $K$ and appropriate depth $D$ are necessary for capturing the underlying compositional symbolic structure.

\paragraph{Runtime behavior with $K$ and $D$.}
For fixed depth $D$, runtime increases approximately linearly in $K$ and for fixed $K$, runtime increases much more rapidly with $D$, reflecting the exponential growth in the number of nodes in a symbolic tree. Importantly, moderate configurations, e.g., $(K=3, D=3)$, already achieve near-optimal \texttt{RMSE} value, offering a favorable accuracy-efficiency trade-off.

\paragraph{Principled heuristics for choosing $K$ and $D$.}
Based on the above ablation studies, we prescribe the following practical guideline for choosing $K$ and $D$. Since \vasstmain\ is designed to recover simple and interpretable scientific expressions, our recommendation follows an Occam's razor principle~\citep{Occams-Razor-1}: choose the smallest model capacity $(K, D)$ that attains stable out-of-sample \texttt{RMSE} and symbolic recovery. Empirically, $K\in \{3, 4\}$ and $D\in \{3, 4\}$ provide a favorable accuracy-interpretability-runtime trade-off in the applications considered in~\hyperref[sec:vasst-in-action]{Section~\ref{sec:vasst-in-action}}, while larger depths increase runtime exponentially and may lead to unnecessarily complex symbolic expressions without substantial predictive gains.
\newpage

\section{Computational Scaling and Gradient Variance Diagnostics for \texorpdfstring{\vasst}{VaSST}}
\label{app:memory-footprint}

In this section, we examine how the computational footprint and stochastic gradient behavior of \vasstmain\ scale with the maximum symbolic tree depth $D$. This diagnostic complements the ablation studies in~\hyperref[app:ablation-study]{Appendix~\ref{app:ablation-study}} by isolating the effect of depth on memory, runtime, and gradient variability in the relaxed symbolic tree parameterization.

\subsection{Scaling of the Relaxed Symbolic Tree Parameterization}

A full binary tree skeleton $\mathsf{S}_j$ of maximum depth $D$ contains $N=2^{D+1}-1$ nodes. For each soft symbolic tree $\{\mathsf{S}_j^{\mathrm{soft}}\}_{j=1}^{K}$, \vasstmain\ maintains relaxed logits for node expansion ($\ell_{j\zeta}$), operator assignment ($\mathbf a_{j\zeta}^{\mathrm{op}}$), and feature assignment ($\mathbf a_{j\zeta}^{\mathrm{ft}}$) for every node $\zeta \in \mathsf Z_D = \{0, 1, \ldots, N-1\}$, along with the variational concentration hyperparameters $\widetilde{\bm \eta}_{\mathrm{op}}$ and $\widetilde{\bm \eta}_{\mathrm{ft}}$. Therefore, if $|\mathbf O|$ denotes the number of operators and $p$ denotes the input dimension, then the total number of variational parameters in $\phi = ((\ell_{j\zeta}, \mathbf{a}^{\mathrm{op}}_{j\zeta}, \mathbf{a}^{\mathrm{ft}}_{j\zeta})_{j=1,\ldots,K;\zeta\in \mathsf{Z}_D}, \widetilde{\bm \eta}_{\mathrm{op}}, \widetilde{\bm \eta}_{\mathrm{ft}})$ is $d_{\mathrm{relax}}(K, D) = |\mathbf{O}| + p + K N(1 + |\mathbf{O}| + p)$. Thus, for fixed $|\mathbf{O}|$ and $p$, the number of variational parameters grows as $d_{\mathrm{relax}}(K, D) = \mathsf{O}(K 2^{D})$.

This scaling gives the dominant depth-dependent memory footprint for the structural relaxation. In addition, during Monte Carlo (\texttt{MC}) \texttt{ELBO} evaluation in the \hyperref[alg:approx-elbo]{\textsc{ApproxELBO} Algorithm~\ref{alg:approx-elbo}} in \hyperref[app:approx-elbo]{Appendix~\ref{app:approx-elbo}}, the soft tree forward pass stores node-level evaluations and stochastic relaxation samples. With $S$ \texttt{MC} samples and $n$ training observations, the activation-level memory is of order $\mathsf{O}(S n K N)$ up to operator-specific constants. Hence, the memory footprint grows approximately linearly in $K$, $S$, and $n$, but exponentially in $D$.

The same node count scaling also determines runtime. A stochastic forward-backward pass evaluates all relaxed nodes and backpropagates through the Binary-Concrete and Gumbel-Softmax relaxations. Therefore, the per-iteration computational cost is approximately $\mathsf{O}(Sn KNC_{\mathbf O})$, where $C_{\mathbf{O}}$ denotes the average cost of evaluating the operator mixture. Since $N=2^{D+1}-1$, the runtime is expected to increase exponentially with depth $D$.

\subsection{Gradient Variance Diagnostic}

We next quantify how the stochastic gradient variability changes with depth $D$. 
For one draw of the relaxed tree variables, the stochastic gradient of the negative \texttt{ELBO} is $g(\phi; \omega) = \nabla_{\phi}\{-\widehat{\mathcal{E}}( \phi; \omega)\}$, where $\omega$ encodes the randomness owing to the Binary-Concrete and Gumbel-Softmax relaxations. We estimate the coordinate-wise gradient variance using $R$ independent stochastic gradients $g^{(1)}, \ldots, g^{(R)}$:
\begin{equation}
\label{eq:coordinate-wise-gradient-variance}
\widehat{\mathbb{V}}_j(g) = \frac{1}{R-1}\sum_{r=1}^{R}(g_j^{(r)} - \overline{g}_j)^2, \quad \overline{g}_{j} = \frac{1}{R}\sum_{r=1}^{R}g_{j}^{(r)}.
\end{equation}
We summarize the gradient variability in~\eqnref{eq:coordinate-wise-gradient-variance} using:
\begin{equation}
\label{eq:var-summary}
\widehat{V}_{\mathrm{coord}} = \frac{1}{d_{\mathrm{relax}}}\sum_{j=1}^{d_{\mathrm{relax}}}\widehat{\mathbb V}_j(g), \quad \widehat{V}_{\mathrm{total}} = \sum_{j=1}^{d_{\mathrm{relax}}}\widehat{\mathbb{V}}_{j}(g),
\end{equation}
which are the mean coordinate-wise gradient variance and the total gradient variance across all relaxed structural logits, respectively. The first quantity in~\eqnref{eq:var-summary} measures the average stochastic fluctuation per relaxed coordinate, while the second quantity in~\eqnref{eq:var-summary} measures the aggregate stochastic fluctuation over the full relaxed symbolic tree parameterization.

\subsection{Results}

We run the diagnostic experiment on the Feynman equation $\eqnref{eq:feynman-cpe}:\; \Delta U = G m_1m_2(\frac{1}{r_2} - \frac{1}{r_1})$ from \hyperref[subsec:feynman-data-study]{Section~\ref{subsec:feynman-data-study}} for which \vasstmain\ is implemented with $K=3$ symbolic trees and depths $D\in \{2, 3, 4, 5\}$. The diagnostic experiment in run on CPU with $n=2000$ subsampled training observations and Gaussian noise level $\sigma^{2} = 0.2^2$ added to the response. The operator set is configured as $\mathbf{O} = \{+,\times,/\}$. All remaining settings follow similarly  from~\hyperref[subsec:vasst-configs]{Appendix~\ref{subsec:vasst-configs}}.
In this experiment, $K=3$, $|\mathbf O|=3$, and $p=3$. Therefore, the size of the relaxed structural parameterization grows as shown in \hyperref[tab:relaxed-logit-scaling]{Table~\ref{tab:relaxed-logit-scaling}}.

\begin{table}[H]
\centering
\caption{Growth of the relaxed symbolic parameterization for $K=3$, $|\mathbf O|=3$, and $p=3$.}
\label{tab:relaxed-logit-scaling}
\begin{tabular}{@{}cccc@{}}
\toprule
\toprule
Depth $D$ & Nodes per tree ($N$) & Total nodes ($K N$) & Variational Parameters ($d_{\mathrm{relax}}(K,D)$) \\
\midrule
$2$ & $7$  & $21$  & $153$ \\
$3$ & $15$ & $45$  & $321$ \\
$4$ & $31$ & $93$  & $657$ \\
$5$ & $63$ & $189$ & $1329$ \\
\bottomrule
\bottomrule
\end{tabular}
\end{table}

For each depth $D$, we measure the runtime of one stochastic forward-backward pass using $S=8$ \texttt{MC} samples in the \texttt{ELBO}. For the gradient variance diagnostic, we generate $R=100$ independent stochastic gradients using one \texttt{MC} sample per gradient. Gradients are collected only with respect $\phi$. Each diagnostic is repeated $5$ times, and the median value is reported. We compute these quantities before training and after $2000$ training iterations.

\begin{figure}[!htp]
    \centering
    \includegraphics[width=0.45\linewidth]{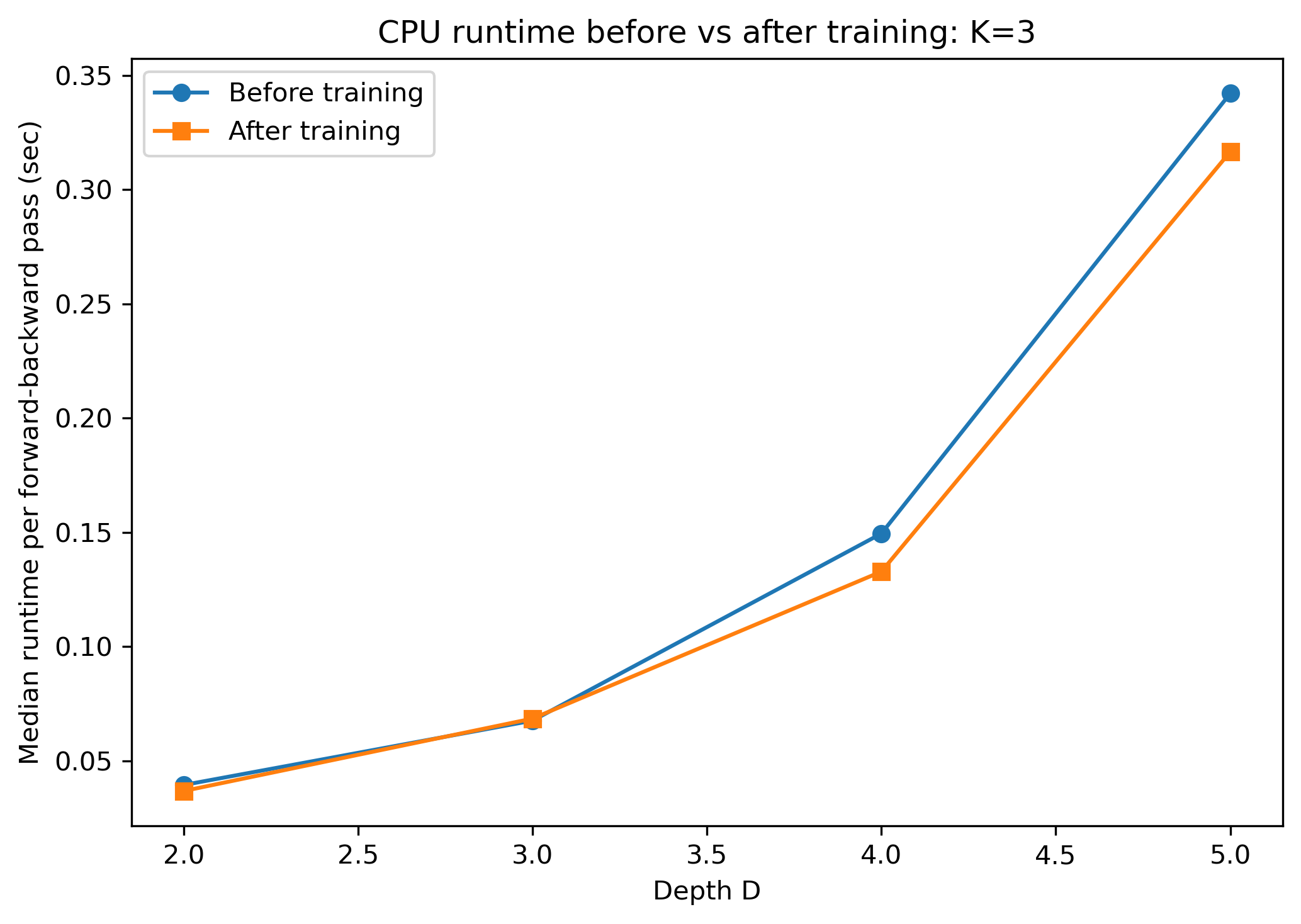}
    \caption{Median (over $5$ repetitions) CPU runtime per forward-backward pass against depth $D$, before and after training for $\eqnref{eq:feynman-cpe}:\;\Delta U = Gm_1m_2(\frac{1}{r_2}-\frac{1}{r_1})$ with $K=3$ and $\sigma^{2}=0.2^{2}$.}
    \label{fig:cpu_runtime}
\end{figure}

\begin{figure}[!htp]
    \centering
    \begin{subfigure}[t]{0.45\linewidth}
        \centering
        \includegraphics[width=\linewidth]{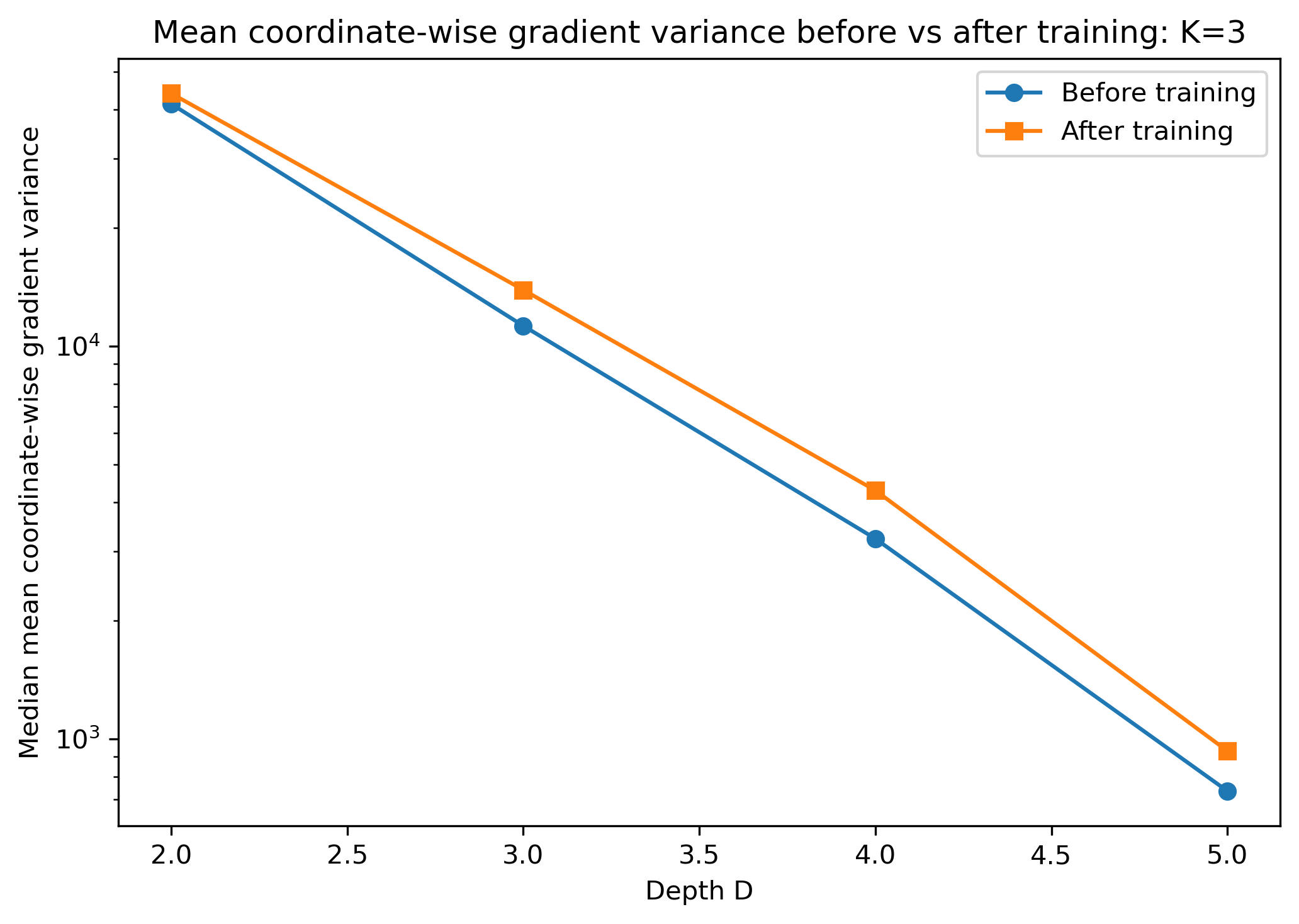}
        \caption{Median (over $5$ repetitions) of $\widehat{V}_{\mathrm{coord}}$ in~\eqnref{eq:var-summary}.}
        \label{fig:mean-coordinate-gradient-var}
    \end{subfigure}
    \hfill
    \begin{subfigure}[t]{0.47\linewidth}
        \centering
        \includegraphics[width=\linewidth]{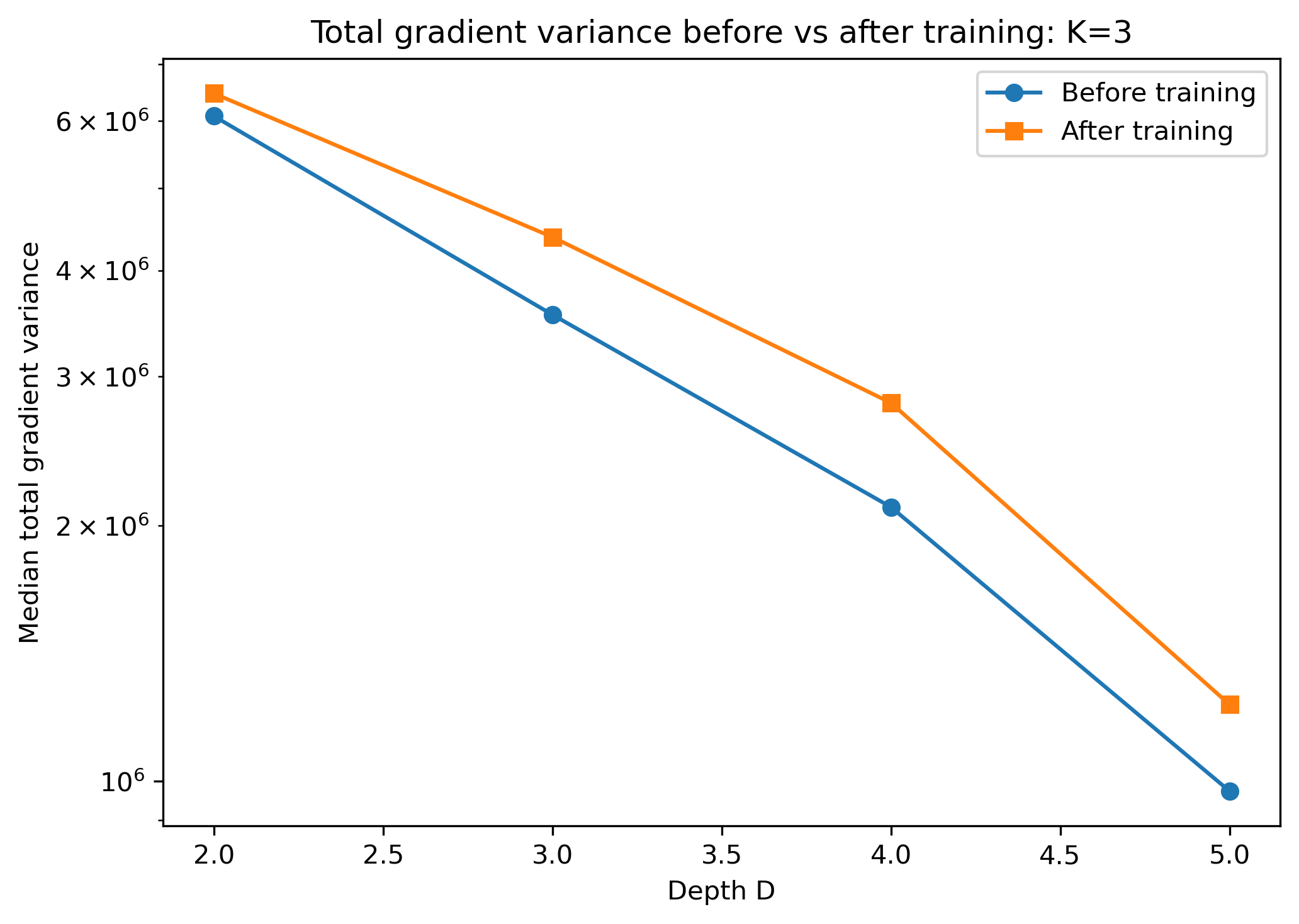}
        \caption{Median (over $5$ repetitions) of $\widehat{V}_{\mathrm{total}}$ in \eqnref{eq:var-summary}.}
        \label{fig:total-coordinate-gradient-var}
    \end{subfigure}
    \caption{Gradient variance diagnostic against depth $D$, before and after training for $\Delta U = Gm_1m_2(\frac{1}{r_2}-\frac{1}{r_1})$ with $K=3$ and $\sigma^{2}=0.2^{2}$.}
    \label{fig:gradvar_diagnostics}
\end{figure}

The runtime behavior in \hyperref[fig:cpu_runtime]{Figure~\ref{fig:cpu_runtime}} confirm the theoretical scaling as noted in \hyperref[tab:relaxed-logit-scaling]{Table~\ref{tab:relaxed-logit-scaling}}. The median CPU time per forward-backward pass increases monotonically with depth $D$. This growth is consistent with the expansion of the full binary tree skeleton $\mathsf{S}_j$ from $N=2^{2+1}-1 = 7$ to $N=2^{5+1}-1 = 63$ nodes per tree. The before-training and after-training curves are nearly identical, indicating that computational cost is governed primarily by the size of the relaxed tree skeleton rather than by the learned parameter values.

The mean coordinate-wise gradient variance ($\widehat{V}_{\mathrm{coord}}$) in \hyperref[fig:mean-coordinate-gradient-var]{Figure~\ref{fig:mean-coordinate-gradient-var}} provide a useful clarification. Increasing the depth $D$ does not lead to an explosion of the gradient variance. Instead, $\widehat{V}_{\mathrm{coord}}$ decreases as $D$ increases. This behavior is consistent with the depth-dependent split prior in \hyperref[subsec:symbolic-tree-prior]{Section~\ref{subsec:symbolic-tree-prior}}: deeper nodes have smaller probability of being active, therefore many deep expansion, operator, and feature logits receive weaker effective gradients and contribute smaller coordinate-wise fluctuations. In other words, although increasing $D$ adds multiple relaxed coordinates, many of these coordinates correspond to deeper nodes that are appropriately penalized for the learning of simple and interpretable scientific expressions~\citep{Occams-Razor-1}.
The total gradient variance ($\widehat{V}_{\mathrm{total}}$) in \hyperref[fig:total-coordinate-gradient-var]{Figure~\ref{fig:total-coordinate-gradient-var}} also decreases over the range of depths $D$. After training, $\widehat{V}_{\mathrm{total}}$ is slightly larger than before training, but the same decreasing trend with $D$ is preserved. Thus, for this Feynman equation~\eqnref{eq:feynman-cpe}, increasing depth $D$ mainly increases the deterministic computation and memory burden through the exponential growth of the full relaxed tree skeleton; it does not destabilize optimization through exploding stochastic gradient variance.

These diagnostics further justify the use of moderate depth values in the applications considered in~\hyperref[sec:vasst-in-action]{Section~\ref{sec:vasst-in-action}}. Even though larger $D$ expands the symbolic expression search space and allows for more complex structures, the memory and runtime scale as $\mathsf{O}(K2^{D})$. Moreover, the ablation studies performed in~\hyperref[app:ablation-study]{Appendix~\ref{app:ablation-study}} shows that $D=3$ already suffices to recover the relevant scientific structures for the Feynman equations viz., \eqnref{eq:feynman-cpe} and \eqnref{eq:feynman-ftc}. Therefore, an appropriate choice of $(K, D)$ (which is $(K=3, D=3)$ in our case) offers a principled accuracy-interpretability-efficiency trade-off, while larger depths can be used when the underlying scientific problem requires more compositional complexity.

\newpage

\section{Additional Experiments with Large Language Models}
\label{app:LLM-SR}

For comparing \vasstmain\ with more recent machine learning-centric \sr\ modules, we conduct additional experiments with \texttt{LLM-SR}~\citep{llmsr}, which casts \sr\ as a large language model (\texttt{LLM})-guided program synthesis task, iteratively generating and evaluating executable symbolic expressions from data. We consider the Feynman equation~\eqnref{eq:feynman-ada}: $f = \beta^{\dagger}(1 + \alpha^{\dagger}\cos\theta)$ from~\hyperref[subsec:feynman-data-study]{Section~\ref{subsec:feynman-data-study}} with $n=2000$ subsampled training observations under additive Gaussian noise level $\sigma^{2}=0.2^{2}$.

\paragraph{\texttt{LLM-SR} configurations and prompts.}
Two open-source language models were used as the \texttt{LLM} backend: \texttt{TinyLlama-1.1B-Chat}~\citep{zhang2024tinyllama} and \texttt{Qwen2.5-Coder-7B-Instruct}~\citep{hui2024qwen25coder}, configured as in~\hyperref[tab:llmsr-config-iii-17-37]{Table~\ref{tab:llmsr-config-iii-17-37}}. Candidate expressions are scored by negative mean squared error after fitting the free continuous parameters using \texttt{BFGS}~\citep{nocedal2006numerical}. We execute three prompting regimes of increasing informativness, as summarized below, ranging from weak syntactic guidance to strongly structured scientific prompting.

\begin{table}[!htp]
\centering
\scriptsize
\renewcommand{\arraystretch}{1.15}
\caption{\texttt{LLM-SR} configurations for learning~\eqnref{eq:feynman-ada} under $\sigma=0.2$.}
\label{tab:llmsr-config-iii-17-37}
\begin{tabularx}{\linewidth}{@{}>{\raggedright\arraybackslash}X c >{\raggedright\arraybackslash}X c c >{\raggedright\arraybackslash}X@{}}
\toprule
\toprule
\texttt{LLM} &
\textbf{Prompt} &
\textbf{Initialization} &
\texttt{MAX\_NPARAMS} &
\texttt{BFGS}  \textbf{maxiter} &
\textbf{Compute Resources} \\
\midrule

\texttt{TinyLlama-1.1B-Chat} &
Prompt 1 &
$c_0\beta^{\dagger}$ &
$3$ &
$100$ &
$1$ GPU, $16$ GB memory \\

\addlinespace[2pt]

\texttt{TinyLlama-1.1B-Chat}  &
Prompt 2 &
$c_0\beta^{\dagger}+c_1\alpha^{\dagger}$ &
$3$ &
$100$ &
$1$ GPU, $16$ GB memory \\

\addlinespace[2pt]

\texttt{TinyLlama-1.1B-Chat}  &
Prompt 3 &
$c_0\beta^{\dagger}(c_1+c_2\alpha^{\dagger}\cos\theta)$ &
$3$ &
$100$ &
$1$ GPU, $16$ GB memory \\

\addlinespace[2pt]

\texttt{Qwen2.5-Coder-7B-Instruct} &
Prompt 1 &
$c_0\beta^{\dagger}$ &
$3$ &
$100$ &
$1$ GPU, $64$ GB memory \\

\addlinespace[2pt]

\texttt{Qwen2.5-Coder-7B-Instruct} &
Prompt 2 &
$c_0\beta^{\dagger}+c_1\alpha^{\dagger}$ &
$3$ &
$100$ &
$1$ GPU, $64$ GB memory \\

\addlinespace[2pt]

\texttt{Qwen2.5-Coder-7B-Instruct} &
Prompt 3 &
$c_0\beta^{\dagger}(c_1+c_2\alpha^{\dagger}\cos\theta)$ &
$3$ &
$100$ &
$1$ GPU, $64$ GB memory \\

\bottomrule
\bottomrule
\end{tabularx}
{
\scriptsize
{
\emph{Note}: $c_0$, $c_1$, and $c_2$ are free continuous parameters to be learned. GPU partition uses NVIDIA A30 GPU.
}
}
\end{table}

\begin{tcolorbox}[
  llmsrpromptbox,
  title={\texttt{LLM-SR} prompt design for \texttt{TinyLlama-1.1B-Chat} while learning~\eqnref{eq:feynman-ada} under $\sigma=0.2$.}
]
\textbf{Prompt 1: Weak syntactic guidance.}
\begin{itemize}
  \item Search target: a simple vectorized \texttt{NumPy} expression mapping
  $(\beta^{\dagger},\alpha^{\dagger},\theta)$ to $f$.
  \item Allowed objects: $\beta^{\dagger},\alpha^{\dagger},\theta$, number of free constants are $3$,
  and elementary trigonometric operations are $\sin$ and $\cos$.
  \item Constraints: exactly two executable lines, no auxiliary variables, reassignment,
  or invalid output shapes.
  \item Initialization: $c_0\beta^{\dagger}$.
\end{itemize}

\textbf{Prompt 2: Scientist-guided moderate prompt.}
\begin{itemize}
  \item Adds domain context: the response is an angular distribution with possible
  forward--backward asymmetry.
  \item Describes the roles of the variables: $\beta^{\dagger}$ acts as an overall scale,
  $\alpha^{\dagger}$ controls asymmetry, and $\theta$ is an angular variable.
  \item Encourages compact expressions with low complexity and cosine-type angular dependence.
  \item Initialization: $c_0\beta^{\dagger}+c_1\alpha^{\dagger}$.
\end{itemize}

\textbf{Prompt 3: Strongly structured scientific prompt.}
\begin{itemize}
  \item States the additive Gaussian noise model and gives the expected scientific form: $\text{response}
    =
    \text{scale}\times
    (\text{baseline}
    +\text{asymmetry}\times\text{angular factor})$.
  \item Specifies that $\beta^{\dagger}$ is the scale, $\alpha^{\dagger}$ controls asymmetry, and the
  angular factor $\theta$ is cosine-like.
  \item Initialization: $c_0\beta^{\dagger}(c_1+c_2\alpha^{\dagger}\cos\theta)$.
\end{itemize}
\end{tcolorbox}

\begin{tcolorbox}[
  llmsrpromptbox,
  title={\texttt{LLM-SR} prompt design for \texttt{Qwen2.5-Coder-7B-Instruct} while learning~\eqnref{eq:feynman-ada} under $\sigma=0.2$.}
]
\textbf{Prompt 1: Weak / fair prompt.}
\begin{itemize}
  \item Search target: a compact vectorized NumPy expression learned from data.
  \item Guidance level: minimal scientific information; only syntactic and execution
  constraints are imposed.
  \item Constraints: two executable lines, no reassignment, valid output shape,
  and restricted elementary operations.
  \item Initialization: $c_0\beta^{\dagger}$.
\end{itemize}

\textbf{Prompt 2: Scientist-guided prompt.}
\begin{itemize}
  \item Provides the scientific target structure $\beta^{\dagger}(1 + \alpha^{\dagger}\cos \theta)$.
  \item Explains variable roles: $\beta^{\dagger}$ is a scale parameter, $\alpha^{\dagger}$ controls
  asymmetry, and $\theta$ is the angular coordinate.
  \item Encourages recovery of the interaction term $\beta^{\dagger}\alpha^{\dagger}\cos\theta$,
  while discouraging unnecessary noise-fitting terms.
  \item Initialization: $c_0\beta^{\dagger}+c_1\alpha^{\dagger}$.
\end{itemize}

\vspace{0.25em}
\textbf{Prompt 3: Strongly structured scientific prompt.}
\begin{itemize}
  \item Fully specifies the expected multiplicative scientific structure: $\text{scale}\times
    (\text{baseline}+\text{angular modulation})$.
  \item Introduces a strong inductive bias toward a cosine-based interaction between
  $\alpha$ and $\theta$.
  \item Initialization: $c_0\beta^{\dagger}(c_1+c_2\alpha^{\dagger}\cos\theta)$.
\end{itemize}
\end{tcolorbox}

\paragraph{Results.}
The results below show that \texttt{LLM-SR} is highly sensitive to prompt specification and initialization across both \texttt{TinyLlama-1.1B-Chat} and \texttt{Qwen2.5-Coder-7B-Instruct}. Under weak or moderately informative prompts, the generated candidates are often incomplete or non-executable, typically consisting only of docstrings or function signatures, and are therefore rejected by the evaluator. Even with strongly structured prompts that describe the underlying scientific mechanism, the generated expressions may only partially recover the interaction
\(\beta^{\dagger}(\cdot+\alpha^{\dagger}\cos\theta)\), while introducing extraneous transformations or violating execution constraints. This leads to degraded predictive performance relative to the noise level \(\sigma=0.2\). These findings suggest that \texttt{LLM}-based \sr\ methods are powerful and complementary, but their practical reliability can depend strongly on prompt design, initialization, model hosting, and filtering of invalid generated expressions.

In contrast, \vasstmain\ provides a computationally efficient probabilistic framework for symbolic structure learning. It directly optimizes a differentiable soft symbolic tree representation and quantifies uncertainty over competing symbolic forms, without relying on prompt engineering or large-scale generative sampling. As reported in~\hyperref[app:additional-feynman-results]{Appendix~\ref{app:additional-feynman-results}}, \vasstmain\ exactly recovers the symbolic structure of~\eqnref{eq:feynman-ada} across the considered noise levels, and its out-of-sample \texttt{RMSE}s on a \(10\%\) held-out test set closely track the corresponding injected noise levels.

\begin{tcolorbox}[
  llmsrpromptbox,
  title={\texttt{LLM-SR} results for \texttt{TinyLlama-1.1B-Chat} while learning~\eqnref{eq:feynman-ada} under \(\sigma=0.2\).}
]
\textbf{Prompt 1: Weak syntactic guidance.}
\begin{itemize}
  \item Generated output: only a function signature and descriptive text were produced; no executable symbolic body was returned.
  \item Interpretation: invalid/incomplete candidate.
  \item Score: \(-10^{20}\).
\end{itemize}

\textbf{Prompt 2: Scientist-guided moderate prompt.}
\begin{itemize}
  \item Generated output: only a function signature and descriptive text were produced; no executable symbolic body was returned.
  \item Interpretation: invalid/incomplete candidate.
  \item Score: \(-10^{20}\).
\end{itemize}

\textbf{Prompt 3: Strongly structured scientific prompt.}
\begin{itemize}
  \item Generated output: the evolved candidate was equivalent to $\widehat f =1+c_0\beta^{\dagger}(c_1+\alpha^{\dagger}\cos\theta)$.
  \item Interpretation: the candidate partially captures the angular interaction
  \(\beta^{\dagger}(\cdot+\alpha^{\dagger}\cos\theta)\), but introduces an extraneous additive shift and violates the required execution format.
  \item Out-of-sample \texttt{RMSE} computed on a $10\%$ held-out test set: \(0.3748\).
\end{itemize}

\textbf{Summary.}
\texttt{TinyLlama-1.1B-Chat} is highly sensitive to prompt specification. Prompts 1-2 fail to produce executable symbolic expressions, while Prompt 3 recovers part of the correct angular structure but adds an incorrect transformation. The initialized Prompt 3 structure $c_0\beta^{\dagger}(c_1+c_2\alpha^{\dagger}\cos\theta)$
achieves out-of-sample \texttt{RMSE} \(\approx 0.2012\) after parameter fitting, close to the injected noise level.
\end{tcolorbox}

\begin{tcolorbox}[
  llmsrpromptbox,
  title={\texttt{LLM-SR} results for \texttt{Qwen2.5-Coder-7B-Instruct} while learning~\eqnref{eq:feynman-ada} under \(\sigma=0.2\).}
]
\textbf{Prompt 1: Weak / fair prompt.}
\begin{itemize}
  \item Generated output: only a function signature and descriptive text were produced; no executable symbolic body was returned.
  \item Interpretation: invalid/incomplete candidate.
  \item Score: \(-10^{20}\).
\end{itemize}

\textbf{Prompt 2: Scientist-guided prompt.}
\begin{itemize}
  \item Generated output: only a function signature and descriptive text were produced, despite the additional scientific guidance.
  \item Interpretation: invalid/incomplete candidate.
  \item Score: \(-10^{20}\).
\end{itemize}

\textbf{Prompt 3: Strongly structured scientific prompt.}
\begin{itemize}
  \item Generated output: the model produced a descriptive docstring with an extended explanation and an additional ``improved version'' section, but no valid executable symbolic body.
  \item Interpretation: invalid/incomplete candidate under the strict execution constraints.
  \item Score: \(-10^{20}\).
\end{itemize}

\textbf{Summary.}
Across all prompting regimes, \texttt{Qwen2.5-Coder-7B-Instruct} frequently generates descriptive rather than executable symbolic programs under the strict \texttt{LLM-SR} formatting constraints. Nevertheless, the structured initialized form $c_0\beta^{\dagger}(c_1+c_2\alpha^{\dagger}\cos\theta)$
achieves out-of-sample \texttt{RMSE} \(\approx 0.2012\) after parameter fitting, indicating that the target expression is representable within the search space but is not reliably recovered through the generated samples.
\end{tcolorbox}


\end{document}